\documentclass[conference]{IEEEtran}
\usepackage[undo-recent-deprecations]{expl3}

\ifCLASSINFOpdf
\else
\fi
\usepackage{ifthen}
\usepackage{xspace}

\newboolean{talk}
\setboolean{talk}{false}
\newboolean{paper}
\setboolean{paper}{false}

\newboolean{IEEEstyle}
\setboolean{IEEEstyle}{false}
\newboolean{lipicsstyle}
\setboolean{lipicsstyle}{false}
\newboolean{eptcsstyle}
\setboolean{eptcsstyle}{false}
\newboolean{entcsstyle}
\setboolean{entcsstyle}{false}
\newboolean{sigplanstyle}
\setboolean{sigplanstyle}{false}
\newboolean{easychairstyle}
\setboolean{easychairstyle}{false}
\newboolean{scrartclstyle}
\setboolean{scrartclstyle}{false}
\newboolean{lncsstyle}
\setboolean{lncsstyle}{false}

\newboolean{acmmart}
\setboolean{acmmart}{false}

\newboolean{needstheorems}
\setboolean{needstheorems}{false}
\newboolean{withimages}
\setboolean{withimages}{false}
\newboolean{withproofs}
\setboolean{withproofs}{true}

\newboolean{french}
\setboolean{french}{false}

\setboolean{IEEEstyle}{true}
\newboolean{techreport}
\setboolean{techreport}{true}
\usepackage{booktabs}   
\usepackage{subcaption} 
\usepackage[all]{xy}

\usepackage{ebproof}
\ebproofset{label separation=0.3em}

\usepackage{booktabs}   
\usepackage{subcaption} 
\usepackage{graphicx}
\usepackage[normalem]{ulem}
\usepackage{multirow}

\usepackage[utf8]{inputenc}

\usepackage{multirow}
\usepackage{cmll}
\usepackage{listings}

\usepackage{amsmath, amsfonts, 
	amsthm, mathtools}
\usepackage{stmaryrd}
\usepackage{mathtools}
\usepackage{enumerate}
\usepackage{listings}

\usepackage{fancybox}
\usepackage{hhline}
\usepackage{marginnote}
\usepackage{soul}
\usepackage{xcolor}

\usepackage[hidelinks]{hyperref}
\usepackage{cleveref}
\usepackage{complexity}

\newcommand{\ignore}[1]{}

\newcommand{\colspace}{@{\hspace{.5cm}}}
\newcommand{\myinput}[1]{\ifthenelse{\boolean{withimages}}{\input{#1}}{}}

\newcommand{\reflemma}[1]{Lemma~\ref{l:#1}}
\newcommand{\reflemmap}[2]{Lemma~\ref{l:#1}.\ref{p:#1-#2}}

\newcommand{\reflemmaeq}[1]{{L.\ref{l:#1}}}
\newcommand{\reflemmaeqp}[2]{{L.\ref{l:#1}.\ref{p:#1-#2}}}
\newcommand{\refpoint}[1]{Point~\ref{p:#1}}

\newcommand{\refpointeq}[1]{P.\ref{p:#1}}

\newcommand{\refcorollary}[1]{Corollary~\ref{c:#1}}

\newcommand{\refth}[1]{Theorem~\ref{th:#1}}
\newcommand{\refthmeq}[1]{{T.\ref{thm:#1}}}
\newcommand{\refthp}[2]{Theorem~\ref{th:#1}.\ref{p:#1-#2}}
\newcommand{\refthm}[1]{Thm.~\ref{thm:#1}}

\newcommand{\refprop}[1]{Prop.~\ref{prop:#1}}

\newcommand{\refpropp}[2]{Prop.~\ref{prop:#1}.\ref{p:#1-#2}} 
\newcommand{\refthpboth}[2]{\refthp{#1}{#2} (proved in Thm.~\ref{thmappendix:#2})}
\newcommand{\refthmboth}[1]{\refthm{#1} (proved in Thm.~\ref{thmappendix:#1})}
\newcommand{\refthboth}[1]{\refth{#1} (proved in Thm.~\ref{thmappendix:#1})}
\newcommand{\reflemmaboth}[1]{\reflemma{#1} (proved in Lemma~\ref{lappendix:#1})}
\newcommand{\reflemmaappendix}[1]{Lemma~\ref{lappendix:#1}}
\newcommand{\reflemmaappendixp}[2]{Lemma~\ref{lappendix:#1}.\ref{p:#1-#2}}
\newcommand{\refcorollaryboth}[1]{\refcorollary{#1} (proved in Corollary~\ref{cappendix:#1})}
\newcommand{\refsect}[1]{Sect.~\ref{sect:#1}}
\newcommand{\refSECT}[1]{Sect.~\ref{SECT:#1}}

\newcommand{\refapp}[1]{Appendix~\ref{app:#1} (p.~\pageref{app:#1})}

\newcommand{\reffig}[1]{Fig.~\ref{fig:#1}}

\newcommand{\refex}[1]{Ex.~\ref{ex:#1}}

\newcommand{\ie}{\textit{i.e.}\xspace}
\newcommand{\eg}{\textit{e.g.}\xspace}
\newcommand{\ih}{\textit{i.h.}\xspace}

\newcommand{\lat}{\mbox{term}\xspace}

\newcommand{\ES}{\text{ES}\xspace}

\newcommand{\full}{\text{full}\xspace}

\renewcommand{\full}{\text{strong}\xspace}

\newcommand{\blue}[1]{{\color{blue} {#1}}}

\newcommand{\darkgreen}[1]{{\color{green!50!black} {#1}}}

\newcommand{\defeq}{\coloneqq} 
\newcommand{\eqdef}{\eqqcolon} 
\newcommand{\grameq}{\Coloneqq} 
\newcommand{\set}[1]{\{#1\}}

\newcommand{\nat}{\mathbb{N}}
\newcommand{\size}[1]{|#1|}

\newcommand{\gc}{{\sf gc}}

\newcommand{\admsym}{{\mathsf{sea}}}

\renewcommand{\l}{\lambda}
\newcommand{\isub}[2]{\{#1/#2\}}
\newcommand{\replace}[2]{#1{\shortleftarrow}#2}
\renewcommand{\isub}[2]{\{\replace{#1}{#2}\}}
\newcommand{\esub}[2]{[\replace{#1}{#2}]}
\renewcommand{\esub}[2]{[#1{\shortleftarrow}#2]}

\newcommand{\fv}[1]{{\sf fv}(#1)}

\newcommand{\bv}[1]{{\sf bv}(#1)}

\newcommand{\nf}{\mathsf{nf}} 
\newcommand{\nfo}[1]{\nf_\osym(#1)} 

\newcommand{\rootRew}[1]{\mapsto_{#1}}
\newcommand{\rootlRew}[1]{\; \mbox{}_{#1}{\mapsfrom}\ }
\newcommand{\Rew}[1]{\rightarrow_{#1}}

\newcommand{\lRew}[1]{\; \mbox{}_{#1}{\leftarrow}\ }

\newcommand{\lto}{\lRew{}}

\newcommand{\rtom}{\rootRew{\msym}}
\newcommand{\rtoe}{\rootRew{\esym}}
\newcommand{\reversertoe}{\rootlRew{\esym}} 

\newcommand{\tob}{\Rew{\beta}}

\newcommand{\betav}{\betaabs} 

\newcommand{\abssym}{\lambda} 
\newcommand{\betaabs}{{\beta_{\!\abssym}}} 

\newcommand{\tobv}{\Rew{\betav}} 

\newcommand{\betain}{\beta_{\isym}} 

\newcommand{\esym}{{\mathsf e}}
\newcommand{\isym}{i}

\newcommand{\msym}{\mathsf{m}}
\newcommand{\fsym}{f}

\newcommand{\wsym}{{\mathsf{o}}} 
\newcommand{\wmsym}{{\wsym\msym}} 
\newcommand{\wesym}{{\wsym\esym}} 
\newcommand{\omsym}{{\wmsym}} 
\newcommand{\oesym}{{\wesym}} 
\newcommand{\smsym}{\esssym{\msym}} 
\newcommand{\sesym}{\esssym{\esym}} 
\newcommand{\vmsym}{\mathsf{shuf}} 
\newcommand{\shuf}{\vmsym} 

\newcommand{\shufeqext}{\shufeqext} 

 \newcommand{\tom}{\Rew{\msym}}
 \newcommand{\toe}{\Rew{\esym}}

\newcommand{\tovsubo}{\Rew{\wsym}}

\newcommand{\tomo}{\Rew{\omsym}}

\newcommand{\toeo}{\Rew{\wsym{\esym}}}
\newcommand{\reversetoeo}{\lRew{\oesym}}

\newcommand{\tovsubs}{\Rew{\esssym}}

\newcommand{\toms}{\Rew{\esssym{\msym}}}
\newcommand{\reversetoms}{\lRew{\smsym}}
\newcommand{\toes}{\Rew{\esssym{\esym}}}
\newcommand{\reversetoes}{\lRew{\sesym}}

\newcommand{\toevals}{\Rew{\esssym \evalsym}}

\newcommand{\eqstruct}{\equiv}
\newcommand{\tostruct}{\eqstruct}

\newcommand{\aplsym}{@\textup{l}}
\newcommand{\aprsym}{@\textup{r}}
\newcommand{\essym}{[\cdot]}
\newcommand{\comsym}{\textup{com}}
\newcommand{\tostructapl}{\tostruct_{\aplsym}}
\newcommand{\tostructapr}{\tostruct_{\aprsym}}
\newcommand{\tostructes}{\tostruct_{\essym}}
\newcommand{\tostructcom}{\tostruct_{\comsym}}

\newcommand{\tm}{t}
\newcommand{\tmtwo}{u}
\newcommand{\tmthree}{r}
\newcommand{\tmfour}{q}
\newcommand{\tmfive}{p}
\newcommand{\tmsix}{s}

\newcommand{\tmp}{\tm'}
\newcommand{\tmtwop}{\tmtwo'}

\newcommand{\tmfourp}{\tmfour'}
\newcommand{\tmfivep}{\tmfive'}

\newcommand{\tmfourpp}{\tmfour''}
\newcommand{\tmfivepp}{\tmfive''}

\newcommand{\var}{x}
\newcommand{\vartwo}{y}
\newcommand{\varthree}{z}
\newcommand{\varfour}{w}

\newcommand{\val}{v}
\newcommand{\valtwo}{\val'}

\newcommand{\mol}{b}
\newcommand{\moltwo}{b'}
\newcommand{\molthree}{b''}

\newcommand{\ctxholep}[1]{\langle #1\rangle}
\newcommand{\ctxhole}{\ctxholep{\cdot}}

\newcommand{\ctx}{C}
\newcommand{\ctxtwo}{\ctx'}

\newcommand{\ctxp}[1]{\ctx\ctxholep{#1}}

\newcommand{\sctx}{L}
\newcommand{\sctxtwo}{\sctx'}

\newcommand{\sctxp}[1]{\sctx\ctxholep{#1}}
\newcommand{\sctxtwop}[1]{\sctxtwo\!\ctxholep{#1}}

\newcommand{\arbctxp}[1]{\arbctxp{#1}}
\newcommand{\arbctxtwop}[1]{\arbctxtwop{#1}}

\newcommand{\evsctx}{S}
\newcommand{\evsctxtwo}{\evsctx'}

\newcommand{\evsctxp}[1]{\evsctx\ctxholep{#1}}
\newcommand{\evsctxtwop}[1]{\evsctxtwo\ctxholep{#1}}

\newcommand{\tomachhole}[1]{\leadsto_{#1}}
\newcommand{\tomach}{\tomachhole{}}

\newcommand{\tomachbv}{\tomachhole{\betaabs}}
\newcommand{\tomachbi}{\tomachhole{\betain}}
\newcommand{\rensym}{\mathsf{ren}}
\newcommand{\tomachsub}{\tomachhole{\rensym}}
\newcommand{\tomachgc}{\tomachhole{\gc}}

\newcommand{\tomacho}{\tomachhole{\osym}}
\newcommand{\tomachb}{\tomachhole{\beta}} 

\newcommand{\tomachc}{\tomachhole{\admsym}}

\newcommand{\tomachcone}{\tomachhole{\admsym_1}}
\newcommand{\tomachctwo}{\tomachhole{\admsym_2}}
\newcommand{\tomachcthree}{\tomachhole{\admsym_3}}
\newcommand{\tomachcfour}{\tomachhole{\admsym_4}}
\newcommand{\tomachcfive}{\tomachhole{\admsym_5}}

\newcommand{\env}{e}
\newcommand{\envtwo}{\env'}
\newcommand{\envthree}{\env''}
\newcommand{\envfour}{\env'''}

\newcommand{\emptyenv}{\epsilon}

\newcommand{\cons}{:}

\ifthenelse{\boolean{acmmart}
}{
  \renewcommand{\state}{s}
  }{
  \newcommand{\state}{s}
  }
\newcommand{\statetwo}{s'}

\newcommand{\rename}[1]{\renamenop{#1}}
\newcommand{\renamenop}[1]{#1^\alpha}

\newcommand{\exec}{\rho}

\newcommand{\exectwo}{\sigma}

\newcommand{\decode}[1]{\llbracket #1\rrbracket}
\renewcommand{\decode}[1]{\underline{#1}}

\newcommand{\good}{\circ}

\newcommand{\deriv}{d}
\newcommand{\derivtwo}{e}

\newcommand{\sizehole}[2]{|#2|_{#1}}

\newcommand{\sizee}[1]{\sizehole{\esym}{#1}}

\newcommand{\sizem}[1]{\sizehole{\msym}{#1}} 

\ifthenelse{\boolean{IEEEstyle}}{
\usepackage{amssymb}
\usepackage{amsthm}
    \newtheorem{theorem}{Theorem}[section]
    \newtheorem{lemma}[theorem]{Lemma}
    \newtheorem{corollary}[theorem]{Corollary}
    \newtheorem{proposition}[theorem]{Proposition}
    \newtheorem{definition}[theorem]{Definition}
}{}

\ifthenelse{\boolean{lncsstyle}}{}{}

\newcommand{\mach}{{\sf M}}

\newcommand{\itm}{i}
\newcommand{\itmtwo}{\itm'}

\newcommand{\fire}{f}
\newcommand{\firetwo}{\fire'}

\newcommand{\sfire}{\fire_\fullsym}

\newcommand{\vsub}{\mathsf{vsc}} 

\newcommand{\tovsub}{\Rew{\vsub}}

\newcommand{\mctx}{{\mathbb C}}
\newcommand{\mctxtwo}{\mctx'}
\newcommand{\mctxthree}{\mctx''}
\newcommand{\mctxp}[1]{\mctx\ctxholep{#1}}
\newcommand{\mctxtwop}[1]{\mctxtwo\ctxholep{#1}}
\newcommand{\mctxthreep}[1]{\mctxthree\ctxholep{#1}}

\newcommand{\unfsym}{\rotatebox[origin=c]{-90}{$\rightarrow$}}
\newcommand{\unf}[1]{#1\unfsym\,}

\newcommand{\csym}{\mathsf{sea}}

\newcommand{\osym}{{\mathsf o}}

\newcommand{\la}[1]{\lambda #1.}

\newcommand{\ictx}{R}

\newcommand{\rctx}{R}
\newcommand{\rctxtwo}{\rctx'}
\newcommand{\rctxthree}{\rctx''}
\newcommand{\rctxp}[1]{\rctx\ctxholep{#1}}

\newcommand{\rmctx}{{\mathbb R}}
\newcommand{\rmctxtwo}{\rmctx'}
\newcommand{\rmctxthree}{\rmctx''}
\newcommand{\rmctxp}[1]{\rmctx\ctxholep{#1}}
\newcommand{\rmctxtwop}[1]{\rmctxtwo\ctxholep{#1}}
\newcommand{\rmctxthreep}[1]{\rmctxthree\ctxholep{#1}}

\newcommand{\myproof}[1]{
\ifthenelse{\boolean{omitproofs}}{\begin{IEEEproof} Proof available but omitted for readability. \end{IEEEproof}}{#1}}

\newcommand{\valES}{\text{answer}\xspace}

\newcommand{\levy}{{L{\'e}vy}\xspace}
\newcommand{\gregoire}{Gr{\'{e}}goire\xspace}

\newcommand{\withproofs}[1]{\ifthenelse{\boolean{withproofs}}{#1}{}}

\newcommand{\withoutproofs}[1]{\ifthenelse{\boolean{withproofs}}{}{#1}}

\newcommand{\NoteProof}[1]{
	\marginnote{{\normalfont\scriptsize{Proof\,p.\,{\pageref{#1}}\,}}}}
\newcommand{\NoteState}[1]{
	\marginnote{{\normalfont\scriptsize{See p.\,{\pageref{#1}}}}}}
\renewcommand{\NoteProof}[1]{\marginnote{{Proof\,p.\,{\pageref{#1}}}}}
\renewcommand{\NoteState}[1]{\marginnote{{See\,p.\,{\pageref{#1}}\\\cref{#1}}}}
\newcommand{\NoteStateP}[2]{\marginnote{{See\,p.\,{\pageref{p:#1-#2}}\\\refthp{#1}{#2}}}}

\crefname{proposition}{Prop.}{Props.}
\crefname{theorem}{Thm.}{Thms.}
\crefname{lemma}{Lemma}{Lemmas}
\crefname{corollary}{Cor.}{Cors.}
\crefname{section}{Sect.}{Sects.}
\Crefname{section}{Section}{Sections}

\newcommand{\indsub}[1]{\sigma_{#1}}
\newcommand{\indenv}[1]{\sctx_{#1}}

\newcommand{\vsubterms}{\Lambda_\vsub}

\newcommand{\vsubcalc}{\lambda_\vsub}

\newcommand{\shufcalc}{\lambda_\shuf}

\newcommand{\doubt}[1]{}

\newcommand{\letexp}{\mathsf{let}}

\newcommand{\lambdamucalc}{\overline\lambda\mu\tilde{\mu}}

\newcommand{\Rule}{\mathsf{r}}

\newcounter{numberone}
\newcounter{numberoneroman}
\newcounter{numberonealph}

\newcommand{\cbn}{CbN\xspace}
\newcommand{\scbn}{Strong CbN\xspace}
\newcommand{\cbv}{CbV\xspace}
\newcommand{\cbneed}{CbNeed\xspace}
\newcommand{\ocbv}{Open \cbv}

\newcommand{\ccbn}{Closed \cbn}
\newcommand{\ccbv}{Closed \cbv}
\newcommand{\scbv}{Strong \cbv}

\newcommand{\compil}[1]{#1^\circ}
\newcommand{\sizebeta}[1]{\size{#1}_\beta}
\newcommand{\tostrat}{\rightarrow}
\newcommand{\tostratx}{\rightarrow_x}
\newcommand{\tomachine}{\tomachhole\mach}

\newcommand{\domain}[1]{\mathsf{dom}(#1)}

\newcommand{\mytr}[1]{\underline{#1}}
\newcommand{\auxtr}[1]{\overline{#1}}

\makeatletter
\newcommand\Copy[2]{
        \marginpar{\scriptsize \ \ \hyperlink{hl-appendix-#1}{Proof p.\,{\pageref*{appendix-#1}}}}
	\immediate\write\@auxout{\unexpanded{\global\long\@namedef{mytext@#1}{#2}
  }}%
	#2%
}

\newcommand\Paste[1]{%
        \hypertarget{hl-appendix-#1}{}\label{appendix-#1}
	\renewcommand{\inappendix}[1]{}
	\ifcsname mytext@#1\endcsname
	\@nameuse{mytext@#1}%
	\else
	``??''
	\fi
	\renewcommand{\inappendix}[1]{#1}
}
\makeatother

\newcommand{\inappendix}[1]{#1}

\newcommand{\weakctx}{O}
\newcommand{\weakctxtwo}{\weakctx'}
\newcommand{\weakctxp}[1]{\weakctx\ctxholep{#1}}
\newcommand{\weakctxtwop}[1]{\weakctxtwo\ctxholep{#1}}

\newcommand{\openctx}{O}
\newcommand{\openctxtwo}{\openctx'}
\newcommand{\openctxthree}{\openctx''}
\newcommand{\openctxp}[1]{\openctx\ctxholep{#1}}
\newcommand{\openctxtwop}[1]{\openctxtwo\ctxholep{#1}}
\newcommand{\openctxthreep}[1]{\openctxthree\ctxholep{#1}}

\newcommand{\subctx}{\sctx}

\newcommand{\strongctx}{\extctx}
\newcommand{\strongctxtwo}{\strongctx'}

\newcommand{\extctx}{S}

\newcommand{\extctxp}[1]{\extctx\ctxholep{#1}}

\newcommand{\strongmctx}{{\mathbb E}}
\newcommand{\strongmctxtwo}{\strongmctx'}
\newcommand{\strongmctxthree}{\strongmctx''}
\newcommand{\strongmctxp}[1]{\strongmctx\ctxholep{#1}}
\newcommand{\strongmctxtwop}[1]{\strongmctxtwo\!\ctxholep{#1}}
\newcommand{\strongmctxthreep}[1]{\strongmctxthree\!\ctxholep{#1}}

\newcommand{\subctxp}[1]{\subctx\ctxholep{#1}}

\newcommand{\Pointed}{Rigid\xspace}
\newcommand{\pointed}{rigid\xspace}
\newcommand{\ptm}{r}
\newcommand{\ptmtwo}{\ptm'}

\newcommand{\rtm}{r}
\newcommand{\rtmtwo}{\rtm'}
\newcommand{\rtmthree}{\rtm''}

\newcommand{\esssym}{\mathsf{x}}

\newcommand{\fullsym}{{\mathsf{f}}}
\renewcommand{\fullsym}{{\mathsf{s}}}
\newcommand{\sitm}{\itm_\fullsym}
\newcommand{\sitmtwo}{\itmtwo_\fullsym}

\newcommand{\sval}{\val_\fullsym}

\newcommand{\rlsep}{\blue{\triangleleft}}
\newcommand{\lrsep}{\darkgreen{\triangleright}}
\newcommand{\gensep}{\bowtie}

\newcommand{\kctx}{K}
\newcommand{\kctxtwo}{\kctx'}
\newcommand{\kctxthree}{\kctx''}
\newcommand{\kctxp}[1]{\kctx\ctxholep{#1}}
\newcommand{\kctxtwop}[1]{\kctxtwo\!\ctxholep{#1}}
\newcommand{\kctxthreep}[1]{\kctxthree\!\ctxholep{#1}}

\newcommand{\varstar}{\star}

\newcommand{\wstenv}[1]{\mathsf{w}_{#1}}

\newcommand{\tosmach}{\Rew{\mathsf M}}

\newcommand{\varnames}{V}
\newcommand{\crnames}{V_{\mathsf{cr}}}
\newcommand{\calcnames}{V_{\mathsf{calc}}}

\renewcommand{\good}{good\xspace}

\newcommand\Crumb\mytr
\newcommand\CrumbAux\auxtr

\newcommand{\evalsym}{\esym_{\val}}
\renewcommand{\evalsym}{\esym_{\l}}

\renewcommand{\frame}{F}
\newcommand{\frametwo}{\frame'}
\newcommand{\framethree}{\frame''}
\newcommand{\framep}[1]{\frame\ctxholep{#1}}
\newcommand{\frametwop}[1]{\frametwo\ctxholep{#1}}
\newcommand{\framethreep}[1]{\framethree\ctxholep{#1}}

\newcommand\SCAM{SCAM\xspace}
\newcommand\OCAM{OCAM\xspace}

\renewcommand\lat{$\lambda$-term}

\newcommand\AlphaEq{=_\alpha}

\renewcommand{\decode}[1]{\unf{#1}}
\newcommand{\allvars}[1]{{\sf vars}(#1)}

\newcommand\Disj{\mathrel\bot}

\makeatletter
\newcommand{\exder}{%
  \def\exderW[##1]{\triangleright_{##1}\ }%
  \def\exderWO{\triangleright\ }%
  \@ifnextchar[\exderW\exderWO%
  }
\makeatother

\newcommand{\dom}[1]{\mathsf{dom}(#1)}

\renewcommand{\wstenv}[1]{\genv_{#1}}
\renewcommand{\wstenv}[1]{\env_{#1}}
\newcommand{\framei}[1]{\frame_{#1}}

\newcommand{\tomachscam}{\tomachhole{\SCAM}}
\DeclarePairedDelimiter{\norm}{\lVert}{\rVert}
\newcommand\tmpMeasure[2][]{\norm{#2}_{#1}}

\newcommand{\tofs}{\Rew{\esssym{\fsym}}}

\newcommand{\bigo}{\mathcal{O}}

\newcommand{\Id}{{\mathsf{I}}}

\newcommand\mydots{\hbox to .6em{.\hss.}}

\usepackage[skip=0pt]{caption}

\renewcommand{\betav}{\beta_v}

\renewcommand{\tobv}{\Rew{\betav}}

\usepackage{bbm}
\newcommand\cvar{\mathbbm{x}}
\newcommand\cvartwo{\mathbbm{y}}
\newcommand\cvarthree{\mathbbm{z}}
\newcommand\cvarfour{\mathbbm{w}}
\newcommand\Placeholder{\textbf{?}}

\newtheorem{example}[theorem]{Example}

\renewcommand{\fullsym}{{\mathsf{s}}}
\renewcommand{\full}{\text{strong}\xspace}

\renewcommand{\extctx}{E}

\newcommand{\symfont}[1]{\mathsf{#1}}
\renewcommand{\tostratx}{\rightarrow_\symfont{s}}

\renewcommand{\tomachbv}{\tomachhole{\betav}}

\renewcommand{\tmthree}{p}
\renewcommand{\tmfive}{r}

\ifthenelse{\boolean{techreport}}{}{\renewcommand{\NoteProof}[1]{}}
\usepackage{textcomp}

\begin{document}
%
\title{Strong Call-by-Value is Reasonable, Implosively}

\author{\IEEEauthorblockN{Beniamino Accattoli}
\IEEEauthorblockA{Inria \& LIX, \'Ecole Polytechnique}
\and
\IEEEauthorblockN{Andrea Condoluci}
\IEEEauthorblockA{Tweag I/O}
\and
\IEEEauthorblockN{Claudio Sacerdoti Coen}
\IEEEauthorblockA{
University of Bologna
}
}



\IEEEoverridecommandlockouts
\IEEEpubid{\makebox[\columnwidth]{978-1-6654-4895-6/21/\$31.00~
\copyright2021 IEEE \hfill} \hspace{\columnsep}\makebox[\columnwidth]{ }}

\maketitle

\begin{abstract}
Whether the number of $\beta$-steps in the $\l$-calculus can be taken as a reasonable time cost model (that is, 
polynomially related to the one of Turing machines) is a delicate problem, which depends on the notion of evaluation 
strategy. Since the nineties, it is known that weak (that is, out of abstractions) call-by-value evaluation is a 
reasonable strategy while \levy's optimal parallel strategy, which is strong (that is, it reduces everywhere), is not. 
The strong case turned out to be subtler than the weak one. In 2014 Accattoli and Dal Lago have shown that strong 
call-by-name is reasonable, by introducing a new form of useful sharing and, later, an abstract machine with an 
overhead quadratic in the number of $\beta$-steps.

Here we show that also strong call-by-value evaluation is reasonable for time, via a new abstract machine realizing useful 
sharing and having a linear overhead. Moreover, our machine 
uses a new mix of sharing techniques, adding on top of useful sharing a form of implosive sharing, which on some terms 
brings an exponential speed-up. We give examples of families that the machine executes in time \emph{logarithmic} in the 
number of $\beta$-steps.

\end{abstract}


\section{Introduction}
\label{sect:intro}
In the last few years, the understanding of the time cost models of the $\l$-calculus has attracted considerable 
attention. The beauty of the $\l$-calculus is that it is an abstract formalism, distant from low-level 
implementation details, while still achieving the same expressive power as Turing machines: it suffices a single $\beta$-rule, 
based on a natural notion of substitution. This is however also its main drawback, as the substitution it is 
based upon is a non-atomic operation that may duplicate whole sub-programs. 
A natural question then is how to measure the time of 
programs expressed in the $\l$-calculus. Of course, one wants a \emph{reasonable} cost model in the sense 
of Slot and van Emde Boas \cite{DBLP:conf/stoc/SlotB84}, that is, preserving the notion of polynomial time complexity as 
defined on Turing machines. 

The candidate measure for time is the number of $\beta$-steps to normal form. At first sight, this approach does not 
seem to work. A first issue is that one has to be more precise because there 
are many different evaluation strategies and notions of normal form in the $\l$-calculus. There is however a 
second bigger issue, called \emph{size explosion}, that affects every evaluation strategy. Namely, there 
are families $\{\tm_n\}_{n\in\nat}$ of $\l$-terms such that $\tm_n$ produces in $n$ 
$\beta$-steps---independently of the strategy---a result $r_n$ of size exponential in $n$. Then the chosen 
measure of time---namely $n$---does not even account for the time to write down the result, whose size is exponential in 
$n$.

\paragraph*{How to Stop Worrying and Love the Bomb} The way out of this apparent \emph{cul-de-sac} is to turn to 
evaluation \emph{up to sharing}, where sharing is used to 
provide compact representations of results, avoiding the explosion. 
It then turns out that (for some natural strategies and notions of 
normal form) the number of $\beta$-steps is a reasonable time cost model. The 
point is subtle, let us be precise. 

The idea is to first fix a strategy $\tostratx$ (together with its notion of normal 
form) in the $\l$-calculus, which is kept as a specification, reference system. Then, to study $\tostratx$ via a 
refined $\l$-calculus with sharing---think of an abstract machine---showing that $\tostratx$ 
can be implemented with an overhead polynomial in the number of $\tostratx$ steps, producing as output a term with 
sharing. The exponential explosion is then moved to the process of \emph{unsharing} output terms. Luckily, unsharing 
can essentially always be avoided (unless one really needs to print the unshared output) as terms with sharing can be 
manipulated efficiently without having to unshare them, see Condoluci, Accattoli, and Sacerdoti Coen 
\cite{DBLP:conf/ppdp/CondoluciAC19}.

\paragraph*{Subterm Sharing and Closed Evaluation} Sharing is an overloaded word, indicating a number of very 
different techniques in the literature about decompositions of the $\l$-calculus. The most basic one can be deemed 
\emph{subterm sharing}---itself coming in a number of variants---that amounts to annotate terms with delayed 
substitutions, coming from $\beta$-steps that have been encountered during the evaluation process. Such annotations may 
take the form of $\letexp$-expressions, explicit substitutions, or environments in abstract machines. Subterm sharing 
is 
enough to show that the number of $\beta$-steps is a reasonable cost model for both weak call-by-name (shortened to 
\cbn) 
and weak call-by-value (\cbv) evaluation (\emph{weak} = out of abstractions) with closed terms. These two settings are here 
referred to as \ccbn and \ccbv. The latter models evaluation in \cbv functional programming languages such as OCaml, 
and the fact that it is reasonable has been the first result in the literature about reasonable strategies for the 
$\l$-calculus, due to Blelloch and Greiner \cite{DBLP:conf/fpca/BlellochG95}. Similar results have also been obtained by Sands, Gustavsson, and Moran \cite{DBLP:conf/birthday/SandsGM02} and Dal Lago 
and Martini \cite{DBLP:journals/corr/abs-1208-0515,DBLP:conf/fopara/LagoM09,DBLP:journals/tcs/LagoM08}.

\paragraph*{Useful Sharing and Strong (\cbn) Evaluation} Subterm sharing is not enough beyond the closed case, that is, 
when evaluation may take place under abstraction and terms may be open---what we refer to as the \emph{strong 
$\l$-calculus}. Namely, there are exploding families whose strong evaluation with subterm sharing takes 
exponential time, independently of the evaluation strategy. For some time, indeed, it has been an open question 
whether there are strong  strategies that can be implemented within a reasonable overhead. The community used to believe 
that it was not 
the case, because of Asperti and Mairson's result that the \levy's optimal (strong) strategy is not reasonable 
\cite{DBLP:journals/iandc/AspertiM01}.

The question was settled by Accattoli and Dal Lago, showing that \scbn is reasonable: the number of call-by-name 
leftmost(-outermost) evaluation $\beta$-steps, which is a strong strategy, is a reasonable time cost model
\cite{DBLP:journals/corr/AccattoliL16}. For 
proving their result, they introduce the new layer of \emph{useful sharing}, operating on top of subterm sharing, and 
show that this is mandatory. Useful sharing amounts to do minimal unsharing work, namely only when it contributes to 
create $\beta$-steps, while avoiding to unfold the sharing when it only makes the term grow in size.

In \cite{DBLP:journals/corr/AccattoliL16}, the authors prove a polynomial overhead without investigating the degree. 
Later on, Accattoli provided an abstract machine---the only reasonable strong machine in the 
literature---with quadratic overhead when implemented on random access machines (RAM) \cite{DBLP:conf/wollic/Accattoli16}. 

Knowing that leftmost evaluation (sometimes referred to as \emph{normal order}) is reasonable is theoretically 
valuable. Because it answers an important questions, but also because leftmost evaluation is a sort of 
canonical strategy for the strong $\l$-calculus. At the same time, however, it is not of much practical value because 
leftmost evaluation 
might be inefficient, and Strong \cbv or Strong Call-by-Need are preferred in practice. For instance, both are used in the implementation of Coq.

\subsection{Contributions of the Paper} 
We prove, for the first time, that also \scbv is reasonable for time. The calculus for \scbv that we adopt is Accattoli and Paolini's \emph{value substitution calculus} \cite{AccattoliPaolini12} (shortened to VSC), for which we consider an \emph{external strategy} playing the same role played by the leftmost strategy in Strong CbN (proved normalizing by Accattoli et al. in \cite{accattoli2021semantic}). 

The main contribution is a new 
abstract 
machine, the \emph{strong crumbling abstract 
machine} (\SCAM{}), that implements the external strategy of the VSC (\refth{machine-final}) and that via subterm and useful sharing does so within a \emph{bilinear} overhead (\refthm{bilinear-scam}), that is, linear in the number of $\beta$-steps and in the size of the initial term, when implemented on RAM, improving over 
Accattoli's quadratic bound. The \SCAM{} actually goes considerably further, adding a form of \emph{implosive sharing} (surveyed below) which on 
some terms brings an 
exponential speed-up, evaluating them in time \emph{logarithmic} (!) in the number of $\beta$-steps, as we show on an 
example (\refprop{implosive-family}). Since implosive sharing forbids to apply the usual proof technique for the correctness of abstract machines, we 
develop a new more flexible one, based on the notion of \emph{relaxed implementation}, in \refSECT{RELAXED}. Last, we provide a prototype 
implementation of our machine in OCaml.

\subsection{Motivations} 
First and foremost, our motivation is foundational. We want to contribute to the study of 
reasonable cost models, showing that \scbv is reasonable for time. We also strive to obtain the best bounds for implementing 
strong evaluations because, after decades of research, how to best implement (strong) $\beta$-reduction is still an open 
problem.

Another motivation comes from the theory of proof assistants, where strong evaluation plays a role. Typically, settings 
such as Coq or Agda use strong evaluation to implement the $\beta$-conversion test, used for type checking with 
dependent types. In particular, one of the abstract machines at work in Coq, due to \gregoire and Leroy 
\cite{DBLP:conf/icfp/GregoireL02}, relies on call-by-value.

\paragraph*{Bounds for $\beta$-Conversion} The \emph{pure algorithm} for testing $\beta$-conversion of two terms $\tm$ 
and $\tmtwo$ first reduces them to normal form and then tests the results for equality. In general, conversion is 
undecidable, because $\tm$ or $\tmtwo$ may diverge, but one may ask---when $\tm$ and $\tmtwo$ are normalizable---what is 
the complexity of checking conversion. Without sharing, the pure algorithm is clearly exponential. By combining our 
results with the linear time algorithm for equality up to sharing by  Condoluci, Accattoli, and Sacerdoti Coen 
\cite{DBLP:conf/ppdp/CondoluciAC19}, we obtain a pure call-by-value algorithm using sharing, and working in time linear 
in the number of (\cbv) $\beta$-steps and in the size of the initial terms. We are not aware of other similar bounds in 
the literature, nor of any algorithmic study of $\beta$-conversion.

Beware: even though this work provides foundations for the implementation of proof 
assistants, we do not aim at direct applications. This is because proof assistants do not usually implement 
conversion via the pure algorithm, as they rest on a number of heuristics to shortcut it, see Sacerdoti Coen~\cite{CSC:strategies}.

\subsection{Implosive Sharing and All That} 
Once subterm sharing is adopted, it is possible to also \emph{evaluate} inside shared 
subterms, thus sharing \emph{evaluations}, not just subterms. The consequence is that one $\beta$-step in the shared 
settings maps to potentially \emph{many} $\beta$-steps in the $\l$-calculus, creating in some cases a \emph{steps 
explosion}, or, dually, an implosion: $n$ $\beta$-steps in the $\l$-calculus may in some cases \emph{implode} up to $\log 
n$ steps in the refinement with sharing. 

The terminology \emph{implosive sharing} is ours but not the (previously nameless) concept: the literature  contains 
implosive evaluation strategies, such as Wadsworth's call-by-need \cite{Wad:SemPra:71} (shortened to \cbneed) or \levy's 
optimal reduction \cite{thesislevy}. 

Implosive sharing poses two technical issues. First, proving correctness, because one needs to relate a 
single $\beta$-step in the sharing setting with potentially many $\beta$-steps in the unshared one, which is always 
involved. The second challenge is complexity analyses, because implosive sharing at times breaks the so-called 
\emph{subterm invariant}: the key property that duplicated terms along the whole evaluation with sharing are subterms of the initial term, which is essential for complexity analyses, and it is used in all existing proofs that a strategy is reasonable.

\paragraph*{Mixing Implosive and Useful Sharing}
There is a degree of freedom in the design of useful sharing. Accattoli and Dal Lago use a non-implosive approach which 
is naturally suggested by the \cbn setting that they study. We adopt here an alternative implosive approach, 
naturally suggested by the \cbv setting. This and other design choices of our  \SCAM{}, such as garbage collection and the light form of compilation called \emph{crumbling}, 
are detailed in \refsect{design}. 

Because of implosive sharing, correctness is the most demanding theorem of the paper, 
for which we develop a new abstract approach, deemed \emph{relaxed implementation}, that we then apply concretely. They 
key idea is modeling the one-to-many phenomenon induced by implosive sharing via a \emph{parallel} strategy on the 
calculus. The complexity analysis of the \SCAM{}, instead, is a smooth adaptation of others in the literature, because 
the implosive sharing of the \SCAM{} is carefully designed as to not clash with the subterm invariant.

\paragraph*{Space} As most environment machines in the literature, the \SCAM{} uses space linearly in its time consumption, and it is then space inefficient. In contrast to the literature however, the \SCAM{} does implement garbage collection, and we strived to make our prototype implementation parsimonious in space.

We do not address the study of a reasonable space cost model---the existence of one for the $\l$-calculus is an open problem. 
There is a recent partial result by Forster, Kunze, and Roth 
\cite{DBLP:journals/pacmpl/ForsterKR20}, but their space cost model---namely, the 
size of the term---can  only measure
linear and super-linear space. In order to study relevant space complexity classes such as \L, one needs to be 
able to measure \emph{sub-linear} space, and thus the result in \cite{DBLP:journals/pacmpl/ForsterKR20} is not a solution for the general problem.


\subsection{Related Work}
\paragraph*{About Abstract Machines} The study of machines for strong evaluation is a blind spot of the 
field, despite the 
relevance for the implementation of proof assistants. There are very few strong machines in the 
literature. The ones by Cr{\'{e}}gut \cite{DBLP:journals/lisp/Cregut07,DBLP:journals/jfp/Garcia-PerezN19} (\cbn), Biernacka et al. 
\cite{DBLP:conf/aplas/BiernackaBCD20} (\cbv), and Biernacka and Charatonik \cite{DBLP:conf/rta/BiernackaC19} 
(\cbneed) all have exponential overhead. The last two works are based on Ager et al. functional correspondence \cite{DBLP:conf/ppdp/AgerBDM03} and Danvy and Nielsen (generalized) refocusing \cite{Danvy04refocusingin,DBLP:conf/rta/BiernackaCZ17}. After submitting our work, we became aware of an independent, concurrent, and currently unpublished work by Biernacka et al. also proving that \scbv is reasonable for time \cite{DBLP:journals/corr/abs-2102-05985}, and using a different approach. 
De Carvalho \cite{deCarvalho18} and  
Ehrhard and Regnier \cite{DBLP:conf/cie/EhrhardR06} study variants of Cr{\'{e}}gut's machine  for denotational purposes.

Coq uses more than one abstract machine for strong evaluation. The \cbv one due to
\gregoire and Leroy \cite{DBLP:conf/icfp/GregoireL02} is---perhaps surprisingly---not really a strong machine. It is 
obtained by iterating under abstractions a \cbv machine for weak evaluation with open terms---a setting sometimes called 
 \emph{\ocbv} \cite{DBLP:conf/aplas/AccattoliG16}. Their machine for \ocbv has exponential overhead (if implemented as defined in \cite{DBLP:conf/icfp/GregoireL02}), and the iteration is also na\"ive and costly 
(namely it unfolds sharing before iterating, thus potentially \emph{exploding} in size), adding a further exponential 
cost. Iterating an open machine is actually very subtle: Accattoli and Guerrieri in \cite{DBLP:journals/scp/AccattoliG19} show 
that, even without sharing unfolding, and even when the open machine is reasonable, iterating may not give a reasonable 
machine for \scbv, as the iteration may 
introduce an exponential blow up. Abstract machines for \ocbv have then been studied in-depth by Accattoli and 
co-authors in \cite{fireballs,DBLP:journals/scp/AccattoliG19,DBLP:conf/ppdp/AccattoliCGC19}, and optimized as to be reasonable 
and with linear overhead, but never extended to \scbv.

Coq also uses a strong machine performing \cbneed evaluation, designed and studied by Barras' in his PhD thesis~\cite{barras-phd}, and for which no complexity results 
are known---its correctness to our knowledge has never been fully proved.

\paragraph*{About Implosive Sharing} \cbneed evaluation---usually considered in the closed 
setting---is an implosive sharing refinement of Closed \cbn. Its correctness is notoriously technical, see for instance 
Maraist, Odersky, and Wadler \cite{DBLP:journals/jfp/MaraistOW98}, and Ariola and Felleisen 
\cite{DBLP:conf/popl/AriolaFMOW95,DBLP:journals/jfp/AriolaF97}. Kesner develops an elegant alternative technique 
resting on multi types \cite{DBLP:conf/fossacs/Kesner16}, that has been adapted to the strong case---becoming quite more 
technical---in \cite{DBLP:journals/pacmpl/BalabonskiBBK17}. 
The correctness of implementations of optimal reductions is extremely involved and sophisticated, see Asperti and 
Guerrini \cite{AG-OIFPL-98}. None of these works are presented using abstract machines. Call-by-need machines do 
exist, but their correctness is always proved relatively to a call-by-need calculus, as in \cite{DBLP:conf/ppdp/DanvyZ13}, with respect to which they are not 
implosive---the standard correctness technique indeed applies.\medskip

\paragraph*{Proofs} \ifthenelse{\boolean{techreport}}{Proofs are in the Appendix. With respect to the LICS 2021 proceedings version of the paper, this version also contains a few more technical details in the body of the paper starting from Section VII.}{Proofs are in the technical report on Arxiv \cite{DBLP:journals/corr/abs-2102-06928}.}

\section{The Value Substitution Calculus}
\label{sect:calculus}
Plotkin's call-by-value $\l$-calculus \cite{DBLP:journals/tcs/Plotkin75} is known to behave perfectly as long as terms are closed (that is, without free variables) and evaluation is weak---let us call such a setting \emph{Closed \cbv},  following Accattoli and Guerrieri \cite{DBLP:conf/aplas/AccattoliG16}. 

It is well known that as soon as one considers open term or strong evaluation then Plotkin's \cbv $\betav$-rule $(\la\var\tm)\val \tobv \tm\isub\var{\val}$ is no longer adequate with various semantical properties---as  first shown by Paolini and Ronchi della Rocca
\cite{DBLP:journals/ita/PaoliniR99,DBLP:conf/ictcs/Paolini01,parametricBook}---and the operational semantics has to be extended somehow. In \cite{DBLP:conf/aplas/AccattoliG16}, Accattoli and Guerrieri compare various ways of doing it, and show that in the open setting they are all equivalent\footnote{One of these calculi is (the \cbv and intuitionistic fragment of) Curien and Herbelin's
$\lambdamucalc$-calculus \cite{DBLP:conf/icfp/CurienH00}, which could be used to reformulate the results in this work, another one is  Guerrieri and Carraro's \emph{shuffling calculus} $\shufcalc$ \cite{DBLP:conf/fossacs/CarraroG14}, which instead could not, because its cost model is unclear, see \cite{DBLP:conf/aplas/AccattoliG16}.}. 

Here we adopt one of those calculi, Accattoli and Paolini's \emph{value substitution calculus}, shortened here to VSC \cite{AccattoliPaolini12}. It was first introduced to study a semantical property of \scbv, \emph{solvability}, and it is isomorphic to the \cbv representation of the $\l$-calculus into linear logic, as shown by Accattoli \cite{DBLP:journals/tcs/Accattoli15}.

It is also well known that values can be defined as variables and abstractions, or simply as abstractions---the resulting theories differ only for inessential details. Restricting values to abstractions is preferred by works on \cbv abstract machines---including this one---because it leads to better performances, as shown by Accattoli and Sacerdoti Coen \cite{DBLP:journals/iandc/AccattoliC17}.

\paragraph*{The Value Substitution Calculus} There are various ingredients in the VSC. First, the syntax of the $\l$-calculus is extended with $\letexp$-expressions, that we here prefer to more compactly write  as explicit substitutions $\tm\esub\var\tmtwo$ (shortened to ES), while we use $\tm\isub\var\tmtwo$ for  meta-level substitution.
\begin{center}
$\arraycolsep=3pt\begin{array}{rrl}
\textsc{VSC Values} & \val & \grameq \la\var\tm 
\\
\textsc{VSC Terms} & \tm,\tmtwo, \tmthree & \grameq \var \mid \la\var\tm \mid \tm\tmtwo 
\mid \tm \esub\var\tmtwo 
\end{array}$
\end{center}
There also is a crucial use of \emph{contexts} to specify the rewriting rules. Contexts are terms with a \emph{hole} $\ctxhole$ intuitively standing for a removed subterm. We shall see various notion of contexts. For now, we need unrestricted contexts $\ctx$ and the special case of substitution contexts $\sctx$ (standing for \emph{L}ist of substitutions).
\begin{center}
$\arraycolsep=3pt\begin{array}{rrl}
\textsc{Contexts} & \ctx & \grameq \ctxhole \mid \ctx \tm \mid \tm \ctx \mid \la{\var}{\ctx} \mid %
\ctx \esub\var\tm \mid \tm \esub \var \ctx \\
\textsc{Sub. Ctxs } &\subctx & \grameq \ctxhole \mid \subctx 
\esub\var\tm
\end{array}$
\end{center}
Replacing the hole of a context $\ctx$ with a term $\tm$ (or another context $\ctxtwo$) is called \emph{plugging} and noted $\ctxp\tm$ (resp. $\ctxp\ctxtwo$).

Given the use of explicit substitutions (shortened to ES), $\beta$-steps are decomposed in two, the introduction of the ES and the turning of an ES into a meta-level substitution. 
The rewrite rules work up to a substitution context $\subctx$, or, if you prefer, up to ES (also called 
\emph{at a distance}). 
\begin{center}
$\begin{array}{rrllll}
\multicolumn{4}{c}{\textsc{VSC Rules at top level}}\\
    \textsc{Multiplicative} & \subctxp{\la\var\tm}\tmtwo &  \rtom  & \subctxp{\tm\esub{\var}{\tmtwo}} \\
    \textsc{Exponential}  & \tm\esub\var{\subctxp{\val}} &  \rtoe  & \subctxp{\tm\isub{\var}{\val}} 
\end{array}$\smallskip
		
		\begin{tabular}{ccc}
\textsc{Contextual closure}:
			&
			\multirow{2}{*}{\begin{prooftree}
					\hypo{\tm \rootRew{a} \tm'}		
					\infer1{\ctxp{\tm} \Rew{a} \ctxp{\tm'}}
			\end{prooftree}}
		\\
		($a \in \set{\msym,\esym}$)
	\end{tabular}\smallskip

		$\begin{array}{cccc}
\textsc{Notation}:& \tovsub  \, {\defeq} \, \tom \cup \toe
	\end{array}$
\end{center}
Examples: $(\la\var\tm)\esub\vartwo\tmtwo \tmthree \tom \tm\esub\var\tmthree \esub\vartwo\tmtwo$ and $\tm\esub\var{\val\esub\vartwo\tmtwo}$ $\toe \tm\isub\var{\val}\esub\vartwo\tmtwo$.
The terminology comes from the connection with linear logic proof nets. Note that the \cbv restriction is not on multiplicative/$\beta$-redexes, but on exponential redexes.

Please note that the VSC can simulate Plotkin's $\betav$ rule, as $(\la\var\tm)\val \tom \tm\esub\var{\val} \toe\tm\isub\var{\val}$. Actually, it does more: in the VSC an open term such as $\tm\defeq (\la\var\delta)(\vartwo\vartwo) \delta$, where $\delta \defeq \la\var \var\var$ is the duplicator, diverges as follows
\begin{center}$
\begin{array}{rllll}
\scriptsize\tm &\tom &\delta\esub\var{\vartwo\vartwo}\delta &\tom 
\\
&&(\varthree\varthree)\esub\varthree\delta\esub\var{\vartwo\vartwo} &\toe (\delta\delta)\esub\var{\vartwo\vartwo} \tovsub \ldots
\end{array}
$\end{center}
while for Plotkin it is normal. 

A key property of the VSC is that while $\tovsub$ obviously does not terminate---being able to simulate Plotkin's $\betav$ rule---its two rules when taken separately are strongly normalizing.
\begin{lemma}[Local termination, \cite{AccattoliPaolini12}]
\label{l:vsc-local-termination}
The reductions $\tom$ and $\toe$ are strongly normalizing.
\end{lemma}
An \emph{evaluation} is a possibly empty sequence $\deriv: \tm \tovsub^* \tmtwo$ of $\tovsub$ steps, whose number of $\tom$ (resp. $\toe$) steps is noted $\sizem\deriv$ (resp. $\sizee\deriv$).

Next, we discuss the simple open fragment, as it allows to introduce some key concepts for the general case, and on top of which we shall define the parallel strategy to be implemented by the \SCAM{}.

\paragraph*{The Open VSC} The open fragment of the VSC is obtained by first defining open contexts---by removing the abstraction case---and then use them to define the open variant of the rewriting rules, which do not evaluate under abstraction.
\begin{center}
		$\begin{aligned}
		\textsc{Open ctxs} && \weakctx & \grameq \ctxhole \mid \weakctx \tm \mid \tm \weakctx \mid \weakctx \esub\var\tm \mid 
\tm\esub\var \weakctx 
		\end{aligned}$
		\smallskip
		
		\begin{tabular}{ccc}
\textsc{Open rewrite rules}:
			&
			\multirow{2}{*}{\begin{prooftree}
					\hypo{\tm \rootRew{a} \tm'}		
					\infer1{\weakctxp{\tm} \Rew{\wsym a} \weakctxp{\tm'}}
			\end{prooftree}}
		\\
		($a \in \set{\msym,\esym}$)
	\end{tabular}\medskip

		$\begin{array}{cccc}
\textsc{Open reduction}:& \tovsubo  \, \defeq \, \tomo \cup \toeo
	\end{array}$
\end{center}
Careful: the open fragment contains closed terms, because terms and contexts are \emph{potentially} (and not necessarily) open.

Note that the grammar of open contexts implies that evaluation is non-deterministic, as rewriting steps can take place on both sides of an application and on both subterms of ES.  For instance, for any step $\tm\tovsubo \tmtwo$, we have the following span $\tmtwo\tm \lRew{\wsym} \tm\tm  \tovsubo \tm\tmtwo$
that  closes on $\tmtwo\tmtwo$ with one $\tovsubo$ step on each side---the same happens with $\tmtwo\esub\var\tm\lRew{\wsym}\tm\esub\var\tm \tovsubo \tm\esub\var\tmtwo$.
 
Such a non-determinism is harmless because it is \emph{diamond}. A rewriting relation $\to$ is diamond if $\tmtwo_1 \lto \tm \to \tmtwo_2$ and $\tmtwo_1 \neq \tmtwo_2$ imply $\tmtwo_1 \to \tmthree \lto \tmtwo_2$ for some $\tmthree$ (it is the 1-step strengthening of confluence).
\begin{proposition}[\cite{AccattoliPaolini12}]
The reduction $\tovsubo$ is diamond.
\end{proposition}
There are two famous consequences of being diamond: \emph{uniform normalization}, that is, if there is a normalizing reduction sequence then there are no diverging sequences, and \emph{random descent}, that is, when a term is normalizable, all sequences to normal form have the same length. Essentially, the diamond is a relaxed form of determinism.

The normal forms of the open fragment have a nice inductive characterization, coming from the so-called \emph{fireball calculus} \cite{DBLP:conf/aplas/AccattoliG16}. \emph{Fireballs} are defined by mutual induction with \emph{inert terms}, and including values, as follows.
\[\begin{array}{r@{\hspace{.5cm}} rll}
  	\textsc{Inert terms} & \itm, \itmtwo & \grameq  &\var \mid \itm \fire \mid \itm \esub{\var}{\itmtwo}
	\\
 	\textsc{Fireballs} &  \fire, \firetwo  &\grameq & \val \mid \itm \mid \fire \esub{\var}{\itm}
\end{array}\]

For instance, $\la\vartwo((\la\var\vartwo)\vartwo)$ is a fireball as a value, while $\var$, $\vartwo(\la\var\var)$, $\var\vartwo$, and $(\varthree(\la\var\Omega))(\varthree\varthree)$ are fireballs as inert terms. 
\begin{proposition}[\cite{AccattoliPaolini12}]
Let $\tm$ be a VSC term. $\tm$ is $\tovsubo$ normal if and only if $\tm$ is a fireball.
\end{proposition}
\paragraph*{The Strong Calculus} Outside of the open fragment, evaluation is not necessarily diamond. For instance, for any step $\tm\tovsub \tmtwo$, the following span
$$(\la\vartwo\tm)(\la\vartwo\tm)\ \lRew{\esym}\ (\var\var)\esub\var{\la\vartwo\tm} \ \tovsub\ (\var\var)\esub\var{\la\vartwo\tmtwo}$$
closes on $(\la\vartwo\tmtwo)(\la\vartwo\tmtwo)$ but not with a diamond diagram. Anyway, the VSC is confluent.
\begin{proposition}[\cite{AccattoliPaolini12}]
The reduction $\tovsub$ is confluent.
\end{proposition}
Strong normal forms also have a nice characterization, iterating inside values the one for open normal forms.
\begin{center}
	$\begin{array}{rlc}
	\textsc{Strong inert terms }
	&
	\sitm \grameq \var \mid \sitm \sfire \mid \sitm \esub\var{\sitmtwo}
	\\
	\textsc{Strong values } 
	&
	\sval  \grameq \la\var\sfire
	\\
	\textsc{Strong fireballs } 
	&
	\sfire \grameq \sitm \mid \sval \mid \sfire \esub\var\sitm
	\end{array}$
\end{center}
For instance, $\la\vartwo(\vartwo(\la\var\vartwo))$ is a strong value, while $\la\vartwo((\la\var\vartwo)\vartwo)$ is not. Similarly, $\var (\la\varthree\varthree\varthree)$ is a strong inert term, while $\var (\la\varthree((\la\vartwo\vartwo)\varthree))$ is not. Note that strong fireballs are similar to the normal forms of the (\cbn) $\l$-calculus, except that they can have \ES containing strong inert terms. 
\begin{lemma}[Characterization of normal forms]
	\label{l:harmony} \label{p:harmony-vsub}\label{p:harmony-fireball}
	Let $\tm$ be a VSC term. 
	\NoteProof{lappendix:harmony}
 $\tm$ is $\tovsub$-normal if and only if $\tm$ is a strong fireball. 
\end{lemma}

\section{The External Strategy}
\label{sect:external-strategy}
Since the VSC is not diamond, we need to isolate an evaluation strategy, playing the role of the leftmost(-outermost) strategy in \cbn. Usually, the strategies implemented by abstract machines are deterministic. Here instead we adopt a \emph{diamond} strategy, that---as we explained in the previous section---can be seen as a form of relaxed determinism.

While \cbn has a clear left-to-right orientation, reflected by its leftmost evaluation strategy, in \cbv there is no such direction of evaluation. For instance, Plotkin standard evaluations are left-to-right
\cite{DBLP:journals/tcs/Plotkin75}, Leroy's ZINC abstract machine \cite{Leroy-ZINC} is right-to-left, and Dal Lago and Martini follow an unspecified non-deterministic order \cite{DBLP:journals/tcs/LagoM08}. Our strategy shall then be liberal, and not impose an order on the evaluation of applications.

On the other hand, we shall keep the \emph{outermost} aspect of the leftmost-outermost \cbn strategy. Our \emph{external} strategy, indeed, shall reduce only redexes that cannot be duplicated or erased by any other redex.

The definition of the strategy requires the auxiliary notion of rigid terms, which are the variation over inert terms where the arguments of the head variable can be whatever term.
\begin{center}
	$\textsc{\Pointed terms} \quad \ptm, \ptmtwo  \grameq  \var \mid \ptm\tm \mid \ptm\esub\var \ptmtwo$
\end{center}
Every (\full) inert term is a \pointed term, but the converse does not hold, consider for instance $\vartwo\,({\delta\delta})$.

\emph{External (evaluation) contexts} are defined by mutual induction with \emph{rigid contexts}. 
\begin{center}
		$\begin{array}{r\colspace r\colspace clccc}
		\multicolumn{4}{c}{\textsc{Term evaluation contexts}}
		\\
		\textsc{External} & \strongctx & \grameq &\ctxhole \mid \la\var \strongctx \mid \tm\esub\var \ictx \mid \strongctx 
\esub\var \ptm \mid \ictx 
		\\
		\textsc{Rigid} & \ictx & \grameq & \ptm \strongctx \mid \ictx \tm \mid \ictx\esub\var \ptm \mid \ptm \esub\var\ictx
		\end{array}$
\end{center}
Finally, the rewriting rules are obtained by closing the open rules with external contexts.
\begin{center}{
		\begin{tabular}{ccc}
\textsc{External rewrite rules:}
			&
			\multirow{2}{*}
			{\begin{prooftree}
					\hypo{\tm \Rew{\wsym a} \tm'}		
					\infer1{\extctxp\tm \Rew{\esssym a} \extctxp{\tm'}}
			\end{prooftree}}
		\\
		($a \in \set{\msym,\esym}$)
	\end{tabular}}
\end{center}

\begin{center}{
		$\begin{array}{cccccc}
\textsc{External reduction:} &
		\tovsubs \, \defeq \, \toms \cup \toes
	\end{array}$}
\end{center}
Key points:
\begin{itemize}
\item \emph{Normalizing}: the strategy 
normalizes the potentially diverging term $(\la\var\vartwo) (\la\varthree\Omega) \toms \vartwo \esub\var{\la\varthree\Omega} \toes \vartwo$, 
and diverges on $\vartwo (\la\varthree\Omega)$. In a companion paper about the semantics of \scbv by Accattoli, Guerrieri, and Leberle \cite{accattoli2021semantic}, it is proved that the external strategy is normalizing, \ie, it reaches a normal form whenever it exists in the VSC.
\item \emph{External}: external steps are not contained in any value that is applied or ready to be substituted. The grammars of external and rigid contexts indeed forbid these situations: given an open step $\tm\tovsubo \tmtwo$, note that $(\la\var\tm) \tmthree \not\tovsubs 
(\la\var\tmtwo) \tmthree$ and $(\var\var)\esub\var{\la\vartwo\tm} \not \tovsubs (\var\var)\esub\var{\la\vartwo\tmtwo}$. On the other hand, the external strategy does enter values that shall not be substituted, for instance $\vartwo\tmthree (\la\var\tm) \tovsubs \vartwo \tmthree(\la\var\tmtwo)$ and $\tmthree \esub\var{\vartwo (\la\var\tm)} \tovsubs \tmthree \esub\var{\vartwo (\la\var\tmtwo)}$.

\item \emph{Non-determinism}: since $\tovsubs$ contains the open rules, it 
is neither left-to-right nor right-to-left---we have both $(\Id\Id) (\Id\Id)\toms (\vartwo\esub\vartwo\Id) (\Id\Id)$ 
and $(\Id\Id) (\Id\Id)\toms (\Id\Id)(\vartwo\esub\vartwo\Id)$. Another example is given by $\tm = \var 
(\la\vartwo(\Id \Id)) \esub\var{\varfour 
(\Id\Id)} \toms \var (\la\vartwo \varthree\esub\varthree \Id) \esub\var{\varfour (\Id\Id)}$, and $\tm \toms \var 
(\la\vartwo(\Id \Id)) \esub\var{\varfour(\varthree\esub\varthree{\Id})}$. 
\end{itemize}
\begin{proposition}[Properties of $\tovsubs$]
	\label{prop:external-properties}
	\label{prop:vsc-diamond}
	\NoteProof{propappendix:vsc-diamond}
	\label{l:fullness}
	Let $\tm$ be a VSC term.
	\begin{enumerate}
		\item\label{p:external-properties-diamond} \emph{Diamond}:	$\tovsubs$ is diamond.	Moreover, every $\tovsubs$ evaluation to normal form (if any) has the same number of $\toms$ steps.
		\item\label{p:external-properties-fullness} \emph{Normal forms}: if $\tm$ is $\esssym$-normal then it is a \full fireball.
	\end{enumerate}
\end{proposition}
\paragraph*{Cost Model of the VSC} As time cost model of the VSC we take the number of $\tom$ steps of the external strategy.  At the end of the paper, we shall prove it reasonable. \emph{Subtlety}: the cost model makes sense despite the non-determinism of $\tovsubs$, because the diamond of $\tovsubs$ in particular preserves the kind of step, and so all evaluations to normal form have the same number of $\tom$ steps, as stated above.

\paragraph*{Structural Equivalence}
The VSC comes with a notion of structural equivalence $\eqstruct$, that equates terms differing only for 
the position of ES. A strong justification comes from the \cbv linear logic interpretation of $\l$-terms with ES, in which structurally equivalent terms 
translate to the same (recursively typed) proof net, see \cite{DBLP:journals/tcs/Accattoli15}. 

The \SCAM{} shall implement the external strategy $\tovsubs$, but only up to $\eqstruct$, which is why we introduce $\eqstruct$ here.

\emph{Structural equivalence} $\eqstruct$ is defined as the least equivalence relation on terms closed by all contexts 
and generated by the following top-level cases:
\begin{center}
	{
	\arraycolsep=.8pt
	$\begin{array}{rcll}
	\tm\esub{\vartwo}{\tmthree}\esub{\var}{\tmtwo} &\tostructcom &\tm\esub{\var}{\tmtwo}\esub{\vartwo}{\tmthree} &\mbox{ 
if 
$\vartwo\notin\fv{\tmtwo}$, $\var\notin\fv{\tmthree}$}
	\\
		\tm\,\tmthree\esub\var\tmtwo &\tostructapr & (\tm\tmthree)\esub\var\tmtwo  &\textrm{ if }\var\not\in\fv{\tm} 
	\\
		\tm\esub{\var}{\tmtwo\esub{\vartwo}{\tmthree}} &\tostructes &\tm\esub{\var}{\tmtwo}\esub{\vartwo}{\tmthree} &\mbox{ 
if $\vartwo\not\in\fv{\tm}$}  
	\\
	\tm\esub\var\tmtwo\tmthree &\tostructapl &(\tm\tmthree)\esub\var\tmtwo  &\textrm{ if }\var\not\in\fv\tmthree 
	\end{array}$}
\end{center}
\label{eq:eqstruct}

Extending the VSC with $\eqstruct$ results in a smooth system, as $\eqstruct$ commutes with evaluation, and can thus be 
postponed. Additionally, the commutation is \emph{strong}, as it preserves the number and kind of steps (thus the cost model)---one says that 
it is a \emph{strong bisimulation} (with respect to $\tovsubs$). 
In particular, the equivalence is not needed to compute and it does not break, or make more complex, any property of the 
calculus---on the contrary, it makes it more flexible.
\begin{proposition}[$\eqstruct$ is a strong bisimulation]
	\label{prop:strong-bisimulation}
	\NoteProof{propappendix:strong-bisimulation} 
	If $\tm\eqstruct\tmtwo$ and $\tm\Rew{\mathsf{a}}\tmp$ then there exists $\tmtwop \in \vsubterms$ such that 
$\tmtwo\Rew{\mathsf{a}}\tmtwop$ and $\tmp\eqstruct\tmtwop$, for $\mathsf{a}\in\set{\msym,\esym,  \osym\msym,\osym\esym, \esssym\msym,\esssym\esym}$. \end{proposition}     

Note that \refprop{strong-bisimulation} implies that $\eqstruct$ preserves normal forms.

\section{A Taste of Useful and Implosive Sharing}
Here we use the VSC to give an informal overview of the various forms of sharing at work in this work.

The VSC comes with ES $\tm\esub\var\tmtwo$, which are a form of \emph{subterm sharing}. The exponential rewriting rule, however, rests on meta-level substitution, and so the system is closer to the $\l$-calculus than to an implementation, which would rather use a micro-step variant of the exponential rule such as
\begin{center}
$\begin{array}{rrllll}
    \multicolumn{3}{c}{\textsc{Micro Exponential} }
    \\
     \ctxp\var\esub\var{\subctxp{\val}} &  \Rew{\mathsf{mi\mbox{-}e}}  & \subctxp{\ctxp\val\esub{\var}{\val}} 
\end{array}$
\end{center}
\paragraph*{Useful Sharing} In \cbn, useful sharing amounts to two modifications of the substitution process, which are mandatory for reasonable implementations of strong evaluation. They are motivated by two paradigmatic cases of size explosions, one related to open terms and one to strong evaluation.

The first example of size-explosion is given by the family of open terms $\{\tm_n \vartwo\}_{n=1,2,\ldots}$ and the family $\{\tmtwo_n\}_{n=1,2,\ldots}$ of 
their normal forms (in the ordinary $\l$-calculus) which are inert terms, where $\tm_{n}$ and $\tmtwo_{n}$ are defined as follows:
\begin{center}
$\begin{array}{rclc|crclccccc}
  \tm_1 & \defeq & \delta = \la{\var}\var\var &&& \tmtwo_1 & \defeq & \vartwo \vartwo\\
  \tm_{n+1} & \defeq & \la{\var} \tm_n (\var\var)  &&&   \tmtwo_{n+1} & \defeq & \tmtwo_n \tmtwo_n
\end{array}$  
\end{center}
We use $\size\tm$ for the size of a term. Size explosion is proved via an auxiliary property. It is worth noticing that in this example the explosion is independent of the evaluation strategy.
\begin{proposition}[Open and strategy-independent size explosion]
Let $n,m>0$. 
\begin{enumerate}
\item \emph{Auxiliary property}: $\tm_n \tmtwo_m \tob^n \tmtwo_{n+m}$.
\item $\tm_n \vartwo \tob^n \tmtwo_{n}$, $\size{\tm_n} = \bigo(n)$, $\size{\tmtwo_n} = \Omega(2^n)$.
\end{enumerate}
\end{proposition}
This case of explosion is avoided by useful sharing by forbidding substitutions of normal terms that are not abstractions, because they do not create $\beta$-redexes---note that the evaluation of the family substitutes $\vartwo$ or instances of $\tmtwo_{i}$, which are inert terms and thus not abstractions. In \scbv as presented via the VSC, this is hardcoded, as only abstractions can be substituted, so nothing needs to be changed. The effect of the optimization can be seen on normal forms: it is accounted by the fact that strong fireballs have ES containing strong inert terms, which are exactly normal terms that are not abstractions.

The second example of size-explosion is a closed variant of the first one, due to Accattoli \cite{DBLP:journals/corr/Accattoli17}. Define:
\begin{center}
$\begin{array}{rclc|crclccccc}
  \tmfour_1  & \defeq &  \la{\var}\la{\vartwo}\vartwo \var \var
   &&&
\tmthree_0  & \defeq & \Id = \la\varthree\varthree\\
\tmfour_{n+1} & \defeq & \la{\var}\tmfour_n (\la{\vartwo}\vartwo \var \var)
  &&&
  \tmthree_{n+1} & \defeq & \la{\vartwo}\vartwo \tmthree_{n} \tmthree_{n}
\end{array}$  
\end{center}

\begin{proposition}[Closed and strategy-independent size explosion, \cite{DBLP:journals/corr/Accattoli17}]
\label{prop:abs-size-explosion}
  Let 
  $n \!>\! 0$. Then $\tmfour_n \Id \tob^n  \tmthree_n$. Moreover, $\size{\tmfour_n \Id} = \bigo(n)$, $\size{\tmthree_n} = \Omega(2^n)$, 
$\tm_n \Id$ is closed, and $\tmthree_n$ is normal.
\end{proposition}

In this second case, substituting only abstractions does not help, because the terms that are substituted along the evaluation are the identity $\Id$ and instances of $\tmthree_{i}$, which are all abstractions. If evaluation is weak, and substitution is done micro-step, then there is no problem because the replaced variables are all instances of $\var$ in some $\tmfour_{i}$, which are under abstraction and which are never replaced in micro-step weak evaluation. With micro-step strong evaluation, however, these replacement do happen, and the size explodes.

To tame this problem, useful sharing rests on an optimization sometimes called \emph{substituting abstractions on-demand}, which is trickier. It requires abstractions to be substituted \emph{only} on applied variable occurrences: note that the explosion is caused by replacements of variables (namely the instances of $\var$) which are \emph{not} applied, and that thus do not create $\beta$-redexes. A step such as $(\var\vartwo)\esub\var\val \toe \val\vartwo$ is accepted, or, \emph{it is useful}, because it creates a $\beta$/multiplicative redex, while a step such as $(\vartwo\var)\esub\var\val \toe \vartwo\val$ is \emph{useless}, and must not be done.

Note however that this optimization makes sense only when one switches to micro-step evaluation via $\Rew{\mathsf{mi\mbox{-}e}}$ above, that is, at the level of machines, because in $(\var\var)\esub\var\val$ there are both a useful and a useless occurrence of $\var$. The implementation of \emph{substituting abstractions on-demand} is very subtle, also because by not performing useless substitutions, it leaves pending ES with values.

\paragraph*{Mixing Implosive and Useful Sharing} There is a case concerning such pending ES where there is some freedom in deciding how to evaluate. Consider a term such as $((\var\vartwo)\vartwo)\esub\vartwo{\la\varthree\tm}$ where $\tm$ is a term with some $\beta$-redexes. Both micro-step substitutions of $\la\varthree\tm$ on $\vartwo$ are useless, but $\tm$ needs to be evaluated. The non-implosive choice is to copy $\la\varthree\tm$, obtaining $(\var(\la\varthree\tm))(\la\varthree\tm)$, and then evaluate $\tm$ twice. This is what Accattoli and Dal Lago do in their useful implementations of the leftmost-outermost \cbn strategy in \cite{DBLP:journals/corr/AccattoliL16,DBLP:conf/wollic/Accattoli16}. It can be seen as useful, because each copy of $\la\varthree\tm$ contains some $\beta$-redexes, so one is not substituting for nothing.

The implosive choice, which is also more \cbv in spirit, is to evaluate $\la\varthree\tm$ only once, keeping it in the ES, that is, reducing $((\var\vartwo)\vartwo)\esub\vartwo{\la\varthree\tm}$ to some $((\var\vartwo)\vartwo)\esub\vartwo{\la\varthree\tmtwo}$. This is what the \SCAM{} shall do.

\paragraph*{A natural question} Why not \emph{always} evaluate values before substituting them? Because, it is unsound with respect to normalization in the \scbv. Consider $\tm\defeq (\var(\la\var\vartwo))\esub\var{\la\varthree\varthree(\la\varfour\Omega)}$. The term $\tm$ would then diverge because it would evaluate $\Omega$, while it normalizes to $\vartwo$, for instance with the external strategy:
\[\begin{array}{lllll}
(\var(\la\var\vartwo))\esub\var{\la\varthree\varthree(\la\varfour\Omega)} 
\toes (\la\varthree\varthree(\la\varfour\Omega)) (\la\var\vartwo)\\ 
\toms\toes 
(\la\var\vartwo) (\la\varfour\Omega)
\toms\toes 
\vartwo
\end{array}\]
Note that in this example the substitution happens on an applied variable. Evaluating a value $\val$ before substituting it can be done safely only when in $\tm\esub\var\val$ the term $\tm$ is normal and all the occurrences of $\var$ in $\tm$ are not applied (that is, the associated micro substitution steps are useless), which is exactly when the \SCAM{} shall do it.

\section{Relaxed Implementations}
\label{SECT:RELAXED}

Here we explain abstractly the subtle and unusual way in which our machine implements the external strategy $\tovsubs$ modulo structural equivalence $\eqstruct$. 

Before giving the details, let us stress a key point. The machine is started on $\l$-terms, not VSC terms, that is, the initial term is not supposed to have any ES. Non-initial states of the machine however shall decode to VSC terms.

\paragraph*{Machines and Structural Strategies}
 A \emph{machine} $\mach = (\state, \tomach, \compil\cdot, \decode\cdot)$ is a transitions system $\tomach$ over a set of states, noted $\state$, with transitions partitioned into \emph{$\beta$-transitions} $\tomachb$ and \emph{overhead transitions} $\tomacho$, together with a \emph{compilation} function $\compil\cdot$ turning $\l$-terms into states, and a read-back function $\decode\cdot$ turning states into VSC terms and satisfying the \emph{initialization constraint} $\decode{\compil\tm} = \tm$ for all $\l$-terms $\tm$. 
 A state $\state$ is \emph{initial} if $\state = \compil\tm$ for some $\l$-term $\tm$, and \emph{final} if no transitions apply.
 An \emph{execution} $\exec: \state \tomach^*\statetwo$ is a possibly empty sequence of transitions from an initial state to a state $\statetwo$ said \emph{reachable}. 

 A \emph{structural strategy} $(\tostrat,\equiv)$ is a rewriting relation $\tostrat$ together with a structural equivalence $\equiv$ on VSC terms, such that $\equiv$ is a \emph{strong bisimulation} with respect to $\tostrat$.

\paragraph*{Relaxed Implementations} In the literature, a machine implements a (structural) strategy when the two are weakly bisimilar, where weakness is given by the fact that the overhead transitions of the machine (that search for redexes and decompose the substitution process) are invisible on the calculus. The bisimulation relates executions of the machine and evaluations on the calculus \emph{locally}, or \emph{small-step}, that is, $\beta$-step-by-$\beta$-step, and for sequences not necessarily reaching a normal form. In particular, there is a bijection between the $\beta$-steps of the strategy and the $\beta$-transitions of the machine.

Our machine does not follow such a simple schema, because it evaluates the body of some shared abstractions, and each $\beta$-transition in these bodies potentially maps (via read-back) to \emph{many} $\beta$-steps on the calculus, breaking the bijection, and forbidding the machine to simulate single steps of the calculus.

We then adopt a relationship between the strategy and the machine that is weaker than a bisimulation and asymmetric: the strategy simulates the machine locally (potentially taking many steps for each $\beta$-transition), while the machine simulates the strategy only \emph{globally}, or \emph{big-step}: preserving divergence and normalizing evaluations, with their cost model, but not $\beta$-step-by-$\beta$-step.
In the VSC, the role of $\beta$ steps on the calculus is played by multiplicative steps, whose number in an evaluation sequence $\deriv$ is noted $\sizem\deriv$.
\begin{definition}[Relaxed implementations]
\label{def:implem}
A machine $\mach = (\state, \tomach, \compil\cdot, \decode\cdot)$ is a \emph{relaxed implementation of a structural strategy} $(\tostrat,\equiv)$ on VSC terms when, given a $\l$-term $\tm$:
\begin{enumerate}
\item \emph{Executions to evaluations}: for any $\mach$-execution $\exec: \compil\tm \tomachine^* \state$ there is a $\tostrat$-evaluation $\deriv: \tm \tostrat^* \equiv\decode\state$ with $\sizebeta\exec \leq \sizem\deriv$.

\item \emph{Normalizing evaluations to executions}: if $\deriv \colon \tm \tostrat^* \tmtwo$ with $\tmtwo$  $\tostrat$-normal then there is an $\mach$-execution $\exec \colon \compil\tm \tomachine^* \state$ with $\state$ final such that $\decode\state \equiv \tmtwo$ with $\sizebeta\exec \leq \sizem{\deriv}$.

\item \emph{Diverging evaluations to executions}: if $\tostrat$ diverges on $\tm$ then $\mach$ diverges on $\compil\tm$ doing infinitely many $\beta$-transitions.
\end{enumerate}
\end{definition}

Next, we isolate sufficient conditions for relaxed implementations, that shall structure our implementative study.

\begin{definition}[Relaxed implementation system]
  \label{def:implementation}
  A relaxed implementation system is given by a \emph{machine} $\mach = (\state, \tomach, \compil\cdot, \decode\cdot)$ and a \emph{structural strategy} $(\tostrat,\equiv)$ such that for every reachable state $\state$:
  \begin{enumerate}
		\item\label{p:def-beta-projection} \emph{Relaxed $\beta$-projection}: $\state \tomachhole\beta \statetwo$ implies that there exists $\deriv:\decode\state \tostrat^+\equiv \decode\statetwo$ such that $\sizem\deriv \geq 1$;
		\item\label{p:def-overhead-transparency} \emph{Overhead transparency}: $\state \tomacho \statetwo$ implies $\decode\state \equiv \decode\statetwo$;
		\item\label{p:def-overhead-terminate}	\emph{Overhead transitions terminate}:  $\tomacho$ terminates;

	\item\label{p:def-progress} \emph{Halt}: if $\state$ is final then $\decode\state$ is $\tostrat$-normal;
	
	\item\label{p:def-determinism} \emph{Lax determinism}:  
 $\tostrat$ is diamond and $\tomach$ is deterministic.
  \end{enumerate}
\end{definition}

\begin{theorem}[Abstract implementation]
\label{thm:abs-impl}
  Let $\mach$ and
  $(\tostrat, \eqstruct)$ form a relaxed implementation system.
    \NoteProof{thmappendix:abs-impl}
  Then, $\mach$ is a relaxed implementation of $(\tostrat,\equiv)$.
\end{theorem}

\section{Introducing the \SCAM{}} 
\label{sect:design}
In the next section we start with the implementative details: let us overview some key points first.

\paragraph*{Crumbling} The \SCAM{} builds on the theory of \cbv abstract machines developed by Accattoli and co-authors  \cite{DBLP:conf/icfp/AccattoliBM14,fireballs,DBLP:journals/scp/AccattoliG19,DBLP:conf/ppdp/AccattoliCGC19}. In particular, it relies on the \emph{crumbling} technique of \cite{DBLP:conf/ppdp/AccattoliCGC19}, which essentially is a specific presentation of the transformation into \emph{administrative normal forms} by Flanagan et al. \cite{DBLP:conf/pldi/FlanaganSDF93a,DBLP:journals/lisp/SabryF93}. As shown in \cite{DBLP:conf/ppdp/AccattoliCGC19}, crumbling allows to reduce the number of data structures required, as it encodes the dump and the stack of \cbv machines inside the environment. This in turn reduces the number of transitions of the machine. Both aspects are extremely valuable when studying strong evaluation, as strong machines tend to have many data structures and at least a dozen transitions. Our \emph{strong crumbling abstract machine} (shortened to \SCAM{})---thanks to crumbling---is compact, having only 1 data structure and 9 transitions. The price to pay are the technicalities of  crumbling, roughly amounting to a light form of compilation.

\paragraph*{Garbage Collection} An unusual but key aspect is that garbage collection is done by the \SCAM{} itself, that is, it is not left to the 
meta-level garbage collector, as it is usually the case with abstract machines. This happens because, in the search for values to evaluate strongly, the \SCAM{} has to avoid the garbage ones, because evaluating their bodies would indeed break correctness with respect to \scbv.

\paragraph*{Zig-Zag} The \SCAM{} shall have two alternating phases, one performing open evaluation, and one searching for a value $\la\var\tm$ to evaluate strongly. Once the \SCAM{} finds it, it switches to the open phase for evaluating its body $\tm$, and so on. The open phase can be implemented exploring the code from left-to-right or from right-to-left---we adopt right-to-left, because this choice induces some stronger invariants. The phase searching for values is instead left-to-right, because it also performs garbage collection, which cannot be done right-to-left. Our mixed order  is hinted at in Biernacka et al. 
\cite{DBLP:conf/aplas/BiernackaBCD20} as a possible optimization of their right-to-left strong machine.

\paragraph*{Two Levels of Implosive Sharing} There are two levels of implosiveness, connected to the out/under abstraction dichotomy. \emph{Shallow implosive sharing} evaluates inside shared subterms but not inside shared abstractions. This happens in \cbneed. 
\emph{Deep implosive sharing}, instead, also enters shared abstractions---an instance is optimal reduction\footnote{We avoid the weak/strong terminology, because there can be strong evaluation with shallow implosive sharing, as in Strong Call-by-Need \cite{DBLP:journals/pacmpl/BalabonskiBBK17}.}.
The \SCAM{} adopts a deep implosive approach to useful sharing.


\section{Compilation and Read-Back}
\label{SECT:CRUMBLING}
\begin{figure*}[t!]
\begin{center}
\begin{tabular}{c|c}
	$\begin{array}{l\colspace l\colspace l}
	\multicolumn{3}{c}{\textsc{Auxiliary}}
	\\
	 \CrumbAux\var \defeq (\var, \epsilon)
	 &
	 \CrumbAux{\la\var\tm}  \defeq (\varthree, \esub\varthree{\la\var{\Crumb\tm}}) 
	 &
	 \CrumbAux{\tm\tmtwo} \defeq (\varthree, \esub\varthree{\var\vartwo}\env\envtwo) 
	\end{array}$
	&
	$\begin{array}{l\colspace l\colspace l}
	 \multicolumn{3}{c}{\textsc{Crumbling}}
	 \\
 	\Crumb\var \defeq \esub\varstar\var 
	&
	\Crumb{\la\var\tm} \defeq \esub\varstar{\la\var{\Crumb\tm}} 
	&
	\Crumb{\tm\tmtwo} \defeq \esub\varstar{\var\vartwo} \env \envtwo
	\end{array}$
	\end{tabular}
where $\CrumbAux\tm = (\var, \env)$  and  $\CrumbAux\tmtwo = (\vartwo, \envtwo)$ in both $\CrumbAux{\tm\tmtwo}$
and $\Crumb{\tm\tmtwo}$; 
and $\varthree \text{ fresh in } \crnames$ in both $\CrumbAux{\la\var\tm}$ and $\CrumbAux{\tm\tmtwo}$.
\end{center}
\caption{\label{fig:crm}Crumbling transformation.}
\end{figure*}

In a \cbv $\l$-calculus with a construct for subterm sharing, such as ES, applications can be decomposed by introducing 
sharing 
points for any non-variable subterm. Here we consider the case where applications are only between 
variables\footnote{For crumbling, we follow \cite{DBLP:conf/ppdp/AccattoliCGC19}. Therein applications can  have abstractions as subterms. Here however 
we 
adopt the minor variant where abstractions are also removed from applications and shared.}. For instance, the crumbling 
representation $\mytr\tm$ of
$\tm \defeq (\la\var(\la\vartwo\vartwo)\,(\var\var))\,(\la\varthree\varthree\varthree)$ (see forthcoming \refex{crumbling})
is
	\begin{center} $
		(\cvar\cvartwo)
		\esub{\cvar}{\la\var{\cvarthree \cvarfour} \esub \cvarthree {\la\vartwo \vartwo} \esub{\cvarfour}{\var\var}}
		\esub \cvartwo {\la\varthree{\varthree\varthree}}
	$ \end{center}
where we denoted the sharing points introduced by the transformation by $\cvar, \cvartwo, \cvarthree, \cvarfour$.
Note that the transformation involves also function bodies (\ie{} $\la\var(\la\vartwo\vartwo)\,(\var\var)$ turns into 
$\la\var{(\cvarthree \cvarfour)} \esub \cvarthree {\la\vartwo \vartwo} \esub{\cvarfour}{\var\var}$), that ES are grouped together unless forbidden by abstractions, and that ES 
are flattened out, i.e. they are not nested unless nesting is forced by abstractions. Here we shall adopt a variant 
of this transformation, having the first subterm $\cvar\cvartwo$ of $\mytr\tm$ in a pending ES $\esub\varstar{\cvar\cvartwo}$ on a 
special variable $\varstar$ dedicated to such pending ES---this is analogous to the initial continuation of 
continuation-passing transformations.

Such a \emph{crumbled representation} of terms impacts on the design of machines for 
\cbv evaluation. 
By removing the applicative structure, 
there is no need for data structures encoding the evaluation context, such as the applicative stack and the dump, that 
get encoded in the environment. The environment is the data structure for sharing that collects the ES obtained 1) at 
compile time, 
\ie by the crumbling transformation and 2) dynamically, during execution. 

In \cite{DBLP:conf/ppdp/AccattoliCGC19}, it is shown that the crumbling technique smoothly accommodates open terms, by designing an abstract 
machine that implements \ocbv within a bilinear overhead (when implemented on RAM). This paper extends that work to the strong 
case, but as explained in the introduction, the extension is non-trivial. We now cover compilation via crumbling; the next section deals with the open machine, and \refSECT{STRONG-MACHINE} presents the strong extension.

\paragraph*{Crumbled Environments}
We first have to define the target language of the translation, which are not terms with ES but \emph{crumbled 
environments}, a slight variant. A crumbled environment is a list of ES containing \emph{bites}, defined below. A key 
point is that we need to distinguish the variables introduced by the crumbling transformation from those originally in 
the term, which is why variables range over a set of names $\varnames= \crnames \uplus \calcnames$ where $\crnames$ is 
the set of crumbling variables and $\calcnames$ the set of variables of the calculus, both infinite---the names in $\crnames$ are sometimes noted $\cvar, \cvartwo, \cvarthree$ for clarity, but in general names from both sets are noted $\var,\vartwo,\varthree$. Moreover, there is 
a distinguished variable $\varstar \in \crnames$.
\begin{center}$
\begin{array}{r rcl}
	\textsc{Bites} & \mol & \grameq & \var \mid \var\vartwo \mid \la\var\env \ \ (*)
	\\
	\textsc{(Crumbled) Envs} & \env & \grameq &\emptyenv \mid \env\cons \esub\var\mol
	\end{array}$\end{center}
Side conditions $(*)$: $\var\neq\varstar\neq\vartwo$ in $\var$ and $\var\vartwo$, and $\var\in\calcnames$ and $\env$ is 
non-empty in $\la\var\env$. The conditions imply that $\varstar$ cannot have free occurrences. As for terms, bites of 
the form $\la\var\env$ are \emph{values}, ranged over by $\val$, while $\var$ and $\var\vartwo$ are \emph{inert bites}.

Environments are defined concatenating on the right, but we shall freely concatenate also on the left, concatenate whole 
environments, and omit the concatenation symbol `$:$'. Environments are also meant to be 
looked up for substitution. \emph{Notation:} $\env(\var) = \mol$ if $\env = \envtwo \esub\var\mol \envthree$ with $\var 
\notin \dom{\envtwo}$, and $\env(\var) = \bot$ otherwise---note that in open/strong settings environments may be 
undefined on some variables. 

\paragraph*{Crumbling $\l$-Terms} Machines start their execution on the compilation of ordinary $\l$-terms (with no ES), 
and the following crumbling transformation $\Crumb\tm$ shall be our notion of compilation. Note that $\varstar$ appears 
always and only as the variable ``bound'' by the leftmost ES in $\Crumb\tm$.
\begin{definition}[Crumbling transformation $\Crumb\cdot$]
Let $\tm$ be a \lat{}. We define its crumbling $\Crumb\tm$
using an auxiliary function $\CrumbAux\cdot$ mapping \lat{s} to pairs of a variable plus an environment. The formal definition of $\Crumb\cdot$ and $\CrumbAux\cdot$ are given in \reffig{crm}, and explained in the next example.
\end{definition}

\begin{example}\label{ex:crumbling}
The main transformation $\Crumb\cdot$ is used at top level, both of the initial term and recursively at top level of every function body. The auxiliary transformation $\CrumbAux\cdot$ instead is used when compiling applications, and it returns the variable that shall be used in place of the original term, plus a crumbled environment that binds additional results of the transformation.

For the sake of example, let us consider the term
%
\[
	\tm\defeq (\la\var I\,(\var\,\var))\,\delta = (\la\var(\la\vartwo\vartwo)\,(\var\,\var))\,(\la\varthree\varthree\,\varthree).
\]
The term $\tm$ consists of an application of two non-variable terms, hence the transformation yields
\begin{center}$\begin{array}{ccc}
	\Crumb\tm & = & \esub\varstar{\cvar\cvartwo}\esub\cvar\Placeholder\cdots\esub\cvartwo\Placeholder\cdots
\end{array}$\end{center}
where $\cvar$ and $\cvartwo$ are two fresh variables generated respectively by $\CrumbAux{\la\var I\,(\var\var)}$ and $\CrumbAux\delta$:
\[\begin{array}{rl}
	\CrumbAux{\la\var I\,(\var\var)} = & (\cvar, \esub\cvar{\la\var{\Crumb{I\,(\var\var)}}}) \\
	= & (\cvar, \esub{\cvar}{\la\var\esub\varstar{\cvarthree \cvarfour} \esub \cvarthree {\la\vartwo \esub \varstar \vartwo} \esub{\cvarfour}{\var\var}}) \\
	\text{\emph{of the form}} & (\cvar,  \esub\cvar{\la\var\esub\varstar{\cvarthree\cvarfour} \esub\cvarthree\Placeholder\cdots\esub\cvarfour\Placeholder\cdots}) \\
\end{array}\]
\[\begin{array}{rll}
	\CrumbAux \delta & = (\cvartwo, \esub \cvartwo {\la\varthree{\Crumb{\varthree\varthree}}}) 
	 = (\cvartwo, \esub \cvartwo {\la\varthree{\esub\varstar{\varthree\varthree}}}).
\end{array}\]
The fully transformed $\Crumb\tm$ is:
\begin{center}\small$
\esub\varstar{\cvar\cvartwo}
\esub{\cvar}{\la\var\esub\varstar{\cvarthree \cvarfour} \esub \cvarthree {\la\vartwo \esub \varstar \vartwo} \esub{\cvarfour}{\var\var}}
\esub \cvartwo {\la\varthree{\esub\varstar{\varthree\varthree}}}.$\end{center}
\end{example}

\paragraph*{Names} A key point is that, as it is standard for abstract machines, crumbled environments and bites are \emph{not} considered modulo 
$\alpha$-equivalence. Some machine transitions shall rename variables: 
$\alpha$-equivalence can rename $\var$ with $\vartwo$ only if they are both in $\crnames\setminus\set\varstar$ (resp. 
both in $\calcnames$), and $\varstar$ cannot be renamed. We also need a notion of well-namedness for both $\l$-terms and 
environments.
\begin{definition}[Well-named]\label{def:well-named-new}
    	A \lat{} $\tm$ is \emph{well-named} if its bound variables are all distinct, and $\fv\tm\cap\bv\tm=\emptyset$. An 
environment $\env$ (resp. a bite $\mol$) is \emph{well-named} if when two binders bind the same variable $\var$ then 
$\var=\varstar$, and $\fv\env \cap \bv\env \subseteq \{\varstar\}$ (resp. $\fv\mol \cap \bv\mol \subseteq 
\{\varstar\}$). 
\end{definition}

\paragraph*{Read-back} Bites and environments are mapped to VSC terms via a read-back function, that in particular inverts 
the crumbling transformation. The distinction between the two kinds of variables plays a role. The intuition is 
that ES are unfolded when they come from crumbling \emph{or} when they contain values, as to include in the read-back the useless part of the work done by $\toe$ on the calculus.

The read-back of a bite $\mol$ and an environment $\env$ are, respectively, the terms $\unf\mol$ and $\unf\env$ defined 
by:
\begin{center}

	$\begin{array}{c}
	\begin{array}{r\colspace rclrclccccc}
	\multicolumn{4}{c}{\textsc{Bites read-back}}
	\\
		\unf\var \defeq  \var
	&
	\unf{(\var\vartwo)} &\defeq&  \var\vartwo
	\\
	&
	\unf{(\la\var\env)} &\defeq&  \la\var\unf\env 
	\end{array}
	\\
	\begin{array}{c}
	\textsc{Crumbled environments read-back}
	\\
	
	\unf\emptyenv \!\defeq  \varstar
	\quad
	\unf{\env\esub\var\mol} \!\defeq
		\begin{cases}
		\unf\env \isub\var{\unf\mol} & \mbox{if $\mol=\val$ or $\var \in \crnames$}\\
		\unf\env \esub\var\mol & \mbox{otherwise.}
		\end{cases}
	\end{array}	
	\end{array}$
\end{center}

\begin{lemma}[Crumbling properties]
\label{l:transl-properties}
    \NoteProof{lappendix:transl-properties}
If $\tm$ is a well-named \lat{} then $\mytr\tm$ is well-named and $\unf{\mytr\tm} = \tm$.	
\end{lemma}

\ifthenelse{\boolean{techreport}}{
In the next sections, we shall need a modular deconstruction of read-back, spelled out below.

\begin{definition}
 Let $\env$ be an environment. Then the \emph{substitution $\indsub\env$} and the \emph{substitution context 
$\indenv\env$  induced by $\env$} are given by (where $(*)$ stands for ``$\mol=\val$ or $\var \in \crnames$'')
\begin{center}
	$\arraycolsep=2pt\begin{array}{rl\colspace rl}
	\multicolumn{4}{c}{\textsc{Substitution $\indsub\env$ induced by $\env$}}
	\\
\indsub\emptyenv & \defeq  Id
&
\indsub{\env\esub\var\mol} & \defeq 
\begin{cases}
\indsub\env\isub\var{\unf\mol}  & \mbox{if $(*)$}\\
\indsub\env & \mbox{otherw.}
\end{cases}
\end{array}$

	$\arraycolsep=2pt\begin{array}{rl\colspace rl}
	\multicolumn{4}{c}{\textsc{Substitution context $\indenv\env$ induced by $\env$}}
	\\
\indenv\emptyenv & \defeq \ctxhole
&
\indenv{\env\esub\var\mol} & \defeq 
\begin{cases}
\indenv\env\isub\var{\unf\mol}  & \mbox{if $(*)$}\\
\indenv\env\esub\var\mol & \mbox{otherw.}
\end{cases}
\end{array}$

\end{center}
\end{definition}
\begin{lemma}[Modular read-back]
\label{l:read-back-decomposition-d}
\NoteProof{lappendix:read-back-decomposition-d}
 $\unf{(\env\envtwo)} = \indenv\envtwo\ctxholep{\unf\env\indsub\envtwo}$.
\end{lemma}
}{}

\begin{example}
Let us consider the environment
\begin{center}$
	\underbrace{\esub\varstar\vartwo}_\env
	\underbrace{\esub\vartwo{\cvartwo\cvartwo}
	\esub \cvartwo {\la\varthree{\esub\varstar{\varthree\,\varthree}}}}_{\envtwo}$
\end{center}
where $\vartwo$ is a ``normal'' variable, and $\cvartwo$ is a crumbling variable.
The read-back $\unf{\env\envtwo}$ proceeds as follows:
\ifthenelse{\boolean{techreport}}{
\[\begin{array}{rl}
	\unf{\env\envtwo}
	& = \unf{\esub\varstar\vartwo \esub\vartwo{\cvartwo\cvartwo}} ~
	\isub \cvartwo {\la\varthree{\esub\varstar{\varthree\,\varthree}}} \\
	& = \unf{\esub\varstar\vartwo} ~ \esub\vartwo{\cvartwo\cvartwo}
	\isub \cvartwo {\la\varthree{\esub\varstar{\varthree\,\varthree}}} \\
	& = \left(\unf{\esub\varstar\vartwo} \isub \cvartwo {\la\varthree{\esub\varstar{\varthree\,\varthree}}}\right)
	\left(\esub\vartwo{\cvartwo\cvartwo } \isub \cvartwo {\la\varthree{\esub\varstar{\varthree\,\varthree}}} \right)\\
\end{array}\]
The equality in~\reflemma{read-back-decomposition-d} holds at this point, since:
\[ \begin{array}{c \colspace c}
	\indenv\envtwo  =  \ctxhole \esub\vartwo{\cvartwo\cvartwo } \isub \cvartwo {\la\varthree{\esub\varstar{\varthree\,\varthree}}} &
	\indsub\envtwo  =  \isub \cvartwo {\la\varthree{\esub\varstar{\varthree\,\varthree}}}
\end{array}\]
The full read-back $\unf{\env\envtwo}$ is $ \vartwo\esub\vartwo{(\la\varthree{\esub\varstar{\varthree\,\varthree}})(\la\varthree{\esub\varstar{\varthree\,\varthree}})}$.
}{
\[\begin{array}{rl}
	\unf{\env\envtwo}
	& = \unf{\esub\varstar\vartwo \esub\vartwo{\cvartwo\cvartwo}} ~
	\isub \cvartwo {\la\varthree{\esub\varstar{\varthree\,\varthree}}} \\
	& = \unf{\esub\varstar\vartwo} ~ \esub\vartwo{\cvartwo\cvartwo}
	\isub \cvartwo {\la\varthree{\esub\varstar{\varthree\,\varthree}}} \\
	&= \vartwo\esub\vartwo{(\la\varthree{\esub\varstar{\varthree\,\varthree}})(\la\varthree{\esub\varstar{\varthree\,\varthree}})}
\end{array}\]
}
\end{example}

\section{The Open Crumbling Machine}
\label{SECT:OPEN-MACHINE}

Here we overview an abstract machine implementing the open VSC $\tovsubo$, that shall be the starting point for the 
strong machine of the next section. We keep following the crumbling technique by \cite{DBLP:conf/ppdp/AccattoliCGC19}, 
slightly adapted. 

The only structure at work in the machine is a crumbled environment, traversed from right to left, together with a 
pointer to where the machine is operating. Then a machine state $\state\defeq \env\,\rlsep\,\envtwo$ is a pair of crumbled 
environments where
\begin{itemize}
 \item \emph{Right}: $\envtwo$ is the part that has already been processed,
 \item \emph{Left}: $\env$ is the part yet to be processed, and
 \item \emph{Separator}: $\rlsep$ represents the pointer to the active point. 
\end{itemize}

The Open Crumbling Abstract Machine (\OCAM) has 4 transitions, two $\beta$ transitions $\tomachbv$ and $\tomachbi$, and 
two overhead transitions $\tomachsub$ and $\tomachc$, detailed below. Compilation is defined as 
$\compil\tm\defeq\mytr\tm\,\rlsep\,\emptyenv$ for a well-named $\l$-term $\tm$, and read-back simply as 
$\unf{(\env\,\rlsep\,\envtwo)} \defeq \unf{(\env\envtwo)}$. By \reflemma{transl-properties}, compilation and read-back 
verify the initialization constraint for the \OCAM.

\begin{figure*}
\begin{center}
	\arraycolsep=2pt
	$\begin{array}{rcllrcllcccc}
\env \esub\var{\vartwo\,\varthree} & \rlsep & \kctx
&  \tomachbv &
\env (\esub\var \mol \envtwo \isub\varfour\varthree)  & \rlsep   & \kctx
& \text{if $(*)$ and $\wstenv\kctx(\varthree)= \val$ for some $\val$;}
\\%
\env \esub\var{\vartwo\,\varthree} & \rlsep &  \kctx
&  \tomachbi &
\env \esub\var \mol \envtwo & \rlsep & \kctxp{ \ctxhole\esub\varfour\varthree   }
& \text{if $(*) $ and $\wstenv\kctx(\varthree)= \itm$ for some $\itm$;}
\\%
\env \esub\var\vartwo & \rlsep & \kctx
& \tomachsub&
\env\isub\var\vartwo & \rlsep & \kctx
& \text{if $\var \neq\varstar$;}
\\
\env \esub\var\mol & \rlsep & \kctx
& \tomachcone &
\env & \rlsep & \kctxp{ \ctxhole\esub\var\mol  }
& \text{if none of the other rules is applicable.}
\\[4pt]
\hline
\\[-6pt]
\emptyenv & \rlsep & \kctx
& \tomachctwo &
\emptyenv & \lrsep & \kctx
\\
\env & \lrsep & \kctxp{ \ctxhole\esub\var\mol  }
& \tomachcthree &
\env \esub\var\mol & \lrsep & \kctx
& \text{if $\mol$ is not a value;}
\\
\env & \lrsep & \kctxp{ \ctxhole\esub\var\val  }
& \tomachgc &
\env & \lrsep & \kctx
& \text{if $\var \notin \fv\env$ and $\env$ is not empty;}
\\
\env & \lrsep & \kctxp{\envtwo \esub\var{\la\vartwo\ctxhole }  }
& \tomachcfour &
\envtwo \esub\var{\la\vartwo\env} & \lrsep  & \kctx
\\
\env & \lrsep & \kctxp{ \ctxhole\esub\var{\la\vartwo\envtwo} } 
& \tomachcfive &
\envtwo & \rlsep & \kctxp{\env \esub\var{\la\vartwo\ctxhole}}  

& \text{if $\var \in \fv\env$.}
\end{array}$
\end{center}
\caption{Open (above) and strong (below) phases of the \SCAM{}.}
\label{fig:scam-phases}
\end{figure*}
\paragraph*{$\beta$-Transitions} They are quite technical unfortunately, because of crumbling. The idea is that there are two cases, $\tomachbv$ for when the argument is a value and $\tomachbi$ for when it is a inert term. In the first case, the machine also does in one single transition both the $\beta$/multiplicative step and the exponential step that is created. Because of crumbling, the $\beta$-redex is given by an application of variables $\vartwo\varthree$, whose abstraction and argument are to be found in the environment. Actually, the transitions also does the copy of the abstraction that replaces $\vartwo$.
They are (further explanations follow):
%
\begin{center}
$\begin{array}{rcl}
\env \esub\var{\vartwo\,\varthree} \,\rlsep\,  \envtwo
 &\tomachbv &
\env (\esub\var \mol \envthree\isub\varfour\varthree)   \,\rlsep\,  \envtwo
\\
\env \esub\var{\vartwo\,\varthree} \,\rlsep\, \envtwo
 &\tomachbi &
\env \esub\var \mol \envthree \,\rlsep\, \esub\varfour\varthree   \envtwo
\end{array}$
\end{center}
where in both cases $\envtwo(\vartwo)$ is a value and $\rename{(\envtwo(\vartwo))} =: \la\varfour(\esub{\varstar}\mol\envthree)$ is a well-named copy of 
$\envtwo(\vartwo)$ with fresh names, and  $\envtwo(\varthree)$ is a value in $\tomachbv$, while in $\tomachbi$ it is an 
inert bite. 

\emph{Explanation.}  First, we explain points that are common to both transitions. The bite under analysis is $\vartwo\varthree$ and $\vartwo$ maps to a value $\val$ in  $\envtwo$, thus the read-back turns $\vartwo\varthree$ into $(\vartwo \varthree)\indsub\envtwo 
=\indsub\envtwo(\vartwo) \indsub\envtwo(\varthree) = \val \indsub\envtwo(\varthree)$ that is a $\beta$/multiplicative redex. 
The machine copies $\val$, obtaining $\la\varfour(\esub{\varstar}\mol\envthree)$, as copying corresponds to 
$\alpha$-renaming. Variables have indeed to be intended as memory locations, and $\alpha$-renaming means making a 
copy somewhere else in the memory. 
Letting the argument $\varthree$ aside, what happens 
to both transitions is: the $\beta$-redex is fired and $\vartwo\varthree$ is replaced by the body 
$\esub{\varstar}\mol\envthree$ of the copied value, that is concatenated with $\env$, obtaining $\env \esub\var \mol 
\envthree$. Via read-back, the multiplicative redex is $\val \indsub\envtwo(\varthree)= (\la\varfour\tm) \indsub\envtwo(\varthree)$ with $\tm = \unf{(\esub{\varstar}\mol\envthree)}$, which takes a $\rtom$ step to $\tm\esub\varfour{\unf{\indsub\envtwo(\varthree)}}$.

Consider now the argument $\varthree$. 
If it is associated with a value $\val$ in $\envtwo$ then its read-back is also a value, namely $\valtwo = 
\val\indsub\envtwo$, and so on the calculus $\tm\esub\varfour{\valtwo}$ is a $\rtoe$ redex. 
The machine substitutes $\varthree$ for $\varfour$ in the body $\esub\var \mol \envthree$ of the copied value. This 
corresponds to performing the $\rtoe$-redex $\tm\esub\varfour{\valtwo}\rtoe \tm\isub\varfour{\valtwo}$ on the calculus. Note that the machine only performs a renaming, it does not duplicate 
$\valtwo$---up to read-back this is equivalent. If instead $\varthree$ is associated to an inert bite in $\envtwo$, let us 
assume for a moment that $\varthree$ reads back to an inert term. Then $\tm\esub\varfour{\unf{\indsub\envtwo(\varthree)}}=\tm\esub\varfour{\itm}$ for some inert term $\itm$, and no substitution happens. 

For $\tomachbv$ and $\tomachbi$ to cover all cases, the environment $\envtwo$ needs to satisfy two properties. First, values are not hidden behind chains of renamings, that is, if $\env(\var)=\vartwo$ then $\env(\vartwo)$ is not a value. Second, if $\env(\var)$ is a inert bite, then it reads back to a inert term. These two invariants are nicely expressed in ``read-back form'' via $\indsub\envtwo$: on any 
reachable state $\env \,\rlsep\, \envtwo$
\begin{itemize}
 \item \emph{$\envtwo$ is a fireball substitution}, that is, $\indsub\envtwo(\var)$ is a fireball for every $\var \in 
\dom{\indsub\envtwo}$.
 \item \emph{$\envtwo$ has immediate values}, that is, if $\indsub\envtwo(\var)$ is a value and $\var \neq \varstar$ 
then $\envtwo(\var)$ is a value. 
\end{itemize}

\paragraph*{Useful Sharing} The \OCAM implements useful sharing because it copies only abstractions and only \emph{on-demand}, that is, only on variable occurrences that are applied, namely on $\vartwo$ in the definitions of $\tomachbv$ and $\tomachbi$.

\paragraph*{Overhead Transitions} The overhead transition $\tomachsub$ eliminates explicit renamings, that is, ES 
containing variables:
\[\arraycolsep=3pt
\begin{array}{rcllrcllcccc}
\env \esub\var\vartwo & \rlsep & \envtwo
& \tomachsub&
\env\isub\var{\vartwo} & \rlsep & \envtwo
\end{array}\]
when $\var \neq\varstar$. \ifthenelse{\boolean{techreport}}{The forthcoming pristine invariant of the machine guarantees that $\var$ always has at 
most one occurrence in $\env$ (\reflemma{aux-most-once} below), so that $\tomachsub$ shall not be costly.}{An invariant of the machine (namely the \emph{pristine invariant}, see the tech report \cite{DBLP:journals/corr/abs-2102-06928}) guarantees that $\var$ always has at 
most one occurrence in $\env$, so that $\tomachsub$ shall not be costly.}

The overhead transition $\tomachc$ simply moves the pointer $\rlsep$ to the left when no other rule is applicable, i.e. when 1) $\mol$ is a value, or 2) when $\mol$ is $\vartwo$ or $\vartwo\varthree$ but $\envtwo(\vartwo)$ is not a 
value:
\begin{center}$\arraycolsep=3pt
\begin{array}{rcllrcllcccc}
\env \esub\var\mol & \rlsep & \envtwo
& \tomachc &
\env & \rlsep & \esub\var\mol\envtwo
\end{array}$\end{center}

\begin{example}
Let us see how the \OCAM reduces the environment from \refex{crumbling}:

\noindent\hspace{-10pt}\scalebox{0.87}{
\bgroup
\setlength\tabcolsep{1.5pt}
\def\arraystretch{1.5}
\begin{tabular}{ll}
   $\esub\varstar{\cvar\cvartwo}
  \esub{\cvar}{\la\var\esub\varstar{\cvarthree \cvarfour} \esub \cvarthree {\la\vartwo \esub \varstar \vartwo} \esub{\cvarfour}{\var\,\var}}
  \esub \cvartwo {\la\varthree{\esub\varstar{\varthree\,\varthree}}} \rlsep\ \tomachc$
  \\ 
  $\esub\varstar{\cvar\cvartwo}
  \esub{\cvar}{\la\var\esub\varstar{\cvarthree \cvarfour} \esub \cvarthree {\la\vartwo \esub \varstar \vartwo} \esub{\cvarfour}{\var\,\var}} \rlsep ~
  \esub \cvartwo {\la\varthree{\esub\varstar{\varthree\,\varthree}}} \tomachc$
  \\  
  $\esub\varstar{\cvar\cvartwo} \rlsep ~
  \esub{\cvar}{\la\var\esub\varstar{\cvarthree \cvarfour} \esub \cvarthree {\la\vartwo \esub \varstar \vartwo} \esub{\cvarfour}{\var\,\var}}
  \esub \cvartwo {\la\varthree{\esub\varstar{\varthree\,\varthree}}} \tomachbv$
  \\ 
  $\esub\varstar{\cvarthree \cvarfour} \esub \cvarthree {\la\vartwo \esub \varstar \vartwo} \esub{\cvarfour}{\cvartwo\,\cvartwo} \rlsep ~
  \esub{\cvar}\cdots
  \esub \cvartwo {\la\varthree{\esub\varstar{\varthree\,\varthree}}} \tomachbv$
  \\  
  $\esub\varstar{\cvarthree \cvarfour} \esub \cvarthree {\la\vartwo \esub \varstar \vartwo} \esub{\cvarfour}{\cvartwo\,\cvartwo} \rlsep ~
  \esub{\cvar}\cdots
  \esub \cvartwo {\la\varthree{\esub\varstar{\varthree\,\varthree}}} \tomachbv$ \ldots
\end{tabular}\egroup}

\medskip

The first two steps apply the search transition: the two rightmost ES bind values, thus they can just be skipped since there is nothing to evaluate. The next steps apply a $\beta_\lambda$ transition, substituting the body of the abstraction bound in the environment by $\cvartwo$. Evaluation then loops infinitely because the environment has no normal form.
\end{example}

\ifthenelse{\boolean{techreport}}{
\subsection{Implementation Theorem} 
The main property for the implementation theorem is relaxed $\beta$-projection. At 
the open level, one $\beta$-transition projects to one $\tomo$ step, plus one $\toeo$ step if the transition is $\tomachbv$. To prove 
it, we need that after read-back 1) the bite $\vartwo\varthree$ rewritten by $\beta$-transitions becomes a 
$\tomo$ redex (with a value as an argument for $\tomachbv$) and 2) it occurs in an open context. Becoming a $\tomo$ redex is guaranteed by the two invariants above. Occurring in an open context instead requires the further invariant below.

Now, if $\state = \env\esub\var{\vartwo\varthree}\,\rlsep\,\envtwo$ performs a $\beta$-transition, by modular read-back (\reflemma{read-back-decomposition-d}) 
 $\unf{\state} = \indenv\envtwo\ctxholep{\unf{(\env\esub\var{\vartwo\varthree})}\indsub\envtwo}$. We have that $\indenv\envtwo$ is an open context. Let's focus on $\env\esub\var{\vartwo\varthree}$. The invariant is that the environments on the left of $\rlsep$ are \emph{pristine}
, that is, that in our case $\env$ unfolds to $\openctxp\var$ for some open context $\openctx$  such that $\var\notin\allvars\openctx$.
\begin{definition}[Pristine]
\emph{Pristine environments} and \emph{pristine bites} are mutually defined as follows:
        \begin{itemize}
            \item \emph{Environments}: $\emptyenv$ is a pristine environment; $\env\esub\var\mol$ is 
pristine if $\env$ and $\mol$ are pristine, $\var\in\crnames$, and $\unf\env=\openctxp\var$ for some open context 
$\openctx$ such that $\var\notin\allvars\openctx$.
            \item \emph{Bites}: inert bites are pristine; $\la\var\env$ is pristine if $\env$ is 
pristine.
        \end{itemize}
\end{definition}
The following property shall be crucial for the complexity analysis in \refSECT{COMPLEXITY}.
\begin{lemma}
 \label{l:aux-most-once}
 \NoteProof{lappendix:aux-most-once}
 If $\env\esub\var\mol$ is pristine and well-named and $\var\neq\varstar$, then $\var$ occurs exactly once in $\env$.
\end{lemma}
Now, if $\state = \env\esub\var{\vartwo\varthree}\,\rlsep\,\envtwo$ the invariants give $\decode\state=\indenv\envtwo\ctxholep{\openctx \ctxholep{\vartwo 
\varthree}\indsub\envtwo} = \indenv\envtwo\ctxholep{\openctx \indsub\envtwo \ctxholep{\vartwo\indsub\envtwo 
\varthree\indsub\envtwo}}$ which has the desired shape because $\indenv\envtwo\ctxholep{\openctx \indsub\envtwo}$ is 
open---both $\openctx$ and $\indenv\envtwo$ are open, and open contexts are stable by substitution and plugging---and 
$\vartwo\indsub\envtwo \varthree\indsub\envtwo$ is a $\rtom$-redex, as seen before. Structural 
equivalence $\eqstruct$ plays a role in the projection of $\tomachbi$, to put the created ES $\esub\varfour\varthree$ at 
the right place.

For the halt property, note that final states have the form $\emptyenv\,\rlsep\,\envtwo$ that by definition of read-back, 
\reflemma{read-back-decomposition-d}, and the fireball substitution invariant, read-backs to $\indenv\envtwo\ctxholep{\fire}$, where $\fire$ is a fireball. A last invariant ensures that the substitution context $\indenv\envtwo$ induced by every reachable state $\env\,\rlsep\,\envtwo$ 
is a \emph{inert (substitution) context} that is, it contains only inert terms---$\indenv\envtwo\ctxholep{\fire}$ is then a fireball, as required.

Spelling out the details, one obtains that the \OCAM is a relaxed implementation of $(\tovsubo, \eqstruct)$.
}{}

\section{The Strong Crumbling Machine}
\label{SECT:STRONG-MACHINE}
Here we define the Strong Crumbling Abstract Machine (\SCAM{}), building on the previous section.

Basically, we extend the \OCAM with a new \emph{strong phase}, identified by a new separator $\lrsep$, whose task is to 
look for ES $\esub\var\val$ containing values and, once one is found, evaluating $\val$ under abstraction if $\var$ 
occurs somewhere in the state, and garbage collect it otherwise. 

\paragraph*{Garbage Collection} Consider the term $\tm \defeq (\la\vartwo\var) (\la\varthree\Omega)$ that evaluates 
in two $\tovsubs$ steps to $\var$ while containing the diverging subterm $\la\varthree\Omega$. The \OCAM executed on 
$\tm$ produces the final state $\state\defeq \rlsep\esub\varstar\var\esub\vartwo{\la{\varthree}\mytr\Omega}$. Now its 
strong extension should search for ES containing values in $\state$, but it has to avoid entering 
$\esub\vartwo{\la{\varthree}\mytr\Omega}$ otherwise it would diverge, breaking the correspondence with $\tovsubs$. It follows that 
the machine has to track which ES are garbage and which are not.
Let us assume for now that the machine can check it easily; the next section discusses how to implement it.

\paragraph*{Deep Implosive Sharing} The machine enters into values whose ES $\esub\var\val$ is such that $\var$ does 
occur, potentially \emph{many} times. This ingredient accounts for \emph{deep implosive} sharing. A key point is that the machine enters only inside ES that shall no longer substitute their values (they are left there pending because of useful sharing, as they bind useless occurrences only)---this is why the crucial subterm invariant for complexity analyses (stating that terms duplicated along the whole evaluation with sharing are subterms of the initial term) is not compromised. Useful sharing, instead, is already implemented by the \OCAM, and thus simply inherited by the \SCAM{}.

\paragraph*{\SCAM{}} First of all, we generalize the right component of states, so as to account for evaluation positions 
under abstraction. The idea is that the right environment $\envtwo$ can be seen as a context, namely $\ctxhole\envtwo$, 
and that going under abstraction simply requires a further context construction.
\begin{center}$\begin{array}{r\colspace rcl \colspace}
\textsc{Machine contexts} & \kctx & \grameq & \ctxhole\envtwo \mid \env\esub\var{\la\vartwo\kctx}\envtwo
\\
\textsc{States} & \state & \grameq & \env\,\rlsep\,\kctx \mid \env\,\lrsep\,\kctx
\end{array}$\end{center}
Often $\ctxhole\emptyenv$ and $\env\esub\var{\la\vartwo\kctx}\emptyenv$ are noted $\ctxhole$ and 
$\env\esub\var{\la\vartwo\kctx}$, respectively. 
Plugging inside machine contexts, of both environments and machine contexts, is defined as expected, and noted 
$\kctxp\env$ and $\kctxp\kctxtwo$. We use $\gensep$ for an unspecified separator, \ie $\gensep{\in}\,\set{\rlsep\,,\lrsep}$.
The 
idea is that a state $\env\gensep\kctx$ represent the environment $\kctxp\env$ and the active point is between $\kctx$ 
and $\env$, possibly deep inside many abstractions in $\kctxp\env$. Compilation is now defined as $\compil\tm \defeq 
\mytr\tm \,\rlsep\, \ctxhole$ for a well-named $\l$-term $\tm$, and read-back as $\unf{(\env\gensep\kctx)} \defeq 
\unf{\kctxp\env}$. A state $\env\gensep\kctx$ is well-named if $\kctxp\env$ is well-named.

\paragraph*{The Open Phase} To lift the transitions of the open machine, we have to define the environment 
$\wstenv\kctx$ induced by a machine context $\kctx$, playing the role played by $\envtwo$ in \cref{SECT:OPEN-MACHINE}.
\begin{definition}
 The \emph{environment} $\wstenv\kctx$ of $\kctx$ is given by $\wstenv{\ctxhole\envtwo}  \defeq  \envtwo$ and $
\wstenv{\env\esub\var{\la\vartwo\kctx}\envtwo}  \defeq \wstenv\kctx\envtwo$.
\end{definition}
The transitions of the 
\OCAM smoothly lift, as shown in \reffig{scam-phases} (where $(*)$ stands for ``$\rename{(\wstenv\kctx(\vartwo))} = 
\la\varfour(\esub{\varstar}\mol\envtwo) $'').
Note that the last transition applies when 1) $\mol$ is a value, or 2) $\mol$ is $\vartwo$ or $\vartwo\varthree$ but  
$\wstenv\kctx(\vartwo)$ is not a value.

\paragraph*{The Strong Phase} There are 5 new transitions, that on a state $\env\,\lrsep\,\kctx$ inspect $\kctx$ rather than 
$\env$, and accumulate in $\env$ the ES that survived the strong phase, which are now fully evaluated. The transitions 
inspect $\kctx$ from the inside, that is, by looking at what is on the right of the hole $\ctxhole$, see \reffig{scam-phases}.

The union of the nine transitions of the \SCAM is noted $\tomachscam$. Let us explain the new ones.
\begin{itemize}
 \item $\tomachctwo$ simply switches to the strong phase, when the current open phase is over.
 \item $\tomachcthree$ moves the next ES $\esub\var\mol$ of the context to the fully evaluated $\env$, if $\mol$ is not 
a value.
 \item $\tomachgc$ garbage collects $\esub\var\val$ when $\var$ does not occur in $\env$. Checking only $\env$ for 
occurrences is correct: a well-named invariant shall guarantee that $\var$ cannot occur in $\kctx$.
 \item $\tomachcfour$ handles the case in which the search has fully processed the current body $\env$ of a value, and 
thus it re-unites the enclosing abstraction $\l\vartwo$ in $\kctx$ with $\env$, adding $\la\vartwo\env$ to the fully 
evaluated environment $\envtwo$, and resumes search at the upper abstraction level.
 \item $\tomachcfive$ enters into the body $\envtwo$ of the next ES in $\kctx$, when its content is a value. Searching for values is 
temporarily over, the machine switches back to the open phase to evaluate $\envtwo$. The fully evaluated environment 
$\env$ and the enclosing abstraction $\l\vartwo$ are then moved to $\kctx$, and the body $\envtwo$ becomes the new left 
component of the state.
\end{itemize}

\section{The Strong Implementation Theorem}
\label{SECT:STRONG-IMPLEMENTATION}

The definition of the \SCAM is a smooth generalization of the \OCAM, but its implementation theorem is considerably more 
sophisticated because of deep implosive sharing. 
We overview the main concepts here, many details are in \ifthenelse{\boolean{techreport}}{\refapp{STRONG-IMPLEMENTATION-app}}{Appendix H of the tech report \cite{DBLP:journals/corr/abs-2102-06928}}. 
The obstacle is the relaxed $\beta$-projection property. 

\paragraph*{Multi Contexts} The crux is showing that the read back $\unf\kctx$ of a machine context (defined below, but 
keep reading) is an external (evaluation) context $\strongctx$, first of all because... it is not. Both strong and machine 
contexts have exactly one hole, but if $\kctx = \env\esub\var{\la\vartwo\kctxtwo}\envtwo$ then $\var$ may have many 
occurrences in $\env$, causing duplications of the hole via read-back---this is the deep implosive sharing ingredient. The first step, then, is to model 
$\unf\kctx$ as a VSC \emph{multi context}. We need generic multi contexts, plus external and rigid variants.
\begin{center}
	{
	$\begin{array}{rrclccc}
	\multicolumn{4}{c}{\textsc{Multi contexts}}
	\\

	\textsc{Generic} & \mctx & \grameq & \ctxhole \mid \var \mid \la\var \mctx \mid \mctx\esub\var\mctx \mid \mctx \mctx
	\\
	
	\textsc{External} & \strongmctx & \grameq & \ctxhole \mid \tm \mid \la\var \strongmctx \mid \rmctx \mid 
\strongmctx\esub\var\rmctx
	\\
	
	\textsc{Rigid} & \rmctx & \grameq & \var \mid \rmctx \strongmctx \mid \rmctx\esub\var\rmctx
	\end{array}$
	}
\end{center}
Multi contexts may have no holes, and thus be a term\footnote{It is easily seen that the grammar of 
$\mctx$ allows to generate all terms, and the one for $\rmctx$ all rigid terms---see \ifthenelse{\boolean{techreport}}{\reflemma{mctx-plugging} in the Appendix (page \pageref{l:mctx-plugging})}{Appendix H in the tech report \cite{DBLP:journals/corr/abs-2102-06928}}. For 
$\strongmctx$, terms are simply injected in by the grammar itself.}. A multi context is \emph{proper} if it has at least 
one hole. Plugging $\mctxp\tm$ plugs $\tm$ in all the holes of $\mctx$, erasing $\tm$ if $\mctx$ is not proper.

The next lemma shows that $\strongmctx$ is the right generalizations of $\strongctx$: for every external step $\tm\Rew{a}\tmtwo$ with $a \in \set{\esssym\msym,\esssym\esym}$, each replaced hole in $\strongmctxp\tm$ is a $\Rew{a}$ redex, and 
the firing of the redex gives the expected result $\strongmctxp\tmtwo$---similarly for $\rmctx$, \ifthenelse{\boolean{techreport}}{see \refapp{STRONG-IMPLEMENTATION-app}.}{see Appendix H in the tech report \cite{DBLP:journals/corr/abs-2102-06928}.} 
We avoid on purpose the definition of a parallel step, that would induce a more complex notion of 
implementation---parallelism here is caught more flexibly by the diamond property of $\tovsubs$. For technical reasons, we prove a more general result:
\begin{lemma}[Multi step]
\label{l:multi-step}
\NoteProof{lappendix:multi-step}
 Let $\strongmctx$ be a proper external multi context with $k$ holes and $\set{a_1, \ldots, a_n}\subseteq\set{\esssym\msym,\esssym\esym}$.\\ If $\tm 
\Rew{a_1} \cdots \Rew{a_n} \tmtwo$ then $\strongmctxp\tm\,(\Rew{a_1} \cdots \Rew{a_n})^k\, \strongmctxp\tmtwo $ where the $i$-th sequence of steps has the shape 
$\strongctx_i\ctxholep\tm \Rew{a_1} \cdots \Rew{a_n} \strongctx_i\ctxholep\tmtwo$ for an external context $\strongctx_i$, for every $i \in 
\set{1,\ldots, k}$.
\end{lemma}
Read back then extends to $\kctx$ via the following clauses (reading back  ES without the hole as before), obtaining a multi 
context: $\unf{\ctxhole} \defeq \ctxhole$ and $\unf{(\env\esub\var{\la\vartwo\kctx})}  \defeq  \unf\env 
\isub\var{\la\vartwo\unf\kctx}$. \ifthenelse{\boolean{techreport}}{Moreover, it factors in a way similar to the open case (\reflemma{read-back-decomposition-d}).
\begin{lemma}[Modular read-back]
\label{l:read-back-decomposition}
\NoteProof{lappendix:read-back-decomposition}
$\unf{\kctxp\env}=\unf\kctx\ctxholep{\unf\env\indsub{\wstenv\kctx}}$.
\end{lemma}
}{}
\paragraph*{Invariants} Let $\kctx$ be the machine context of a reachable state $\state$. Proving that $\unf\kctx$ is a 
proper external multi context $\strongmctx$ requires delicate and involved invariants, that build on those for the \OCAM. 
An essential concept is the \emph{frame} of $\kctx$, that isolates its fully evaluated part plus the hole.
\begin{center}$\begin{array}{r\colspace rcl}
\textsc{Frames} & \frame & \grameq & \ctxhole \mid \env\esub\var{\la\vartwo\frame}
\end{array}$\end{center}
 The \emph{frame $\framei\kctx$ of a context $\kctx$} is defined by $\framei{\ctxhole\envtwo}  \defeq  \ctxhole$ and 
$\framei{\env\esub\var{\la\vartwo\kctx}\envtwo}  \defeq \env\esub\var{\la\vartwo\framei\kctx}$.

Proving that $\unf\kctx$ is \emph{proper}, requires an invariant ensuring that $\framei\kctx$ is \emph{garbage-free}, so 
that every variable bound by an ES (but $\varstar$) occurs, implying that $\unf{(\env\esub\var{\la\vartwo\kctx})} = \unf\env 
\isub\var{\la\vartwo\unf\kctx}$ does not erase $\la\vartwo\unf\kctx$ and that $\unf\kctx$ is itself proper. 

Proving that $\unf\kctx$ is \emph{external}, requires a sophisticated \emph{goodness} invariant, building on the 
invariants of the \OCAM. \ifthenelse{\boolean{techreport}}{Roughly, $\unf\kctx$ is given by $\unf{\framei\kctx}$ where one applies the substitution 
$\indsub{\wstenv\kctx}$ and shuffles around the ES in $\indenv{\wstenv\kctx}$. Oversimplifying, goodness says that 
$\unf{\framei\kctx}$ is an external multi context that stays so also after the application of $\indsub{\wstenv\kctx}$ and 
the shuffling of $\indenv{\wstenv\kctx}$.}
{Roughly, $\unf\kctx$ is given by $\unf{\framei\kctx}$ where one applies the meta-level substitutions and the ES given by reading back the part of $\kctx$ not accounted by $\framei\kctx$. Oversimplifying, goodness states that 
$\unf{\framei\kctx}$ is an external multi context that stays so also after the application of these meta-level substitutions and ES.
}

\begin{theorem}[Contextual read-back]
\label{thm:ctx-read-back}
\NoteProof{thmappendix:ctx-read-back}
Let $\state = \env\gensep\kctx$ be a reachable state. Then $\unf\kctx$ is a proper external multi context.
\end{theorem}
Obtaining \refthm{ctx-read-back} is the difficult and involved step. Then, $\beta$-transitions are smoothly projected on $\tovsubs$ steps, 
and the implementation theorem easily follows.

\begin{theorem}[\SCAM{} implementation]
\label{th:machine-final}
\NoteProof{thmappendix:machine-final}~
\begin{enumerate}
\item \label{p:machine-final-projection-ms}\NoteProof{thmappendix:projection-ms} \emph{Relaxed $\beta$-projection}: let $\state$ be 
a reachable state. If $\state \tomachbv \statetwo$ then $\unf\state (\toms\toes)^+\unf\statetwo$, and if 
$\state \tomachbi \statetwo$ then $\unf\state \toms^+\eqstruct \unf\statetwo$.

\item \label{p:machine-final-impl-scam}  \NoteProof{thmappendix:impl-scam}
\emph{Strong implementation}: the \SCAM is a relaxed implementation of the external strategy $(\tovsubs,\eqstruct)$.
\end{enumerate}
\end{theorem}

%

\section{Complexity Analysis of the \SCAM{}}
\label{SECT:COMPLEXITY}

Here we prove that the \SCAM{} can be implemented within a bilinear overhead. The proof is simple and mostly follows a 
standard schema: we bound the number of overhead steps, then bound the cost of single steps, and end by combining the 
two. It all 
rests on the \emph{size invariant} below, that is a quantitative form of the subterm invariant needed for all 
complexity analyses of abstract machines.

The size $\size\tm$ of the initial term $\tm$ is one of the two parameters of the analysis, linearly preserved by 
compilation $\compil\tm = \mytr\tm \,\rlsep\, \ctxhole$. 
\begin{center}$\begin{array}{c@{\hspace{.4cm}} c@{\hspace{.4cm}} c}
		\textsc{$\l$-terms} & \textsc{Bites} & \textsc{Envs}
		\\
		\size{\var} \defeq 1 
		& 
		\size\var \defeq 1 
		&
		  \size\epsilon \defeq 0
		\\
		\size{\tm\tmtwo} \defeq \size{\tm} + \size{\tmtwo} + 1
		& 
		\size{\var\vartwo} \defeq 1
 		&
		\size{\env\esub\var\mol} \defeq 
		\\
		\size{\la{\var}{\tm}} \defeq \size{\tm} + 1
		& 
  		\size{\la\var\env} \defeq 1 + \size\env
		&
		\ \ \ \ \ \  1 + \size\env + \size\mol
			\end{array}$\end{center}
\begin{lemma}[Linear compilation] 
   \label{l:bound-measure-size} 
   \NoteProof{lappendix:bound-measure-size}
   Let $\tm$ be a \lat{}. Then $\size{\mytr\tm} \leq 2\size\tm$.
\end{lemma}

The size invariant provides a size bound on the values that may be duplicated by $\beta$-transitions, which are the only 
ones making the size of states grow.
\begin{lemma}[Size invariant]
 \label{p:invariants-quantitative}
 \NoteProof{thm:size-invariant}
 Let $\state=\env \gensep\kctx$ be a state reachable from $\state_0 = \env_0 \,\rlsep\, \ctxhole$. 
Then $\size\val \leq \size{\env_0}$ for every value $\val$ that occurs in $\env$ or $\wstenv\kctx$ if ${\gensep} = 
{\rlsep}$\, , or in $\wstenv\kctx$ if ${\gensep} = {\lrsep}$.
\end{lemma}

\paragraph*{Number of Overhead Transitions}  
We bound the number of overhead transitions in a modular way with respect to the two phases of the machine.
First, a global analysis shows that the number of transitions of all strong phases is bounded by the number of 
transitions of all open phases:

\begin{lemma}[Open phases bound strong phases]
  \label{l:strong-bound-open}
  \NoteProof{lappendix:strong-bound-open}
  Let $\exec\colon\state_0 \tomach{} \state$ an execution of the \SCAM{}. Then:
  $\size\exec_{\admsym_2} + \size\exec_{\admsym_3} + \size\exec_{\gc} + \size\exec_{\admsym_4} + \size\exec_{\admsym_5} 
\leq \size\exec_{\betain} + 4\size\exec_{\admsym_1} + 1$.
\end{lemma}

To quantify the overhead of the open phase, we introduce a new measure $\tmpMeasure\cdot$ over machine states, tailored to decrease on all open transitions but $\beta$s.
\begin{center}
\begin{tabular}{cc}
$ \arraycolsep=1.4pt\begin{array}{rclrclrcllll}
  \multicolumn{3}{c}{\textsc{Contexts}}
  \\
  \tmpMeasure{\ctxhole} & \defeq & 0 \qquad
  \\
  \tmpMeasure{\env\esub\var{\la\vartwo\kctx}} & \defeq & \tmpMeasure{\kctx} \qquad
  \\\
  \tmpMeasure{\kctx\esub\var\mol} &\defeq & \tmpMeasure{\kctx} + \size{\mol}
\end{array}$
&
$ \arraycolsep=1.4pt\begin{array}{rclrclrcllll}
\multicolumn{3}{c}{\textsc{States}}
\\
  \tmpMeasure{\env ~\rlsep~ \kctx} &\defeq & \size\env + \tmpMeasure\kctx \quad
  \\
  \tmpMeasure{\env ~\lrsep~ \kctx} &\defeq &\tmpMeasure\kctx
\end{array}$
\end{tabular}
\end{center}

\begin{lemma}[Measure during execution]
  \label{l:meas-during-exec}
  \NoteProof{lappendix:meas-during-exec}
  Let $\state$ be a state reachable from $\state_0$, and $\state \tomachhole a \statetwo$.
  \begin{enumerate}
     \item Beta transitions increase the measure: if $a \in\set{\betaabs, \betain}$ then $\tmpMeasure\statetwo \leq 
\tmpMeasure\state + \tmpMeasure{\state_0}$.
     \item Open overhead decreases the measure: if $a \in\set{\rensym, \admsym_1}$ then $\tmpMeasure\statetwo < 
\tmpMeasure\state$.
     \item Strong phase does not increase the measure: if $a \in\set{\admsym_2, \admsym_3, \gc, \admsym_4, \admsym_5 }$ 
then $\tmpMeasure\statetwo \leq \tmpMeasure\state$.
  \end{enumerate}
\end{lemma}

By combining the two previous lemmas, we obtain a bilinear bound on the total overhead:
 \begin{corollary}[Bilinear number of overhead transitions]
    \label{c:bound-com}
    \NoteProof{cappendix:bound-com}
    Let $\tm$ be a \lat{} and $\exec\colon \compil\tm \tomachscam^* \state$ be a  \SCAM{} execution. Then $\size\exec 
\in \bigo((1 + 
\sizebeta\exec) \size\tm)$.
 \end{corollary}

\paragraph*{Cost of Single Steps} We need high-level assumptions on how bites and crumbled 
environments are concretely implemented---a reference implementation in OCaml is  
explained in \ifthenelse{\boolean{techreport}}{Appendix \ref{sect:ocaml}, p. \pageref{sect:ocaml}}{Appendix J of the tech report \cite{DBLP:journals/corr/abs-2102-06928}} 
 and can be
downloaded at \url{https://github.com/sacerdot/SCAM}.

As it is standard for machines with global environments (see Accattoli and Barras \cite{DBLP:conf/ppdp/AccattoliB17}), variables are represented as memory locations, variable 
occurrences as pointers, and an ES $\esub\var\mol$ is the fact that the location associated with $\var$ contains 
$\mol$---this allows $\bigo(1)$ look-up in environments. The copy/renaming in 
$\beta$-transitions costs $\bigo(\size\tm)$ by the size invariant, following  \cite{DBLP:conf/ppdp/AccattoliB17}, 
essentially implementing the proof-nets representation of a term, that can be seen as a pointer-based DAG.

The \SCAM{} visits the proof-net/DAG in bi-directional ways: environments are visited both right-to-left (open phases) 
and left-to-right (strong phases), and abstractions are entered/exited at phase switches. Rather than 
having doubly linked nodes we adapt to our more general DAG framework the subtler space-conscious technique by McBride 
\cite{Mcbride01thederivative}, itself generalizing the standard zipper technique for lists by Huet~\cite{DBLP:journals/jfp/Huet97}, 
obtaining bi-directional moves over the graph in $\bigo(1)$.

To implement $\tomachsub$ in $\bigo(1)$ we need to jump to the occurrence of the renamed variable. By \ifthenelse{\boolean{techreport}}{\reflemma{aux-most-once}}{one of the invariants of the \SCAM{} (namely, the pristine one, see the tech report \cite{DBLP:journals/corr/abs-2102-06928})}, the renamed
variable has at most one occurrence when the rule fires. An implementation can thus keep with no overhead a 
bi-directional link between the variable and its occurrence --- as long as the variable occurs once.

For $\tomachgc$, variables also carry a reference counter, to test 
if they occur in $\bigo(1)$. By the size invariant, the size of erased values is $\bigo(\size\tm)$, but since there can 
be $\bigo((1 + \sizebeta\exec) \size\tm)$ $\tomachgc$ steps (\refcorollary{bound-com}), we obtain a bound quadratic in 
$\size\tm$, which is too loose.
The stricter bilinear bound is inferred via global analysis.
Indeed, by the size invariant the size of a state is bounded by $\bigo((1 + \sizebeta\exec) \size\tm)$. 
So, the global cost of erasing cannot be more than $\bigo((1 + \sizebeta\exec) \size\tm)$. Summing it all up,

\begin{theorem}[The \SCAM{} is bilinear]
   \label{thm:bilinear-scam}
   \NoteProof{c:bilinear-scam}
   Let $\tm$ be a $\l$-term and  $\exec\colon \compil\tm \tomachscam^* \state$ a \SCAM{} execution. Then $\exec$ can be 
implemented on a RAM in $\bigo((1 + \sizebeta\exec)\size\tm)$.
\end{theorem}
Via the implementation theorem (\refth{machine-final}), we obtain reasonable cost models for \scbv.
\begin{corollary}[Reasonable cost models]
Let $\tm$ be a $\l$-term, $\deriv: \tm \tovsubs^* \tmtwo$, and $\exec\colon \compil\tm 
\tomachscam^* \state$. Then $\sizem\deriv$ and $\sizem\exec$ are reasonable time cost models for Strong \cbv.
\end{corollary}

\section{Implosiveness at Work}
\label{SECT:COMPLEMENTS}

Let us give an example of the implosive phenomenon. Consider the following 
families of terms $\tm_1 \defeq \pi I$ and $\tm_{n+1} \defeq \pi (\la\varthree \tm_n)$, where $\pi \defeq 
\la\var\la\vartwo ((\vartwo \var)\var)$, and $\tmtwo_1 \defeq \la\vartwo((\vartwo I) I)$ and $\tmtwo_{n+1} \defeq 
\la\vartwo((\vartwo (\la\varthree\tmtwo_n)) (\la\varthree\tmtwo_n))$.

An easy induction shows that $\tovsubs$ takes on $\tm_n$ a number of steps exponential in $n$. Note that 
1) all the 
$\beta$-redexes of $\tm_n$ are given by a copy of $\pi$ with just one argument, which is a value, 2) the argument is duplicated, and that 
3) these facts are stable by evaluation. Therefore, substitution always happens on occurrences of $\var$ as arguments. 
The \SCAM{} never substitutes on arguments, and evaluates them while they are shared, thus avoiding the exponential 
duplication of redexes.
\begin{proposition}[Implosive family]
\label{prop:implosive-family}
\NoteProof{propappendix:implosive-family}
Let $\tm_n$ and $\tmtwo_n$ as above.
\begin{enumerate}
	\item \emph{External strategy (exponentially many steps)}: $\tm_n (\toms\toes)^{2^n-1} \tmtwo_n$ and $\tmtwo_n$ is a strong fireball.
\item \emph{\SCAM{} Implosion (linearly many steps)}: $\exec_n: \compil{\tm_n} \tomachscam^* \state_n$ with $\unf{\state_n} = \tmtwo_n$ and $\sizebeta{\exec_n} = n$.
\end{enumerate}
\end{proposition}

\noindent\textbf{Acknowledgements.}
This work has been partially funded by the ANR JCJC grant
  'COCA HOLA' (ANR-16-CE40-004-01); the third author by the MIUR-PRIN project 'Analysis of Program Analyses' (ASPRA, ID 201784YSZ5 004) and the Italian INdAM -- GNCS
project 2020 'Reversible Concurrent Systems: from Models to 
Languages'.


%
%
\IEEEpeerreviewmaketitle


\bibliographystyle{IEEEtranS}
\bibliography{main.bbl}

\ifthenelse{\boolean{techreport}}{
\clearpage
\appendices
\onecolumn
\section*{\Large Technical Appendix}
\label{SECT:APPENDIX}

\bigskip

{\newcommand\tatocitem[3]{\item[\ref{#1}] #2 \dotfill \pageref{#1} #3} 
\begin{enumerate}
  \tatocitem{sect:preliminaries}{Preliminaries and Notations in Rewriting}{}
  \tatocitem{app:calculus-proofs}{Proofs of Section~\ref*{sect:calculus} (VSC)}{}
  \tatocitem{app:external-strategy-proofs}{Proofs of Section~\ref*{sect:external-strategy} (External Strategy)}{}
  \tatocitem{app:relaxed-implementation}{Proofs of Section~\ref*{SECT:RELAXED} (Relaxed Implementation)}{}
  \tatocitem{app:compilation-read-back}{Proofs of Section~\ref*{SECT:CRUMBLING} (Compilation \& Read-back)}{}
  \tatocitem{app:ocam}{Proofs of Section~\ref*{SECT:OPEN-MACHINE} (Open Crumbling Machine)}{}
  \tatocitem{app:scam}{Proofs of Section~\ref*{SECT:STRONG-MACHINE} (Strong Crumbling Machine)}{}
  \tatocitem{app:STRONG-IMPLEMENTATION-app}{Proofs of Section~\ref*{SECT:STRONG-IMPLEMENTATION} (Strong Implementation Theorem)}{}
  \begin{enumerate}
    \item[] Well-named invariant \dotfill \pageref{thm:named-invariant}
    \item[] Pristine invariant \dotfill \pageref{thm:pristine-invariant-app}
    \item[] Well-crumbled invariant \dotfill \pageref{thm:well-crumbled-invariant-app}
    \item[] Garbage-free invariant \dotfill \pageref{thm:garbage-free-invariant-app}
    \item[] Goodness invariant \dotfill \pageref{thm:goodness-invariants}
  \end{enumerate}
  \tatocitem{app:complexity}{Proofs of Section~\ref*{SECT:COMPLEXITY} (Complexity)}{}
  \tatocitem{sect:ocaml}{Implementation in OCaml}{}
  \tatocitem{app:complements}{Proofs of Section~\ref*{SECT:COMPLEMENTS} (Implosiveness at Work)}{}
\end{enumerate}
}

\bigskip

\section{Preliminaries and Notations in Rewriting}
\label{sect:preliminaries}
For a relation $R$ on a set of terms, $R^*$ is its reflexive-transitive closure. 
Given a relation $\Rew{\Rule}$, an $\Rule$-\emph{evaluation}
(or simply evaluation if unambiguous) $\deriv$ is a finite sequence of terms $(\tm_i)_{0 \leq i \leq n}$ (for some $n \geq 0$) such that $\tm_i \Rew{\Rule} \tm_{i+1}$ for all $1 \leq i < n$, and we write $\deriv \colon \tm \Rew{\Rule}^* \tmtwo$ if $\tm_0 = \tm$ and $\tm_n = \tmtwo$. The
\emph{length} $n$ of $\deriv$ is denoted by $\size{\deriv}$, and $\size{\deriv}_a$ is the number of $a$-\emph{steps} (\ie the number of $\tm_i \Rew{a} \tm_{i+1}$ for some $1 \leq i \leq n$) in $\deriv$, for a given subrelation $\Rew{a}$ of $\Rew{\Rule}$.

A term $\tm$ is $\Rule$-\emph{normal} if there is no $\tmtwo$ such that $\tm \Rew{\Rule} \tmtwo$.
An evaluation $\deriv \colon \tm \Rew{\Rule}^* \tmtwo$ is \emph{$\Rule$-normalizing} if $\tmtwo$ is $\Rule$-normal.
A term $\tm$ is \emph{weakly $\Rule$-normalizing} if there is a $\Rule$-normalizing evaluation $\deriv \colon \tm \Rew{\Rule}^* \tmtwo$; and $\tm$ is \emph{strongly $\Rule$-normalizing} if there no infinite sequence $(\tm_i)_{i \in \nat}$  such that $\tm_0  = \tm$ and $\tm_i \Rew{\Rule} \tm_{i+1}$ for all $i \in \nat$.
Clearly, strong $\Rule$-normalization implies weak $\Rule$-normalization.

A relation $\Rew{\Rule}$ is \emph{diamond} if $\tmtwo_1 \,{}_\Rule\!\!\lto \tm \Rew{\Rule} \tmtwo_2$ and $\tmtwo_1 \neq \tmtwo_2$ imply $\tmtwo_1 \Rew{\Rule} \tmthree \, {}_\Rule\!\!\lto \tmtwo_2$ for some $\tmthree$. 
As a consequence:
\begin{enumerate}
	\item $\Rew{\Rule}$ is confluent (\ie $\tmtwo_1 \,{}_\Rule^*\!\!\lto \tm \Rew{\Rule}^* \tmtwo_2$  implies $\tmtwo_1 \Rew{\Rule}^* \tmthree \, {}_\Rule^*\!\!\lto \tmtwo_2$ for some $\tmthree$); \item any term $\tm$ has at most one normal form (\ie if $\tm \Rew{\Rule}^* \tmtwo$  and $\tm \Rew{\Rule}^* \tmthree$ with $\tmtwo$ and $\tmthree$ $\Rule$-normal, then $\tmtwo = \tmthree$);
	\item all $\Rule$-evaluations with the same start and end terms have the same length (\ie if $\deriv \colon \tm \Rew{\Rule}^* \tmtwo$  and $\deriv' \colon \tm \Rew{\Rule}^* \tmtwo$ then $\size{\deriv} = \size{\deriv'}$);
	\item $\tm$ is weakly $\Rule$-normalizing iff it is strongly $\Rule$-normalizing.
\end{enumerate}

Two relations $\Rew{\Rule_1}$ and $\Rew{\Rule_2}$ \emph{strongly commute} if $\tmtwo_1 \,{}_{\Rule_1}\!\!\!\lto \tm \Rew{\Rule_2} \tmtwo_2$ implies $\tmtwo_1 \Rew{\Rule_2} \tmthree \, {}_{\Rule_2}\!\!\!\lto \tmtwo_2$ for some $\tmthree$. 
If $\Rew{\Rule_1}$ and $\Rew{\Rule_2}$ strongly commute and are diamond, then 
\begin{enumerate}
	\item $\Rew{\Rule} \, = \, \Rew{\Rule_1} \!\cup \Rew{\Rule_2}$ is diamond,
	\item all $\Rule$-evaluations with the same start and end terms have the same number of any kind of steps (\ie if $\deriv \colon \tm \Rew{\Rule}^* \tmtwo$  and $\deriv' \colon \tm \Rew{\Rule}^* \tmtwo$ then $\size{\deriv}_{\Rule_1} = \size{\deriv'}_{\Rule_1}$ and $\size{\deriv}_{\Rule_2} = \size{\deriv'}_{\Rule_2}$).
\end{enumerate}

It is a strong form of confluence and implies \emph{uniform normalization} (if there is a normalizing sequence from $\tm$ then there are no diverging sequences from $\tm$) and the 
\emph{random descent property} (all normalizing sequences from $\tm$ have the same length)

\section{Proofs of Section~\ref*{sect:calculus} (VSC)}
\label{app:calculus-proofs}

\begin{lemma}[Shape of strong fireballs]
	\label{l:shape-of-strong-fireballs}
	Let $\tm$ be a strong fireball. Then exactly one of the following holds:
	\begin{itemize}
		\item either $\tm$ is a strong inert term,
		
		\item or $\tm$ is a value fireball.
		
	\end{itemize}
\end{lemma}

\begin{proof}
	Proving that at least one of the two holds is left for the reader. 
	We now prove that only one of them holds:
	
	\begin{itemize}
		\item Let $\tm$ be a strong inert term. We prove that $\tm$ is not a value fireball by structural induction on $\tm$:
		\begin{itemize}
			\item \emph{Variable}: Trivial.
			
			\item \emph{Application}: Trivial.
			
			\item \emph{$\mathsf{ES}$}; \ie, $\tm = \sitm \esub{\var}{\sitmtwo}$: Then $\sitm$ is not a value fireball ---by \ih---, and so neither is $\tm$.
			
		\end{itemize}
		
		\item Let $\tm = \sctxp{\la{\var}{\sfire}}$, with $\sctx = \esub{\var_{1}}{\itm_{s,1}} \dots \esub{\var_{n}}{\itm_{s,n}}$, with $n \geq 0$. We prove that $\tm$ is not a strong inert term by induction on $n$:
		\begin{itemize}
			\item \emph{Empty substitution context}; \ie, $\sctx = \ctxhole$: Trivial.
			
			\item \emph{Non-empty substitution context}; \ie, $\sctx = \sctxtwo \esub{\var}{\itm_{s,n+1}}$: Since $\sctxtwop{\la{\var}{\fire}}$ is not a strong inert term ---by \ih---, then neither is $\sctxtwop{\la{\var}{\fire}} \esub{\var}{\itm_{s,n+1}} = \tm$.
			\qedhere
		\end{itemize}
		
	\end{itemize}
\end{proof}

\begin{proposition}[Characterization  of  normal  forms]
  \label{lappendix:harmony}
  \NoteState{l:harmony}
  Let $\tm$ be a VSC term.
 $\tm$ is $\tovsub$-normal if and only if $\tm$ is a strong fireball. 
\end{proposition}

\begin{proof}
	We prove the two directions of the equivalence separately.
		\begin{enumerate}
		\item \label{p:asfa} Let $\tm$ be $\vsub$-normal. We shall prove that $\tm$ is a \full fireball by induction on $\tm$:
		\begin{itemize}
			\item \emph{Variable}. Trivial.
			
			\item \emph{Abstraction}; \ie, $\tm = \la{\var}{\tmtwo}$. 
			As $\tm$ is $\tovsub$-normal, so is $\tmtwo$.
			By \ih, $\tmtwo$ is a \full fireball, and then so~is~$\tm$.
			
			\item \emph{Application}; \ie, $\tm = \tm_{1} \tm_{2}$. 
			Since $\tm$ is $\tovsub$-normal,  so are $\tm_1$ and $\tm_2$.
			By \ih, $\tm_{1}$ and $\tm_{2}$ are \full fireballs. 
			Note that  $\tm_{1} \neq \sctxp{\la{\var}{\tmtwo}}$, otherwise $\tm \rtom \sctxp{\tmtwo \esub{\var}{\tm_{2}}}$ which contradicts $\vsub$-normality of $\tm$. 
			So, $\tm_{1}$ is a \full inert term (by \Cref{l:shape-of-strong-fireballs}), thus $\tm$ is a \full fireball.
			
			\item \emph{Explicit substitution}; \ie, $\tm = \tm_{1} \esub{\var}{\tm_{2}}$. 
			Since $\tm$ is $\vsub$-normal, then so are $\tm_1$ and $\tm_2$.
			By \ih, $\tm_{1}$ and $\tm_{2}$ are \full fireballs. 
			Note that $\tm_{2} \neq \sctxp{\val}$, otherwise $\tm \rtoe \sctxp{\tm_{1} \isub{\var}{\val}}$ which contradicts the $\vsub$-normality of $\tm$. 
			Thus $\tm_{2}$ is a \full inert term (by \Cref{l:shape-of-strong-fireballs}), and so $\tm$ is a \full fireball.
			
		\end{itemize}
		
		\item\label{p:fireball-to-normal} Let $\tm$ be a \full fireball. 
		We prove that $\tm$ is $\vsub$-normal by induction on the definition of \full fireball.
		\begin{itemize}
			\item \emph{Variable}. Trivial.
			
			\item \emph{Abstraction}; \ie, $\tm = \la{\var}{\sfire}$. 
			By \ih, $\sfire$ is $\vsub$-normal, and hence so is $\tm$.
			
			\item \emph{Application}; \ie, $\tm = \sitm \sfire$.
			By \ih, $\sitm $ and $\sfire$  are $\vsub$-normal. 
			Since $\sitm$ is not of the of the form $\sctxp{\la{\var}{\tmtwo}}$ (by \Cref{l:shape-of-strong-fireballs}), then $\tm$ is also $\vsub$-normal.
			
			\item \emph{Explicit substitution}; \ie, $\tm = \sfire \esub{\var}{\sitm}$ (it includes the case when $\sfire$ is a \full inert term). 
			By \ih, both $\sfire$ and $\sitm$ are $\vsub$-normal. 
			Since $\sitm$ is not of the form $\sctxp{\val}$ (by \reflemma{shape-of-strong-fireballs}), then $\tm$ is also $\vsub$-normal.
			\qedhere
		\end{itemize}
	\end{enumerate}
\end{proof}

\begin{lemma}[Basic Properties of $\vsubcalc$]
	\label{l:basic-value-substitution}
	\hfill
	\begin{enumerate}
		\item\label{p:basic-value-substitution-tom-toe-terminates} $\tom$ and $\toe$ are strongly normalizing (separately).
		\item\label{p:basic-value-substitution-tom-toe-diamond-open} $\tomo$ and $\toeo$ are diamond  (separately).
		
		\item\label{p:basic-value-substitution-tom-toe-commute-open}  $\tomo$ and $\toeo$ strongly commute.
	\end{enumerate}
\end{lemma}

\begin{proof}
	The statements of \reflemma{basic-value-substitution} are a refinement of some results proved in \cite{AccattoliPaolini12}, where $\tovsubo$ is denoted by $\to_\mathsf{w}$.
	\begin{enumerate}
		\item See \cite[Lemma~3]{AccattoliPaolini12}.
		
		\item We prove that $\tomo$ is diamond, \ie if $\tmtwo \lRew{\wmsym} \tm \tomo \tmthree$ with $\tmtwo \neq \tmthree$ then there exists $\tmp \in \Lambda_\vsub$ such that $\tmtwo \tomo \tmp \lRew{\wmsym} \tmthree$.
		The proof is by induction on the definition of $\tomo$. 
		Since there $\tm \tomo \tmthree \neq \tmtwo$ and the reduction $\tomo$ is weak, there are only eight cases:
		\begin{itemize}
			\item \emph{Step at the Root for $\tm \!\tomo\! \tmtwo$ and Application Right for $\tm \!\tomo\! \tmthree$}, \ie $\tm \defeq \sctxp{\la\var\tmfive}\tmfour \rtom \sctxp{\tmfive\esub{\var}{\tmfour}} \eqdef \tmtwo$ and $\tm \!\rtom\! \sctxp{\la\var\tmfive}\tmfourp\! \eqdef \tmthree$ with $\tmfour \!\tomo\! \tmfourp$: then, $\tmtwo \!\tomo\! \sctxp{\tmfive\esub{\var}{\tmfourp}} \!\lRew{\wmsym}\! \tmthree$;
			
			\item \emph{Step at the Root for $\tm \tomo \tmtwo$ and Application Left for $\tm \tomo \tmthree$}, \ie, for some $n > 0$, 
			$$
			\tm \defeq (\la\var\tmfive)\esub{\var_1}{\tm_1}\dots\esub{\var_n}{\tm_n}\tmfour \allowbreak\rtom \tmfive\esub{\var}{\tmfour}\esub{\var_1}{\tm_1}\dots\esub{\var_n}{\tm_n} \eqdef \tmtwo
			$$
			whereas $\tm \tomo \allowbreak (\la\var\tmfive)\esub{\var_1}{\tm_1}\dots\esub{\var_j}{\tmp_j}\dots\esub{\var_n}{\tm_n}\tmfour \eqdef \tmthree$ with $\tm_j \tomo \tmp_j$ for some $1 \leq j \leq n$: then, 
			\begin{align*}
			\tmtwo \tomo \allowbreak \tmfive\esub{\var}{\tmfour}\esub{\var_1}{\tm_1}\dots\esub{\var_j}{\tmp_j}\dots\esub{\var_n}{\tm_n} \lRew{\wmsym} \tmthree;
			\end{align*}
			\item \emph{Application Left for $\tm \tomo \tmtwo$ and Application Right for $\tm \tomo \tmthree$}, \ie $\tm \defeq \tmfour\tmfive \tomo \tmfourp\tmfive \eqdef \tmtwo$ and $\tm \tomo \tmfour\tmfivep \eqdef \tmthree$ with $\tmfour \tomo \tmfourp$ and $\tmfive \tomo \tmfivep$: then, $\tmtwo \tomo \tmfourp\tmfivep\! \lRew{\wmsym} \tmthree$;
			\item \emph{Application Left for both $\tm \tomo \tmtwo$ and $\tm \tomo \tmthree$}, \ie $\tm \defeq \tmfour\tmfive \tomo \tmfourp\tmfive \eqdef \tmtwo$ and $\tm \tomo \tmfour''\tmfive \eqdef \tmthree$ with $\tmfourp \lRew{\wmsym} \tmfour \tomo \tmfour''$: by \ih, there exists $\tmfour_0 \in \Lambda_\vsub$ such that $\tmfourp \tomo \tmfour_0 \lRew{\msym} \tmfour''$, hence $\tmtwo \tomo \tmfour_0\tmfive \lRew{\msym} \tmthree$;
			\item \emph{Application Right for both $\tm \tomo \tmtwo$ and $\tm \tomo \tmthree$}, \ie $\tm \defeq \tmfive\tmfour \tomo \tmfive\tmfourp \eqdef \tmtwo$ and $\tm \tomo \tmfive\tmfour'' \eqdef \tmthree$ with $\tmfourp \lRew{\wmsym} \tmfour \tomo \tmfour''$: by \ih, there exists $\tmfour_0 \in \Lambda_\vsub$ such that $\tmfourp \tomo \tmfour_0 \lRew{\wmsym} \tmfour''$, hence $\tmtwo \tomo \tmfive\tmfour_0 \lRew{\wmsym} \tmthree$;
			\item \emph{$\mathsf{ES}$ left for $\tm \tomo \tmtwo$ and $\mathsf{ES}$ right for $\tm \tomo \tmthree$}, \ie $\tm \defeq \tmfour\esub\var\tmfive \tomo \tmfourp\esub\var\tmfive \eqdef \tmtwo$ and $\tm \tomo \tmfour\esub\var\tmfivep \eqdef \tmthree$ with $\tmfour \tomo \tmfourp$ and $\tmfive \tomo \tmfivep$: then, 
			$$
			\tmtwo \tomo \tmfourp\esub\var\tmfivep\! \lRew{\wmsym} \tmthree
			$$
			\item \emph{$\mathsf{ES}$ left for both $\tm \tomo \tmtwo$ and $\tm \tomo \tmthree$}, \ie $\tm \defeq \tmfour\esub\var\tmfive \tomo \tmfourp\esub\var\tmfive \eqdef \tmtwo$ and $\tm \tomo \tmfour''\esub\var\tmfive \eqdef \tmthree$ with $\tmfourp \lRew{\wmsym} \tmfour \tomo \tmfour''$: by \ih, there exists $\tmfour_0 \in \Lambda_\vsub$ such that $\tmfourp \tomo \tmfour_0 \lRew{\wmsym} \tmfour''$, hence $\tmtwo \tom \tmfour_0\esub\var\tmfive \lRew{\wmsym} \tmthree$;
			\item \emph{$\mathsf{ES}$ right for both $\tm \tomo \tmtwo$ and $\tm \tomo \tmthree$}, \ie $\tm \defeq \tmfive\esub\var\tmfour \tomo \tmfive\esub\var\tmfourp \eqdef \tmtwo$ and $\tm \tomo \tmfive\esub\var{\tmfour''} \eqdef \tmthree$ with $\tmfourp \lRew{\wmsym} \tmfour \tom \tmfour''$: by \ih, there exists $\tmfour_0 \in \Lambda_\vsub$ such that $\tmfourp \tomo \tmfour_0 \lRew{\wmsym} \tmfour''$, hence $\tmtwo \tom \tmfive\esub\var{\tmfour_0} \lRew{\wmsym} \tmthree$.
		\end{itemize}
		
		\smallskip
		We prove that $\toeo$ is diamond, \ie if $\tmtwo \lRew{\wesym} \tm \toeo \tmthree$ with $\tmtwo \neq \tmthree$ then there exists $\tmfour \in \Lambda_\vsub$ such that $\tmtwo \toeo \tmp \lRew{\esym} \tmthree$.
		The proof is by induction on the definition of $\toeo$. 
		Since there $\tm \toeo \tmthree \neq \tmtwo$ and the reduction $\toeo$ is weak, there are only eight cases:
		\begin{itemize}
			\item \emph{Step at the Root for $\tm \!\toeo\! \tmtwo$} and \emph{$\mathsf{ES}$ left for $\tm \!\toeo\! \tmthree$}, \ie $\tm \defeq \tmfour\esub\var{\sctxp{\val}} \rtoe \sctxp{\tmfour\isub{\var}{\val}} \eqdef \tmtwo$ and $\tm \!\rtoe\! \tmfourp\esub\var{\sctxp{\val}}\! \eqdef \tmthree$ with $\tmfour \!\toeo\! \tmfourp$: then, 
			$$
			\tmtwo \!\toeo\! \sctxp{\tmfourp\esub{\var}{\val}} \!\lRew{\wesym}\! \tmthree
			$$
			
			\item \emph{Step at the Root for $\tm \toeo \tmtwo$} and \emph{$\mathsf{ES}$ right for $\tm \toeo \tmthree$}, \ie, for some $n > 0$, $\tm \defeq \tmfour\esub\var{\val\esub{\var_1}{\tm_1}\dots\esub{\var_n}{\tm_n}} \allowbreak\rtoe \tmfour\isub{\var}{\val}\esub{\var_1}{\tm_1}\dots\esub{\var_n}{\tm_n} \eqdef \tmtwo$ whereas $\tm \toeo \allowbreak \tmfour\esub{\var}{\val\esub{\var_1}{\tm_1}\dots\esub{\var_j}{\tmp_j}\dots\esub{\var_n}{\tm_n}} \eqdef \tmthree$ with $\tm_j \toeo \tmp_j$ for some $1 \leq j \leq n$: then, 
			\begin{align*}
			\tmtwo \toeo \allowbreak \tmfour\isub{\var}{\val}\esub{\var_1}{\tm_1}\dots\esub{\var_j}{\tmp_j}\dots\esub{\var_n}{\tm_n} \lRew{\wesym} \tmthree;
			\end{align*}
			\item \emph{Application Left for $\tm \toeo \tmtwo$} and \emph{Application Right for $\tm \toeo \tmthree$}, \ie $\tm \defeq \tmfour\tmfive \toeo \tmfourp\tmfive \eqdef \tmtwo$ and $\tm \toeo \tmfour\tmfivep \eqdef \tmthree$ with $\tmfour \toeo \tmfourp$ and $\tmfive \toeo \tmfivep$: then, $\tmtwo \toeo \tmfourp\tmfivep\! \lRew{\wesym} \tmthree$;
			\item \emph{Application Left for both $\tm \toeo \tmtwo$ and $\tm \toeo \tmthree$}, \ie $\tm \defeq \tmfour\tmfive \toeo \tmfourp\tmfive \eqdef \tmtwo$ and $\tm \toeo \tmfour''\tmfive \eqdef \tmthree$ with $\tmfourp \lRew{\wesym} \tmfour \toeo \tmfour''$: by \ih, there exists $\tmfour_0 \in \Lambda_\vsub$ such that $\tmfourp \toeo \tmfour_0 \lRew{\wesym} \tmfour''$, hence $\tmtwo \toeo \tmfour_0\tmfive \lRew{\wesym} \tmthree$;
			\item \emph{Application Right for both $\tm \toeo \tmtwo$ and $\tm \toeo \tmthree$}, \ie $\tm \defeq \tmfive\tmfour \toeo \tmfive\tmfourp \eqdef \tmtwo$ and $\tm \toeo \tmfive\tmfour'' \eqdef \tmthree$ with $\tmfourp \lRew{\wesym} \tmfour \toeo \tmfour''$: by \ih, there exists $\tmfour_0 \in \Lambda_\vsub$ such that $\tmfourp \toeo \tmfour_0 \lRew{\wesym} \tmfour''$, hence $\tmtwo \toeo \tmfive\tmfour_0 \lRew{\wesym} \tmthree$;
			\item \emph{$\mathsf{ES}$ left for $\tm \toeo \tmtwo$} and \emph{$\mathsf{ES}$ right for $\tm \toeo \tmthree$}, \ie $\tm \defeq \tmfour\esub\var\tmfive \toeo \tmfourp\esub\var\tmfive \eqdef \tmtwo$ and $\tm \toeo \tmfour\esub\var\tmfivep \eqdef \tmthree$ with $\tmfour \toeo \tmfourp$ and $\tmfive \toeo \tmfivep$: then, $\tmtwo \toeo \tmfourp\esub\var\tmfivep\! \lRew{\wesym} \tmthree$;
			\item \emph{$\mathsf{ES}$ left for both $\tm \toe \tmtwo$ and $\tm \toe \tmthree$}, \ie $\tm \defeq \tmfour\esub\var\tmfive \toe \tmfourp\esub\var\tmfive \eqdef \tmtwo$ and $\tm \toe \tmfour''\esub\var\tmfive \eqdef \tmthree$ with $\tmfourp \lRew{\esym} \tmfour \toe \tmfour''$: by \ih, there exists $\tmfour_0 \in \Lambda_\vsub$ such that $\tmfourp \toe \tmfour_0 \lRew{\esym} \tmfour''$, hence $\tmtwo \toe \tmfour_0\esub\var\tmfive \lRew{\esym} \tmthree$;
			\item \emph{$\mathsf{ES}$ right for both $\tm \toe \tmtwo$ and $\tm \toe \tmthree$}, \ie $\tm \defeq \tmfive\esub\var\tmfour \toe \tmfive\esub\var\tmfourp \eqdef \tmtwo$ and $\tm \toe \tmfive\esub\var{\tmfour''} \eqdef \tmthree$ with $\tmfourp \lRew{\esym} \tmfour \toe \tmfour''$: by \ih, there exists $\tmfour_0 \in \Lambda_\vsub$ such that $\tmfourp \toe \tmfour_0 \lRew{\esym} \tmfour''$, hence $\tmtwo \toe \tmfive\esub\var{\tmfour_0} \lRew{\esym} \tmthree$.
		\end{itemize}
		
		\smallskip
		
		Note that in \cite[Lemma~11]{AccattoliPaolini12} it has just been proved the strong confluence of $\tovsub$, not of $\tom$ or $\toe$.
		
		\item We show that $\toeo$ and $\tomo$ strongly commute, \ie if $\tmtwo \lRew{\wesym} \tm \tomo \tmthree$, then $\tmtwo \neq \tmthree$ and there is $\tmp \in \Lambda_\vsub$ such that $\tmtwo \tomo \tmp \lRew{\wesym} \tmthree$. 
		The proof is by induction on the definition of $\tm \toeo \tmtwo$. 
		The proof that $\tmtwo \neq \tmthree$ is left to the reader.
		Since the $\toe$ and $\tom$ cannot reduce under $\l$'s, all values are $\omsym$-normal and $\oesym$-normal. So, there are the following cases.
		\begin{itemize}
			\item \emph{Step at the Root for $\tm \toeo \tmtwo$} and \emph{$\mathsf{ES}$ left for $\tm \tomo \tmthree$}, \ie $\tm \defeq \tmfour\esub\varthree{\sctxp\val} \toeo \sctxp{\tmfour\isub\varthree{\val}} \eqdef \tmtwo$ and $\tm \tomo \tmfourp\esub\varthree{\sctxp{\val}} \eqdef \tmthree$ with $\tmfour \tomo \tmfourp$: then 
			$$
			\tmtwo \tomo \sctxp{\tmfourp\isub\varthree{\val}} \lRew{\wesym} \tmtwo
			$$
			\item \emph{Step at the Root for $\tm \toeo \tmtwo$} and \emph{$\mathsf{ES}$ right for $\tm \tomo \tmthree$}, \ie 
			\begin{align*}
			\tm &\defeq \tmfour\esub\varthree{\val\esub{\var_1}{\tm_1}\dots\esub{\var_n}{\tm_n}} \\
			&\toeo \tmfour\isub\varthree{\val}\esub{\var_1}{\tm_1}\dots\esub{\var_n}{\tm_n} \eqdef \tmtwo
			\end{align*}
			and $\tm \tomo \tmfour\esub\varthree{\val\esub{\var_1}{\tm_1}\dots\esub{\var_j}{\tmp_j}\dots\esub{\var_n}{\tm_n}} \eqdef \tmthree$
			for some $n > 0$, and $\tm_j \tomo \tmp_j$ for some $1 \leq j \leq n$: then, $\tmtwo \tomo \tmfour\isub\varthree{\val}\esub{\var_1}{\tm_1}\dots\esub{\var_j}{\tmp_j}\dots\esub{\var_n}{\tm_n} \lRew{\wesym} \tmthree$;
			
			\item \emph{Application Left for $\tm \toe \tmtwo$} and \emph{Application Right for $\tm \tom \tmthree$}, \ie $\tm \defeq \tmfour\tmfive \toeo \tmfourp\tmfive \eqdef \tmtwo$ and $\tm \tomo \tmfour\tmfivep \eqdef \tmthree$ with $\tmfour \toe \tmfourp$ and $\tmfive \tomo \tmfivep$: then, $\tm \tomo \tmfourp\tmfivep \lRew{\wesym} \tmtwo$;
			\item \emph{Application Left for both $\tm \toeo \tmtwo$ and $\tm \tomo \tmthree$}, \ie $\tm \defeq \tmfour\tmfive \toeo \tmfourp\tmfive \eqdef \tmtwo$ and $\tm \tomo \tmfour''\tmfive \eqdef \tmthree$ with $\tmfourp \lRew{\wesym} \tmfour \tomo \tmfour''$: by \ih, there exists $\tmsix \in \Lambda_\vsub$ such that $\tmfourp \tomo \tmsix \lRew{\wesym} \tmfour''$, hence $\tmtwo \tomo \tmsix\tmfive \lRew{\wesym} \tmthree$;
			\item \emph{Application Left for $\tm \toeo \tmtwo$} and \emph{Step at the Root for $\tm \tomo \tmthree$}, \ie $\tm \defeq (\la\var\tmfive){\esub{\var_1}{\tm_1}\dots\esub{\var_n}{\tm_n}}\tmfour \toeo (\la\var\tmfive){\esub{\var_1}{\tm_1}\dots\esub{\var_j}{\tmp_j}\dots\esub{\var_n}{\tm_n}}\tmfour \eqdef \tmtwo$ with $n > 0$ and $\tm_j \toeo \tmp_j$ for some $1 \leq j \leq n$, and 
			$$
			\tm \tomo \allowbreak \tmfive\esub\var\tmfour{\esub{\var_1}{\tm_1}\dots\esub{\var_n}{\tm_n}} \eqdef \tmthree
			$$ 
			Then,
			\begin{equation*}
			\tmtwo \tomo \tmfive\esub\var\tmfour{\esub{\var_1}{\tm_1}\dots\esub{\var_j}{\tmp_j}\dots\esub{\var_n}{\tm_n}} \lRew{\wesym} \tmthree;
			\end{equation*}
			\item \emph{Application Right for $\tm \toeo \tmtwo$} and \emph{Application Left for $\tm \tom \tmthree$}, \ie $\tm \defeq \tmfive\tmfour \toeo \tmfive\tmfourp \eqdef \tmtwo$ and $\tm \tomo \tmfivep\tmfour \eqdef \tmthree$ with $\tmfour \toeo \tmfourp$ and $\tmfive \tomo \tmfivep$: then, $\tmtwo \tomo \tmfivep\tmfourp \lRew{\wesym} \tmthree$;
			\item \emph{Application Right for both $\tm \toeo \tmtwo$ and $\tm \tomo \tmthree$}, \ie $\tm \defeq \tmfive\tmfour \toeo \tmfive\tmfourp \eqdef \tmtwo$ and $\tm \tomo \tmfive\tmfour'' \eqdef \tmthree$ with $\tmfourp \lRew{\wesym} \tmfour \tomo \tmfour''$: by \ih, there exists $\tmsix \in \Lambda_\vsub$ such that $\tmfourp \tomo \tmsix \lRew{\wesym} \tmfour''$, hence $\tmtwo \tomo \tmfive\tmsix \lRew{\wesym} \tmthree$;
			\item \emph{Application Right for $\tm \toeo \tmtwo$} and \emph{Step at the Root for $\tm \tomo \tmthree$}, \ie $\tm \defeq \sctxp{\la\var\tmfive}\tmfour \toeo \sctxp{\la\var\tmfive}\tmfourp \eqdef \tmtwo$ with $\tmfour \toeo \tmfourp$\!, and $\tm \tomo \allowbreak \sctxp{\tmfive\esub\var\tmfour} \eqdef \tmthree$: then, $\tmtwo \tomo \sctxp{\tmfive\esub\var\tmfourp} \lRew{\wesym} \tmthree$;
			\item \emph{$\mathsf{ES}$ left for $\tm \toeo \tmtwo$} and \emph{$\mathsf{ES}$ right for $\tm \tomo \tmthree$}, \ie $\tm \defeq \tmfour\esub\var\tmfive \toeo \tmfourp\esub\var\tmfive \eqdef \tmtwo$ and $\tm \tomo \tmfour\esub\var\tmfivep \eqdef \tmthree$ with $\tmfour \toe \tmfourp$ and $\tmfive \tomo \tmfivep$: then, $\tmtwo \tomo \tmfourp\esub\var\tmfivep \lRew{\wesym} \tmthree$;
			\item \emph{$\mathsf{ES}$ left for both $\tm \toeo \tmtwo$ and $\tm \tomo \tmthree$}, \ie $\tm \defeq \tmfour\esub\var\tmfive \toe \tmfourp\esub\var\tmfive \eqdef \tmtwo$ and $\tm \tomo \tmfour''\esub\var\tmfive \eqdef \tmthree$ with $\tmfourp \lRew{\wesym} \tmfour \tomo \tmfour''$: by \ih,there is $\tmsix \in \Lambda_\vsub$ such that $\tmfourp \tomo \tmsix \lRew{\wesym} \tmfour''$, hence $\tmtwo \tomo \tmsix\esub\var\tmfive \lRew{\wesym} \tmthree$;
			\item \emph{$\mathsf{ES}$ right for $\tm \toeo \tmtwo$} and \emph{$\mathsf{ES}$ left for $\tm \tomo \tmthree$}, \ie $\tm \defeq \tmfive\esub\var\tmfour \toeo \tmfive\esub\var\tmfourp \eqdef \tmtwo$ and $\tm \tomo \tmfivep\esub\var\tmfour \eqdef \tmthree$ with $\tmfour \toeo \tmfourp$ and $\tmfive \tomo \tmfivep$: then, $\tmtwo \tomo \tmfivep\esub\var\tmfourp \lRew{\wesym} \tmthree$;
			\item \emph{$\mathsf{ES}$ right for both $\tm \toeo \tmtwo$ and $\tm \tomo \tmthree$}, \ie $\tm \defeq \tmfive\esub\var\tmfour \toeo \tmfive\esub\var{\tmfourp} \eqdef \tmtwo$ and $\tm \tomo \tmfive\esub\var{\tmfour''} \eqdef \tmthree$ with $\tmfour \lRew{\esym} \tmfourp \tomo \tmfour''$: by \ih, there is $\tmsix \in \Lambda_\vsub$ such that $\tmfour \tomo \tmsix \lRew{\wesym} \tmfour''$, so $\tmtwo \tomo \tmfive\esub\var\tmsix \lRew{\wesym} \tmthree$.
			\qedhere
		\end{itemize}
	\end{enumerate}
\end{proof}

\section{Proofs of Section~\ref*{sect:external-strategy} (External Strategy)}
\label{app:external-strategy-proofs}

\begin{lemma}[Properties of \pointed terms]
  \label{l:properties-of-rigid-terms}
  \hfill
  \begin{enumerate}
    \item \label{p:properties-of-rigid-terms-plugging-is-rigid} Given $\tm \in \vsubterms$ and a rigid context $\rctx$, $\rctxp{\tm}$ is a \pointed term.
    
    \item \label{p:properties-of-rigid-terms-open-strategy-preserves-rigid} Let $\rtm$ be a \pointed term and $\tm \in \vsubterms$ such that $\rtm \tovsubo \tm$. Then $\tm$ is \pointed.
    
    \item \label{p:properties-of-rigid-terms-strong-strategy-preserves-rigid} Let $\rtm$ be a \pointed term and $\tm \in \vsubterms$ such that $\rtm \tovsubs \tm$. Then $\tm$ is \pointed.
  
    
  \end{enumerate}
  \end{lemma}
  
  \begin{proof}~
  \begin{enumerate}
    \item By induction on the definition of $\rctx$.
    
    \item Let $\weakctx$ be an open evaluation context such that $\rtm = \weakctxp{\tmtwo} \tovsubo \weakctxp{\tmtwop} = \tm$, with $\tmtwo \rtom \tmtwop$ or $\tmtwo \rtoe \tmtwop$. We prove this for $\tomo$ by structural induction on $\weakctx$; the proof for $\toeo$ follows the same schema.
    \begin{itemize}
      \item \emph{Empty context}; \ie, $\weakctx = \ctxhole$. This case is not possible, because it would imply that $\rtm = \tmtwo = \sctxp{\la{\var}{\tmthree}} \tmfour$, which is not a \pointed term.
      
      \item \emph{Application right}; \ie, $\weakctx = \tmfour \weakctxtwo$. 
      As $\rtm = \tmfour \weakctxtwop{\tmtwo}$, then $\tmfour$ is a \pointed term, and so $\tmfour \weakctxtwop{\tmtwop} = \tm$ is \pointed.
      
      \item \emph{Application left}; \ie, $\weakctx = \weakctxtwo \tmfour$. 
      As $\rtm = \weakctxtwop{\tmtwo} \tmfour$ and $\weakctxtwop{\tmtwop}$ is \pointed by \ih, then $\weakctxtwop{\tmtwop} \tmfour = \tm$ is \pointed too.
      
      \item \emph{$\mathsf{ES}$ left}; \ie, $\weakctx = \weakctxtwo \esub{\var}{\tmfour}$. Since $\rtm = \weakctxtwop{\tmtwo} \esub{\var}{\tmfour}$, then both $\weakctxtwop{\tmtwo}$ and $\tmfour$ are \pointed. Moreover, $\weakctxtwop{\tmtwop}$ is \pointed by \ih, and so $\weakctxtwop{\tmtwop} \esub{\var}{\tmfour} = \tm$ is \pointed too.
      
      \item \emph{$\mathsf{ES}$ right}; \ie, $\weakctx = \tmfour \esub{\var}{\weakctxtwo}$. Since $\rtm = \tmfour \esub{\var}{\weakctxtwop{\tmtwo}}$, then both $\tmfour$ and $\weakctxtwop{\tmtwo}$ are \pointed. Moreover,  $\weakctxtwop{\tmtwop}$ is \pointed by \ih, and so $\tmfour \esub{\var}{\weakctxtwop{\tmtwop}} = \tm$ is \pointed too.
      
    \end{itemize}
    
    \item Let $\evsctx$ be a strong evaluation context such that $\rtm = \evsctxp{\tmtwo} \tovsubs \evsctxp{\tmtwop} = \tm$, with $\tmtwo \tomo \tmtwop$ or $\tmtwo \toeo \tmtwop$. We prove this for $\tmtwo \tomo \tmtwop$ by structural induction on $\evsctx$; the proof for $\toeo$ follows the same schema.
    \begin{itemize}
      \item \emph{Empty context}; \ie, $\evsctx = \ctxhole$. Then $\rtm = \tmtwo \tomo \tmtwop = \tm$ and the statement holds by \reflemmap{properties-of-rigid-terms}{open-strategy-preserves-rigid}
      
      \item \emph{Under $\lambda$-abstraction right}; \ie, $\evsctx = \la{\var}{\evsctxtwo}$. This case is not possible, because it would imply that $\rtm = \la{\var}{\evsctxtwop{\tmtwo}}$, which is not a \pointed term.
      
      \item \emph{External context, $\mathsf{ES}$ right}; \ie, $\evsctx = \tmfour \esub{\var}{\rctx}$. Since $\rtm = \tmfour \esub{\var}{\rctxp{\tmtwo}}$, then both $\tmfour$ and $\rctxp{\tmtwo}$ are \pointed terms. Moreover, $\rctxp{\tmtwop}$ is \pointed by \ih, and so $\tmfour \esub{\var}{\rctxp{\tmtwop}} = \tmthree$ is \pointed too.
      
      \item \emph{External context, $\mathsf{ES}$ left}; \ie, $\evsctx = \evsctxtwo \esub{\var}{\tmfour}$, with $\tmfour$ a \pointed term. Since $\rtm = \evsctxtwop{\tmtwo} \esub{\var}{\tmfour}$, then $\evsctxtwop{\tmtwo}$ is \pointed. Moreover, $\evsctxtwop{\tmtwop}$ is \pointed by \ih, and so $\evsctxtwop{\tmtwop} \esub{\var}{\tmfour} = \tmthree$ is \pointed too.
      
      \item \emph{Rigid context, application right}; \ie, $\evsctx = \tmfour \evsctxtwo$, with $\tmfour$ a \pointed term. Then $\tmfour \evsctxtwop{\tmtwop} = \tmthree$ is \pointed too.
      
      \item \emph{Rigid context, application left}; \ie, $\evsctx = \rctx \tmfour$. Since $\rtm = \rctxp{\tmtwo} \tmfour$, then $\rctxp{\tmtwo}$ is \pointed. Moreover, $\rctxp{\tmtwop}$ is \pointed by \ih, and so $\rctxp{\tmtwop} \tmfour$ is \pointed too.
      
      \item \emph{Rigid context, $\mathsf{ES}$ left}; \ie, $\evsctx = \rctx \esub{\var}{\tmfour}$, with $\tmfour$ a \pointed term. Since $\rtm = \rctxp{\tmtwo} \esub{\var}{\tmfour}$, then $\rctxp{\tmtwo}$ is \pointed. Moreover, $\rctxp{\tmtwop}$ is \pointed by \ih, and so $\rctxp{\tmtwop} \esub{\var}{\tmfour} = \tmthree$ is \pointed too.
      
      \item \emph{Rigid context, $\mathsf{ES}$ right}; \ie, $\evsctx = \tmfour \esub{\var}{\rctx}$, with $\tmfour$ a \pointed term. Since $\rtm = \tmfour \esub{\var}{\rctxp{\tmtwo}}$, then $\rctxp{\tmtwo}$ is \pointed. Moreover,  $\rctxp{\tmtwop}$ is \pointed by \ih, and so $\tmfour \esub{\var}{\rctxp{\tmtwop}} = \tmthree$ is \pointed too.
    \end{itemize}

  \end{enumerate}
  \end{proof}

The properties stated by the next proposition are the building blocks for the proof that $\tovsubs$ is diamond, coming right next. In particular, the strong commutation of $\toms$ and $\toes$ shall ensure that the diamond of $\tovsubs$ preserves the kind of step, thus strengthening the diamond of $\tovsubs$ with the guarantee that all $\tovsubs$ sequences to normal form (if any) have the same number of $\tom$ steps.
\begin{lemma}[Basic Properties of $\tovsubs$]
	\label{l:basic-value-substitution-external}
	\hfill
	\begin{enumerate}
%
		\item\label{p:basic-value-substitution-external-tom-toe-diamond-strong} $\toms$ and $\toes$ are diamond  (separately).
		
		\item\label{p:basic-value-substitution-external-tom-toe-commute-strong}  $\toms$ and $\toes$ strongly commute.
	\end{enumerate}
\end{lemma}

\begin{proof} 		
	\begin{enumerate}

		\item We prove that $\toms$ is diamond, \ie if $\tmtwo \reversetoms \tm \toms \tmthree$ with $\tmtwo \neq \tmthree$ then there exists $\tmp \in \Lambda_\vsub$ such that $\tmtwo \toms \tmp \lRew{\smsym} \tmthree$ (the proof that $\toes$ is diamond is analogue). 
		The proof is by structural induction on $\tm$, doing case analysis on $\tm \toms \tmtwo$ and $\tm \toms \tmthree$: 
		\begin{itemize}
			\item \emph{Under $\lambda$-abstraction for both $\tm \toms \tmtwo$ and $\tm \toms \tmthree$}; \ie, $\tm = \la{\var}{\tmfour} \toms \la{\var}{\tmfive} = \tmtwo$ and $\tm = \la{\var}{\tmfour} \toms \la{\var}{\tmsix} = \tmthree$, 
			with $\tmfour \toms \tmfive$ and $\tmfour \toms \tmsix$. 
			By \ih there exists $\tmfourp$ such that $\tmfive \toms \tmfourp \reversetoms \tmsix$ and so $\tmtwo = \la{\var}{\tmfive} \toms \la{\var}{\tmfourp} \reversetoms \la{\var}{\tmsix} = \tmthree$.
			
			\item \emph{Application right for $\tm \toms \tmtwo$ and application left for $\tm \toms \tmthree$}; \ie, $\tm = \tmfour \tmfive \toms \tmfour \tmfivep = \tmtwo$ and $\tm = \tmfour \tmfive \toms \tmfourp \tmfive = \tmthree$. There are several sub-cases to this:
			\begin{itemize}
				\item Let $\tm = \tmfour \tmfive = \tmfour \openctxp{\tilde{\tmfive}} \tomo \tmfour \openctxp{\tilde{\tmfivep}} = \tmtwo$, with $\tilde{\tmfive} \rtom \tilde{\tmfivep}$, and $\tm = \rctxp{\tilde{\tmfour}} \tmfive \toms \rctxp{\tilde{\tmfourp}} \tmfive$, with $\tilde{\tmfour} \tomo \tilde{\tmfourp}$. Let $\tmp \defeq \rctxp{\tilde{\tmfourp}} \openctxp{\tilde{\tmfivep}}$, having that 
				$$
				\tmtwo = \rctxp{\tilde{\tmfour}} \openctxp{\tilde{\tmfivep}} \toms \tmp \reversetoms \rctxp{\tilde{\tmfourp}} \openctxp{\tilde{\tmfive}} = \tmthree
				$$
				Note that $\tmtwo \toms \tmp$ holds because every rigid context is an open context.

				\item Let $\tm = \tmfour \tmfive = \tmfour \openctx_{1} \ctxholep{\tilde{\tmfive}} \tomo \tmfour \openctx_{1} \ctxholep{\tilde{\tmfivep}} = \tmtwo$, with $\tilde{\tmfive} \rtom \tilde{\tmfivep}$, and $\tm = \openctx_{2} \ctxholep{\tilde{\tmfour}} \tmfive \tomo \openctx_{2} \ctxholep{\tilde{\tmfourp}} \tmfive$, with $\tilde{\tmfour} \rtom \tilde{\tmfourp}$. 
				Then the statement holds by \reflemmap{basic-value-substitution}{tom-toe-diamond-open}
				
				\item Let $\tm = \tmfour \evsctxp{\tilde{\tmfive}} \toms \tmfour \evsctxp{\tilde{\tmfivep}} = \tmtwo$, with $\tmfour$ a rigid term and $\tilde{\tmfive} \tomo \tilde{\tmfivep}$, and $\tm = \rctxp{\tilde{\tmfour}} \tmfive \toms \rctxp{\tilde{\tmfourp}} \tmfive = \tmthree$, with $\tilde{\tmfour} \tomo \tilde{\tmfourp}$. Let $\tmp \defeq \rctxp{\tilde{\tmfourp}} \evsctxp{\tilde{\tmfivep}}$, having that 
				$$
				\tmtwo = \rctxp{\tilde{\tmfour}} \evsctxp{\tilde{\tmfivep}} \toms\tmp \reversetoms \rctxp{\tilde{\tmfourp}} \evsctxp{\tilde{\tmfive}} = \tmthree
				$$
				
				Note that $\tmp \reversetoms \tmthree$ holds because $\rctxp{\tilde{\tmfourp}}$ is a \pointed term ---by \reflemmap{properties-of-rigid-terms}{plugging-is-rigid}.
				
				\item Let $\tm = \tmfour \evsctxp{\tilde{\tmfive}} \toms \tmfour \evsctxp{\tilde{\tmfivep}} = \tmtwo$, with $\tmfour$ a rigid term and $\tilde{\tmfive} \tomo \tilde{\tmfivep}$, and $\tm = \openctxp{\tilde{\tmfour}} \tmfive \toms \openctxp{\tilde{\tmfourp}} \tmfive$, with $\tilde{\tmfour} \rtom \tilde{\tmfourp}$. Let $\tmp \defeq \openctxp{\tilde{\tmfourp}} \evsctxp{\tilde{\tmfivep}}$, having that
				$$
				\tmtwo = \openctxp{\tilde{\tmfour}} \evsctxp{\tilde{\tmfivep}} \toms \tmp \reversetoms \openctxp{\tilde{\tmfourp}} \evsctxp{\tilde{\tmfive}} = \tmthree
				$$
				
				Note that $\tmp \reversetoms \tmthree$ holds because the fact that $\tmfour$ is a \pointed term and that $\tmfour = \openctxp{\tilde{\tmfour}} \tomo \openctxp{\tilde{\tmfourp}}$ imply that $\openctxp{\tilde{\tmfourp}}$ is a \pointed term ---by \reflemmap{properties-of-rigid-terms}{open-strategy-preserves-rigid}.
				
			\end{itemize}
			
			\item \emph{Application right for both $\tm \toms \tmtwo$ and $\tm \toms \tmthree$}; \ie, $\tm = \tmfour \tmfive \toms \tmfour \tmfivep = \tmtwo$ and $\tm = \tmfour \tmfive \toms \tmfour \tmfivepp = \tmthree$. By \ih there exists $\tmsix \in \vsubterms$ such that $\tmfivep \toms \tmsix \reversetoms \tmfivepp$. The analysis of the sub-cases, depending on the open/strong/rigid type contexts involved in $\tm \toms \tmtwo$ and $\tm \toms \tmthree$, follows the same schema as for the previous item, all showing that
			$$
			\tmtwo = \tmfour \tmfivep \toms \tmfour \tmsix \reversetoms \tmfour \tmfivepp = \tmthree
			$$

			\item \emph{Application left for both $\tm \toms \tmtwo$ and $\tm \toms \tmthree$}; \ie, $\tm = \tmfour \tmfive \toms \tmfourp \tmfive = \tmtwo$ and $\tm = \tmfour \tmfive \toms \tmfourpp \tmfive = \tmthree$. By \ih there exists $\tmsix \in \vsubterms$ such that $\tmfourp \toms \tmsix \reversetoms \tmfourpp$. The analysis of the sub-cases, depending on the open/strong/rigid type contexts involved in $\tm \toms \tmtwo$ and $\tm \toms \tmthree$, follows the same schema as for the previous item, all showing that
			$$
			\tmtwo = \tmfourp \tmfive \toms \tmsix \tmfive \reversetoms \tmfourpp \tmfive = \tmthree
			$$
			
			\item \emph{$\mathsf{ES}$ right for $\tm \toms \tmtwo$ and $\mathsf{ES}$ left for $\tm \toms \tmthree$}; \ie, $\tm = \tmfour \esub{\var}{\tmfive} \toms \tmfour \esub{\var}{\tmfivep} = \tmtwo$ and $\tm = \tmfour \esub{\var}{\tmfive} \toms \tmfourp \esub{\var}{\tmfive}$. There are several sub-cases to this:
			\begin{itemize}
				\item Let $\tm = \tmfour \esub{\var}{\openctxp{\tilde{\tmfive}}} \toms \tmfour \esub{\var}{\openctxp{\tilde{\tmfivep}}} = \tmtwo$, with $\tilde{\tmfive} \rtom \tilde{\tmfivep}$, and $\tm = \openctxp{\tilde{\tmfour}} \esub{\var}{\tmfive} \toms \openctxp{\tilde{\tmfourp}} \esub{\var}{\tmfive} = \tmthree$, with $\tilde{\tmfour} \rtom \tilde{\tmfourp}$. Then the statement holds by \reflemmap{basic-value-substitution}{tom-toe-diamond-open}.
				
				\item Let $\tm = \tmfour \esub{\var}{\openctxp{\tilde{\tmfive}}} \toms \tmfour \esub{\var}{\openctxp{\tilde{\tmfivep}}} = \tmtwo$, with $\tilde{\tmfive} \rtom \tilde{\tmfivep}$, and $\tm = \evsctxp{\tilde{\tmfour}} \esub{\var}{\tmfive} \toms \evsctxp{\tilde{\tmfourp}} \esub{\var}{\tmfive} = \tmthree$, with $\tilde{\tmfour} \tomo \tilde{\tmfourp}$ and $\tmfive$ is a \pointed term. Let $\tmp \defeq \evsctxp{\tilde{\tmfourp}} \esub{\var}{\openctxp{\tilde{\tmfivep}}}$, having that 
				$$
				\tmtwo = \evsctxp{\tilde{\tmfour}} \esub{\var}{\openctxp{\tilde{\tmfivep}}} \toms \tmp \reversetoms \evsctxp{\tilde{\tmfourp}} \esub{\var}{\openctxp{\tilde{\tmfive}}} = \tmthree
				$$
				
				Note that $\tmtwo \toms \tmp$ holds because the fact that$\tmfive$ is a \pointed term and that $\tmfive = \openctxp{\tilde{\tmfive}} \toms \openctxp{\tilde{\tmfivep}}$ imply that  $\openctxp{\tilde{\tmfivep}}$ is a \pointed term ---by \reflemmap{properties-of-rigid-terms}{strong-strategy-preserves-rigid}.
				
				\item Let $\tm = \tmfour \esub{\var}{\openctxp{\tilde{\tmfive}}} \toms \tmfour \esub{\var}{\openctxp{\tilde{\tmfivep}}} = \tmtwo$, with $\tilde{\tmfive} \rtom \tilde{\tmfivep}$, and $\tm = \rctxp{\tilde{\tmfour}} \esub{\var}{\tmfive} \toms \rctxp{\tilde{\tmfourp}} \esub{\var}{\tmfive} = \tmthree$, with $\tilde{\tmfour} \tomo \tilde{\tmfourp}$ and $\tmfive$ is a \pointed term. Let $\tmp \defeq \rctxp{\tilde{\tmfourp}} \esub{\var}{\openctxp{\tilde{\tmfivep}}}$, having that 
				$$
				\tmtwo = \rctxp{\tilde{\tmfour}} \esub{\var}{\openctxp{\tilde{\tmfivep}}} \toms \tmp \reversetoms \rctxp{\tilde{\tmfourp}} \esub{\var}{\openctxp{\tilde{\tmfive}}} = \tmthree
				$$
				
				Note that $\tmtwo \toms \tmp$ holds because the fact that $\tmfive$ is a \pointed term and that $\tmfive = \openctxp{\tilde{\tmfive}} \toms \openctxp{\tilde{\tmfivep}}$ imply $\openctxp{\tilde{\tmfivep}}$ is a \pointed term ---by \reflemmap{properties-of-rigid-terms}{open-strategy-preserves-rigid}.
				
				\item Let $\tm = \tmfour \esub{\var}{\rctxp{\tilde{\tmfive}}} \toms \tmfour \esub{\var}{\rctxp{\tilde{\tmfivep}}} = \tmtwo$, with $\tilde{\tmfive} \tomo \tilde{\tmfivep}$, and $\tm = \openctxp{\tilde{\tmfour}} \esub{\var}{\tmfive} \toms \openctxp{\tilde{\tmfourp}} \esub{\var}{\tmfive} = \tmthree$, with $\tilde{\tmfour} \rtom \tilde{\tmfourp}$. Let $\tmp \defeq \openctxp{\tilde{\tmfourp}} \esub{\var}{\rctxp{\tilde{\tmfivep}}}$, having that 
				$$
				\tmtwo = \openctxp{\tilde{\tmfour}} \esub{\var}{\rctxp{\tilde{\tmfivep}}} \toms \tmp \reversetoms \openctxp{\tilde{\tmfourp}} \esub{\var}{\rctxp{\tilde{\tmfive}}} = \tmthree
				$$
				
				\item Let $\tm = \tmfour \esub{\var}{\rctxp{\tilde{\tmfive}}} \toms \tmfour \esub{\var}{\rctxp{\tilde{\tmfivep}}} = \tmtwo$, with $\tilde{\tmfive} \tomo \tilde{\tmfivep}$, and $\tm = \evsctxp{\tilde{\tmfour}} \esub{\var}{\tmfive} \toms \evsctxp{\tilde{\tmfourp}} \esub{\var}{\tmfive} = \tmthree$, with $\tilde{\tmfour} \tomo \tilde{\tmfourp}$ and $\tmfive$ is a \pointed term. Let $\tmp \defeq \evsctxp{\tilde{\tmfourp}} \esub{\var}{\rctxp{\tilde{\tmfivep}}}$, having that 
				$$
				\tmtwo = \evsctxp{\tilde{\tmfour}} \esub{\var}{\rctxp{\tilde{\tmfivep}}} \toms \tmp \reversetoms \evsctxp{\tilde{\tmfourp}} \esub{\var}{\rctxp{\tilde{\tmfive}}} = \tmthree
				$$
				Note that $\tmtwo \toms \tmp$ holds because the fact that $\tmfive$ is a \pointed term and that $\tmfive = \rctxp{\tilde{\tmfive}} \toms \rctxp{\tilde{\tmfivep}}$ imply that $\rctxp{\tilde{\tmfivep}}$ is a \pointed term ---by \reflemmap{properties-of-rigid-terms}{strong-strategy-preserves-rigid}.
				
				\item Let $\tm = \tmfour \esub{\var}{\rctx_{1} \ctxholep{\tilde{\tmfive}}} \toms \tmfour \esub{\var}{\rctx_{1} \ctxholep{\tilde{\tmfivep}}} = \tmtwo$, with $\tilde{\tmfive} \tomo \tilde{\tmfivep}$, and $\tm = \rctx_{2} \ctxholep{\tilde{\tmfour}} \esub{\var}{\tmfive} \toms \rctx_{2} \ctxholep{\tilde{\tmfourp}} \esub{\var}{\tmfive} = \tmthree$, with $\tilde{\tmfour} \tomo \tilde{\tmfourp}$ and $\tmfive$ is a \pointed term. Let $\tmp \defeq \rctx_{2} \ctxholep{\tilde{\tmfourp}} \esub{\var}{\rctx_{1} \ctxholep{\tilde{\tmfivep}}}$, having that 
				$$
				\tmtwo = \rctx_{2} \ctxholep{\tilde{\tmfour}} \esub{\var}{\rctx_{1} \ctxholep{\tilde{\tmfivep}}} \toms \tmp \reversetoms \rctx_{2} \ctxholep{\tilde{\tmfourp}} \esub{\var}{\rctx_{1} \ctxholep{\tilde{\tmfive}}} = \tmthree
				$$
				Note that $\tmtwo \toms \tmp$ holds because the fact that $\tmfive$ is a \pointed term and that $\tmfive = \rctx_{1} \ctxholep{\tilde{\tmfive}} \toms \rctx_{1} \ctxholep{\tilde{\tmfivep}}$ imply that $\rctx_{1} \ctxholep{\tilde{\tmfivep}}$ is a \pointed term ---by \reflemmap{properties-of-rigid-terms}{strong-strategy-preserves-rigid}.
				
				\item Let $\tm = \tmfour \esub{\var}{\rctxp{\tilde{\tmfive}}} \toms \tmfour \esub{\var}{\rctxp{\tilde{\tmfivep}}} = \tmtwo$, with $\tilde{\tmfive} \tomo \tilde{\tmfivep}$ and $\tmfour$ is a \pointed term, and $\tm = \openctxp{\tilde{\tmfour}} \esub{\var}{\tmfive} \toms \openctxp{\tilde{\tmfourp}} \esub{\var}{\tmfive} = \tmthree$, with $\tilde{\tmfour} \rtom \tilde{\tmfourp}$. Let $\tmp \defeq \openctxp{\tilde{\tmfourp}} \esub{\var}{\rctxp{\tilde{\tmfivep}}}$, having that
				$$
				\tmtwo = \openctxp{\tilde{\tmfour}} \esub{\var}{\rctxp{\tilde{\tmfivep}}} \toms \tmp \reversetoms \openctxp{\tilde{\tmfourp}} \esub{\var}{\rctxp{\tilde{\tmfive}}} = \tmthree
				$$
				Note that $\tmthree \reversetoms \tmp$ holds because the fact that $\tmfive$ is a \pointed term and that $\tmfive = \rctxp{\tilde{\tmfive}} \toms \rctxp{\tilde{\tmfivep}}$ imply that $\rctxp{\tilde{\tmfivep}}$ is a \pointed term ---by \reflemmap{properties-of-rigid-terms}{strong-strategy-preserves-rigid}.
				
				\item Let $\tm = \tmfour \esub{\var}{\rctxp{\tilde{\tmfive}}} \toms \tmfour \esub{\var}{\rctxp{\tilde{\tmfivep}}} = \tmtwo$, with $\tilde{\tmfive} \tomo \tilde{\tmfivep}$ and $\tmfour$ is a \pointed term, and $\tm = \evsctxp{\tilde{\tmfour}} \esub{\var}{\tmfive} \toms \evsctxp{\tilde{\tmfourp}} \esub{\var}{\tmfive} = \tmthree$, with $\tilde{\tmfour} \tomo \tilde{\tmfourp}$ and $\tmfive$ is a \pointed term. Let $\tmp \defeq \evsctxp{\tilde{\tmfourp}} \esub{\var}{\rctxp{\tilde{\tmfivep}}}$, having that 
				$$
				\tmtwo = \evsctxp{\tilde{\tmfour}} \esub{\var}{\rctxp{\tilde{\tmfivep}}} \toms \tmp \reversetoms \evsctxp{\tilde{\tmfourp}} \esub{\var}{\rctxp{\tilde{\tmfive}}} = \tmthree
				$$ 
				Note that $\tmtwo \toms \tmp$ holds because the fact that $\tmfive$ is a \pointed term and that $\tmfive = \rctxp{\tilde{\tmfive}} \toms \rctxp{\tilde{\tmfivep}}$ imply that $\rctxp{\tilde{\tmfivep}}$ is a \pointed term ---by \reflemmap{properties-of-rigid-terms}{strong-strategy-preserves-rigid}. Moreover, note that $\tmp \reversetoms \tmthree$ holds because the fact that $\tmfour$ is a \pointed term and that $\tmfour = \evsctxp{\tilde{\tmfour}} \toms \evsctxp{\tilde{\tmfourp}}$ imply that $\evsctxp{\tilde{\tmfourp}}$ is a \pointed term ---by \reflemmap{properties-of-rigid-terms}{strong-strategy-preserves-rigid}.
				
				\item Let $\tm = \tmfour \esub{\var}{\rctxp{\tilde{\tmfive}}} \toms \tmfour \esub{\var}{\rctxp{\tilde{\tmfivep}}} = \tmtwo$, with $\tilde{\tmfive} \tomo \tilde{\tmfivep}$ and $\tmfour$ is a \pointed term, and $\tm = \rctxp{\tilde{\tmfour}} \esub{\var}{\tmfive} \toms \rctxp{\tilde{\tmfourp}} \esub{\var}{\tmfive} = \tmthree$, with $\tilde{\tmfour} \tomo \tilde{\tmfourp}$ and $\tmfive$ is a \pointed term. Let $\tmp \defeq \rctxp{\tilde{\tmfourp}} \esub{\var}{\rctxp{\tilde{\tmfivep}}}$, having that 
				$$
				\tmtwo = \rctxp{\tilde{\tmfour}} \esub{\var}{\rctxp{\tilde{\tmfivep}}} \toms \tmp \reversetoms \rctxp{\tilde{\tmfourp}} \esub{\var}{\rctxp{\tilde{\tmfive}}} = \tmthree
				$$
				Note that $\tmtwo \toms \tmp$ holds because the fact that $\tmfive$ is a \pointed term and that $\tmfive = \rctxp{\tilde{\tmfive}} \toms \rctxp{\tilde{\tmfivep}}$ imply that $\rctxp{\tilde{\tmfivep}}$ is a \pointed term ---by \reflemmap{properties-of-rigid-terms}{strong-strategy-preserves-rigid}. Moreover, note that $\tmp \reversetoms \tmthree$ holds because the fact that $\tmfour$ is a \pointed term and that $\tmfour = \rctxp{\tilde{\tmfour}} \toms \rctxp{\tilde{\tmfourp}}$ imply that $\rctxp{\tilde{\tmfourp}}$ is a \pointed term ---by \reflemmap{properties-of-rigid-terms}{strong-strategy-preserves-rigid}.
				
			\end{itemize}
			
			\item \emph{$\mathsf{ES}$ right for both $\tm \toms \tmtwo$ and $\tm \toms \tmthree$}; \ie, $\tm = \tmfour \esub{\var}{\tmfive} \toms \tmfour \esub{\var}{\tmfivep} = \tmtwo$ and $\tm = \tmfour \esub{\var}{\tmfive} \toms \tmfour \esub{\var}{\tmfivepp} = \tmthree$. By \ih there exists $\tmsix \in \vsubterms$ such that $\tmfivep \toms \tmsix \reversetoms \tmfivepp$. The analysis of the sub-cases, depending on the open/strong/rigid type contexts involved in $\tm \toms \tmtwo$ and $\tm \toms \tmthree$, follows the same schema as for the previous item, all showing that
			$$
			\tmtwo = \tmfour \esub{\var}{\tmfivep} \toms \tmfour \esub{\var}{\tmsix} \reversetoms \tmfour \esub{\var}{\tmfivepp} = \tmthree
			$$
			
			\item \emph{$\mathsf{ES}$ left for both $\tm \toms \tmtwo$ and $\tm \toms \tmthree$}; \ie, $\tm = \tmfour \esub{\var}{\tmfive} \toms \tmfourp \esub{\var}{\tmfive} = \tmtwo$ and $\tm = \tmfourpp \esub{\var}{\tmfive}$. By \ih there exists $\tmsix \in \vsubterms$ such that $\tmfourp \toms \tmsix \reversetoms \tmfourpp$. The analysis of the sub-cases, depending on the open/strong/rigid type contexts involved in $\tm \toms \tmtwo$ and $\tm \toms \tmthree$, follows the same schema as for the previous item, all showing that
			$$
			\tmtwo = \tmfourp \esub{\var}{\tmfive} \toms \tmsix \esub{\var}{\tmfive} \reversetoms \tmfourpp \esub{\var}{\tmfive} = \tmthree
			$$
			
		\end{itemize}

		The proof that $\toes$ is diamond (\ie, if $\tmtwo \reversetoes \tm \toes \tmthree$ with $\tmtwo \neq \tmthree$ then there exists $\tmp \in \Lambda_\vsub$ such that $\tmtwo \toes \tmp \reversetoes \tmthree$) follows the same schema as for $\toms$.
		
		\item We show that $\toes$ and $\toms$ strongly commute; \ie, if $\tmtwo \reversetoes \tm \toms \tmthree$, then $\tmtwo \neq \tmthree$ and there is $\tmp \in \Lambda_\vsub$ such that $\tmtwo \toms \tmp \reversetoes \tmthree$. The proof is by structural induction on $\tm$, doing case analysis on $\tm \reversetoes \tmtwo$ and $\tm \toms \tmthree$:
		\begin{itemize}
			\item \emph{Under $\lambda$-abstraction for both $\tm \reversetoes \tmtwo$ and $\tm \toms \tmthree$}; \ie, $\tm = \la{\var}{\tmfive} \reversetoes \la{\var}{\tmfour} = \tmtwo$ and $\tm = \la{\var}{\tmfour} \toms \la{\var}{\tmsix} = \tmthree$, with $\tmfive \reversetoeo \tmfour \tomo \tmsix$. By \reflemmap{basic-value-substitution}{tom-toe-commute-open}, there exists $\tmfourp$ such that $\tmfive \toms \tmfourp \reversetoes \tmsix$, and so $\tmtwo = \la{\var}{\tmfive} \tomo \la{\var}{\tmfourp} \reversetoeo \la{\var}{\tmsix}$.
			
			\item \emph{Application right for $\tmtwo \reversetoes \tm$ and application left for $\tm \toms \tmthree$}; \ie, $\tmtwo = \tmfour \tmfivep \reversetoes \tmfour \tmfive = \tm$ and $\tm = \tmfour \tmfive \toms \tmfourp \tmfive = \tmthree$. There are several sub-cases to this:
			\begin{itemize}
				\item Let $\tmtwo = \tmfour \openctxp{\tilde{\tmfivep}} \reversetoes \tmfour \openctxp{\tilde{\tmfive}} = \tm$, with $\tilde{\tmfivep} \reversetoeo \tilde{\tmfive}$, and $\tm = \rctxp{\tilde{\tmfour}}  \tmfive \toms \rctxp{\tilde{\tmfourp}} \tmfive = \tmthree$, with $\tilde{\tmfour} \tomo \tilde{\tmfourp}$. Let $\tmp = \rctxp{\tilde{\tmfourp}} \openctxp{\tilde{\tmfivep}}$, having that 
				$$
				\tmtwo = \rctxp{\tilde{\tmfour}} \openctxp{\tilde{\tmfivep}} \toms \tmp \reversetoes \rctxp{\tilde{\tmfourp}} \openctxp{\tilde{\tmfive}} = \tmthree
				$$
				
				\item Let $\tmtwo = \tmfour \openctx_{1} \ctxholep{\tilde{\tmfivep}} \reversetoes \tmfour \openctx_{1} \ctxholep{\tilde{\tmfive}} = \tm$, and $\tm = \openctx_{2} \ctxholep{\tilde{\tmfour}} \tmfive \toms \openctx_{2} \ctxholep{\tilde{\tmfourp}} \tmfive = \tmthree$, with $\tilde{\tmfour} \toms \tilde{\tmfourp}$. Then the statement holds by \reflemmap{basic-value-substitution}{tom-toe-commute-open}
				
				\item Let $\tmtwo = \tmfour \evsctxp{\tilde{\tmfivep}} \reversetoes \tmfour \evsctxp{\tilde{\tmfive}} = \tm$, with $\tmfour$ a \pointed term and $\tilde{\tmfivep} \reversetoeo \tilde{\tmfive}$, and $\tm = \rctxp{\tilde{\tmfour}} \tmfive \toms \rctxp{\tilde{\tmfourp}} \tmfive = \tmthree$, with $\tilde{\tmfour} \tomo \tilde{\tmfourp}$. Let $\tmp = \rctxp{\tilde{\tmfourp}} \evsctxp{\tilde{\tmfivep}}$, having that
				$$
				\tmtwo = \rctxp{\tilde{\tmfour}} \evsctxp{\tilde{\tmfivep}} \toms \tmp \reversetoes \rctxp{\tilde{\tmfourp}} \evsctxp{\tilde{\tmfive}} = \tmthree
				$$
				Note that $\tmp \reversetoes \tmthree$ holds because $\rctxp{\tilde{\tmfourp}}$ is a \pointed term ---by \reflemmap{properties-of-rigid-terms}{strong-strategy-preserves-rigid}.
				
				\item Let $\tmtwo = \tmfour \evsctxp{\tilde{\tmfivep}} \reversetoes \tmfour \evsctxp{\tilde{\tmfive}} = \tm$, with $\tmfour$ a \pointed term and $\tilde{\tmfivep} \reversetoeo \tilde{\tmfive}$, and $\tm = \openctxp{\tilde{\tmfour}} \tmfive \toms \openctxp{\tilde{\tmfourp}} \tmfive = \tmthree$, with $\tilde{\tmfour} \rtom \tilde{\tmfourp}$. Let $\tmp = \openctxp{\tilde{\tmfourp}} \evsctxp{\tilde{\tmfivep}}$, having that
				$$
				\tmtwo = \openctxp{\tilde{\tmfour}} \evsctxp{\tilde{\tmfivep}} \toms \tmp \reversetoes \openctxp{\tilde{\tmfourp}} \evsctxp{\tilde{\tmfive}} = \tmthree
				$$
				Note that $\tmp \reversetoes \tmthree$ holds because $\openctxp{\tilde{\tmfourp}}$ is \pointed ---by \reflemmap{properties-of-rigid-terms}{open-strategy-preserves-rigid}.
				
			\end{itemize}
			
			\item \emph{Application left for $\tmtwo \reversetoes \tm$ and application right for $\tm \toms \tmthree$}; \ie, $\tmtwo = \tmfourp \tmfive \reversetoes \tmfour \tmfive = \tm$ and $\tm = \tmfour \tmfive \toms \tmfour \tmfivep = \tmthree$. There are several sub-cases to this:
			\begin{itemize}
				\item Let $\tmtwo = \rctxp{\tilde{\tmfourp}} \tmfive \reversetoes \rctxp{\tilde{\tmfour}} \tmfive = \tm$, with $\tilde{\tmfourp} \reversetoeo \tilde{\tmfour}$, and $\tm = \tmfour \openctxp{\tilde{\tmfive}} \toms \tmfour \openctxp{\tilde{\tmfivep}} = \tmthree$, with $\tilde{\tmfive} \rtom \tilde{\tmfivep}$. Let $\tmp = \rctxp{\tilde{\tmfourp}} \openctxp{\tilde{\tmfivep}}$, having that 
				$$
				\tmtwo = \rctxp{\tilde{\tmfourp}} \openctxp{\tilde{\tmfive}} \toms \tmp \reversetoes \rctxp{\tilde{\tmfour}} \openctxp{\tilde{\tmfivep}} = \tmthree
				$$
				
				\item Let $\tmtwo = \openctx_{1} \ctxholep{\tilde{\tmfourp}} \tmfive \reversetoes \openctx_{1} \ctxholep{\tilde{\tmfour}} \tmfive = \tm$, with $\tilde{\tmfourp} \reversetoeo \tilde{\tmfour}$, and $\tm = \tmfour \openctx_{2} \ctxholep{\tilde{\tmfive}} \toms \tmfour \openctx_{2} \ctxholep{\tilde{\tmfivep}} = \tmthree$, with $\tilde{\tmfive} \rtom \tilde{\tmfivep}$. Then the statement holds by \reflemmap{basic-value-substitution}{tom-toe-commute-open}
				
				\item Let $\tmtwo = \rctxp{\tilde{\tmfourp}} \tmfive \reversetoes \rctxp{\tilde{\tmfour}} \tmfive = \tm$, with $\tilde{\tmfourp} \reversetoeo \tilde{\tmfour}$, and $\tm = \tmfour \evsctxp{\tilde{\tmfive}} \toms \tmfour \evsctxp{\tilde{\tmfivep}} = \tmthree$, with $\tilde{\tmfive} \tomo \tilde{\tmfivep}$ and $\tmfour$ a \pointed term. Let $\tmp = \rctxp{\tilde{\tmfourp}} \evsctxp{\tilde{\tmfivep}}$, having that
				$$
				\tmtwo = \rctxp{\tilde{\tmfourp}} \evsctxp{\tilde{\tmfive}} \toms \tmp \reversetoes \rctxp{\tilde{\tmfour}} \evsctxp{\tilde{\tmfivep}} = \tmthree
				$$
				Note that $\tmtwo \toms \tmp$ holds because $\rctxp{\tilde{\tmfourp}}$ is a \pointed term ---by \reflemmap{properties-of-rigid-terms}{strong-strategy-preserves-rigid}.
				
				\item Let $\tmtwo = \openctxp{\tilde{\tmfourp}} \tmfive \reversetoes \openctxp{\tilde{\tmfour}} \tmfive = \tm$, with $\tilde{\tmfourp} \reversetoeo \tilde{\tmfour}$, and $\tm = \tmfour \evsctxp{\tilde{\tmfive}} \toms \tmfour \evsctxp{\tilde{\tmfivep}} = \tmthree$, with $\tmfour$ a \pointed term and $\tilde{\tmfive} \tomo \tilde{\tmfivep}$. Let $\tmp = \openctxp{\tilde{\tmfourp}} \evsctxp{\tilde{\tmfivep}}$, having that
				$$
				\tmtwo = \openctxp{\tilde{\tmfourp}} \evsctxp{\tilde{\tmfive}} \toms \tmp \reversetoes \openctxp{\tilde{\tmfour}} \evsctxp{\tilde{\tmfivep}} = \tmthree
				$$
				Note that $\tmtwo \toms \tmp$ holds because $\openctxp{\tilde{\tmfourp}}$ is \pointed ---by \reflemmap{properties-of-rigid-terms}{open-strategy-preserves-rigid}.
				
			\end{itemize}
			
			\item \emph{Application right for both $\tmtwo \reversetoes \tm$ and $\tm \toms \tmthree$}; \ie, $\tmtwo = \tmfour \tmfivep \reversetoes \tmfour \tmfive = \tm$ and $\tm = \tmfour \tmfive \toms \tmfour \tmfivepp = \tmthree$. By \ih, there exists $\tmsix \in \vsubterms$ such that $\tmfivep \toms \tmsix \reversetoes \tmfivepp$. The analysis of the sub-cases, depending on the open/strong/rigid type contexts involved in $\tmtwo \reversetoes \tm$ and $\tm \toms \tmthree$ follows the same schema as for the previous item, all showing that
			$$
			\tmtwo = \tmfour \tmfivep \toms \tmfour \tmsix \reversetoes \tmfour \tmfivepp = \tmthree
			$$
			
			\item \emph{Application left for both $\tmtwo \reversetoes \tm$ and $\tm \toms \tmthree$}; \ie, $\tmtwo = \tmfourp \tmfive \reversetoes \tmfour \tmfive = \tm$ and $\tm = \tmfour \tmfive \toms \tmfourpp \tmfive = \tmthree$. By \ih, there exists $\tmsix \in \vsubterms$ such that $\tmfourp \toms \tmsix \reversetoes \tmfourpp$. The analysis of the sub-cases, depending on the open/strong/rigid type contexts involved in $\tmtwo \reversetoes \tm$ and $\tm \toms \tmthree$ follows the same schema as for the previous item, all showing that
			$$
			\tmtwo = \tmfourp \tmfive \toms \tmsix \tmfive \reversetoes \tmfourpp \tmfive = \tmthree
			$$

\item \emph{$\mathsf{ES}$ right for $\tmtwo \reversetoes \tm$ and $\mathsf{ES}$ left for $\tm \toms \tmthree$}; \ie, $\tmtwo = \tmfour \esub{\var}{\tmfivep} \reversetoes \tmfour \esub{\var}{\tmfive} = \tm$ and $\tm = \tmfour \esub{\var}{\tmfive} \toms \tmfourp \esub{\var}{\tmfive} = \tmthree$. There are several sub-cases to this:
\begin{itemize}
	\item Let $\tmtwo = \tmfour \esub{\var}{\openctxp{\tilde{\tmfivep}}} \reversetoes \tmfour \esub{\var}{\openctxp{\tilde{\tmfive}}} = \tm$, with $\tilde{\tmfivep} \reversertoe \tilde{\tmfive}$, and $\tm = \openctxp{\tilde{\tmfour}} \esub{\var}{\tmfive} \toms \openctxp{\tilde{\tmfourp}} \esub{\var}{\tmfive} = \tmthree$, with $\tilde{\tmfour} \rtom \tilde{\tmfourp}$. Then the statement holds by \reflemmap{basic-value-substitution}{tom-toe-commute-open}.
	
	\item Let $\tmtwo = \tmfour \esub{\var}{\openctxp{\tilde{\tmfivep}}} \reversetoes \tmfour \esub{\var}{\openctxp{\tilde{\tmfive}}} = \tm$, with $\tilde{\tmfivep} \reversertoe \tilde{\tmfive}$, and $\tm = \evsctxp{\tilde{\tmfour}} \esub{\var}{\tmfive} \toms \evsctxp{\tilde{\tmfourp}} \esub{\var}{\tmfive} = \tmthree$, with $\tilde{\tmfour} \tomo \tilde{\tmfourp}$ and $\tmfive$ is a \pointed term. Let $\tmp \defeq \evsctxp{\tilde{\tmfourp}} \esub{\var}{\openctxp{\tilde{\tmfivep}}}$, having that 
	$$
	\tmtwo = \evsctxp{\tilde{\tmfour}} \esub{\var}{\openctxp{\tilde{\tmfivep}}} \toms \tmp \toms \evsctxp{\tilde{\tmfourp}} \esub{\var}{\openctxp{\tilde{\tmfive}}} = \tmthree
	$$
	Note that $\tmtwo \toms \tmp$ holds because $\openctxp{\tilde{\tmfivep}}$ is a \pointed term ---by \reflemmap{properties-of-rigid-terms}{open-strategy-preserves-rigid}
	
	\item Let $\tmtwo = \tmfour \esub{\var}{\openctxp{\tilde{\tmfivep}}} \reversetoes \tmfour \esub{\var}{\openctxp{\tilde{\tmfive}}} = \tm$, with $\tilde{\tmfivep} \reversertoe \tilde{\tmfive}$, and $\tm = \rctxp{\tilde{\tmfour}} \esub{\var}{\tmfive} \toms \rctxp{\tilde{\tmfourp}} \esub{\var}{\tmfive} = \tmthree$, with $\tilde{\tmfour} \tomo \tilde{\tmfourp}$ and $\tmfive$ is a \pointed term. Let $\tmp \defeq \rctxp{\tilde{\tmfourp}} \esub{\var}{\openctxp{\tilde{\tmfivep}}}$, having that 
	$$
	\tmtwo = \rctxp{\tilde{\tmfour}} \esub{\var}{\openctxp{\tilde{\tmfivep}}} \toms \tmp \toms \rctxp{\tilde{\tmfourp}} \esub{\var}{\openctxp{\tilde{\tmfive}}} = \tmthree
	$$
	Note that $\tmtwo \toms \tmp$ holds because $\openctxp{\tilde{\tmfivep}}$ is a \pointed term ---by \reflemmap{properties-of-rigid-terms}{open-strategy-preserves-rigid}
	
	\item Let $\tmtwo = \tmfour \esub{\var}{\rctxp{\tilde{\tmfivep}}} \reversetoes \tmfour \esub{\var}{\rctxp{\tilde{\tmfive}}} = \tm$, with $\tmfivep \reversetoeo \tmfive$, and $\tm = \openctxp{\tilde{\tmfour}} \esub{\var}{\tmfive} \toms \openctxp{\tilde{\tmfourp}} \esub{\var}{\tmfive} = \tmthree$, with $\tilde{\tmfive} \toms \tilde{\tmfivep}$. Let $\tmp \defeq \openctxp{\tilde{\tmfourp}} \esub{\var}{\rctxp{\tmfivep}}$, having that
	$$
	\tmtwo = \openctxp{\tilde{\tmfour}} \esub{\var}{\rctxp{\tilde{\tmfivep}}} \toms \tmp \reversetoes \openctxp{\tilde{\tmfourp}} \esub{\var}{\rctxp{\tilde{\tmfive}}} = \tmthree
	$$
	
	\item Let $\tmtwo = \tmfour \esub{\var}{\rctxp{\tilde{\tmfivep}}} \reversetoes \tmfour \esub{\var}{\rctxp{\tilde{\tmfive}}} = \tm$, with $\tilde{\tmfivep} \reversetoeo \tilde{\tmfive}$, and $\tm = \evsctxp{\tilde{\tmfour}} \esub{\var}{\tmfive} \toms \evsctxp{\tilde{\tmfourp}} \esub{\var}{\tmfive} = \tmthree$, with $\tilde{\tmfour} \tomo \tilde{\tmfourp}$ and $\tmfive$ is a \pointed term. Let $\tmp \defeq \evsctxp{\tilde{\tmfourp}} \esub{\var}{\rctxp{\tilde{\tmfivep}}}$, having that 
	$$
	\tmtwo = \evsctxp{\tilde{\tmfour}} \esub{\var}{\rctxp{\tilde{\tmfivep}}} \toms \tmp \reversetoes \evsctxp{\tilde{\tmfourp}} \esub{\var}{\rctxp{\tilde{\tmfive}}} = \tmthree
	$$
	Note that $\tmtwo \toms \tmp$ holds because $\rctxp{\tilde{\tmfivep}}$ is a \pointed term ---by \reflemmap{properties-of-rigid-terms}{strong-strategy-preserves-rigid}.
	
	\item Let $\tmtwo = \tmfour \esub{\var}{\rctx_{1} \ctxholep{\tilde{\tmfivep}}} \reversetoes \tmfour \esub{\var}{\rctx_{1} \ctxholep{\tilde{\tmfive}}} = \tm$, with $\tilde{\tmfivep} \reversetoeo \tilde{\tmfive}$, and $\tm = \rctx_{2} \ctxholep{\tilde{\tmfour}} \esub{\var}{\tmfive} \toms \rctx_{2} \ctxholep{\tilde{\tmfourp}} \esub{\var}{\tmfive}$, with $\tilde{\tmfour} \tomo \tilde{\tmfourp}$ and $\tmfive$ is a \pointed term. Let 
	
		$$
			\tmp \defeq \rctx_{2} \ctxholep{\tilde{\tmfourp}} \esub{\var}{\rctx_{1} \ctxholep{\tilde{\tmfivep}}}
		$$
	
	having that
	
	$
	\begin{array}{rcl}
		\tmtwo 
		& = & \rctx_{2} \ctxholep{\tilde{\tmfour}} \esub{\var}{\rctx_{1} \ctxholep{\tilde{\tmfivep}}} \\
		& \toms & \tmp \reversetoes \rctx_{2} \ctxholep{\tilde{\tmfourp}} \esub{\var}{\rctx_{1} \ctxholep{\tilde{\tmfive}}} \\
	\end{array}
	$
	
	Note that $\tmtwo \toms \tmp$ holds because $\rctx_{1} \ctxholep{\tilde{\tmfivep}}$ is a \pointed term ---by \reflemmap{properties-of-rigid-terms}{strong-strategy-preserves-rigid}.
	
	\item Let $\tmtwo = \tmfour \esub{\var}{\rctxp{\tilde{\tmfivep}}} \reversetoes \tmfour \esub{\var}{\rctxp{\tilde{\tmfive}}} = \tm$, with $\tilde{\tmfivep} \reversetoeo \tilde{\tmfive}$ and $\tmfour$ is a \pointed term, and $\tm = \openctxp{\tilde{\tmfour}} \esub{\var}{\tmfive} \toms \openctxp{\tilde{\tmfourp}} \esub{\var}{\tmfive} = \tmthree$, with $\tilde{\tmfour} \rtom \tilde{\tmfourp}$. Let $\tmp \defeq \openctxp{\tilde{\tmfourp}} \esub{\var}{\rctxp{\tilde{\tmfivep}}}$, having that 
	$$
	\tmtwo = \openctxp{\tilde{\tmfour}} \esub{\var}{\rctxp{\tilde{\tmfivep}}} \toms \tmp \reversetoes \openctxp{\tilde{\tmfourp}} \esub{\var}{\rctxp{\tilde{\tmfive}}} = \tmthree
	$$
	Note that $\tmp \reversetoes \tmthree$ holds because $\openctxp{\tilde{\tmfourp}}$ is a \pointed term ---by \reflemmap{properties-of-rigid-terms}{open-strategy-preserves-rigid}.
	
	\item Let $\tm = \tmfour \esub{\var}{\rctxp{\tilde{\tmfivep}}} \reversetoes \tmfour \esub{\var}{\rctxp{\tilde{\tmfive}}} = \tm$, with $\tilde{\tmfivep} \reversetoeo \tilde{\tmfive}$ and $\tmfour$ is a \pointed term, and $\tm = \evsctxp{\tilde{\tmfour}} \esub{\var}{\tmfive} \toms \evsctxp{\tilde{\tmfourp}} \esub{\var}{\tmfive}$, with $\tilde{\tmfour} \tomo \tilde{\tmfourp}$ and $\tmfive$ is a \pointed term. Let 
		$$
			\tmp \defeq \evsctxp{\tilde{\tmfourp}} \esub{\var}{\rctxp{\tilde{\tmfivep}}}
		$$ 
	having that 
		$$
			\tmtwo = \evsctxp{\tilde{\tmfour}} \esub{\var}{\rctxp{\tilde{\tmfivep}}} \toms \tmp \reversetoes \evsctxp{\tilde{\tmfourp}} \esub{\var}{\rctxp{\tilde{\tmfive}}} = \tmthree
		$$
	Note that $\tmtwo \toms \tmp$ holds because $\rctxp{\tilde{\tmfivep}}$ is a \pointed term ---by \reflemmap{properties-of-rigid-terms}{strong-strategy-preserves-rigid}---, and that $\tm \reversetoes \tmthree$ holds because $\evsctxp{\tilde{\tmfourp}}$ is a \pointed term ---by \reflemmap{properties-of-rigid-terms}{strong-strategy-preserves-rigid}.			
	
	\item Let $\tmtwo = \tmfour \esub{\var}{\rctx_{1} \ctxholep{\tilde{\tmfivep}}} \reversetoes \tmfour \esub{\var}{\rctx_{1} \ctxholep{\tilde{\tmfive}}} = \tm$, with $\tilde{\tmfivep} \reversetoeo \tilde{\tmfive}$ and $\tmfour$ is a \pointed term, and $\tm = \rctx_{2} \ctxholep{\tilde{\tmfour}} \esub{\var}{\tmfive} \toms \rctx_{2} \ctxholep{\tilde{\tmfourp}} \esub{\var}{\tmfive} = \tmthree$, with $\tilde{\tmfour} \tomo \tilde{\tmfourp}$ and $\tmfive$ is a \pointed term. Let $\tmp \defeq \rctx_{2} \ctxholep{\tilde{\tmfourp}} \esub{\var}{\rctx_{1} \ctxholep{\tilde{\tmfivep}}}$, having that 
	$$
	\tmtwo \rctx_{2} \ctxholep{\tilde{\tmfour}} \esub{\var}{\rctx_{1} \ctxholep{\tilde{\tmfivep}}} \toms \tmp \reversetoes \rctx_{2} \ctxholep{\tilde{\tmfourp}} \esub{\var}{\rctx_{1} \ctxholep{\tilde{\tmfive}}} = \tmthree
	$$
	Note that $\tmtwo \toms \tmp$ holds because $\rctx_{1} \ctxholep{\tilde{\tmfivep}}$ is a \pointed term ---by \reflemmap{properties-of-rigid-terms}{strong-strategy-preserves-rigid}---, and that $\tmp \reversetoes \tmthree$ because $\rctx_{2} \ctxholep{\tilde{\tmfourp}}$ is a \pointed term ---by \reflemmap{properties-of-rigid-terms}{strong-strategy-preserves-rigid}.
	
\end{itemize}

\item \emph{$\mathsf{ES}$ left for $\tmtwo \reversetoes \tm$ and $\mathsf{ES}$ right for $\tm \toms \tmthree$}; \ie, $\tmtwo = \tmfourp \esub{\var}{\tmfive} \reversetoes \tmfour \esub{\var}{\tmfive} = \tm$ and $\tm = \tmfour \esub{\var}{\tmfive} \toms \tmfour \esub{\var}{\tmfivep} = \tmthree$. There are several sub-cases to this, all of which follow the same kind of reasoning as for the case \emph{$\mathsf{ES}$ right for $\tmtwo \reversetoes \tm$ and $\mathsf{ES}$ left for $\tm \toms \tmthree$}. Therefore, we shall leave this case for the reader.

\item \emph{$\mathsf{ES}$ right for both $\tmtwo \reversetoes \tm$ and $\tm \toms \tmthree$}; \ie, 
		$$
			\tmtwo = \tmfour \esub{\var}{\tmfivep} \reversetoes \tmfour \esub{\var}{\tmfive} = \tm
		$$
	and 
		$$
			\tm = \tmfour \esub{\var}{\tmfive} \toms \tmfour \esub{\var}{\tmfivepp} = \tmthree
		$$
	
	By \ih there exists $\tmsix \in \vsubterms$ such that $\tmfivep \toms \tmfive \reversetoes \tmfivepp$. The analysis of the sub-cases, depending on the open/strong/rigid type contexts involved in $\tm \toms \tmtwo$ and $\tm \toms \tmthree$, follows the same schema as for the previous item, all showing that
$$
\tmtwo = \tmfour \esub{\var}{\tmfivep} \toms \tmfour \esub{\var}{\tmsix} \reversetoes \tmfour \esub{\var}{\tmfivepp} = \tmthree
$$

\item \emph{$\mathsf{ES}$ left for both $\tmtwo \reversetoes \tm$ and $\tm \toms \tmthree$}; \ie, 
		$$
			\tmtwo = \tmfourp \esub{\var}{\tmfive} \reversetoes \tmfour \esub{\var}{\tmfive} = \tm
		$$
	and 
		$$
			\tm = \tmfour \esub{\var}{\tmfive} \toms \tmfourpp \esub{\var}{\tmfive} = \tmthree
		$$
	By \ih there exists $\tmsix \in \vsubterms$ such that $\tmfivep \toms \tmfive \reversetoes \tmfivepp$. The analysis of the sub-cases, depending on the open / strong / rigid type contexts involved in $\tm \toms \tmtwo$ and $\tm \toms \tmthree$, follows the same schema as for the previous item, all showing that
$$
\tmtwo = \tmfourp \esub{\var}{\tmfive} \toms \tmsix \esub{\var}{\tmfive} \reversetoes \tmfourpp \esub{\var}{\tmfive} = \tmthree.
\qedhere
$$
		\end{itemize}
	\end{enumerate}
\end{proof}

\begin{proposition}[Properties of $\tovsubs$]
	\label{propappendix:vsc-diamond}
  \NoteState{prop:vsc-diamond}
	\label{lappendix:fullness}
	Let $\tm$ be a VSC term.
	\begin{enumerate}
		\item \label{p:vsc-diamond-diamond}\emph{Diamond}:	$\tovsubs$ is diamond. Moreover, every $\tovsubs$ evaluation to normal form (if any) has the same number of $\toms$ steps.
		\item \emph{Normal forms}: if $\tm$ is $\tovsubs$ normal then it is a strong fireball.
	\end{enumerate}
\end{proposition}

In the following proof we refer to terms of the form $\subctxp\tm$ as \emph{answers}.

\begin{proof}
	\begin{enumerate}
		\item Follows from strong commutation of $\toms$ and $\toes$ (\reflemmap{basic-value-substitution-external}{tom-toe-commute-strong}), from diamond for $\toms$ and $\toes$ (\reflemmap{basic-value-substitution-external}{tom-toe-diamond-strong}),
		and from Hindley-Rosen lemma (\cite[Prop. 3.3.5]{Barendregt84}). By the random descent property (which is a well-known corollary of the diamond recalled in \refsect{calculus}), all $\tovsubs$ evaluation sequences to normal form have the same number of steps. By strong commutation of $\toms$ and $\toes$, they also have the number of $\toms$ steps.
	
	\item To have the right \ih, we prove simultaneously, by induction on $\tm$, the following stronger statements (we recall that all \full inert terms are \full fireballs): 
\begin{enumerate}
	\item \emph{Fireball property}: If $\tm$ is $\esssym$-normal, then $\tm$ is a \full fireball.
	
	\item \emph{Non-value property}: If $\tm$ is $\esssym$-normal and not an \valES, then $\tm$ is a \full inert term.
\end{enumerate}

Cases:

\begin{itemize}
	\item \emph{Variable}, \ie, $\tm = \var$: both properties trivially hold, since $\tm$ is a \full inert term and so a \full fireball.
	
	\item \emph{Abstraction}, \ie, $\tm = \la{\var}{\tmtwo}$:
	\begin{enumerate}
		\item \emph{Non-value property}: vacuously true, as $\tm$ is an abstraction and hence an \valES.			
		\item \emph{Fireball property}: Since $\tm$ is $\esssym$-normal, so is $\tmtwo$.
		By \ih applied to $\tmtwo$ (fireball property), $\tmtwo$ is a \full fireball and hence so is $\tm$ (as a \full value).
	\end{enumerate}
	
	\item \emph{Application}; \ie, $\tm = \tm_{1} \tm_{2}$ (which is not an \valES): 
	\begin{enumerate}
		\item \emph{Non-value property}: Since $\tm$ is $\esssym$-normal, so are $\tm_{1}$ and $\tm_{2}$.  
		Moreover, $\tm_{1}$ is not an \valES (otherwise $\tm$ would be a $\toms$-redex).
		By \ih applied to $\tm_{1}$ (non-value property) and to $\tm_{2}$ (fireball property),
		$\tm_{1}$ is a \full inert term and $\tm_{2}$ is a \full fireball.
		Thus, $\tm$ is a \full inert term.
		
		\item \emph{Fireball property}: We have just proved that $\tm$ is a \full inert term, and hence it is a \full fireball.
		
	\end{enumerate}
	
	\item \emph{Explicit substitutions}, \ie, $\tm = \tm_{1} \esub{\var}{\tm_{2}}$:
	\begin{enumerate}
		\item \emph{Fireball property}: Since $\tm$ is $\esssym$-normal, so are $\tm_{1}$ and $\tm_{2}$.  
		Moreover, $\tm_{2}$ is not an \valES (otherwise $\tm$ would be a $\toevals$-redex).
		By \ih applied to $\tm_{1}$ (fireball property) and to $\tm_{2}$ (non-value property),
		$\tm_{1}$ is a \full fireball and $\tm_{2}$ is a \full inert term.
		Thus, $\tm$ is a \full fireball.
		
		\item \emph{Non-value property}: We have just proved that $\tm$ is a \full fireball.
		If moreover $\tm$ is not a \valES, then $\tm_{1}$ is not an \valES and hence, by \ih applied to $\tm_{1}$ (non-value property), $\tm_{1}$ is a \full inert term.
		Therefore, $\tm$ is a \full inert term.
		\qedhere
	\end{enumerate}
	
\end{itemize}

\end{enumerate}
\end{proof}

\begin{proposition}[$\eqstruct$ is a strong bisimulation]
	\label{propappendix:strong-bisimulation}
  \NoteState{prop:strong-bisimulation}
  If $\tm\eqstruct\tmtwo$ and $\tm\Rew{\mathsf{a}}\tmp$ then there exists $\tmtwop \in \vsubterms$ such that 
$\tmtwo\Rew{\mathsf{a}}\tmtwop$ and $\tmp\eqstruct\tmtwop$, for $\mathsf{a}\in\set{\msym,\esym,  \osym\msym,\osym\esym, \esssym\msym,\esssym\esym}$.
\end{proposition}

\begin{proof}
	Easy adaptation of the proof in \cite[Lemma 12]{AccattoliPaolini12}.
\end{proof}

\section{Proofs of Section~\ref*{SECT:RELAXED} (Relaxed Implementation)}
\label{app:relaxed-implementation}
This section proves \refthmboth{abs-impl} using the following auxiliary lemma.

\begin{lemma}[One-step transfer]
  \label{l:one-step-transfer}
  Let $\mach$ and $(\tostrat, \eqstruct)$ be a relaxed implementation system.
  For any state $\state$ of $\mach$, if $\decode\state \tostrat \tmtwo$ then there is a state $\statetwo$ of $\mach$ 
such that $\state \tomacho^*\tomachb \statetwo$.
\end{lemma}

\begin{proof}
  For any state $\state$ of $\mach$, let $\nfo{\state}$ be 
  the normal form of $\state$ with respect to $\tomacho$: such a state exists and is unique because overhead transitions 
terminate (\refpoint{def-overhead-terminate}) and $\mach$ is deterministic (\refpoint{def-determinism}).
  Since $\tomacho$ is mapped on identities (\refpoint{def-overhead-transparency}), one has $\decode{\nfo{\state}} = 
\decode\state$.
  As $\decode\state$ is not $\to$-normal by hypothesis, the halt property (\refpoint{def-progress}) entails that 
$\nfo{\state}$ is not final, therefore $\state \tomacho^* \nfo{\state} \tomachb \statetwo$ for some state $\statetwo$.
\end{proof}

\begin{theorem}[Sufficient condition for implementations]
\label{thmappendix:abs-impl}
  Let $\mach$ and $(\tostrat, \eqstruct)$
  \NoteState{thm:abs-impl}
  be a relaxed implementation system.
  Then, $\mach$ is a relaxed implementation of $(\tostrat,\equiv)$.
\end{theorem}

\begin{proof}
\begin{enumerate}
\item \emph{Executions to evaluations}: by induction on the length of the execution. It  follows easily from relaxed 
$\beta$-projection and overhead transparency, plus the strong bisimulation of $\equiv$ with respect to $\tostrat$.

\item \emph{Normalizing evaluations to executions}: we prove a more general statement where we replace $\compil\tm$ with 
a general state $\state$, and $\tm$ with $\decode\state$: \emph{if $\deriv:\decode\state \tostrat^* \tmtwo$ with $\tmtwo$  
normal form then there exists an $\mach$-execution $\exec: \state \tomach^* \statetwo$ with $\statetwo$ final such that 
$\decode\statetwo \equiv \tmtwo$ with $\sizebeta\exec \leq \sizem\deriv$.} Then if we instantiate this more general statement on 
$\state = \compil\tm$ we obtain the official statement, because by the initialization constraint of 
the machine we have $\decode{\compil\tm} = \tm$.

The proof of the generalized statement is by induction on $\sizem\deriv$. If $\sizem\deriv=0$ then consider 
$\nfo{\state}$, that by overhead transparency (\refpoint{def-overhead-transparency}) satisfies $\decode{\nfo{\state}} = 
\decode{\state}$. Now, if $\nfo{\state}$ has a $\beta$-transition then from $\decode{\state}$ it is eventually possible to do a $\tom$ steps by relaxed 
$\beta$-projection (\refpoint{def-beta-projection}), which is impossible, because $\sizem\deriv=0$ and by the diamond property all evaluations sequences from a term have the same number and kind of steps. Then, $\nfo{\state}$ is a final state and there is an execution $\exec: \state \tomacho^*\nfo{\state}$ such 
that $\sizebeta\exec= 0$. 

If $\sizem\deriv>0$ then $\decode\state \tostrat^+ \tmtwo$ and $\decode\state \tostrat \tmthree \tostrat^* \tmtwo$ for some 
$\tmthree$. By the one-step transfer lemma (\reflemma{one-step-transfer}), we obtain $\state\tomacho^*\tomachb 
\statetwo$ for some $\statetwo$. By overhead transparency and relaxed $\beta$ projection, we obtain 
$\deriv':\decode\state\tostrat^+\tmfour\equiv\decode\statetwo$ for some $\tmfour$ and with $\sizem{\deriv'} \geq 1$. By diamond of $\tostrat$, we obtain an 
evaluation $\derivtwo:\tmfour \tostrat^* \tmtwo$ such that $\sizem\derivtwo \leq \sizem{\deriv}-1$.  Note also that by strong bisimulation 
of $\equiv$ applied to $\derivtwo$ we obtain $\derivtwo':\decode\statetwo \tostrat^{\size\derivtwo} \tmtwo'$ with $\tmtwo' \equiv 
\tmtwo$. We can then apply the \ih, obtaining an execution $\exectwo:\statetwo \tomach^* \statetwo'$ with $\statetwo'$ 
final and such that $\decode{\statetwo'} \equiv \tmtwo'$, and $\sizebeta\exectwo \leq \sizem\derivtwo \leq \sizem\deriv-1$. Note that 
the execution $\exec:\state \tomacho^*\tomachb \statetwo\tomach^* \statetwo'$ satisfies the statement because 
$\decode{\statetwo'} \equiv \tmtwo' \equiv \tmtwo$ and $\sizebeta\exec = \sizebeta\exectwo +1 \leq \sizem\derivtwo +1 \leq \sizem\deriv$.

\item \emph{Diverging evaluations to executions}: suppose that $\tostrat$ diverges on $\tm$ but $\mach$ terminates, that 
is, that there is an execution $\exec:\compil\tm \tomach^* \state$ with $\state$ final. Then the projection 
$\decode\exec:\tm \tostrat^* \tmtwo \equiv \decode\state$ for some $\tmtwo$ given by point 1 (of this theorem) is a normalizing sequence 
by the halt property (guaranteeing that $\decode\state$ is $\tostrat$ normal), and by strong bisimulation so is the 
evaluation $\tm \tostrat^* \tmtwo$. Then $\tostrat$ normalizes $\tm$ and so, by diamond (precisely by \emph{uniform 
normalization} implied by the diamond property, see footnote \refsect{calculus}), $\tostrat$ cannot diverge on 
$\tm$---absurd. Therefore $\mach$ diverges on $\compil\tm$. Now, since $\tomacho$ terminates, the diverging execution 
from $\compil\tm$ must have infinitely many $\beta$ transitions. Note that a diverging $\tostrat$ sequence necessarily has an infinity of $\tom$-steps, because $\toe$ terminates (\reflemma{vsc-local-termination}).
\qedhere
\end{enumerate}
\end{proof}

\section{Proofs of Section~\ref*{SECT:CRUMBLING} (Compilation \& Read-back)}
\label{app:compilation-read-back}

In this section, we introduce formally the notions of \emph{well-namedness}, crumbling, and read-back. The main results are some fundamental properties about the crumbling translation -- \reflemmaboth{transl-properties} -- and the modular read-back of crumbled terms --
\reflemmaboth{read-back-decomposition-d}. To achieve the goal we will also provide first a number of additional
technical properties on (well-named) environments and on the definitions of read-back.

\subsection{\lat{s} \& well-namedness}

\begin{definition}[Capture-avoiding substitution]
    We denote by $\tm\isub\var\tmtwo$ the term obtained by replacing the variable $\var$ with $\tmtwo$ in $\tm$. The operation of replacing may rename bound variables in order to avoid captures, but we assume that it minimizes the number of renamings by performing only the strictly necessary ones.
\end{definition}

\begin{definition}
    A \emph{substitution} $\sigma$ is a mapping from variables to terms such that it is the identity on all but a finite number of variables.
    Its \emph{domain} $\domain\sigma$ is the finite set of variables that are not mapped to themselves.
    The set $\fv\sigma$ of free variables of $\sigma$ is the set of all variables that occur free in at least one $\sigma(\vartwo)$ for
    $\vartwo \in \domain\sigma$.
   It is a \emph{fireball substitution} if it maps variables to fireballs. We write $\indsub1\indsub2$ for the composed substitution defined by $(\indsub1\indsub2) (\var) \defeq \indsub2(\indsub1(\var))$. We write $\isub\var\tm$ for a substitution on a single variable and $\tm\sigma$ for the term obtained from $\tm$ by the capture-avoiding replacement of every $\var\in\fv\tm$ with $\sigma\var$. Similarly for every other syntactic category, \eg $\env\sigma$. 
   \end{definition}

The terms that we operate on are subject to well-namedness, which basically amounts to Barendregt's variable convention. To lighten up the proofs, we introduce the following shorthand to denote disjoint sets of variables:
\begin{definition}[Disjoint variables -- $\lambda$-calculus version]
	Let $X$ and $Y$ be sets of variable names. We say that $X$ and $Y$ are disjoint (in symbols, $X \Disj Y$) if $X \cap Y = \emptyset$.
\end{definition}

\begin{definition}[Well-named \lat{s}]
	A \lat{} $\tm$ is \emph{well-named} if its bound variables are all distinct, and $\fv\tm\Disj\bv\tm$.
\end{definition}

\subsection{Crumbled environments}

\begin{definition}[Capture-avoiding substitution]
    We denote by $\env\isub\var\vartwo$ and $\mol\isub\var\vartwo$ the environment (resp. bite) obtained by replacing the variable $\var$ with $\vartwo$ in $\env$ (resp. $\mol$). The operation of replacing may rename bound variables in order to avoid captures, but we assume that it minimizes the number of renamings by performing only the strictly necessary ones.
\end{definition}

\begin{definition}[Free and bound variables]
    We define the sets of \emph{free} and \emph{bound variables} of crumbled environments and bites, and the \emph{domain} of crumbled environments, as expected:
    \[\begin{array}{cc}
        \bv\emptyenv \defeq \emptyset &
        \bv{\env\esub\var\mol} \defeq \bv\env \cup \set\var \cup \bv\mol \\
        \fv\emptyenv \defeq \set\varstar &
        \fv{\env\esub\var\mol} \defeq \fv\env \setminus \set\var \cup \fv\mol \\
        \domain\emptyenv \defeq \emptyset &
        \domain{\env\esub\var\mol} \defeq \domain\env \cup \set\var
    \end{array}\]
    \[\begin{array}{ccc}
        \bv{\var} \defeq \emptyset &
        \bv{\var\vartwo} \defeq \emptyset &
        \bv{\la\var\env} \defeq \set\var \cup \bv\env
        \\
        \fv{\var} \defeq \set\var &
        \fv{\var\vartwo} \defeq \set{\var,\vartwo} &
        \fv{\la\var\env} \defeq \fv\env \setminus \set\var
    \end{array}\]
    Moreover, $\allvars\env \defeq \fv\env \cup \bv\env$ and $\allvars\mol \defeq \fv\mol \cup \bv\mol$.
\end{definition}

\begin{definition}[$\alpha$-equality $\AlphaEq$]
    \label{def:crumbling-alpha}
    We define the relation $\AlphaEq$ as the smallest equivalence relation (reflexive, symmetric, transitive) over crumbled forms, satisfying the following properties:
    \begin{enumerate}
        \item\label{p:crumbling-alpha-stuct-ES}
            \emph{Structural, ES:}
            If $\env \AlphaEq \envtwo$ and $\mol \AlphaEq \moltwo$ then $\env\esub\var\mol \AlphaEq \envtwo\esub\var\moltwo$.
        \item\label{p:crumbling-alpha-ren-ES}
            \emph{Rename, ES:}
            $\env\esub\var\mol \AlphaEq \env\isub\var\vartwo\esub\vartwo\mol$ when $\vartwo\not\in\allvars\env$.
        \item\label{p:crumbling-alpha-ren-abs}
            \emph{Rename, abstraction:}
            $\la\var\env \AlphaEq \la\vartwo{(\env\isub\var\vartwo)}$ when $\vartwo\not\in\allvars\env$.
        \item\label{p:crumbling-alpha-stuct-abs}
            \emph{Structural, abstraction:}
            If $\env \AlphaEq \envtwo $ then $\la\var\env \AlphaEq \la\var\envtwo$.
    \end{enumerate}
    Moreover, in points \ref{p:crumbling-alpha-ren-ES} and \ref{p:crumbling-alpha-ren-abs}, $\alpha$-equivalence can rename $\var$ with $\vartwo$ only if they are both in $\crnames\setminus\set\varstar$ (resp. both in $\calcnames$); $\varstar$ cannot be renamed.
\end{definition}

Given an environment $\env$, we denote by $\rename\env$ any environment that is $\alpha$-equivalent to $\env$: we call $\rename\env$ a ``copy'' of $\env$. Later, when necessary, we will attach additional requirements on the variables of $\rename\env$, for instance that $\bv{\rename\env$} is disjoint from a certain set of variables.

In crumbled environments may occur so-called crumbling variables ($\crnames$). We do not handle these variables any different with respect to well-namedness: the only difference is that we ignore the special variable $\varstar$, both when defining disjoint set of variables, and in well-namedness.
\begin{definition}[Disjoint variables -- crumbled version]
    Let $X$ and $Y$ be sets of variable names, \ie{} $X, Y \subseteq \calcnames \cup \crnames$. We say that $X$ and $Y$ are disjoint (in symbols, $X \Disj Y$) if $X \cap Y \subseteq \set\varstar$.
\end{definition}

\begin{definition}[Well-named crumbled forms]
    An environment $\env$ (resp. a bite $\mol$) is \emph{well-named} if its bound variables (not including $\varstar$) are all distinct, and $\fv\env \Disj \bv\env$ (resp. $\fv\mol \Disj \bv\mol$).
\end{definition}

\begin{lemma}[Variables and concatenation]
    \label{l:fv-join-cup}
    For all crumbled environments $\env$ and $\envtwo$:
    \begin{itemize}
        \item $\fv{\env\envtwo} \subseteq \fv\env \setminus \domain\envtwo \cup \fv\envtwo$
        \item $\bv{\env\envtwo} = \bv\env \cup \bv\envtwo$
        \item $\allvars{\env\envtwo} = \allvars\env \cup \allvars\envtwo$
        \item $\domain{\env\envtwo} = \domain\env \cup \domain\envtwo$.    
    \end{itemize}
\end{lemma}
\begin{proof}
    Easy, by induction on the structure of $\envtwo$.
\end{proof}

We provide a formal definition of $\alpha$-equality for crumbled environments.

\begin{definition}[$\alpha$-equality $\AlphaEq$]
    \label{def:es-calculus-alpha}
    We define the relation $\AlphaEq$ as the smallest equivalence relation (reflexive, symmetric, transitive) over terms, satisfying the following properties:
    \begin{enumerate}
        \item\label{p:es-calculus-alpha-stuct-ES}
            \emph{Structural, ES:}
            If $\tm \AlphaEq \tm'$ and $\tmtwo \AlphaEq \tmtwo'$ then $\tm\esub\var\tmtwo \AlphaEq \tm'\esub\var{\tmtwo'}$.
        \item\label{p:es-calculus-alpha-ren-ES}
            \emph{Rename, ES:}
            $\tm\esub\var\tmtwo \AlphaEq \tm\isub\var\vartwo\esub\vartwo\tmtwo$ when $\vartwo\not\in\allvars\tm$.
        \item\label{p:es-calculus-alpha-ren-abs}
            \emph{Rename, abstraction:}
            $\la\var\tm \AlphaEq \la\vartwo{(\tm\isub\var\vartwo)}$ when $\vartwo\not\in\allvars\tm$.
        \item\label{p:es-calculus-alpha-stuct-abs}
            \emph{Structural, abstraction:}
            If $\tm \AlphaEq \tmtwo $ then $\la\var\tm \AlphaEq \la\var\tmtwo$.
    \end{enumerate}
\end{definition}

Free variables of a crumbled environment are stable under $\alpha$-renaming:

\begin{lemma}[Alpha and free variables]\label{l:env-fv-alpha}
    For every environments $\env,\envtwo$: if $\env \AlphaEq \envtwo$ then $\fv{\env} = \fv{\envtwo}$.
\end{lemma}
\begin{proof}
    Easy, by structural induction on the derivation of $\env \AlphaEq \envtwo$.
\end{proof}

\subsection{Properties of well-named environments}

\begin{lemma}[Concatenation of well-named environments]
    \label{l:append-well-named}
    Let $\env, \envtwo$ be well-named crumbled environments such that:
    \begin{itemize}
        \item $\bv\env \Disj \allvars\envtwo$,
        \item $\fv\env \setminus \domain\envtwo \Disj \bv\envtwo$.
    \end{itemize}
    Then $\env\envtwo$ is well-named.
\end{lemma}
\begin{proof}
    By induction on the structure of $\envtwo$:
    \begin{itemize}
        \item If $\envtwo = \emptyenv$, then $\env\envtwo = \env$ and we conclude.
        \item If $\envtwo = \envthree\esub\var\mol$, in order to apply the \ih{} obtaining that $\env\envthree$ is well-named, we need the following properties:
        \begin{itemize}
            \item $\bv\env \Disj \allvars\envthree$. Follows from the hypothesis $\bv\env \Disj \allvars\envtwo$ because $\allvars\envthree \subseteq \allvars\envtwo$.
            \item $\fv\env \setminus \domain\envthree \Disj \bv\envthree$. Note that $\domain\envthree = \domain\env \setminus \set\varthree$, hence $\fv\env \setminus \domain\envthree \subseteq \fv\env \setminus \domain\envtwo \cup \set\varthree$. Also, $\bv\envthree \subseteq \bv\envtwo$. Finally, we use the fact that $\varthree\not\in\bv\envthree$, which follows from the hypothesis that $\envthree\esub\var\mol$ is well-named.
        \end{itemize}
        We have just obtained that $\env\envthree$ is well-named. We need to show that $\env\envthree\esub\var\mol$ is well-named.
        \begin{itemize}
            \item Bound variables are all distinct. It suffices to show that $\bv{\env\envthree}\cup\set\var \Disj \bv\mol$, and that $\var\not\in\bv{\env\envthree}$. Both follow from $\bv\env \Disj \allvars\envtwo$ and from the well-namedness of $\envtwo$.
            \item Free variables are distinct from bound ones. By definition, $\fv{\env\envthree\esub\var\mol} = \fv{\env\envthree}\setminus\set\var\cup\fv\mol$ and $\bv{\env\envthree\esub\var\mol} = \bv{\env\envthree}\cup\set\var\cup\bv\mol$.
            \begin{itemize}
                \item $\fv{\env\envthree}\setminus\set\var \Disj \bv{\env\envthree}\cup\set\var$ because $\fv{\env\envthree} \Disj \bv{\env\envthree}$ by the well-namedness of $\env\envthree$;
                \item $\fv\mol \Disj \bv\mol$ by well-namedness of $\envtwo$;
                \item $\var\not\in\fv\mol$ by well-namedness of $\envtwo$;
                \item $\fv\mol \Disj \bv{\env\envthree} $ if and only if $\fv\mol\Disj \bv\env$ and $\fv\mol \Disj \bv\envthree$. The first follows from $\bv\env \Disj \allvars\envtwo$, the second by well-namedness of $\envtwo$.
                \item $\bv\mol \Disj \fv{\env\envthree}\setminus\set\var$. Note that $\var\not\in\bv\mol$ by well-namedness of $\envtwo$, therefore we prove instead the equivalent $\bv\mol \Disj \fv{\env\envthree}$. The hypothesis $\fv\env \setminus \domain\envtwo \Disj \bv\envtwo$ implies $\fv\env \setminus \domain\envtwo \Disj \bv\mol$, and since by well-namedness of $\envtwo$ $\bv\mol \Disj \domain\envtwo$, $\fv\env \setminus \domain\envtwo \Disj \bv\mol$ is equivalent to $\fv\env \Disj \bv\mol$. Again by well-namedness of $\envtwo$, $\fv\envthree \Disj \bv\mol$, and we conclude because $\fv{\env\envthree} \subseteq \fv\env \cup \fv\envthree$ by \reflemma{fv-join-cup}.
            \qedhere
            \end{itemize}
        \end{itemize}
    \end{itemize}
\end{proof}

An auxiliary lemma that will be used in later sections:
\begin{lemma}\label{l:decomp-well-named-two}
    If $\env\esub\var\mol$ is well-named then $\env$ is well-named and $\fv\mol \Disj \bv\env$.
    \end{lemma}
    \begin{proof}
        Easy by the definition of well-named.
    \end{proof}

\subsection{Read-back properties}

Before proving some important properties of the read-back in \reflemma{properties-unfolding}, we need the following additional properties of VSC terms.

\begin{lemma}[Alpha \& substitution]
	\label{l:alpha-substitivity}
	If $\tm \AlphaEq \tm'$ and $\tmtwo \AlphaEq \tmtwo'$, then $\tm\isub\var\tmtwo \AlphaEq \tm'\isub\var{\tmtwo'}$.
\end{lemma}
\begin{proof}
	Easy by structural induction on the derivation of $\tm \AlphaEq \tm'$, and by the definition of substitution.
\end{proof}

\begin{lemma}[Variables \& substitution]
    \label{l:vars-after-subst}
    Let $\tm,\tmtwo$ be terms. Then:
    \begin{enumerate}
        \item\label{p:vars-after-subst-fv}
        $\fv{\tm\isub\var\tmtwo} \subseteq \fv\tm \setminus\set\var\cup \fv\tmtwo$.
        \item\label{p:vars-after-subst-bv}
        $\bv{\tm\isub\var\tmtwo} \subseteq \bv\tm \cup \bv\tmtwo$ when $\bv\tm \Disj \fv\tmtwo$.
    \end{enumerate}
  \end{lemma}
  \begin{proof}
    Easy, by induction on the structure of $\tm$.
  \end{proof}

  Similar properties hold for crumbled environments:
\begin{lemma}
    \label{l:vars-after-subst-crumbled}
    Let $\env$ be an environment such that $\vartwo\not\in\bv\env$. Then:
    \begin{enumerate}
        \item\label{p:vars-after-subst-crumbled-fv}
        $\fv{\env\isub\var\vartwo} \subseteq \fv\env \setminus\set\var\cup \set\vartwo$.
        \item\label{p:vars-after-subst-crumbled-bv}
        $\bv{\env\isub\var\vartwo} = \bv\env$.
    \end{enumerate}
\end{lemma}
\begin{proof}
    Easy, by induction on the structure of $\env$.
\end{proof}

\begin{lemma}[Properties of the read-back]\label{l:properties-unfolding}
    For all environments $\env,\envtwo$:
    \begin{enumerate}
        \item\label{p:properties-unfolding-alpha}
            \emph{$\alpha$-Equality:}
            if $\env \AlphaEq \envtwo$ and $\env,\envtwo$ are well-named, then 
            $\unf\env \AlphaEq \unf\envtwo$.
        \item\label{p:properties-unfolding-renaming}
            \emph{Renaming:}
            $\unf{\env\isub\var\vartwo} = \unf\env\isub\var\vartwo$ when $\vartwo\not\in\bv\env$.
        \item\label{p:properties-unfolding-fv}
            \emph{Free variables:}
            $\fv{\unf\env} \subseteq \fv\env$.
        \item\label{p:properties-unfolding-bv}
            \emph{Bound variables:}
            $\bv{\unf\env} \subseteq \bv\env$ when $\env$ is well-named.
        \item \label{p:properties-unfolding-tmpintegrare}
         \emph{Dissociation:}
    if $\fv\envtwo \Disj \domain\env$ and $\domain\env \subset \crnames$
    then $\unf{\envtwo\esub\var\mol\env} = \unf\envtwo\isub\var{\unf{\esub\varstar\mol\env}}$.
    \end{enumerate} 
\end{lemma}
\begin{proof}
    We prove the points by mutual induction on the size of $\env$:
    \begin{enumerate}
        \item We prove at the same time the corresponding statement for bites, \ie{} that $\mol \AlphaEq \moltwo$ implies $\unf\mol \AlphaEq \unf\moltwo$ for all bites $\mol,\moltwo$. By structural induction on the derivation of respectively $\env \AlphaEq \envtwo$ or $\mol \AlphaEq \moltwo$:
            \begin{itemize}
                \item \emph{Reflexivity:} in this case $\env = \envtwo$ and $\mol=\moltwo$, and therefore trivially $\unf\env = \unf\envtwo$ and $\unf\mol = \unf\moltwo$.
                \item \emph{Symmetry:} $\env \AlphaEq \envtwo$ because $\envtwo \AlphaEq \env$. Then by \ih{} we obtain $\unf\envtwo \AlphaEq \unf\env$, and we conclude because $\AlphaEq$ is symmetric. Similarly for bites.
                \item \emph{Transitivity:}
                    $\env \AlphaEq \envtwo$ because $\env \AlphaEq \envthree$ and $\envthree \AlphaEq \envtwo$. Just use the \ih{} and use symmetry of $\AlphaEq$. Similarly for bites.
                \item \emph{Structural, ES:}
                    $\env \AlphaEq \envtwo$ because $\env = \envthree\esub\var\mol$, $\envtwo = \envfour\esub\var\moltwo$ where $\envthree \AlphaEq \envfour$ and $\mol \AlphaEq \moltwo$.
                    Then $\unf\env = \unf{\envthree\esub\var\mol}$, and there are two subcases:
                    \begin{itemize}
                        \item $\var\in\crnames$ or $\mol$ abstraction: then $\unf{\envthree\esub\var\mol} = \unf\envthree\isub\var{\unf\mol}$ and $\unf{\envfour\esub\var\moltwo} = \unf\envfour\isub\var{\unf\moltwo}$. By \ih{} $\unf\envthree \AlphaEq \unf\envfour$ and $\unf\mol \AlphaEq \unf\moltwo$. We conclude by \reflemma{alpha-substitivity}.
                        \item $\var\in\calcnames$: then $\unf{\envthree\esub\var\mol} = \unf\envthree\esub\var{\unf\mol}$ and $\unf{\envfour\esub\var\moltwo} = \unf\envfour\esub\var{\unf\moltwo}$. By \ih{} $\unf\envthree \AlphaEq \unf\envfour$ and $\unf\mol \AlphaEq \unf\moltwo$. Conclude by the property \emph{Structural, ES} of $\AlphaEq$ for the ES calculus.
                    \end{itemize}
                \item \emph{Rename, ES}:
                    $\env \AlphaEq \envtwo$ because $\env = \envthree\esub\var\mol$ and $\envtwo = \envthree\isub\var\vartwo\esub\vartwo\mol$ with $\vartwo\not\in\allvars\env$. Again two subcases:
                    \begin{itemize}
                        \item $\var\in\crnames$ or $\mol$ abstraction: then $\unf\env = \unf\envthree\isub\var{\unf\mol}$ and $\unf\envtwo = \unf{\envthree\isub\var\vartwo}\isub\vartwo{\unf\mol}$. \sloppy By \refpoint{properties-unfolding-renaming} $\unf{\envthree\isub\var\vartwo}\isub\vartwo{\unf\mol} = \unf\envthree\isub\var\vartwo\isub\vartwo{\unf\mol}$ because $\vartwo\not\in\bv\envthree$. Since $\vartwo\not\in\allvars\envthree$, by \refpoint{properties-unfolding-fv} $\vartwo\not\in\bv{\unf\envthree}$. Therefore $\unf\envthree\isub\var\vartwo\isub\vartwo{\unf\mol} = \unf\envthree\isub\var{\unf\mol}$ and we conclude.
                        \item Otherwise, \sloppy $\unf\env = \unf\envthree\esub\var{\unf\mol}$ and $\unf\envtwo = \unf{\envthree\isub\var\vartwo}\esub\vartwo{\unf\mol}$. As in the point above, $\unf{\envthree\isub\var\vartwo}\esub\vartwo{\unf\mol} = \unf\envthree\isub\var\vartwo\esub\vartwo{\unf\mol}$ and we conclude by the points ``Rename, ES'' of $\AlphaEq$ for the ES calculus.
                    \end{itemize}
                \item \emph{Rename, abstraction:} $\mol \AlphaEq \moltwo$ because $\mol = \la\var\env$ and $\moltwo = \la\vartwo{(\env\isub\var\vartwo)}$ for $\vartwo\not\in\allvars\env$. Then $\unf\mol = \la\var{(\unf\env)}$ and $\unf\moltwo = \la\vartwo{(\unf{\env\isub\var\vartwo})}$. By \refpoint{properties-unfolding-renaming} $\unf{\env\isub\var\vartwo} = \unf\env\isub\var\vartwo$. By \refpoint{properties-unfolding-fv} and \refpoint{properties-unfolding-bv} $\vartwo\not\in\allvars{\unf\env}$, and we conclude by the point ``Rename, abstraction'' of $\AlphaEq$ for the ES calculus.
                \item \emph{Structural, abstraction:} $\mol \AlphaEq \mol$ because $\mol=\la\var\env$ and $\moltwo=\la\var\envtwo$ and $\env\AlphaEq\envtwo$. Use the \ih{} and conclude by the point ``Structural, abstraction'' of $\AlphaEq$ for the ES calculus.
            \end{itemize}

        \item We prove at the same time the corresponding statement for bites, \ie{} that $\unf{\mol\isub\var\vartwo} = \unf\mol \isub\var\vartwo$ when $\vartwo\not\in\bv\mol$. By structural induction on $\env$ and $\mol$:
            \begin{itemize}
                \item The cases $\env=\emptyenv$, $\mol = \varthree$, and $\mol=\varthree\varfour$ are trivial.
                \item If $\mol = \la\varthree\env$, then  $\unf\mol\isub\var\vartwo = (\la\varthree{(\unf\env)})\isub\var\vartwo$. There are two subcases:
                \begin{itemize}
                    \item If $\var = \varthree$, then $(\la\varthree\env)\isub\var\vartwo = \la\varthree\env$ and $(\la\varthree{(\unf\env)})\isub\var\vartwo = \la\varthree{(\unf\env)}$, and we can conclude.
                    \item If $\var\neq\varthree$, then $(\la\varthree\env)\isub\var\vartwo = \la\varthree{(\env\isub\var\vartwo)}$ and $(\la\varthree{(\unf\env)})\isub\var\vartwo = \la\varthree{(\unf\env\isub\var\vartwo)}$ (recall that $\vartwo\not\in\bv\env$, hence $\vartwo\not\in\bv{\unf\env}$ by \refpoint{properties-unfolding-bv}, and also $\vartwo\neq\varthree$).
                    By \ih{} $\unf{\env\isub\var\vartwo} = \unf\env\isub\var\vartwo$ and we conclude.
                \end{itemize}
                \item If $\env = \envtwo\esub\varthree\mol$, then $\unf\env\isub\var\vartwo = \unf{\envtwo\esub\varthree\mol}\isub\var\vartwo$. There are two subcases:
                \begin{itemize}
                    \item \sloppy If $\varthree\in\crnames$ or $\mol$ abstraction, $\unf{\envtwo\esub\varthree\mol}\isub\var\vartwo = \unf\envtwo\isub\varthree{\unf\mol}\isub\var\vartwo$. Note that $\vartwo\neq\varthree$ by the hypothesis that $\vartwo\not\in\bv\env$. Two subsubcases:
                    \begin{itemize}
                        \item \sloppy If $\var=\varthree$, then $\unf\envtwo\isub\varthree{\unf\mol}\isub\var\vartwo = \unf\envtwo\isub\varthree{\unf\mol\isub\var\vartwo}$. By \ih{} $\unf\mol\isub\var\vartwo = \unf{\mol\isub\var\vartwo}$, and we conclude with $\unf\envtwo\isub\varthree{\unf\mol\isub\var\vartwo} = \unf{\envtwo\esub\varthree{\mol\isub\var\vartwo}} = \unf{\envtwo\esub\varthree\mol\isub\var\vartwo}$.
                        \item \sloppy If $\var\neq\varthree$, then $\unf\envtwo\isub\varthree{\unf\mol}\isub\var\vartwo = \unf\envtwo\isub\var\vartwo\isub\varthree{\unf\mol\isub\var\vartwo}$. By \ih{} $\unf\envtwo\isub\var\vartwo = \unf{\envtwo\isub\var\vartwo}$ and $\unf\mol\isub\var\vartwo = \unf{\mol\isub\var\vartwo}$, hence $\unf\envtwo\isub\var\vartwo\isub\varthree{\unf\mol\isub\var\vartwo} = \unf{\envtwo\isub\var\vartwo}\isub\varthree{\unf{\mol\isub\var\vartwo}} = \unf{\envtwo\isub\var\vartwo\esub\varthree{\mol\isub\var\vartwo}} = \unf{\envtwo\esub\varthree\mol\isub\var\vartwo}$.
                    \end{itemize}
                    \item Otherwise, $\unf{\envtwo\esub\varthree\mol}\isub\var\vartwo = \unf\envtwo\esub\varthree{\unf\mol}\isub\var\vartwo$. Note that $\vartwo\neq\varthree$ by the hypothesis that $\vartwo\not\in\bv\env$. Two subsubcases:
                    \begin{itemize}
                        \item \sloppy If $\var=\varthree$, then $\unf\envtwo\esub\varthree{\unf\mol}\isub\var\vartwo = \unf\envtwo\esub\varthree{\unf\mol\isub\var\vartwo}$. By \ih{} $\unf\mol\isub\var\vartwo = \unf{\mol\isub\var\vartwo}$, and we conclude with $\unf\envtwo\esub\varthree{\unf\mol\isub\var\vartwo} = \unf{\envtwo\esub\varthree{\mol\isub\var\vartwo}} = \unf{\envtwo\esub\varthree\mol\isub\var\vartwo}$.
                        \item \sloppy If $\var\neq\varthree$, then $\unf\envtwo\esub\varthree{\unf\mol}\isub\var\vartwo = \unf\envtwo\isub\var\vartwo\esub\varthree{\unf\mol\isub\var\vartwo}$. By \ih{} $\unf\envtwo\isub\var\vartwo = \unf{\envtwo\isub\var\vartwo}$ and $\unf\mol\isub\var\vartwo = \unf{\mol\isub\var\vartwo}$, hence $\unf\envtwo\isub\var\vartwo\esub\varthree{\unf\mol\isub\var\vartwo} = \unf{\envtwo\isub\var\vartwo}\esub\varthree{\unf{\mol\isub\var\vartwo}} = \unf{\envtwo\isub\var\vartwo\esub\varthree{\mol\isub\var\vartwo}} = \unf{\envtwo\esub\varthree\mol\isub\var\vartwo}$.
                    \end{itemize}
                \end{itemize}
            \end{itemize}
        \item 
        We prove at the same time the corresponding statement for bites, \ie{} that $\fv{\unf\mol} \subseteq \fv\mol$ for each well-named $\mol$.
        By induction on the structure of $\env$ and $\mol$:
        \begin{itemize}
            \item If $\env = \emptyenv$, then $\fv{\unf\env} = \fv{\varstar} = \set\varstar$ and $ \fv\env = \set\varstar $.
            \item If $\env = \envtwo\esub\var\mol$, then there are two subcases:
                \begin{itemize}
                    \item If $\var\in\crnames$ or $\mol$ is an abstraction,
                    then $\unf\env = \unf\envtwo\isub\var{\unf\mol}$.
                    By \reflemmap{vars-after-subst}{fv} $\fv{\unf\envtwo\isub\var{\unf\mol}} \subseteq \fv{\unf\envtwo} \setminus \set\var \cup \fv{\unf\mol}$.
                    By \ih{} $\fv{\unf\envtwo} \subseteq \fv\envtwo$ and $\fv{\unf\mol} \subseteq \fv\mol$, and we conclude because $\fv\env = \fv\envtwo \setminus \set\var \cup \fv\mol$.
                    \item Otherwise $\unf\env = \unf\envtwo\esub\var{\unf\mol}$, and $\fv{\unf\env} = \fv{\unf\envtwo} \setminus \set\var \cup \fv{\unf\mol}$. By \ih{} $\fv{\unf\envtwo} \subseteq \fv\envtwo$ and $\fv{\unf\mol} \subseteq \fv\mol$, and we conclude.
                \end{itemize}
            \item If $\mol = \var$, then $\fv{\unf\var} = \fv\var = \set\var$.
            \item If $\mol = \var\vartwo$, then $\fv{\unf{\var\vartwo}} = \fv{\var\vartwo} = \set{\var,\vartwo}$.
            \item If $\mol = \la\var\env$, then $\fv{\unf{\la\var\env}} = \fv{\la\var{(\unf\env)}} = \fv{\unf\env} \setminus \set\var$. By \ih{} $\fv{\unf\env} \subseteq \fv\env$, and we conclude.
        \end{itemize}
        \item
            We prove at the same time the corresponding statement for bites, \ie{} that $\bv{\unf\mol} \subseteq \bv\mol$ for each well-named $\mol$.
            By induction on the structure of $\env$ and $\mol$:
            \begin{itemize}
                \item If $\env = \emptyenv$, then $\bv{\unf\env}=\bv\varstar = \emptyset$ and $\bv\env = \emptyset$.
                \item If $\env = \envtwo\esub\var\mol$, then there are two subcases:
                \begin{itemize}
                    \item If $\var\in\crnames$ or $\mol$ is an abstraction,
                    then $\unf\env = \unf\envtwo\isub\var{\unf\mol}$.
                    By \reflemmap{vars-after-subst}{bv} $\bv{\unf\envtwo\isub\var{\unf\mol}} \subseteq \bv{\unf\envtwo} \cup \bv{\unf\mol}$.
                    By \ih{} $\bv{\unf\envtwo} \subseteq \bv\envtwo$ and $\bv{\unf\mol} \subseteq \bv\mol$, and we conclude because $\bv\env = \bv\envtwo \cup \set\var \cup \bv\mol$.
                    \item Otherwise $\unf\env = \unf\envtwo\esub\var{\unf\mol}$, and $\bv{\unf\env} = \bv{\unf\envtwo} \cup \set\var \cup \bv{\unf\mol}$. By \ih{} $\bv{\unf\envtwo} \subseteq \bv\envtwo$ and $\bv{\unf\mol} \subseteq \bv\mol$, and we conclude.
                \end{itemize}
                \item If $\mol = \var$ then $\fv{\unf\var} = \bv\var = \emptyset$.
                \item If $\mol = \var\vartwo$ then $\bv{\unf{\var\vartwo}} = \bv{\var\vartwo} = \emptyset$.
                \item If $\mol = \la\var\env$, then $\bv{\unf{\la\var\env}} = \bv{\unf\env} \setminus \set\var$ and $\bv{\la\var\env} = \bv\env \setminus\set\var$. By \ih{} $\bv{\unf\env} \subseteq \bv\env$, and we conclude.
            \qedhere
            \end{itemize}
     \item
    We proceed by induction over the structure of $\env$:
    \begin{itemize}
        \item If $\env = \emptyenv$, then $\unf{\envtwo\esub\var\mol\env} = \unf{\envtwo\esub\var\mol} = \unf\envtwo\isub\var{\unf\mol} = \unf\envtwo\esub\var{\unf{\esub\varstar\mol}}$.
        \item \sloppy If $\env = \envthree\esub\vartwo\moltwo$, then $\unf{\envtwo\esub\var\mol\env} = \unf{\envtwo\esub\var\mol\envthree\esub\vartwo\moltwo} = \unf{\envtwo\esub\var\mol\envthree}\isub\vartwo{\unf\moltwo}$. By \ih{} $\unf{\envtwo\esub\var\mol\envthree} = \unf\envtwo \isub\var{\unf{\esub\varstar\mol\envthree}}$, and therefore $\unf{\envtwo\esub\var\mol\envthree}\isub\vartwo{\unf\moltwo} = \unf\envthree \isub\var{\unf{\esub\varstar\mol\envthree}}\isub\vartwo{\unf\moltwo}$.
        
        From the hypothesis $\fv\envtwo \Disj \domain\env$ it follows that $\vartwo\not\in\fv\envtwo$, and therefore $\vartwo\not\in\fv{\unf\envtwo}$ by \refpoint{properties-unfolding-fv}. Hence $\unf\envthree \isub\var{\unf{\esub\varstar\mol\envthree}}\isub\vartwo{\unf\moltwo} = \unf\envthree \isub\var{\unf{\esub\varstar\mol\envthree}\isub\vartwo{\unf\moltwo}} = \unf\envthree \isub\var{\unf{\esub\varstar\mol\envthree\esub\vartwo\moltwo}}$.
    \qedhere
    \end{itemize}
    \end{enumerate}
\end{proof}

The following lemma is an easy consequence of \reflemmap{properties-unfolding}{tmpintegrare}, but we state it separately because it is used later in the proofs about the machine.
\begin{lemma}[Pristine dissociation]
    \label{l:weak-unfolding-pristine-new}
    Let $\env\esub\var\mol$ be a pristine and well-named environment.
    Then there exists an open context $\openctx$ such that for every pristine environment $\envtwo$ satisfying $\fv\env \Disj \domain\envtwo$, it holds:
    \[\unf{(\env\esub\var\mol\envtwo)} = \openctxp{\unf{\esub\varstar\mol\envtwo}}.\]
  \end{lemma}
  \begin{proof}
      By pristinity of $\env\esub\var\mol$, there exists an open context $\openctx$ such that $\unf\env = \openctxp\var$ and $\var\not\in\allvars\openctx$.
      Let $\envtwo$ be any environment satisfying the hypotheses.
      By \reflemmap{properties-unfolding}{tmpintegrare}, $\unf{\env\esub\var\mol\envtwo} = \unf\env\isub\var{\unf{\esub\varstar\mol\envtwo}} = \openctxp{\unf{\esub\varstar\mol\envtwo}}$.
  \end{proof}

\subsection{Read-back is inverse to translation}
We now want to prove \reflemmaboth{transl-properties}, i.e. that the translation of a well-named \lat{} $\tm$ read-backs to $\tm$ itself.

The proof depends on the following technical lemma that relates the free and bound variables of
a translated term to those of the term to be translated.

\begin{lemma}[On the variables of a translated term]
    \label{l:vars-of-crumb}
    Let $\tm$ be a \lat{}.
    \begin{itemize}
        \item $\fv{\mytr\tm} = \fv\tm$
        \item $\bv{\mytr\tm} \cap \calcnames = \bv\tm$
        \item $\domain{\mytr\tm} \subset \crnames$
    \end{itemize}
    Moreover, if $\tm$ is not a variable and $\auxtr\tm = (\vartwo, \env)$:
    \begin{itemize}
        \item $\fv{\env} = \fv\tm$
        \item $\bv{\env} \cap \calcnames = \bv\tm$
        \item $\domain{\env} \subset \crnames$
    \end{itemize}
\end{lemma}
\begin{proof}
    The proof proceeds by mutual induction on the structure of $\tm$. Here however we provide the proof only for the principal crumbling translation $\mytr{(\cdot)}$, as the proof for $\auxtr{(\cdot)}$ is similar.
    \begin{itemize}
        \item If $\tm=\var$, then $\fv{\mytr\tm} = \fv{\esub\varstar\var} = \set{\var}$, $\bv{\mytr\tm} = \bv{\esub\varstar\var} = \set{\varstar}$, and $\domain{\mytr\tm} = \set\varstar$, and we can conclude.
        \item If $\tm=\la\var\tmtwo$, then $\fv{\mytr\tm} = \fv{\esub\varstar{\la\var\mytr\tmtwo}} = \fv{\mytr\tmtwo}\setminus\set{\var}$, $\bv{\mytr\tm} = \set{\varstar,var} \cup \bv{\mytr\tmtwo}$, and $\domain{\mytr\tm} = \set\varstar$. We conclude by using the \ih{}
        \item If $\tm = \tmtwo\tmthree$, then then $\auxtr\tm=\esub\varstar{\var\vartwo}\env\envtwo$ where $(\var,\env)\defeq\auxtr\tmtwo$ and $(\vartwo,\envtwo)\defeq\auxtr\tmthree$. By cases on $\tmtwo$ and $\tmthree$:
        \begin{itemize}
            \item If $\tmtwo$ and $\tmthree$ are both variables, then $\env=\envtwo=\emptyenv$, $\mytr\tm = \esub\varstar{\var\vartwo}$ and $\auxtr\tm = (\varthree, \esub\varthree{\var\vartwo})$. Clearly $\fv{\mytr\tm} = \fv{\esub\varthree{\var\vartwo}} = \set{\var,\vartwo}$, $\bv{\mytr\tm} = \set\varstar$ and $\bv\env = \set\varthree$, and we can conclude because $\set{\varstar, \varthree} \subset \crnames$.
            \item If $\tmtwo$ and $\tmthree$ are both not variables, then $\env,\envtwo\neq\emptyenv$ and $\var\in\domain\env$ and $\vartwo\in\domain\envtwo$. By \ih{}, $\fv{\env} = \fv\tmtwo$, $\fv{\envtwo} = \fv\tmthree$, $\bv{\env} \cap \calcnames = \bv\tmtwo$, $\bv{\envtwo} \cap \calcnames = \bv\tmthree$, and $\domain\env \cup \domain\envtwo \subset \crnames$.
            
            \begin{itemize}
                \item The requirement $\domain{\mytr\tm} \subset \crnames$ follows directly from the \ih{}, since $\domain{\mytr\tm} = \set\varstar \cup \domain\env \cup \domain\envtwo$ by \reflemma{fv-join-cup}.

            \item The requirement $\bv{\env} \cap \calcnames = \bv\tm$ follows easily, since $\bv\env = \set\varstar \cup \bv\env \cup \bv\envtwo$ by \reflemma{fv-join-cup}. 

            \item As for the free variables, note that by \reflemma{fv-join-cup} $\fv{\esub\varstar{\var\vartwo}\env\envtwo} = (\set{\var,\vartwo}\setminus \domain\env \cup \fv\env) \setminus \domain\envtwo \cup \fv\envtwo$. By \ih{} $(\set{\var,\vartwo}\setminus \domain\env \cup \fv\env) \setminus \domain\envtwo \cup \fv\envtwo = (\set{\var,\vartwo}\setminus \domain\env \cup \fv\tmtwo) \setminus \domain\envtwo \cup \fv\tmthree$, which is simply $\fv\tmtwo \cup \fv\tmthree$ because $\domain\env \Disj \domain\envtwo$ since the crumbling variables are chosen as globally fresh during crumbling, and $\fv\tmtwo \Disj \domain\envtwo$ because $\domain\envtwo$ are fresh crumbling variables for the same reason.
            \end{itemize}
        \end{itemize}
        \item The cases when only one among $\tmtwo$ and $\tmthree$ is a variable are similar to the two cases proved above.
        \qedhere
    \end{itemize}
\end{proof}

We also need the following technical and rather uninteresting lemma on the auxiliary translation.
\begin{lemma}[On the auxiliary translation]
    \label{l:crumbling--aux}
    Let $\tm$ be a \lat{} such that $\auxtr\tm = (\var,\env)$. Then:
    \begin{enumerate}
        \item If $\env=\emptyenv$, then $\tm=\var$ and $\var\in\calcnames$.
        \item If $\env\neq\emptyenv$, say $\env\defeq\esub\var\mol\envtwo$, then $\mytr\tm=\esub\varstar\mol\envtwo$ and $\var\in\crnames$.
    \end{enumerate} 
\end{lemma}
\begin{proof}
    By inspection of the rules defining the crumbling transformation.
\end{proof}

\begin{lemma}[Crumbling properties] 
\label{lappendix:transl-properties}
\NoteState{l:transl-properties}
Let $\tm$ be a well-named \lat{}. Then 
\begin{enumerate}
    \item \label{p:transl-properties-init} \emph{Inverse}:  $\unf{\mytr\tm} = \tm$.
    \item \label{p:transl-properties-wn} \emph{Name}: $\mytr\tm$ is well-named.
\end{enumerate}
\end{lemma}
\begin{proof}~
    \begin{enumerate}
        \item 
        We proceed by induction on the structure of $\tm$:
        \begin{itemize}
            \item If $\tm=\var$, then $\mytr\tm = \esub\varstar\var$, and clearly $\unf{\esub\varstar\var} = \var = \tm$.
            \item If $\tm = \la\var\tmtwo$, then $\mytr\tm = \esub\varstar{\la\var\mytr\tmtwo}$. By definition $\unf{\esub\varstar{\la\var\mytr\tmtwo}} = \la\var(\unf{\mytr\tmtwo})$, and we conclude by using the \ih{} $\unf{\mytr\tmtwo} = \tmtwo$.
            \item If $\tm=\tmtwo\tmthree$, then $\mytr\tm=\esub\varstar{\var\vartwo}\env\envtwo$ where $(\var,\env)=\auxtr\tmtwo$ and $(\vartwo,\envtwo)=\auxtr\tmthree$.
            We proceed by cases on $\tmthree$:
            \begin{itemize}
                \item If $\tmthree = \vartwo$, then $\vartwo\in\calcnames$ and $\envtwo = \emptyenv$.
                We proceed by cases on $\tmtwo$:
                \begin{itemize}
                    \item If $\tmtwo = \var$, then also $\env = \emptyenv$. Thus $\mytr\tm=\esub\varstar{\var\vartwo}$ and $\unf{\mytr\tm} = \unf{\esub\varstar{\var\vartwo}} = \var\vartwo=\tmtwo\tmthree=\tm$.
                    \item \sloppy If $\tmtwo \neq \var$, then $\env\neq\emptyenv$, say $\env=\esub\var\mol\envthree$. By \reflemma{crumbling--aux} $\mytr\tmtwo = \esub\varstar\mol\envthree$ and $\var\in\crnames$, and by \reflemma{vars-of-crumb} $\vartwo\not\in\domain\env\cup\set\var$. Hence by \reflemmap{properties-unfolding}{tmpintegrare} $\unf{\esub\varstar{\var\vartwo}\env} = \unf{\esub\varstar{\var\vartwo}}\isub\var{\unf{\esub\varstar\mol\envthree}}$. By \ih{} $\unf{\esub\varstar{\var\vartwo}}\isub\var{\unf{\esub\varstar\mol\envthree}} = \unf{\esub\varstar{\var\vartwo}}\isub\var\tm = \tm\vartwo$.
                \end{itemize}
                \item \sloppy If $\tmthree \neq \vartwo$, then $\envtwo\neq\emptyenv$, say $\envtwo=\esub\var\mol\envthree$. By \reflemma{crumbling--aux} $\mytr\tmthree = \esub\varstar\mol\envthree$ and $\vartwo\in\crnames$. By \reflemmap{properties-unfolding}{tmpintegrare} $\unf{\esub\varstar{\var\vartwo}\env\envtwo} = \unf{\esub\varstar{\var\vartwo}\env}\isub\vartwo{\unf{\esub\varstar\mol\envthree}}$. By \ih{} $\unf{\esub\varstar{\var\vartwo}\env}\isub\vartwo{\unf{\esub\varstar\mol\envthree}} = \unf{\esub\varstar{\var\vartwo}\env}\isub\vartwo\tmthree$. We now proceed by cases on $\tmtwo$:
                \begin{itemize}
                    \item If $\tmtwo = \var$, then also $\env = \emptyenv$ and $\var\in\calcnames$. Thus $\unf{\esub\varstar{\var\vartwo}\env}\isub\vartwo\tmthree = \unf{\esub\varstar{\var\vartwo}}\isub\vartwo\tmthree = \var\vartwo\isub\vartwo\tmthree = \var\tmthree$. 
                    \item \sloppy If $\tmtwo \neq \var$, then $\env\neq\emptyenv$, say $\env=\esub\var\mol\envthree$. By \reflemma{crumbling--aux} $\mytr\tmtwo = \esub\varstar\mol\envthree$ and $\var\in\crnames$. Moreover, $\var\neq\vartwo$ because they are fresh crumbling variables created in distinct branches of the crumbling transformation. Hence by \reflemmap{properties-unfolding}{tmpintegrare} $\unf{\esub\varstar{\var\vartwo}\env}\isub\vartwo\tmthree = \unf{\esub\varstar{\var\vartwo}}\isub\var{\unf{\esub\varstar\mol\envthree}}\isub\vartwo\tmthree$. By \ih{} $\unf{\esub\varstar{\var\vartwo}}\isub\var{\unf{\esub\varstar\mol\envthree}}\isub\vartwo\tmthree = \unf{\esub\varstar{\var\vartwo}}\isub\var\tmtwo\isub\vartwo\tmthree = \var\vartwo\isub\var\tmtwo\isub\vartwo\tmthree$. We conclude with $\var\vartwo\isub\var\tmtwo\isub\vartwo\tmthree = \tmtwo\tmthree$ because $\vartwo\not\in\fv\tmtwo$ since $\vartwo\in\crnames$.
                \end{itemize}
            \end{itemize}
        \end{itemize}
        \item
        We prove at the same time the corresponding statement for $\auxtr{(\cdot)}$, \ie{} that if $\tm$ is well-named and $\auxtr\tm=(\var,\env)$, then $\env$ is well-named.
        By mutual induction on the size of $\tm$:
        \begin{itemize}
            \item If $\tm=\var$ for some variable $\var$, then $\mytr\tm = \esub\varstar\var$ which is well-named. $\auxtr\tm = (\var, \emptyenv)$ and $\emptyenv$ is clearly well-named.
            \item If $\tm=\la\var\tmtwo$, then $\mytr\tm = \esub\varstar{\la\var{\mytr\tmtwo}}$. By \ih{} $\mytr\tmtwo$ is well-named, thus $\esub\varstar{\la\var{\mytr\tmtwo}}$ is well-named as well, if $\var\not\in\bv{\mytr\tmtwo}$ (which holds by \reflemma{vars-of-crumb} and well-namedness of $\tm$).
                Similarly for $\auxtr\tm = (\varthree, \esub\varthree{\la\var{\mytr\tmtwo}})$, since $\varthree\not\in\allvars{\la\var{\mytr\tmtwo}}$ because it is a fresh local variable.
            \item If $\tm = \tmtwo\tmthree$, then $\mytr\tm=\esub\varstar{\var\vartwo}\env\envtwo$ where $(\var,\env)\defeq\auxtr\tmtwo$ and $(\vartwo,\envtwo)\defeq\auxtr\tmthree$. By \ih{} both $\env$ and $\envtwo$ are well-named.
            We apply twice \reflemma{append-well-named} to conclude, but we first need to prove the following disjointedness conditions:
            \begin{itemize}
                \item $\bv{\esub\varstar{\var\vartwo}} \Disj \allvars\env$ follows from the definition of $\Disj$ because $\bv{\esub\varstar{\var\vartwo}} = \set\varstar$.
                \item $\fv{\esub\varstar{\var\vartwo}} \setminus \domain\env \Disj \bv\env$, that is $\set{\var,\vartwo} \setminus \domain\env \Disj \bv\env$.
                We proceed by cases on $\tmtwo$ and $\tmthree$:
                \begin{itemize}
                    \item If $\tmtwo = \var$, then $\env=\emptyenv$ and we conclude because $\bv\env = \emptyset$.
                    \item If $\tmtwo\neq\var$ and $\tmthree\neq\vartwo$, then $\var\in\domain\env$ and $\vartwo\in\domain\envtwo$. Note that $\domain\env \Disj \domain\envtwo$ by the definition of crumbling, which always generates fresh crumbling variables. Thus $\set{\var,\vartwo} \setminus \domain\env = \set\vartwo$. By \reflemma{vars-of-crumb}, $\bv\env\cap\calcnames = \bv\tmtwo$. Conclude again by the property of freshness during crumbling.
                    \item If $\tmtwo\neq\var$ and $\tmthree=\vartwo$, then $\var\in\domain\env$ and $\envtwo = \emptyenv$. In this case $\set{\var,\vartwo} \setminus \domain\env = \set\vartwo$ because $\vartwo\not\in\crnames$ while $\domain\env \subset \crnames$ by \reflemma{vars-of-crumb}. By well-namedness of $\tm$, $\vartwo \not\in\bv\tmtwo$, and conclude by \reflemma{vars-of-crumb}.
                \end{itemize}
                \item $\bv{\esub\varstar{\var\vartwo}\env} \Disj \allvars\envtwo$, that is $\bv\env \Disj \allvars\envtwo$, \ie{} $\bv\env \Disj \fv\envtwo$ and $\bv\env \Disj \bv\envtwo$. 
                By \reflemma{vars-of-crumb} $\fv\envtwo = \fv\tmthree$, $\bv\env \cap \calcnames = \bv\tm$ and $\bv\envtwo \cap \calcnames = \bv\tmthree$. From the hypothesis that $\tm$ is well-named, it follows that $\fv\tmtwo \Disj \bv\tmthree$, and that $\bv\tmtwo \Disj \bv\tmthree$. The remaining requirements follow from the definition of crumbling, where crumbling variables are always chosen to be globally fresh.
                \item \sloppy $\fv{\esub\varstar{\var\vartwo}\env} \setminus \domain\envtwo \Disj \bv\envtwo$. By \reflemma{fv-join-cup}, $\fv{\esub\varstar{\var\vartwo}\env} = \set{\var,\vartwo}\setminus\domain\env \cup \fv\env$, and by \reflemma{vars-of-crumb} $\set{\var,\vartwo}\setminus\domain\env \cup \fv\env = \set{\var,\vartwo}\setminus\domain\env \cup \fv\tmtwo$.
                First of all, note that $\domain\envtwo \Disj \fv\tmtwo$ because $\domain\envtwo$ only contains crumbling variables (\reflemma{vars-of-crumb}); therefore $\fv{\esub\varstar{\var\vartwo}\env} \setminus \domain\envtwo = \set{\var,\vartwo}\setminus\domain\env \setminus \domain\envtwo \cup \fv\tmtwo$. We proceed by cases on $\tmtwo, \tmthree$:
                \begin{itemize}
                    \item If $\tmthree = \vartwo$ then $\envtwo = \emptyenv$, and we conclude because $\bv\envtwo = \emptyset$.
                    \item If $\tmtwo=\var$ and $\tmthree\neq\vartwo$, then $\var\in\calcnames$, $\env = \emptyenv$, and $\vartwo\in\domain\envtwo$. In this case $\set{\var,\vartwo}\setminus\domain\env \setminus \domain\envtwo \cup \fv\tmtwo = \set{\var} $. By well-namedness of $\tm$, $\set{\var,\vartwo} \Disj \bv\tmtwo$, and we conclude with $\set{\var,\vartwo} \Disj \bv\tmtwo$ by \reflemma{vars-of-crumb}.
                    \item If $\tmtwo\neq\var$ and $\tmthree\neq\vartwo$, then $\var\in\domain\env$ and $\vartwo\in\domain\envtwo$. In this case $\set{\var,\vartwo}\setminus\domain\env \setminus \domain\envtwo \cup \fv\tmtwo = \fv\tmtwo $. By well-namedness of $\tm$, $\fv\tmtwo \Disj \bv\tmtwo$, and we conclude with $\fv\tmtwo \Disj \bv\tmtwo$ by \reflemma{vars-of-crumb}.
                \end{itemize}
            \end{itemize}
            As for the case of $\auxtr\tm=(\varthree,\esub\varthree{\var\vartwo}\env\envtwo)$ where $(\var,\env)\defeq\auxtr\tmtwo$ and $(\vartwo,\envtwo)\defeq\auxtr\tmthree$, the proof proceeds in a similar way as above.
        \end{itemize}
    \end{enumerate}
\end{proof}

\subsection{On the properties of $\indsub\cdot$ and $\indenv\cdot$}

Here we collect several properties that relate an environment $\env$ and its associated $\indsub\env$ and $\indenv\env$.

\begin{lemma}\label{l:indsub-domain-b}
 $\domain{\indsub\env} \subseteq \domain\env$
\end{lemma}
\begin{proof}
By structural induction over $\env$:
\begin{itemize}
 \item Case $\emptyenv$: $\domain{\indsub\env} = \domain{Id} = \emptyset = \domain\emptyenv$.
 \item Case $\envtwo\esub\var\mol$:
  $\indsub{\envtwo\esub\var\mol}$ is either $\indsub\envtwo$ or $\indsub\envtwo\isub\var{\unf\mol}$.
  Composing $\indsub\envtwo$ with $\isub\var{\unf\mol}$ can add $\var$ to the domain of $\indsub\envtwo$
  and remove an arbitrary number of other variables that, after the substitution, can now be mapped to
  themselves. Therefore we conclude $\domain{\indsub\envtwo\isub\var{\unf\mol}} \subseteq
   \domain{\indsub\envtwo} \cup \set\var \subseteq_\ih \domain\envtwo \cup \set\var
   = \domain{\envtwo\esub\var\mol}$.
\end{itemize}
\end{proof}

\begin{lemma}\label{l:sigma-composition}
$\indsub{\env\envtwo} = \indsub\env \indsub\envtwo$
\end{lemma}
\begin{proof}
By structural induction over $\envtwo$:
\begin{itemize}
 \item Case $\emptyenv$: $\indsub\env = \indsub\env\, Id = \indsub\env\indsub\emptyenv$.
 \item Case $\envthree\esub\var\mol$:
  $\indsub{\env\envthree\esub\var\mol} =
   \indsub{\env\envthree}\indsub{\esub\var\mol} =_\ih
   \indsub\env \indsub\envthree \indsub{\esub\var\mol}
   = \indsub\env \indsub{\envthree\esub\var\mol}$.
\end{itemize}
\end{proof}

In the following statement we commit a little abuse of notation: we write
$\isub{\varfour}{\unf\mol\indsub{\env}}\cup\indsub{\env}$ to mean
$\isub{\varfour}{\unf\mol\indsub{\env}}\cup(\indsub{\env} \setminus \isub\varfour{\indsub{\env}(\varfour)})$,
i.e.  the total function that maps $\varfour$ to $\unf\mol\indsub{\env}$ and every other variable $\var$ to
$\indsub{\env}(\var)$.

\begin{lemma}[Left-to-right induced substitutions and substitution contexts ]
  \label{l:sigma-wk-cons} 
  \hfill
  \begin{enumerate}
    \item \label{p:sigma-wk-cons-a}
      $\indsub{\esub\varfour\mol\env} = \isub{\varfour}{\unf\mol\indsub{\env}}\cup\indsub{\env}$
      if $\mol=\val$ or $\varfour\in\crnames$ and  
      $\indsub{\esub\varfour\mol\env} = \indsub{\env}$ otherwise
    \item \label{p:sigma-wk-cons-c}
      $\indenv{\esub\varfour\mol\env} = \indenv{\env}$ if $\mol=\val$ or $\varfour\in\crnames$
      and 
      $\indenv{\esub\varfour\mol\env} = \esub{\varfour}{\mol\indsub{\env}} \indenv{\env}$ otherwise
  \end{enumerate}
\end{lemma}
\begin{proof}
We have to prove that:
\begin{enumerate}
 \item
  $\indsub{\esub\varfour\mol\env} = \isub{\varfour}{\unf\mol\indsub{\env}}\cup\indsub{\env}$
  if $\mol=\val$ or $\varfour\in\crnames$ and  
  $\indsub{\esub\varfour\mol\env} = \indsub{\env}$ otherwise.
  By \reflemma{sigma-composition},
  $\indsub{\esub\varfour\mol\env} = \indsub{\esub\varfour\mol} \indsub\env$.
  The thesis follows from the definition of composition of substitutions and the definition
  of $\indsub{\esub\varfour\mol}$.
 \item
  $\indenv{\esub\varfour\mol\env} = \indenv{\env}$ if $\mol=\val$ or $\varfour\in\crnames$
  and 
  $\indenv{\esub\varfour\mol\env} = \esub{\varfour}{\mol\indsub{\env}} \indenv{\env}$ otherwise.
  We proceed by structural induction on $\env$.
  \begin{itemize}
   \item Case $\emptyenv$. The property holds by definition of induced substitution context,
    noticing that $\indsub\emptyenv = Id$.
   \item Case $\envtwo\esub\var\moltwo$.
    $$\begin{array}{lll}
     & \indenv{\esub\varfour\mol\envtwo\esub\var\moltwo}\\
    = & \indenv{\esub\varfour\mol\envtwo}\indsub{\esub\var\moltwo}\indenv{\esub\var\moltwo}\\
    =_\ih &
     \left\{\begin{array}{l\colspace l}
      \indenv{\envtwo}\indsub{\esub\var\moltwo}\indenv{\esub\var\moltwo} = \indenv{\envtwo\esub\var\moltwo}
      \\ \ \ \ \mbox{if } \mol=\val$ or $\varfour\in\crnames\\~\\
      \esub{\varfour}{\mol\indsub{\envtwo}} \indenv{\envtwo}\indsub{\esub\var\moltwo}\indenv{\esub\var\moltwo} =\\ \ \
      \esub{\varfour}{\mol\indsub{\envtwo}} \indenv{\envtwo\esub\var\moltwo}
       \ \ \mbox{otherwise.}
     \end{array}\right.
    \end{array}$$
  \end{itemize}
\end{enumerate}
\end{proof}

\begin{lemma}\label{l:lookup-indsub}
If $\env$ is well-named and $\env(\var) = \val$ then
$\indsub\env(\var)$ is a value.
\end{lemma}
\begin{proof}
If $\env(\var) = \val$ then $\env = \env_1\esub\var\val\env_2$ and, by \reflemma{sigma-composition},
$\indsub\env(\var) = \indsub{\env_1\esub\var\val\env_2}(\var) = (\indsub{\env_1} \isub\var{\unf\val} \indsub{\env_2})(\var)$.
Since $\env_1\esub\var\val\env_2$ is well-named,
$\var \not\in \bv{\env_1} \supseteq \domain{\env_1} \supseteq_\reflemmaeq{indsub-domain-b} \domain{\indsub{\env_1}}$
and thus
$(\indsub{\env_1} \isub\var{\unf\val} \indsub{\env_2})(\var) = \unf\val \indsub{\env_2}$, which is a value
by definition of substitution.
\end{proof}

\subsection{Read-back is modular}

\begin{lemma}
\label{lappendix:read-back-decomposition-d}
\NoteState{l:read-back-decomposition-d}
$\unf{(\env\envtwo)} = \indenv\envtwo\ctxholep{\unf\env\indsub\envtwo}$.
\end{lemma}
\begin{proof}
  By induction on $\envtwo$. \emph{Base case}: if $\envtwo = \emptyenv$ then $\indenv\envtwo = \ctxhole$ and $\indsub\envtwo$ is the identity, so that $\indenv\envtwo\ctxholep{\unf\env\indsub\envtwo} = \unf\env$. \emph{Inductive case}: if $\envtwo = \envthree\esub\var\mol$ there are two cases.
	\begin{itemize}
		\item $\mol = \val$ or $\var \in \crnames$ then 
			\[\begin{array}{rcllllll}
				\unf{(\env\envthree\esub\var\mol)} 
				& = & 
				\unf{(\env\envthree)}\isub\var{\unf\mol} 
				\\ & =_{\ih} & 
				\indenv\envthree\ctxholep{\unf\env\indsub\envthree}\isub\var{\unf\mol} 
				\\ & = & 
				\indenv\envthree\isub\var{\unf\mol}\ctxholep{\unf\env\indsub\envthree\isub\var{\unf\mol}} 
				\\ & = & 
				\indenv{\envthree\esub\var\mol}\ctxholep{\unf\env\indsub{\envthree\esub\var\mol}}.
	\end{array}\]
		\item Otherwise:
					\[\begin{array}{rcllllll}
						\unf{(\env\envthree\esub\var\mol)} 
						& = & 
						\unf{(\env\envthree)}\esub\var{\unf\mol}
						\\ & =_{\ih} & 
						\indenv\envthree\ctxholep{\unf\env\indsub\envthree}\esub\var{\unf\mol}
						\\ & = & 
						\indenv{\envthree\esub\var\mol}\ctxholep{\unf\env\indsub{\envthree}}
						\\ & = & 
						\indenv{\envthree\esub\var\mol}\ctxholep{\unf\env\indsub{\envthree\esub\var\mol}}.
					\end{array}\]
	\end{itemize}
\end{proof}

\section{Proofs of Section~\ref*{SECT:OPEN-MACHINE} (Open Crumbling Machine)}
\label{app:ocam}

The aim of this section is to provide all the lemmas required to prove the open machine correct.
We do not provide the latter proof explicitly. Instead we shall later extend the open machine to a
strong machine by adding a new phase and we shall prove that correct. Therefore the proof of correctness of
the strong machine shall entail the one for the open machine and --- of course --- requires all the lemmas
provided in this section.

\subsection{Fundamental property of pristine environments}
One fundamental property of pristine environments is the fact that the crumbling variables introduced by the crumbling transformation occur at most once in the resulting environment. In order to prove \reflemmaboth{aux-most-once}, we first prove the following two auxiliary lemmas:

\begin{lemma}[Pristine free variables]
    \label{l:aux-aux-aux-most-once}
    If $\env$ is pristine and well-named, then $\fv{\env} = \fv{\unf\env}$.
\end{lemma}
\begin{proof}
    One direction is proved in \reflemma{properties-unfolding}. We now prove the inclusion $\fv{\env} \subseteq \fv{\unf\env}$.
    We prove the required statement by mutual induction with the corresponding statement for bites (proved below).
    By induction on the structure of $\env$:
    \begin{itemize}
        \item If $\env=\emptyenv$ then $\fv\emptyenv = \fv{\unf\emptyenv} = \set\varstar$.
        \item If $\env = \envtwo\esub\vartwo\mol$ then, by the definition of pristine, both $\envtwo$ and $\mol$ are pristine, and $\unf\envtwo = \openctxtwop\vartwo$ for an open context $\openctxtwo$ such that $\vartwo\not\in\allvars\openctxtwo$. By \ih{} $\fv\envtwo = \fv{\unf\envtwo}$ and $\fv\mol = \fv{\unf\mol}$. We conclude by \reflemmap{vars-after-subst}{fv} and \reflemmap{vars-after-subst-crumbled}{fv}.
    \end{itemize}
    Now the corresponding statement for bites. We prove that $\fv\mol \subseteq \fv{\unf\mol}$ for every pristine and well-named bite $\mol$:
    \begin{itemize}
        \item If $\mol$ is a variable, then clearly $\fv\mol=\fv{\unf\mol}=\set\mol$.
        \item If $\mol$ is an application $\mol=\var\vartwo$, then clearly $\fv\mol=\fv{\unf\mol}=\set{\var,\vartwo}$.
        \item If $\mol$ is an abstraction $\mol = \la\var\env$, then by \ih{} $\fv\env=\fv{\unf\env}$. We conclude with $\fv\mol=\fv\env\setminus\set\var = \fv{\unf\env} \setminus\set\var = \fv{\unf\mol}$.
    \qedhere
    \end{itemize}
\end{proof}

\begin{lemma}
\label{l:aux-aux-most-once}
    Let $\env$ be a pristine and well-named environment. If $\unf\env=\openctxp\var$ for some open context $\openctx$ such that $\var\not\in\allvars\openctx$, then $\var$ occurs at most once in $\env$.
\end{lemma}
\begin{proof}
    We prove the required statement by mutual induction with the corresponding statement for bites (proved below).
    By induction on the structure of $\env$:
    \begin{itemize}
        \item If $\env=\emptyenv$ then $\unf\env=\varstar$, and the only option is   $\openctx=\ctxhole$ and $\var=\varstar$.
        \item If $\env = \envtwo\esub\vartwo\mol$ then, by the definition of pristine, both $\envtwo$ and $\mol$ are pristine, and $\unf\envtwo = \openctxtwop\vartwo$ for an open context $\openctxtwo$ such that $\vartwo\not\in\allvars\openctxtwo$. By \reflemma{sub-O} $\unf\env = \openctxtwop{\unf\moltwo}$, and therefore we obtain that $\openctxp\var = \openctxtwop{\unf\moltwo}$. Recall the hypothesis $\var\not\in\allvars\openctx$.
        We proceed by cases:
        \begin{itemize}
            \item If $\var\in\allvars\openctxtwo$, then $\var\not\in\allvars{\unf\mol}$ and $\unf\envtwo = \openctxtwop\vartwo = \openctxthreep\var$ for some $\openctxthree$ such that $\var\not\in\allvars\openctxthree$. By \ih{}, $\var$ occurs at most once in $\envtwo$. In order to conclude, it suffices to show that $\var\not\in\allvars{\esub\vartwo\mol}$. As for the bound variables, $\var\not\in\bv{\esub\vartwo\mol}$ holds by well-namedness of $\env$, because $\var$ is free in $\env$. As for the free variables, it follows by \reflemma{aux-aux-aux-most-once}.
            \item If $\var\not\in\allvars\openctxtwo$, then $\openctx = \openctxtwop\openctxthree$ where $\openctxthreep\var = \unf\mol$. By \ih{}, we obtain that $\var$ occurs at most once in $\mol$. Note that by well-namedness $\var\neq\vartwo$. From $\var\not\in\allvars\openctxtwo$ we obtain $\var\not\in\allvars{\openctxtwop\vartwo} = \allvars{\unf\envtwo}$, and we conclude like in the case above by \reflemma{aux-aux-aux-most-once}.
        \end{itemize}
    \end{itemize}
    Now the corresponding statement for bites. We prove that for every pristine and well-named bite $\mol$, if $\unf\mol=\openctxp\var$ for some open context $\openctx$ such that $\var\not\in\allvars\openctx$, then $\var$ occurs at most once in $\mol$:
    \begin{itemize}
        \item If $\mol$ is a variable, then necessarily $\mol=\var$ which clearly occurs once in $\mol$.
        \item If $\mol$ is an application, then necessarily $\mol=\var\vartwo$ ($\openctx=\ctxhole\vartwo$) or $\mol=\vartwo\var$ ($\openctx=\vartwo\ctxhole$) for $\var\neq\vartwo$ (because $\var\not\in\allvars\openctx$), and therefore $\var$ clearly occurs once in $\mol$.
        \item The case when $\mol$ is an abstraction is not possible, because in this case $\unf\mol$ is an abstraction as well, and thus it cannot hold that $\unf\mol=\openctxp\var$ for some open context $\openctx$.
    \end{itemize}
\end{proof}

As a corollary:

\begin{lemma}
    \label{lappendix:aux-most-once}
\NoteState{l:aux-most-once}
 If $\env\esub\var\mol$ is pristine and well-named and $\var\neq\varstar$, then $\var$ occurs exactly once in $\env$.
\end{lemma}
\begin{proof}
    \reflemma{aux-aux-most-once} and the definition of pristine imply that $\var$ occurs at most once in $\env$. From the definition of pristine it also follows that $\var\in\fv{\unf\env}$. To conclude, just note that the unfolding preserves the names of free variables by \reflemma{properties-unfolding} (the only exception is $\varstar$, which can be syntactically present in the unfolding, but not in the original environment, for example in the case $\unf\emptyenv = \varstar$). Thus when a free variable different that $\varstar$ occurs in the unfolding, it must also occur syntactically in the original environment.
\end{proof}

\subsection{Other fundamental properties of pristine environments}
Pristine environments are stable under a certain number of operations that are performed on them during a run of the open machine.
These properties are used in the proof of correctness of the machine to show the invariant that certain environments of the machine
remain pristine during execution, assuming that the property holds for the initial machine state.

In order to prove the essential properties of pristine environments in \reflemma{pristine-properties}, we first introduce some auxiliary properties of open contexts.

The first one is the fact that open contexts are closed under composition:
\begin{lemma}[Composition of open contexts]\label{l:wctxs-composition}
    Let $\openctx$ and $\openctxtwo$ be open contexts. Then their composition $\openctxp\openctxtwo$ is still a open context.
  \end{lemma}
  \begin{proof}
    By an easy inspection of the grammar of open evaluation contexts. It can be proved formally by induction on $\openctx$.
  \end{proof}

  An easy property about the variables of plugged open contexts:
\begin{lemma}
\label{l:ctxs-fv-bv}
Let $\openctx$ be an open context, and $\tm$ be a term.
Then
\begin{enumerate}
        \item $\fv\openctx \subseteq \fv{\openctxp\tm}$.
        \item $\bv\openctx \subseteq \bv{\openctxp\tm}$.
\end{enumerate}
\end{lemma}
\begin{proof}
	Easy, by induction on the structure of $\openctx$.
\end{proof}

The following two lemmas prove properties of open contexts under certain substitutions of terms.

\begin{lemma}[Open contexts and substitutions]
	\label{l:horror}
	Let $\openctx$ be an open context and $\sigma$ be a substitution. Then:
  \begin{enumerate}
    \item\label{p:horror-1} $\openctxp{\mol}\sigma = \openctx\sigma\ctxholep{\mol\sigma}$ when $\bv\openctx \Disj \fv\sigma \cup \domain\sigma$.
    \item\label{p:horror-2} $\openctx\sigma$ is a open context.
  \end{enumerate}
\end{lemma}
\begin{proof}~
	\begin{enumerate}
		\item By induction over the structure of $\openctx$:
		\begin{itemize}
			\item The case $\openctx=\ctxhole$ is trivial.
			\item If $\openctx = \openctxtwo\tm$, then $\openctxp{\mol}\sigma = (\openctxtwop\mol\tm)\sigma = \openctxtwop\mol\sigma(\tm\sigma)$. By \ih{} $\openctxtwop\mol\sigma = \openctxtwo\sigma\ctxholep{\mol\sigma}$. Since $\openctx\sigma = (\openctxtwo\sigma)(\tm\sigma)$ we can conclude.
			\item The case $\openctx = \tm\openctxtwo$ is similar to the one above.
			\item Case $\openctx = \openctxtwo\esub\var\tm$. $\bv\openctx \Disj \fv\sigma \cup \domain\sigma$ implies $\var\not\in\fv\sigma\cup\domain\sigma$. Therefore $\openctxtwo\esub\var\tm\sigma = \openctxtwo\sigma\esub\var{\tm\sigma}$, and we conclude by \ih{}
			\item  Case $\openctx = \tm\esub\var\openctx$. Again, $\bv\openctx \Disj \fv\sigma \cup \domain\sigma$ implies $\var\not\in\fv\sigma\cup\domain\sigma$. Therefore $\tm\esub\var\openctxtwo\sigma = \tm\sigma\esub\var{\openctxtwo\sigma}$, and we conclude by \ih{}
		\end{itemize}
		\item Easy, by induction on the structure of $\openctx$.
		\qedhere
	\end{enumerate}
\end{proof}

\begin{lemma}
  \label{l:sub-O}
  Let $\tm,\tmtwo$ be terms and $\openctx$ be an open context such that $\var\not\in\allvars\openctx$ and $\bv\openctx \Disj \fv\tmtwo$. Then $\openctxp\tm\isub\var\tmtwo = \openctxp{\tm\isub\var\tmtwo}$.
\end{lemma}
\begin{proof}
  By induction on the structure of $\openctx$:
  \begin{itemize}
    \item When $\openctx = \ctxhole$, clearly $\openctxp\tm\isub\var\tmtwo = \tm\isub\var\tmtwo=\openctxp{\tm\isub\var\tmtwo}$.
    \item When $\openctx = \openctxtwo\tmthree$,
       $\openctxp\tm\isub\var\tmtwo = \openctxtwop\tm\isub\var\tmtwo\tmthree\isub\var\tmtwo$. Note that $\bv\openctxtwo \subseteq \bv\openctx$ and $\allvars\openctxtwo \subseteq \allvars\openctx$,
      and therefore by \ih{} $\openctxtwop\tm\isub\var\tmtwo = \openctxtwop{\tm\isub\var\tmtwo}$. Moreover, from $\var\not\in\allvars\openctx$ it follows that $\var\not\in\fv\tmthree$, and therefore $\tmthree\isub\var\tmtwo = \tmthree$. As a consequence, $\openctxtwop\tm\isub\var\tmtwo\tmthree\isub\var\tmtwo = \openctxtwop{\tm\isub\var\tmtwo}\tmthree = \openctxp{\tm\isub\var\tmtwo}$.
    \item The case $\openctx = \tmthree\openctxtwo$ is similar to the case above.
    \item \sloppy When $\openctx = \openctxtwo\esub\vartwo\tmthree$,
      $\openctxp\tm\isub\var\tmtwo = \openctxtwop\tm\esub\vartwo\tmthree\isub\var\tmtwo$. Since by hypothesis $\var\not\in\allvars\openctx$ and $\vartwo\not\in\fv\tmtwo$, $\openctxtwop\tm\esub\vartwo\tmthree\isub\var\tmtwo = \openctxtwop\tm\isub\var\tmtwo\esub\vartwo{\tmthree\isub\var\tmtwo}$.
      Since $\bv\openctxtwo \subseteq \bv\openctx$ and $\allvars\openctxtwo \subseteq \allvars\openctx$, by \ih{} $\openctxtwop\tm\isub\var\tmtwo = \openctxtwop{\tm\isub\var\tmtwo}$, and from the hypothesis that $\var\not\in\allvars\openctx$ it follows that $\var\not\in\fv\tmthree$ and hence $\tmthree\isub\var\tmtwo = \tmthree$.
    \item The case $\openctx = \tmthree\esub\vartwo\openctxtwo$ is similar to the case above.
  \qedhere
  \end{itemize}
\end{proof}

We now have all the ingredients to prove the following important properties of every pristine environment:

\begin{lemma}[Pristine properties]
\label{l:pristine-properties}\hfill 
\begin{enumerate}
 \item \emph{Replacement}: \label{p:pristine-properties-replacement}
 if $\env\esub\var\mol$ is pristine, then $\env\esub\var\moltwo$ is pristine for any pristine $\moltwo$.
 
 \item \emph{Concatenation}: \label{p:pristine-properties-concatenation}
 if $\env\esub\var\mol$ and $\esub\star\moltwo\envtwo$ are pristine, then $\env\esub\var\moltwo\envtwo$ is pristine, under the requirement that $\bv\env \Disj \fv{\esub\varstar\moltwo\envtwo}$ and $\fv\env \Disj \domain\envtwo$.
 
 \item \emph{$\alpha$-Conversion}: \label{p:pristine-properties-alpha}
 if $\env$ is well-named and pristine and $\env \AlphaEq \envtwo$ for $\envtwo$ well-named, then $\envtwo$ is pristine as well.
 
 \item \emph{Renaming}: \label{p:pristine-properties-renaming}
 if $\env$ is well-named and pristine then $\env\isub\var\vartwo$ is pristine when $\vartwo\not\in\bv\env$.

\end{enumerate}
\end{lemma}
\begin{proof}~
\begin{enumerate}
    \item Trivial by the definition of pristine.
    \item
    By induction on $\envtwo$. If $\envtwo=\emptyenv$ then the statement follows directly from \refpoint{pristine-properties-replacement}.
    Let us suppose now that $\envtwo=\envthree\esub\vartwo\molthree$. By the hypothesis that $\esub\star\moltwo\envtwo$ is pristine, it follows that $\vartwo\in\crnames$, that $\molthree$ is pristine, and that $\esub\star\moltwo\envthree$ is pristine. Therefore by \ih{} $\env\esub\var\moltwo\envthree$ is pristine. To prove that $\env\esub\var\moltwo\envthree\esub\vartwo\moltwo$ is pristine, it remains to prove that $\unf{\env\esub\var\moltwo\envthree} = \openctxp\vartwo$ for some open context $\openctx$ such that $\vartwo\not\in\allvars\openctx$.

    First of all, by the pristine hypothesis, $\unf\env = \openctxtwop\var$ and $\unf{\esub\varstar\moltwo\envthree} = \openctxthreep\vartwo$ for some open contexts $\openctxtwo$ and $\openctxthree$ such that $\var\not\in\allvars\openctxtwo$ and $\vartwo\not\in\allvars\openctxthree$.
    We then apply \reflemmap{properties-unfolding}{tmpintegrare} and obtain that $\unf{(\env\esub\var\moltwo\envthree)} = \unf\env\isub\var{\unf{\esub\varstar\moltwo\envthree}} = \openctxtwop\var\isub\var{\openctxthreep\vartwo} = \openctxtwop{\openctxthreep\vartwo}$, where $\openctx \defeq \openctxtwop\openctxthree$ is an open context because the composition of weak contexts is still a weak context (see \reflemma{wctxs-composition}), and $\vartwo\not\in\allvars{\openctxtwop\openctxthree}$ because $\vartwo\not\in\allvars\openctxthree$ and $\vartwo\not\in\allvars\openctxtwo$.
    \item 
        We prove at the same time the corresponding statement for bites, \ie{} that if $\mol$ is pristine and well-named and $\mol \AlphaEq \moltwo$ for $\moltwo$ well-named, then $\moltwo$ is pristine. By structural induction on the derivation of respectively $\env \AlphaEq \envtwo$ or $\mol \AlphaEq \moltwo$:
        \begin{itemize}
            \item \emph{Reflexivity:} in this case $\env = \envtwo$ and $\mol=\moltwo$, and therefore trivially $\envtwo$ and $\moltwo$ are pristine.
            \item \emph{Symmetry:} $\env \AlphaEq \envtwo$ because $\envtwo \AlphaEq \env$. Then by \ih{} we obtain $\envtwo$ pristine, and we conclude because $\AlphaEq$ is symmetric. Similarly for bites.
            \item \emph{Transitivity:}
                $\env \AlphaEq \envtwo$ because $\env \AlphaEq \envthree$ and $\envthree \AlphaEq \envtwo$. Just use the \ih{} and use symmetry of $\AlphaEq$. Similarly for bites.
            \item \emph{Structural, ES:}
                $\env \AlphaEq \envtwo$ because $\env = \envthree\esub\var\mol$, $\envtwo = \envfour\esub\var\moltwo$ where $\envthree \AlphaEq \envfour$ and $\mol \AlphaEq \moltwo$.
                By \ih{} $\envfour$ and $\moltwo$ are pristine, and it remains to show that $\unf\envfour = \openctxp\var$ for some open context $\openctx$ such that $\var\not\in\allvars\openctx$. Since $\env$ is pristine, $\unf\envthree = \openctxp\var$ for some open context $\openctx$ such that $\var\not\in\allvars\openctx$. Conclude by \reflemmap{properties-unfolding}{alpha}.
            \item \emph{Rename, ES}:
                $\env \AlphaEq \envtwo$ because $\env = \envthree\esub\var\mol$ and $\envtwo = \envthree\isub\var\vartwo\esub\vartwo\mol$ with $\vartwo\not\in\allvars\env$. By \ih{} $\envthree$ and $\mol$ are pristine, and it remains to show that $\vartwo\in\crnames$ and $\unf{\envthree\isub\var\vartwo} = \openctxp\vartwo$ for some open context $\openctx$ such that $\vartwo\not\in\allvars\openctx$. $\vartwo\in\crnames$ because $\alpha$-equality renames crumbling variables to crumbling variables. As for the readback, since $\env$ is pristine, $\unf\envthree = \openctxtwop\var$ for some open context $\openctxtwo$ such that $\var\not\in\allvars\openctxtwo$. Let us take $\openctx \defeq \openctxtwo$. By \reflemmap{properties-unfolding}{renaming}, $\unf{\envthree\isub\var\vartwo} = \unf\envthree\isub\var\vartwo = \openctxtwop\var\isub\var\vartwo$. Note from $\vartwo\not\in\allvars\env$, \reflemma{ctxs-fv-bv} and \reflemma{properties-unfolding} follows that $\vartwo\not\in\allvars\openctx$. Finally, by 
                \reflemma{sub-O}, $ \openctxp\var\isub\var\vartwo = \openctxp\vartwo $ (which requires $\var\not\in\allvars\openctx$ and $\vartwo\not\in\bv\openctx$), and we conclude. 
            \item \emph{Rename, abstraction:} $\mol \AlphaEq \moltwo$ because $\mol = \la\var\env$ and $\moltwo = \la\vartwo{(\env\isub\var\vartwo)}$ for $\vartwo\not\in\allvars\env$. Since $\env$ is pristine, by \refpoint{pristine-properties-renaming} $\env\isub\var\vartwo$ is pristine as well, and we conclude.
            \item \emph{Structural, abstraction:} $\mol \AlphaEq \mol$ because $\mol=\la\var\env$ and $\moltwo=\la\var\envtwo$ and $\env\AlphaEq\envtwo$. By \ih{} $\envtwo$ is pristine, and we conclude.
        \end{itemize}
    \item We prove the statement by induction on the structure of $\env$, mutually with the following corresponding statement for bites: if $\mol$ is well-named and pristine then $\mol\isub\var\vartwo$ is pristine when $\vartwo\not\in\bv\env$.
    
    If $\env=\emptyenv$, then $\env\isub\var\vartwo=\emptyenv$ and it is therefore pristine. If $\env=\envtwo\esub\varthree\mol$, there are two subcases:
    \begin{itemize}
        \item If $\var=\varthree$, then $\env\isub\var\vartwo=\envtwo\esub\varthree{\mol\isub\var\vartwo}$. By \ih{} $\mol\isub\var\vartwo$ is pristine, and we can conclude by \refpoint{pristine-properties-replacement}.
        \item If $\var\neq\varthree$, then $\env\isub\var\vartwo=\envtwo\isub\var\vartwo\esub\varthree{\mol\isub\var\vartwo}$ because $\vartwo\neq\varthree$ by hypothesis. By \ih{} $\envtwo\isub\var\vartwo$ and $\mol\isub\var\vartwo$ are pristine, and in order to conclude it suffices to show that $\unf{\envtwo\isub\var\vartwo} = \openctxp\varthree$ for some open context $\openctx$ such that $\varthree\not\in\allvars\openctx$. By pristinity of $\env$ it follows that $\unf\envtwo = \openctxtwop\varthree$ for some open context $\openctxtwo$ such that $\varthree\not\in\allvars\openctxtwo$.
        
        By \reflemmap{properties-unfolding}{renaming} $\unf{\envtwo\isub\var\vartwo} = \unf{\envtwo}\isub\var\vartwo$ because $\vartwo\not\in\bv\envtwo$ and thus also $\vartwo\not\in\bv{\unf\envtwo}$ by \reflemmap{properties-unfolding}{bv}. This implies that $\unf{\envtwo\isub\var\vartwo} = \openctxtwop\varthree\isub\var\vartwo$, which equals $ \openctxtwo\isub\var\vartwo\ctxholep{\varthree\isub\var\vartwo} $ by \reflemmap{horror}{1}, and thus $\openctxtwo\isub\var\vartwo\ctxholep\varthree$. We take as the required context $\openctxtwo \defeq \openctxtwo\isub\var\vartwo$, which is open by \reflemmap{horror}{2}. It remains to prove that $\varthree\not\in\allvars{\openctxtwo\isub\var\vartwo}$, which follows from the fact that $\var\not\in\allvars\openctx$ and by \reflemma{ctxs-fv-bv} and \reflemma{vars-after-subst}.
    \end{itemize}
    Finally, the statement for bites. If $\mol$ is a variable or an application then $\mol\isub\var\vartwo$ is clearly pristine; if $\mol=\la\varthree\envtwo$, there are again two cases:
    \begin{itemize}
        \item If $\var=\varthree$, then $\mol\isub\var\vartwo = \mol$, and we conclude.
        \item If $\var\neq\varthree$, then $\mol\isub\var\vartwo=\la\varthree{(\envtwo\isub\var\vartwo)}$ because $\vartwo\not\in\bv\mol$ \ie{} $\vartwo\neq\varthree$, and we conclude because by \ih{} $\envtwo\isub\var\vartwo$ is pristine.
    \end{itemize}
\end{enumerate}
\end{proof}

\subsection{Compilation produces a pristine environment}
The initial state of an open machine is obtained by compiling a well-named \lat{}, obtaining a pristine environment (\reflemma{transl-properties-pristine}). As we will see, that property of being pristine is preserved during a machine execution, but only on the environments on the left of $\rlsep$.

\begin{lemma}[Properties of the translation]
\label{l:transl-properties-pristine}
Let $\tm$ be a well-named \lat{}. Then $\mytr\tm$ is pristine.
\end{lemma}
\begin{proof}
        By induction on the size of $\tm$:
        \begin{itemize}
            \item If $\tm=\var$, then $\mytr\tm = \esub\varstar\var$, which is pristine.
            \item If $\tm=\la\var\tmtwo$, then $\mytr\tm = \esub\varstar{\la\var{\mytr\tmtwo}}$. Since $\tmtwo$ is also well-named, the \ih{} provides that $\mytr\tmtwo$ is pristine, and we can conclude.
            \item If $\tm=\tmtwo\tmthree$, then $\mytr\tm=\esub\varstar{\var\vartwo}\env\envtwo$ where $(\var,\env)=\auxtr\tmtwo$ and $(\vartwo,\envtwo)=\auxtr\tmthree$.
            
            We proceed to prove that every sub-environment of $\esub\varstar{\var\vartwo}\env\envtwo$ is pristine. The thesis then follows easily. We proceed by induction on the length of the chosen sub-environment:
            \begin{itemize}
                \item Case length equals 1, \ie{} $\esub\varstar{\var\vartwo}$. Clearly pristine.
                \item Case $\esub\varstar{\var\vartwo}\envthree\esub\varthree\mol$ where the ES $\esub\varthree\mol$ belongs to $\env$. By \ih{} $\esub\varstar{\var\vartwo}\envthree$ is pristine. $\varthree\in\crnames$ because $\varthree\in\domain\env$ and by \reflemma{vars-of-crumb}.
                It remains to prove that $\unf{\esub\varstar{\var\vartwo}\envthree} = \openctxp\varthree$ for some open context $\openctx$ such that $\varthree\in\allvars\openctx$.
                \begin{itemize}
                    \item If $\varthree=\var$ then $\envthree = \emptyenv$ and we conclude with $\openctx\defeq \ctxhole\vartwo$.
                    \item \sloppy If $\varthree\neq\var$ then $\envthree \neq \emptyenv$. In this case, suppose $\envthree = \esub\var\moltwo\envfour\esub\varfour\molthree$. By inspection of the definition of crumbling and \reflemma{crumbling--aux}, the environment $\esub\varstar\moltwo\envfour\esub\varfour\molthree$ is an initial sub-environment of $\mytr\tmtwo$, which by \ih{} is pristine. Hence $\unf{\esub\varstar\moltwo\envfour\esub\varfour\molthree} =\openctxtwop\varfour$ for some open context $\openctxtwo$ such that $\varthree\not\in\allvars\openctxtwo$. By \reflemmap{properties-unfolding}{tmpintegrare}, $\unf{\esub\varstar{\var\vartwo}\envthree} = \var\vartwo\isub\var{\openctxtwop\varfour}$. Conclude by taking $\openctx\defeq\openctxtwop\varfour\vartwo$.
                \end{itemize}
                \item Case $\esub\varstar{\var\vartwo}\env\envthree\esub\varthree\mol$ where the ES $\esub\varthree\mol$ belongs to $\envtwo$. By \ih{} $\esub\varstar{\var\vartwo}\env\envthree$ is pristine. $\varthree\in\crnames$ because $\varthree\in\domain\envtwo$ and by \reflemma{vars-of-crumb}.
                It remains to prove that $\unf{\esub\varstar{\var\vartwo}\env\envthree} = \openctxp\varthree$ for some open context $\openctx$ such that $\varthree\in\allvars\openctx$.
                First of all, note that $\unf{\esub\varstar{\var\vartwo}\env} = \unf{\esub\varstar{\var\vartwo}}\isub\var{\unf{\mytr\tm}} = \unf{\esub\varstar{\var\vartwo}}\isub\var\tm = \tmtwo\vartwo$ by the discussion on the previous point and by \refpoint{transl-properties-init}.
                \begin{itemize}
                    \item If $\varthree=\vartwo$ then $\envthree = \emptyenv$ and we conclude with $\openctx\defeq \tm\ctxhole$.
                    \item \sloppy If $\varthree\neq\vartwo$ then $\envthree \neq \emptyenv$. In this case, suppose $\envthree = \esub\vartwo\moltwo\envfour\esub\varfour\molthree$. By inspection of the definition of crumbling and by \reflemma{crumbling--aux}, the environment $\esub\varstar\moltwo\envfour\esub\varfour\molthree$ is an initial sub-environment of $\mytr\tmthree$, which by \ih{} is pristine. Hence $\unf{\esub\varstar\moltwo\envfour\esub\varfour\molthree} =\openctxtwop\varfour$ for some open context $\openctxtwo$ such that $\varthree\not\in\allvars\openctxtwo$. By \reflemmap{properties-unfolding}{tmpintegrare}, $\unf{\esub\varstar{\var\vartwo}\env\envthree} = \tmtwo\vartwo\isub\vartwo{\openctxtwop\varfour} = \tmtwo\openctxtwop\varfour$. Conclude by taking $\openctx\defeq\tmtwo\openctxtwop\varfour$.
                \qedhere
                \end{itemize}
            \end{itemize}
        \end{itemize}
\end{proof}

\section{Proofs of Section~\ref*{SECT:STRONG-MACHINE} (Strong Crumbling Machine)}
\label{app:scam}

There are no proofs to be done for \cref{SECT:STRONG-MACHINE}.
However we include here a bunch of technical lemmas on $\wstenv{(\cdot)}$ that will be used in the rest of the paper.

\begin{lemma}\label{l:wstenv-well-named}
For every $\kctx$:
\begin{enumerate}
 \item \label{p:wstenv-well-named-a} $\fv{\wstenv\kctx} \subseteq \fv\kctx$
 \item \label{p:wstenv-well-named-b} $\bv{\wstenv\kctx} \subseteq \bv\kctx$
 \item \label{p:wstenv-well-named-c} If $\kctx$ is well-named then $\wstenv\kctx$ is well-named
 \item \label{p:wstenv-well-named-d} If $\kctxp\env$ is well-named then $\env$ is well-named
 \item \label{p:wstenv-well-named-e} If $\kctxp\env$ is well-named then $\kctx$ is well-named
\end{enumerate}
\end{lemma}
\begin{proof}
Easy structural induction on $\kctx$.
%
\end{proof}

\begin{lemma}
\label{l:unsenv-prefix}
Let $\kctx$ be a context. Then
\begin{enumerate}
\item 
 \label{p:unsenv-prefix-one}
 $\wstenv{\kctxp{\ctxhole\esub\var\mol}}= \esub\var\mol \wstenv\kctx$.

 \item \label{p:unsenv-prefix-two}
 $\wstenv{\kctxp{\env \esub\var{\la\vartwo\ctxhole}}} = \wstenv\kctx$.
\end{enumerate}
\end{lemma}

\begin{proof}
By induction on $\kctx$. 
\begin{enumerate}
 \item If $\kctx = \ctxhole\envtwo$ then $\wstenv{\kctxp{\ctxhole\esub\var\mol}} = \esub\var\mol\envtwo = \esub\var\mol\wstenv\kctx$ and the statement holds.  If $\kctx = \env_1\esub\var{\la\vartwo\kctxtwo}\env_2$ then $\wstenv{\kctxp{\ctxhole\esub\var\mol}} = \wstenv{\kctxtwop{\ctxhole\esub\var\mol}} \env_2$. By \ih, $\wstenv{\kctxtwop{\ctxhole\esub\var\mol}}= \esub\var\wstenv\kctxtwo$, and so $\wstenv{\kctxp{\ctxhole\esub\var\mol}} = \esub\var\mol \wstenv\kctxtwo\env_2 = \esub\var\mol \wstenv\kctx$.
 
 \item If $\kctx = \ctxhole\envtwo$ then $\wstenv{\kctxp{\env \esub\var{\la\vartwo\ctxhole}}} = \wstenv{\env \esub\var{\la\vartwo\ctxhole}\envtwo} = \wstenv\ctxhole \envtwo = \wstenv\kctx$ and the statement holds.  If $\kctx = \env_1\esub\var{\la\vartwo\kctxtwo}\env_2$ then $\wstenv{\kctxp{\env \esub\var{\la\vartwo\ctxhole}}} = \wstenv{\env_1\esub\var{\la\vartwo\kctxtwop{\env \esub\var{\la\vartwo\ctxhole}}}\env_2} = 
 \wstenv{\kctxtwop{\env \esub\var{\la\vartwo\ctxhole}}}\env_2 =_{\ih} \wstenv{\kctxtwo}\env_2 = \wstenv{\env_1\esub\var{\la\vartwo\kctxtwo}\env_2} = \wstenv\kctx$.
\end{enumerate}

\end{proof}

\paragraph{Variables of plugged machine contexts}
We group here a couple of technical lemmas about the variables of plugged machine contexts.
\begin{lemma}
	\label{l:fv-mol-k}
	$\fv\mol \subseteq \fv\env \cup \bv{\kctxp{\env\esub\var\mol}}$.
\end{lemma}
\begin{proof}
    Easy by induction on the structure of $\env$.
\end{proof}

The following is a generalization of \reflemma{fv-join-cup}:
\begin{lemma}\label{l:fv-join-cup-k}
	For all machine contexts $\kctx$ and environments $\env$:
	\begin{enumerate}
			\item $\fv{\kctxp\env} \subseteq \fv\env \setminus V \cup \fv\kctx$ where $V$ is a set of variables that depends on $\kctx$ but not on $\env$, such that $V \subseteq \bv\kctx$.
			\item $\bv{\kctxp\env} = \bv\kctx \cup \bv\env$
			\item $\allvars{\kctxp\env} = \allvars\kctx \cup \allvars\env$
	\end{enumerate}
\end{lemma}
\begin{proof}
	Easy by structural induction on $\kctx$.
\end{proof}

\section{Proofs of Section~\ref*{SECT:STRONG-IMPLEMENTATION} (Strong Implementation Theorem)}
\label{app:STRONG-IMPLEMENTATION-app}
Following the main part of the paper, we first introduce multi-contexts, their properties
and multi-step reduction. Lastly, we address the Strong Implementation Theorem \refthpboth{machine-final}{impl-scam} that requires them.

\subsection{Multi-contexts and their properties}

We begin studying properties of multi-contexts and proving
\reflemmaboth{multi-step}.

\begin{subsection}{Multi-contexts and multi-steps reduction}
We recall here the terminology introduced in the paper and add some more.

\begin{definition}[Kinds of multi context]
A multi context $\mctx$ is 
\begin{itemize}
\item \emph{Normal} if $\mctxp\sfire$ is a strong fireball for every strong fireball $\sfire$;
\item \emph{Proper} if it has at least one hole;
\item \emph{Fine} if it is strong and proper.
\end{itemize}
\end{definition}

We state an auxiliary lemma that shows properties of strong and rigid multi-contexts
that are required to prove \reflemma{multi-step}.

 
\begin{lemma}
\label{l:mctx-plugging} 
 Let $\strongmctx$ and $\rmctx$ be respectively a strong and a rigid multi contexts, and let $\tm$ be a term. 
 \begin{enumerate}
  \item \label{p:mctx-plugging-strong}
  There exists a term $\tmtwo$ such that $\strongmctxp\tm=\tmtwo$.
  
  \item \label{p:mctx-plugging-rigid}
  There exists a rigid term $\rtm$ such that $\rmctxp\tm=\rtm$.
 \end{enumerate}
 Note that the two points imply that if $\strongmctx$ and $\rmctx$ have no holes then they are a term and a rigid term respectively.
\end{lemma}

\begin{proof}
 By mutual induction on $\strongmctx$ and $\rmctx$. 
 \begin{enumerate}
  \item Cases of $\strongmctx$:
  \begin{itemize}
   \item $\strongmctx = \ctxhole$: obvious.
   \item $\strongmctx = \tm $: obvious.
   \item $\strongmctx =  \la\var \strongmctxtwo $: it follows by the \ih
   \item $\strongmctx =  \rmctx $: by \ih on rigid contexts.
   \item $\strongmctx =  \strongmctxtwo\esub\var\rmctx$: by \ih there exists a term $\tmthree$ and a rigid term $\rtm$ such that $\strongmctxtwop\tm = \tmthree$ and $\rmctxp\tm = \rtm$. Therefore, $\strongmctxp\tm = \tmthree\esub\var\rtm$.

  \end{itemize}
  \item Cases of $\rmctx$:
   \begin{itemize}
  \item $\rmctx = \vartwo$: obvious.
  \item $\rmctx = \rmctxtwo \strongmctx$: by \ih there exists a rigid term $\rtmtwo$ and a term $\tmtwo$ such that $\rmctxtwop\tm = \rtmtwo$ and $\strongmctxp\tm = \tmtwo$ . Therefore, $\rmctxp\tm = \rtmtwo\tmtwo$.
  
  \item $\rmctx = \rmctxtwo\esub\var\rmctxthree$: by \ih there exist  rigid terms $\rtmtwo$ and $\rtmthree$ such that $\rmctxtwop\tm = \rtmtwo$ and $\rmctxthreep\tm = \rtmthree$ . Therefore, $\rmctxp\tm = \rtmtwo \esub\var\rtmthree$.
  \end{itemize}
 \end{enumerate}
 \end{proof}


We can now proceed proving \reflemma{multi-step} mutually
with the corresponding statement for rigid multi contexts:

\begin{lemma}[Multi step]
\label{lappendix:multi-step}
\NoteState{l:multi-step}
 Let $\strongmctx$ and $\rmctx$ be respectively an external and rigid proper multi context with $k$ holes and $\set{a_1,\ldots,a_n}\subseteq\set{\esssym\msym,\esssym\esym}$. If $\tm 
 \Rew{a_1} \cdots \Rew{a_n} \tmtwo$ then 
 \begin{enumerate}
  \item $\rmctxp\tm \,(\Rew{a_1} \cdots \Rew{a_n})^k\, \rmctxp\tmtwo$ where the $i$-th sequence of steps has the shape $\rctx_i\ctxholep\tm \Rew{a_1} \cdots \Rew{a_n} \rctx_i\ctxholep\tmtwo$ for a rigid context $\rctx_i$, for every $i \in \set{1,\ldots, k}$;
  \item $\strongmctxp\tm \,(\Rew{a_1} \cdots \Rew{a_n})^k\, \strongmctxp\tmtwo $ where the $i$-th sequence of steps has the shape $\strongctx_i\ctxholep\tm \Rew{a_1} \cdots \Rew{a_n} \strongctx_i\ctxholep\tmtwo$ for an external context $\strongctx_i$, for every $i \in \set{1,\ldots, k}$.
 \end{enumerate}
\end{lemma}

\begin{proof}
 Let us lighten the notation by writing $\Rew{a_1\cdots a_n}$ in place of $({\Rew{a_1} \cdots \Rew{a_n}})$.
 By mutual induction on $\strongmctx$ and $\rmctx$:
 \begin{enumerate}

     \item Cases of $\rmctx$:
   \begin{itemize}
  \item $\rmctx = \vartwo$: trivial.
  
  \item $\rmctx = \rmctxtwo \strongmctx$: we have $\rmctxp\tm = \rmctxtwop\tm \strongmctxp\tm$ where $\rmctxtwo$ has $k_1$ holes, $\strongmctx$ has $k_2$ holes, and $k_1 + k_2 = k$. We deal with the case where $k_1 \neq 0 \neq k_2$. If $k_1=0$ then by \reflemma{mctx-plugging} $\rmctxtwo = \rmctxtwop\tm$ is a rigid term and we consider only $\strongmctxp\tm$, and dually if $k_2=0$.
    
  By \ih, $\strongmctxp\tm \Rew{a_1\cdots a_n}^{k_2} \strongmctxp\tmtwo$ where for the $i$-th sequence of steps there is an external context $\strongctxtwo_i$ such that the step has the shape $\strongctxtwo_i\ctxholep\tm\Rew{a_1\cdots a_n} \strongctxtwo_i\ctxholep\tmtwo$ for $i \in \set{1,\ldots, k_2}$. Then $\rmctxtwop\tm \strongctxtwo_i$ is a rigid context because by \reflemmap{mctx-plugging}{rigid} $\rmctxtwop\tm$ is a rigid term. Then
  $$\rmctxtwop\tm  \strongmctxp\tm \Rew{a_1\cdots a_n}^{k_2} \rmctxtwop\tm \strongmctxp\tmtwo$$ 
  
  By \ih, $\rmctxtwop\tm \Rew{a_1\cdots a_n}^{k_1} \rmctxtwop\tmtwo$ where for the $j$-th sequence of steps there is a rigid context $\rctxtwo_j$ such that the step has the shape $\rctxtwo_j\ctxholep\tm\Rew{a_1\cdots a_n} \rctxtwo_j\ctxholep\tmtwo$ for $j \in \set{1,\ldots, k_1}$. Then $\rctxtwo_j\strongmctxp\tmtwo$ is a rigid context for every $j$, given that by \reflemmap{mctx-plugging}{strong} $\strongmctxp\tmtwo$ is a term. Therefore, we obtain:
   $$\rmctxtwop\tm\strongmctxp\tmtwo \Rew{a_1\cdots a_n}^{k_1} \rmctxtwop\tmtwo\strongmctxp\tmtwo $$ 
   Summing up,
   $$\rmctxtwop\tm\strongmctxp\tm \Rew{a_1\cdots a_n}^{k_1 +k_2} \rmctxtwop\tmtwo\strongmctxp\tmtwo$$
  The $k_1 +k_2$ rigid contexts of the statement are given by $\rmctxtwop\tm \strongctxtwo_i$ with $i \in \set{1,\ldots, k_2}$ followed by $\rctxtwo_j\strongmctxp\tmtwo$ with $j \in \set{1,\ldots, k_1}$.
  
  \item $\rmctx = \rmctxtwo\esub\vartwo\rmctxthree$: we have $\rmctxp\tm = \rmctxtwop\tm\esub\vartwo{\rmctxthreep\tm}$ where $\rmctxtwo$ has $k_1$ holes, $\rmctxthree$ has $k_2$ holes, and $k_1 + k_2 = k$. We deal with the case where $k_1 \neq 0 \neq k_2$. If $k_1=0$ then by \reflemma{mctx-plugging} $\rmctxtwo = \rmctxtwop\tm$ is a rigid term and we consider only $\rmctxthreep\tm$, and dually if $k_2=0$.
  
  By \ih, $\rmctxtwop\tm \Rew{a_1\cdots a_n}^{k_1} \rmctxtwop\tmtwo$ where for the $i$-th sequence of steps there is a rigid context $\rctxtwo_i$ such that the step has the shape $\rctxtwo_i\ctxholep\tm\Rew{a_1\cdots a_n} \rctxtwo_i\ctxholep\tmtwo$ for $i \in \set{1,\ldots, k_1}$. Then $\rctxtwo_i \esub\vartwo{\rmctxthreep\tm}$ is a rigid context for every $i$ because by \reflemmap{mctx-plugging}{rigid} $\rmctxthreep\tm$ is a rigid term, and so    
   $$\rmctxtwop\tm \esub\vartwo{\rmctxthreep\tm}\Rew{a_1\cdots a_n}^{k_1} \rmctxtwop\tmtwo\esub\vartwo{\rmctxthreep\tm}$$
   
   By \ih, $\rmctxthreep\tm \Rew{a_1\cdots a_n}^{k_2} \rmctxthreep\tmtwo$ where for the $j$-th sequence of steps there is a rigid context $\rctxthree_j$ such that the step has the shape $\rctxthree_j\ctxholep\tm\Rew{a_1\cdots a_n} \rctxthree_j\ctxholep\tmtwo$ for $j \in \set{1,\ldots, k_2}$. Then $\rmctxtwop\tmtwo\esub\vartwo{\rctxthree_j}$ is a rigid context for every $j$, given that by \reflemmap{mctx-plugging}{rigid} $\rmctxtwop\tmtwo$ is a term. Therefore, we obtain:
   $$\rmctxtwop\tmtwo\esub\vartwo{\rmctxthreep\tm} \Rew{a_1\cdots a_n}^{k_2} \rmctxtwop\tmtwo \esub\vartwo{\rmctxthreep\tmtwo}$$ 
   Summing up,
   $$\rmctxtwop\tm \esub\vartwo{\rmctxthreep\tm} \Rew{a_1\cdots a_n}^{k_1 +k_2} \rmctxtwop\tmtwo \esub\vartwo{\rmctxthreep\tmtwo}$$
   The $k_1 +k_2$ rigid contexts of the statement are given by $\rctxtwo_i \esub\vartwo{\rmctxthreep\tm}$ with $i \in \set{1,\ldots, k_1}$ followed by $\rmctxtwop\tmtwo\esub\vartwo{\rctxthree_j}$ with $j \in \set{1,\ldots, k_2}$.
 \end{itemize}

  \item Cases of $\strongmctx$:
  \begin{itemize}
   \item $\strongmctx = \ctxhole$: trivial.
   \item $\strongmctx = \tmthree $: trivial. 
   \item $\strongmctx =  \la\vartwo \strongmctxtwo $: it follows by the \ih
   \item $\strongmctx =  \rmctx $: by \ih on rigid contexts.
   \item $\strongmctx =  \strongmctxtwo\esub\vartwo\rmctx$: we have $\strongmctxp\tm = \strongmctxtwop\tm\esub\vartwo{\rmctxp\tm}$ where $\strongmctxtwo$ has $k_1$ holes, $\rmctx$ has $k_2$ holes, and $k_1 + k_2 = k$. We deal with the case where $k_1 \neq 0 \neq k_2$. If $k_1=0$ then by \reflemma{mctx-plugging} $\strongmctxtwo = \strongmctxtwop\tm$ is a term and we consider only $\rmctxp\tm$, and dually if $k_2=0$.
By \ih, $\strongmctxtwop\tm \Rew{a_1\cdots a_n}^{k_1} \strongmctxtwop\tmtwo$ where for the $i$-th sequence of steps there is an external context $\strongctxtwo_i$ such that the step has the shape $\strongctxtwo_i\ctxholep\tm\Rew{a_1\cdots a_n} \strongctxtwo_i\ctxholep\tmtwo$ for $i \in \set{1,\ldots, k_1}$. Then $\strongctxtwo_i \esub\vartwo{\rmctxp\tm}$ is an external context for every $i$ because by \reflemmap{mctx-plugging}{rigid} $\rmctxp\tm$ is a rigid term, and so    
   $$\strongmctxtwop\tm \esub\vartwo{\rmctxp\tm}\Rew{a_1\cdots a_n}^{k_1} \strongmctxtwop\tmtwo\esub\vartwo{\rmctxp\tm}$$
   
   By \ih, $\rmctxp\tm \Rew{a_1\cdots a_n}^{k_2} \rmctxp\tmtwo$ where for the $j$-th sequence of steps there is a rigid context $\rctx_j$ such that the step has the shape $\rctx_j\ctxholep\tm\Rew{a_1\cdots a_n} \rctx_j\ctxholep\tmtwo$ for $j \in \set{1,\ldots, k_2}$. Then $\strongmctxtwop\tmtwo\esub\vartwo{\rctx_j}$ is an external context for every $j$, given that by \reflemmap{mctx-plugging}{strong} $\strongmctxp\tmtwo$ is a term. Therefore, we obtain:
   $$\strongmctxtwop\tmtwo\esub\vartwo{\rmctxp\tm} \Rew{a_1\cdots a_n}^{k_2} \strongmctxtwop\tmtwo \esub\vartwo{\rmctxp\tmtwo}$$ 
   Summing up,
   $$\strongmctxtwop\tm \esub\vartwo{\rmctxp\tm} \Rew{a_1\cdots a_n}^{k_1 +k_2} \strongmctxtwop\tmtwo \esub\vartwo{\rmctxp\tmtwo}$$
   The $k_1 +k_2$ external contexts of the statement are given by $\strongctxtwo_i \esub\vartwo{\rmctxp\tm}$ with $i \in \set{1,\ldots, k_1}$ followed by $\strongmctxtwop\tmtwo\esub\vartwo{\rctx_j}$ with $j \in \set{1,\ldots, k_2}$.
  \end{itemize}
 \end{enumerate}
\end{proof}
\end{subsection}

\begin{subsection}{Modular read-back}
We prove \reflemma{read-back-decomposition} as \reflemmaappendixp{read-back-decomposition}{c}.
The proof requires two auxiliary and uninteresting subparts, $\reflemmap{read-back-decomposition}{a}$ and
$\reflemmap{read-back-decomposition}{b}$.
\begin{lemma}[Modular read back]
\label{lappendix:read-back-decomposition}
\NoteState{l:read-back-decomposition}
For all environments $\env,\envtwo$ and machine contexts $\kctx$:
\begin{enumerate}
	\item \label{p:read-back-decomposition-a}
	$\unf{(\ctxhole \envtwo)}  =   \indenv\envtwo$
	\item \label{p:read-back-decomposition-b}
	$\unf{(\env\esub\var{\la\vartwo\kctx}\envtwo)}  =  \indenv\envtwo\ctxholep{\unf\env \isub\var{\la\vartwo\unf\kctx}\indsub\envtwo}$
	\item \label{p:read-back-decomposition-c}
	$\unf{\kctxp\env}=\unf\kctx\ctxholep{\unf\env\indsub{\wstenv\kctx}}$
\end{enumerate}
\end{lemma}
\begin{proof}
\hfill
\begin{enumerate}
\item By induction on $\envtwo$. \emph{Base case}: if $\envtwo = \emptyenv$ then $\unf{(\ctxhole \envtwo)} = \ctxhole = \indenv\emptyenv$. \emph{Inductive case}: let $\envtwo = \envthree\esub\var\mol$. If $\mol=\val$ or $\var \in \crnames$ then $\unf{(\ctxhole \envthree\esub\var\mol)} = \unf{(\ctxhole \envthree)}\isub\var{\unf\mol} =_{\ih} \indenv\envthree\isub\var{\unf\mol} = \indenv{\envthree\esub\var\mol}$. Otherwise $\unf{(\ctxhole \envthree\esub\var\mol)} = \unf{(\ctxhole \envthree)}\esub\var\mol =_{\ih} \indenv\envthree\esub\var\mol = \indenv{\envthree\esub\var\mol}$.
\item By induction on $\envtwo$, as in the previous point.
\item By induction on $\kctx$. Cases:
\begin{itemize}
	\item \emph{Base}: $\kctx = \ctxhole\envtwo$. Note that $\envtwo = \wstenv\kctx$. Then $\unf{\kctxp\env} = \unf{\env\envtwo} =_{\reflemmaeq{read-back-decomposition-d}} \indenv\envtwo\ctxholep{\unf\env\indsub\envtwo} =_{\refpointeq{read-back-decomposition-a}}\unf\kctx\ctxholep{\unf\env\indsub\envtwo} = \unf\kctx\ctxholep{\unf\env\indsub{\wstenv\kctx}}$.
	
	\item \emph{Inductive}: $\kctx = \envthree\esub\var{\la\vartwo\kctxtwo}\envtwo$. Then 
	\[\begin{array}{rcllllll}
	\unf{\kctxp\env} & = & \unf{\envthree\esub\var{\la\vartwo\kctxtwop\env}\envtwo} 
	\\
	& =_{\refpointeq{read-back-decomposition-b}} & \indenv\envtwo\ctxholep{\unf{\envthree\esub\var{\la\vartwo\kctxtwop\env}}\indsub\envtwo} 	\\
	& = & \indenv\envtwo\ctxholep{\unf\envthree\isub\var{\la\vartwo\unf{\kctxtwop\env}}\indsub\envtwo }
	\\
	& =_{\ih} & \indenv\envtwo\ctxholep{\unf\envthree\isub\var{\la\vartwo\unf\kctxtwo\ctxholep{\unf\env\indsub{\wstenv\kctxtwo}}}\indsub\envtwo }
	\\
	& = & \indenv\envtwo\ctxholep{\unf\envthree\indsub\envtwo\isub\var{\la\vartwo\unf\kctxtwo\indsub\envtwo\ctxholep{\unf\env\indsub{\wstenv\kctxtwo}\indsub\envtwo}} }	
	\\
	& = & (\indenv\envtwo\ctxholep{\unf\envthree\indsub\envtwo\isub\var{\la\vartwo\unf\kctxtwo\indsub\envtwo }	)\ctxholep{\unf\env\indsub{\wstenv\kctxtwo}\indsub\envtwo}} 
		\\
	& = & (\indenv\envtwo\ctxholep{\unf\envthree\isub\var{\la\vartwo\unf\kctxtwo }\indsub\envtwo	)\ctxholep{\unf\env\indsub{\wstenv\kctxtwo}\indsub\envtwo}}
		\\
	& =_{\refpointeq{read-back-decomposition-b}} & \unf\kctx\ctxholep{\unf\env\indsub{\wstenv\kctxtwo}\indsub\envtwo}
		\\
	& = & \unf\kctx\ctxholep{\unf\env\indsub{\wstenv\kctx}}
	\end{array}\]
	
\end{itemize}
\end{enumerate}
\end{proof}
\end{subsection}

\subsection{Properties of frames}

Here we collect a number of technical lemmas on frames and frames of a context.
They are used in the proof of propagation of the invariants since a few invariants are formulated on frames.

\begin{lemma}\label{l:wk-st-lam}
  \hfill
  \begin{enumerate}    
    \item \label{p:wk-st-lam-a}
    $\framei{\kctxp{\env \esub\var{\la\vartwo\ctxhole}}} = \framei\kctx\ctxholep{\env\esub\var{\la\vartwo\ctxhole}}$.
    \item \label{p:wk-st-lam-b}
    $\framei{\kctxp{\ctxhole \esub\var\mol}} = \framei\kctx$.
  \end{enumerate}
\end{lemma}
\begin{proof}
We have to prove:
\begin{enumerate}
  \item $\framei{\kctxp{\env \esub\var{\la\vartwo\ctxhole}}} = \framei\kctx\ctxholep{\env\esub\var{\la\vartwo\ctxhole}}$.
   We proceed by structural induction on $\kctx$.
   \begin{itemize}
    \item Case $\ctxhole\envtwo$:
      $$\begin{array}{lll}
      \framei{\env \esub\var{\la\vartwo\ctxhole}\envtwo}
      & = & \env\esub\var{\la\vartwo\ctxhole} \\
      & = & \framei{\ctxhole\envtwo}\ctxholep{\env\esub\var{\la\vartwo\ctxhole}}
      \end{array}$$
    \item Case $\envtwo\esub\varthree{\la\varfour\kctxtwo}\envthree$:
      $$\begin{array}{lll}
      \framei{\envtwo\esub\varthree{\la\varfour{\kctxtwop{\env \esub\var{\la\vartwo\ctxhole}}}}\envthree}
      & = & \envtwo\esub\varthree{\la\varfour{\framei{\kctxtwop{\env \esub\var{\la\vartwo\ctxhole}}}}} \\
      & =_\ih & \envtwo\esub\varthree{\la\varfour{\framei\kctxtwo\ctxholep{\env \esub\var{\la\vartwo\ctxhole}}}} \\

      & = & \envtwo\esub\varthree{\la\varfour{\framei{\kctxtwo}}}\ctxholep{\env\esub\var{\la\vartwo\ctxhole}} \\
      & = & \framei{\envtwo\esub\varthree{\la\varfour\kctxtwo}\envthree}\ctxholep{\env\esub\var{\la\vartwo\ctxhole}}
      \end{array}$$
   \end{itemize}
  \item $\framei{\kctxp{\ctxhole \esub\var\mol}} = \framei\kctx$.
   We proceed by structural induction on $\kctx$.
   \begin{itemize}
    \item Case $\ctxhole\envtwo$:
      $$\begin{array}{lll}
      \framei{\ctxhole \esub\var\mol \envtwo}
      & = & \ctxhole\\
      & = & \framei{\ctxhole\envtwo}
      \end{array}$$
    \item Case $\envtwo\esub\varthree{\la\varfour\kctxtwo}\envthree$:
      $$\begin{array}{lll}
      \framei{\envtwo\esub\varthree{\la\varfour{\kctxtwop{\ctxhole \esub\var\mol}}}\envthree}
      & = & \envtwo\esub\varthree{\la\varfour{\framei{\kctxtwop{\ctxhole \esub\var\mol}}}}\\

      & =_\ih & \envtwo\esub\varthree{\la\varfour{\framei\kctxtwo}} \\
      & = & \framei{\envtwo\esub\varthree{\la\varfour\kctxtwo}\envthree}
      \end{array}$$
   \end{itemize}
\end{enumerate}
\end{proof}

\begin{lemma}\label{l:frame-unf-factorization}
\hfill
\begin{enumerate}
\item \label{p:frame-unf-factorization-b}
  $\unf{\framep\env} = \unf\frame\ctxholep{\unf\env}$  
\item \label{p:frame-unf-factorization-a}
  $\unf{\framep\frametwo} = \unf\frame\ctxholep{\unf\frametwo}$  
\end{enumerate}
\end{lemma}

\begin{proof}
We show the first point, the proof of the second point is obtained by simply replacing $\env$ with a frame. By induction on $\frame$. Cases:
\begin{itemize}
	\item \emph{Base}: $\frame = \ctxhole$. Then $\unf{\framep\env} = \unf{\env} =\unf\ctxhole \ctxholep{\unf\env}$.
	
	\item \emph{Inductive}: $\frame = \envtwo\esub\var{\la\vartwo\frametwo}$. Then 
	\[\begin{array}{rlllllll}
	\unf{\framep\env} & = & \unf{\envtwo\esub\var{\la\vartwo\frametwop\env}} 
	\\
	& = & \unf\envthree\isub\var{\la\vartwo\unf{\frametwop\env}}
	\\
	& =_{\ih} & \unf\envthree\isub\var{\la\vartwo\unf\frametwo\ctxholep{\unf\env}}	
	\\
	& = & (\unf\envthree\isub\var{\la\vartwo\unf\frametwo})\ctxholep{\unf\env}	
	\\
	& = & \unf\frame\ctxholep{\unf\env}	
	\end{array}\]	
\end{itemize}
\end{proof}

\subsection{Invariants}

The aim of this section is to define enough invariants on the reachable machine states
to be able to prove relaxed $\beta$-projection \refthpboth{machine-final}{projection-ms}: for each reachable state $\state$ we need to prove that:
  \begin{enumerate}
    \item If $\state \tomachbv \statetwo$ then $\unf\state \mathrel{(\toms\toes)}^+\unf\statetwo$.
    \item If $\state \tomachbi \statetwo$ then $\unf\state \toms^+\eqstruct \unf\statetwo$.
  \end{enumerate}
  Let $\state = \env \esub\var{\vartwo\varthree} \rlsep \kctx$.
  The proof requires a few intermediate results, the most important are
  \begin{enumerate}
   \item \emph{Open unfolding}: $\unf{\env\esub\var\ctxhole}$ must be an open context,
     so that $\unf{\env\esub\var{\vartwo\varthree}}$ is a top-level redex in an open context.\\
     To guarantee Open unfolding we ask every body in $\wstenv\kctx$ to be pristine
     for each reachable state $\state$ and moreover $\env$ to be pristine if
     $\state = \env \rlsep \kctx$ and $\env \neq \emptyenv$. We call this the \emph{Pristine} invariant.
     The same invariant is required also to prove the open machine correct, where
     $\wstenv\kctx$ collapses to the the environment on the right of $\rlsep$.
   \item \emph{Proper unfolding}: $\unf\kctx$ must be proper, otherwise the redex would
     disappear during the read-back.\\
     To guarantee Proper unfolding we introduce a notion of \emph{garbage-free state} and
     we prove every reachable state to be garbage-free. We call this the
     \emph{Garbage} invariant.\\
     In order to show Garbage to be invariant, we shall also
     introduce an additional, technical invariant that we call \emph{Well-crumbling}
     and that basically says that some properties of pristine environments are propagated
     also to the evaluated, no longer pristine parts of the state.\\
     The proof of the Garbage invariant requires the Well-crumbling invariant to hold.
     The latter in turn requires the Pristine invariant.
   \item \emph{Strong unfolding}: $\unf\kctx$ must be an external multi-context,
     so that $\unf{\kctxp{\env\esub\var{\vartwo\varthree}}}$ is a top-level redex in an open
     context in a strong, context, i.e. a redex according to the strong strategy.\\
     To guarantee Strong unfolding we introduce the notion of \emph{good state} and
     we prove each reachable state to be good. We call this the \emph{Goodness} invariant.\\
     In order to show Goodness to be an invariant, we shall also introduce the last invariant,
     called \emph{Well-named}, that asks every reachable state to be well-named.
  \end{enumerate}

The four invariants introduced so far, namely Well-Named, Pristine, Well-Crumbled, Garbage, and Good, are sufficient to prove
also the other requirements of a relaxed implementation system.

We now define and show the invariance of each of these statements in the following subsections, \ref{subsect:first-invariant}--\ref{subsect:last-invariant}.

\subsubsection{Well-named invariant}\label{subsect:first-invariant}

The Well-named invariant, required to prove the Goodness invariant,
is customary in the abstract machine literature since it guarantees
the possibility to drop variable names and use the address of variables in memory instead.
As a consequence the $\alpha$-renaming operation $\rename\cdot$ is implemented just by physically copying the term.

\begin{definition}[Well-named states]
  We say that a state $\env \gensep\kctx$ is well-named if $\kctxp\env$ is well-named.
\end{definition}


\begin{theorem}[Well-named invariant]
  \label{thm:named-invariant}
  Let $\state=\env \gensep\kctx$ be a state reachable from an initial state $\state_0$. Then $\state$ is well-named.
\end{theorem}
\begin{proof}

By induction on the length of the execution $\exec:\state_0 \tosmach^* \state$.
 If $\exec$ is empty then $\state = \state_0$ and, by definition of initial state, $\state_0 = \env_0 \rlsep \ctxhole$ for some well-named and pristine environment $\env_0$: then $\state_0$ is well-named because $\ctxhole\ctxholep{\env_0} = \env_0$ is well-named by hypothesis.

 If $\exec$ is non-empty we look at the last transition $\statetwo  \tosmach \state$, knowing by \ih that the well-named invariant holds
for $\statetwo$:

\begin{itemize}
  
\item 
$\env \esub\var{\vartwo\,\varthree} \rlsep \kctx
 \tomachbv 
\env (\esub\var \mol \envtwo \isub\varfour\varthree) \rlsep    \kctx$
with $\rename{(\wstenv\kctx(\vartwo))} = \la\varfour(\esub{\varstar}\mol\envtwo) $ and $\wstenv\kctx(\varthree)= \val$ for some $\val$.

      We need to prove that $\kctxp{\env (\esub\var \mol \envtwo \isub\varfour\varthree)}$ is well-named, under the hypothesis that $\kctxp{\env \esub\var{\vartwo\,\varthree}}$ is well-named.
      \begin{itemize}
        \item The bound variables of $\kctxp{\env (\esub\var \mol \envtwo \isub\varfour\varthree)}$ are all distinct: by \reflemma{fv-join-cup-k}, the bound variables in $\state$ are the ones bound in $\statetwo$ plus the bound variables of $\esub\var \mol \envtwo \isub\varfour\varthree$ (excluding $\var$, which was already present in $\statetwo$). The bound variables of $\esub\var \mol \envtwo \isub\varfour\varthree$, however, are all globally fresh due to $\alpha$-renaming (excluding $\var$) and by \reflemmap{vars-after-subst-crumbled}{bv}.
        \item The bound variables of $\kctxp{\env (\esub\var \mol \envtwo \isub\varfour\varthree)}$ are distinct from its free variables: we prove that $\fv\state\subseteq\fv\statetwo$ and $\bv\state \subseteq \bv\statetwo \cup W$ (with $W$ a set of globally fresh new variables disjoint from $\allvars\statetwo$), and then conclude using the \ih{} $\fv\statetwo \Disj \bv\statetwo$.
        \begin{itemize}
          \item \emph{Free variables.}
          By \reflemma{fv-join-cup-k}, $ \fv\statetwo = \fv{\env\esub\var{\vartwo\varthree}} \setminus V \cup \fv\kctx $ and $ \fv\state = \fv{\env (\esub\var \mol \envtwo \isub\varfour\varthree)} \setminus V \cup \fv\kctx $ for some set of variables $V$.
          By \reflemma{env-fv-alpha}, $\fv{\esub\varstar \mol \envtwo} \subseteq \fv\val \cup \set\varfour $ and by \reflemmap{vars-after-subst-crumbled}{fv} $\fv{\esub\varstar \mol \envtwo \isub\varfour\varthree} \subseteq \fv\val \cup \set\varthree $. 
          By \reflemma{fv-join-cup}, $\fv{\env (\esub\var \mol \envtwo \isub\varfour\varthree)} \subseteq \fv\env \setminus \domain{\esub\var \mol \envtwo \isub\varfour\varthree} \cup \fv{\esub\var \mol \envtwo \isub\varfour\varthree} $. Because of the $\alpha$-renaming performed, the variables in $\domain{\esub\var \mol \envtwo \isub\varfour\varthree}\setminus \set\var $ are globally fresh, and therefore $\fv\env \setminus \domain{\esub\var \mol \envtwo \isub\varfour\varthree} \cup \fv{\esub\var \mol \envtwo \isub\varfour\varthree} = \fv\env \setminus \set\var \cup \fv{\esub\var \mol \envtwo \isub\varfour\varthree} \subseteq \fv\env \setminus \set\var \cup \fv\val \cup \set\varthree $.
          \item \emph{Bound variables.}
            \sloppy By \reflemma{fv-join-cup} and \reflemma{fv-join-cup-k}, $ \bv\statetwo \subseteq \bv{\env\esub\var{\vartwo\varthree}} \cup \bv\kctx = \bv\env \cup \bv{\esub\var{\vartwo\varthree}} \cup \bv\kctx $ and $ \bv\state = \bv{\env (\esub\var \mol \envtwo \isub\varfour\varthree)} \cup \bv\kctx = \bv\env \cup \bv{\esub\var \mol \envtwo \isub\varfour\varthree} \cup \bv\kctx$. Because of the $\alpha$-renaming performed and \reflemmap{vars-after-subst-crumbled}{bv}, $\bv{\esub\var \mol \envtwo \isub\varfour\varthree} = \set\var \cup W$ where $W$ is a set of new, globally fresh variables.
        \end{itemize}
      \end{itemize}
  
     \medskip

\item 
$\env \esub\var{\vartwo\,\varthree}  \rlsep   \kctx
  \tomachbi 
\env \esub\var \mol \envtwo  \rlsep  \kctxp{ \ctxhole\esub\varfour\varthree   }$
with $\rename{(\wstenv\kctx(\vartwo))} = \la\varfour(\esub{\varstar}\mol\envtwo) $ and $\wstenv\kctx(\varthree)= \itm$ for some inert term $\itm$.

      We need to prove that $\kctxp{\env (\esub\var \mol \envtwo \esub\varfour\varthree)}$ is well-named, under the hypothesis that $\kctxp{\env \esub\var{\vartwo\,\varthree}}$ is well-named.
      \begin{itemize}
        \item The bound variables of $\kctxp{\env (\esub\var \mol \envtwo \esub\varfour\varthree)}$ are all distinct: by \reflemma{fv-join-cup-k}, the bound variables in $\state$ are the ones bound in $\statetwo$ plus the bound variables of $\esub\var \mol \envtwo \esub\varfour\varthree$ (excluding $\var$, which was already present in $\statetwo$). The bound variables of $\esub\var \mol \envtwo \esub\varfour\varthree$, however, are all globally fresh due to $\alpha$-renaming (excluding $\var$).
        \item The bound variables of $\kctxp{\env (\esub\var \mol \envtwo \esub\varfour\varthree)}$ are distinct from its free variables: we prove that $\fv\state\subseteq\fv\statetwo$ and $\bv\state \subseteq \bv\statetwo \cup W$ (with $W$ a set of globally fresh new variables disjoint from $\allvars\statetwo$), and then conclude using the \ih{} $\fv\statetwo \Disj \bv\statetwo$.
        \begin{itemize}
          \item \emph{Free variables.}
          By \reflemma{fv-join-cup-k}, $ \fv\statetwo = \fv{\env\esub\var{\vartwo\varthree}} \setminus V \cup \fv\kctx $ and $ \fv\state = \fv{\env (\esub\var \mol \envtwo \esub\varfour\varthree)} \setminus V \cup \fv\kctx $ for some set of variables $V$.
          By \reflemma{env-fv-alpha}, $\fv{\esub\varstar \mol \envtwo} \subseteq \fv\val \cup \set\varfour $ and by definition $\fv{\esub\varstar \mol \envtwo \esub\varfour\varthree} \subseteq \fv\val \cup \set\varthree $. 
          \sloppy By \reflemma{fv-join-cup}, $\fv{\env \esub\var \mol \envtwo \esub\varfour\varthree} \subseteq \fv\env \setminus \domain{\esub\var \mol \envtwo \esub\varfour\varthree} \cup \fv{\esub\var \mol \envtwo \esub\varfour\varthree} $. Because of the $\alpha$-renaming performed, the variables in $\domain{\esub\var \mol \envtwo \esub\varfour\varthree}\setminus \set\var $ are globally fresh, and therefore $\fv\env \setminus \domain{\esub\var \mol \envtwo \esub\varfour\varthree} \cup \fv{\esub\var \mol \envtwo \esub\varfour\varthree} = \fv\env \setminus \set\var \cup \fv{\esub\var \mol \envtwo \esub\varfour\varthree} \subseteq \fv\env \setminus \set\var \cup \fv\val \cup \set\varthree $.
          \item \emph{Bound variables.}
            \sloppy By \reflemma{fv-join-cup} and \reflemma{fv-join-cup-k}, $ \bv\statetwo \subseteq \bv{\env\esub\var{\vartwo\varthree}} \cup \bv\kctx = \bv\env \cup \bv{\esub\var{\vartwo\varthree}} \cup \bv\kctx $ and $ \bv\state = \bv{\env (\esub\var \mol \envtwo \esub\varfour\varthree)} \cup \bv\kctx = \bv\env \cup \bv{\esub\var \mol \envtwo \esub\varfour\varthree} \cup \bv\kctx$. Because of the $\alpha$-renaming performed and \reflemmap{vars-after-subst-crumbled}{bv}, $\bv{\esub\var \mol \envtwo \esub\varfour\varthree} = \set\var \cup W$ where $W$ is a set of new, globally fresh variables.
        \end{itemize}
      \end{itemize}

     \medskip

\item ${\env \esub\var\vartwo}  \rlsep  \kctx
 \tomachsub
 {\env \isub\var\vartwo}  \rlsep  \kctx $
with $\var\neq\varstar$.

We need to prove that $\kctxp{{\env \isub\var\vartwo}}$ is well-named, under the hypothesis that $\kctxp{\env \esub\var\vartwo}$ is well-named.
\begin{itemize}
  \item The bound variables of $\kctxp{{\env \isub\var\vartwo}}$ are all distinct: by \reflemma{fv-join-cup-k}, the bound variables in $\state$ are the ones bound in $\kctx$ plus the ones in $\env\isub\var\vartwo$. Note that $\vartwo\not\in\allvars\env$ by well-namedness of $\statetwo$, and hence by \reflemmap{vars-after-subst-crumbled}{bv} the bound variables of $\env\isub\var\vartwo$ are the same of $\env$. Therefore we can conclude by using the hypothesis that $\statetwo$ is well-named.

  \item The bound variables of $\kctxp{{\env \isub\var\vartwo}}$ are distinct from its free variables: we prove that $\fv\state\subseteq\fv\statetwo$ and $\bv\state \subseteq \bv\statetwo$, and then conclude using the \ih{} $\fv\statetwo \Disj \bv\statetwo$.
  \begin{itemize}
    \item \emph{Free variables.} By \reflemma{fv-join-cup-k} and \reflemmap{vars-after-subst-crumbled}{fv}, $\fv{\kctxp{{\env \isub\var\vartwo}}} = \fv{\env\isub\var\vartwo} \setminus V \cup \fv\kctx \subseteq \fv{\env} \setminus\set\var \cup\set\vartwo \setminus V \cup \fv\kctx$ and $\fv{\kctxp{{\env \esub\var\vartwo}}} = \fv{\env\esub\var\vartwo} \setminus V \cup \fv\kctx = \fv\env \setminus \set\var \cup \set\vartwo \setminus V \cup \fv\kctx$ for some $V\subseteq \bv\kctx$, and we conclude.
    \item \emph{Bound variables.}
    By \reflemma{fv-join-cup-k} and \reflemmap{vars-after-subst-crumbled}{bv}, $\bv{\kctxp{{\env \isub\var\vartwo}}} = \bv{\env\isub\var\vartwo} \cup \bv\kctx \subseteq \bv\env \cup \bv\kctx$ and $\bv{\kctxp{{\env \esub\var\vartwo}}} = \bv{\env\esub\var\vartwo} \cup \bv\kctx = \bv\env \cup \set\vartwo \cup \fv\kctx$, and we conclude.
  \end{itemize}
\end{itemize}

     \medskip

     \item ${\env \esub\var\mol}  \rlsep  \kctx
      \tomachcone 
     \env  \rlsep  \kctxp{ {\ctxhole\esub\var\mol}  }
     $
      when $\mol$ is an abstraction or when $\mol$ is $\vartwo$ or $\vartwo\varthree$ but $\vartwo$ is not defined in $\wstenv\kctx$ or $\wstenv\kctx(\vartwo)$ is not a value.

      $\env  \rlsep  \kctxp{ {\ctxhole\esub\var\mol}  }$ is obviously well-named because plugging the crumbled environment in the machine context has the same result in $\statetwo$ and $\state$.

     \medskip

\item $\emptyenv  \rlsep  \kctx
 \tomachctwo 
\emptyenv  \lrsep  \kctx$

     ${\env \esub\var\mol}  \lrsep  \kctx$ is obviously well-named
     because plugging the crumbled environment in the machine context has the same result in $\statetwo$ and $\state$.

     \medskip

\item $\env  \lrsep  \kctxp{ {\ctxhole\esub\var\mol  }}
 \tomachcthree 
{\env \esub\var\mol}  \lrsep  \kctx
$
where $\mol$ is a variable or an application.

$\emptyenv  \lrsep  \kctx$ is obviously well-named
     because plugging the crumbled environment in the machine context has the same result in $\statetwo$ and $\state$.

     \medskip

\item $\env  \lrsep  \kctxp{ {\ctxhole\esub\var\val  }}
 \tomachgc 
\env  \lrsep  \kctx$
with $\var \notin \fv\env$.

      We need to prove that $\kctxp\env$ is well-named, under the hypothesis that $\kctxp{{\env\esub\var\val}}$ is well-named.
      \begin{itemize}
        \item The bound variables of $\state$ are all distinct: clearly $\state$ contains fewer occurrences of bound variables than $\statetwo$. In fact, by \reflemma{fv-join-cup} and \reflemma{fv-join-cup-k}, the bound variables in $\statetwo$ are the ones bound in $\state$ plus $\var$ and the bound variables in $\mol$. Therefore all the bound variables in $\state$ are distinct because all bound variables of $\statetwo$ are distinct by well-namedness.
        \item The bound variables of $\kctxp\env$ are distinct from its free variables: first of all, by \reflemma{fv-join-cup-k}, $ \fv\statetwo = \fv{\env\esub\var\val} \setminus V \cup \fv\kctx $ and $ \fv\state = \fv\env \setminus V \cup \fv\kctx $ for some set of variables $V$. Moreover, from the discussion on the point above, $\bv\state\subseteq\bv\statetwo$.
        Now, note that $\fv{\env\esub\var\val} = \fv\env \setminus\set\var \cup \fv\mol = \fv\env \cup \fv\mol \supseteq \fv\env$ because $\var\not\in\fv\env$. Therefore also $\fv\state \subseteq \fv\statetwo$, and we conclude by the well-named hypothesis $\fv\statetwo \Disj \bv\statetwo$.
      \end{itemize}

     \medskip

\item $\env  \lrsep  \kctxp{\envtwo {\esub\var{\la\vartwo\ctxhole }  }}
 \tomachcfour 
{\envtwo \esub\var{\la\vartwo\env}}  \lrsep   \kctx$. 

${\envtwo \esub\var{\la\vartwo\env}}  \lrsep   \kctx$ is obviously well-named because plugging the crumbled environment in the machine context has the same result in $\statetwo$ and $\state$.

     \medskip

\item $\env  \lrsep  \kctxp{ {\ctxhole\esub\var{\la\vartwo\envtwo} } }
 \tomachcfive 
\envtwo  \rlsep  \kctxp{\env {\esub\var{\la\vartwo\ctxhole}}  }
$ with $\var \in \fv\env$.

$\envtwo  \rlsep  \kctxp{\env {\esub\var{\la\vartwo\ctxhole}}  }$ is obviously well-named because plugging the crumbled environment in the machine context has the same result in $\statetwo$ and $\state$.

\qedhere
\end{itemize}
\end{proof}

\subsubsection{Pristine invariant}

Previously, we defined pristine environments so to characterize the good properties that are enforced by the translation from \lat{s}. We extend the notion of pristine environments to machine states as follows, by considering all the unevaluated environments therein contained:

\begin{definition}[Pristine state]
A state $\env \gensep \kctx$ is \emph{pristine} if
every body in $\wstenv\kctx$ is pristine, and also if
${\gensep} = {\rlsep}$ and $\env \neq \emptyenv$ then $\env$ is pristine. 
\end{definition}

The fundamental property of a pristine state was basically already stated in
\reflemmaboth{aux-most-once}. We now lift that result to machine states.


\begin{theorem}[Pristine invariant]
  \label{thm:pristine-invariant-app}
  Let $\state=\env \gensep\kctx$ be a state reachable from an initial state $\state_0$. Then $\state$ is pristine.
\end{theorem}
\begin{proof}

By induction on the execution $\exec:\state_0 \tosmach^* \state$.
 If $\exec$ is empty then $\state = \state_0$ and, by definition of initial state, $\state_0 = \env_0 \rlsep \ctxhole$ for some well-named environment $\env_0$.
 \begin{itemize}
   \item \emph{Every body in $\wstenv\ctxhole$ is pristine:} obvious since there are none.
   \item \emph{$\env_0$ is pristine:} by the definition of initial state and \reflemma{transl-properties-pristine}.
 \end{itemize}

 If $\exec$ is non-empty we look at the last transition $\statetwo  \tosmach \state$, knowing by \ih that the pristine invariant holds
for $\statetwo$:

\begin{itemize}

    \item 
  $\env \esub\var{\vartwo\,\varthree} \rlsep \kctx
   \tomachbv 
  \env (\esub\var \mol \envtwo \isub\varfour\varthree) \rlsep    \kctx$
  with $\rename{(\wstenv\kctx(\vartwo))} = \la\varfour(\esub{\varstar}\mol\envtwo) $ and $\wstenv\kctx(\varthree)= \val$ for some $\val$.
\begin{itemize}
  \item \emph{Every body in $\wstenv\kctx$ is pristine:} obvious because the property holds by \ih.
\item 
\emph{$\env (\esub\var \mol \envtwo \isub\varfour\varthree) $ is pristine:} by \ih{} $\env \esub\var{\vartwo\,\varthree}$ is pristine. Note that $\wstenv\kctx(\vartwo)$ is a value, and therefore by the previous point, the body $\envthree$ of $\wstenv\kctx(\vartwo)$ is pristine. Therefore, $\esub{\varstar}\mol\envtwo$, which is an $\alpha$-renaming of $\envthree$, is pristine by \reflemmap{pristine-properties}{alpha}. By \reflemmap{pristine-properties}{renaming} also $\esub\varstar\mol\envtwo \isub\varfour\varthree$ is pristine (note that $\varthree\not\in\bv\envtwo$ because $\esub{\varstar}\mol\envtwo$ is an $\alpha$-renaming of $\envthree$ where all bound variables are globally fresh).
  In order to apply \reflemmap{pristine-properties}{concatenation} to conclude that $\env (\esub\var \mol \envtwo\isub\varfour\varthree) $ is pristine, we need to show that $\allvars\env \Disj \bv{\envtwo\isub\varfour\varthree}$ (which follows from the fact that $\esub{\varstar}\mol\envtwo$ is an $\alpha$-renaming of $\envthree$ where all bound variables are globally fresh and \reflemmap{vars-after-subst-crumbled}{bv}) and $\bv\env \Disj \fv{\esub\varstar\mol\envtwo\isub\varfour\varthree}$ ($\fv{\esub\varstar\mol\envtwo\isub\varfour\varthree} \subseteq \fv{\esub\varstar\mol\envtwo} \setminus\set\varfour\cup\set\varthree$ by \reflemmap{vars-after-subst-crumbled}{fv}, $\fv{\esub\varstar\mol\envtwo} = \fv{\envthree}$ by \reflemma{env-fv-alpha}, which are different than the bound variables in $\env$ by well-namedness).
\end{itemize}

  \medskip

   \item 
  $\env \esub\var{\vartwo\,\varthree}  \rlsep   \kctx
    \tomachbi 
  \env \esub\var \mol \envtwo  \rlsep  \kctxp{ \ctxhole\esub\varfour\varthree   }$
  with $\rename{(\wstenv\kctx(\vartwo))} = \la\varfour(\esub{\varstar}\mol\envtwo) $ and $\wstenv\kctx(\varthree)= \itm$ for some inert term $\itm$.

  \begin{itemize}
    \item \emph{Every body in $\wstenv{\kctxp{\ctxhole\esub\varfour\varthree}} =_\reflemmaeqp{unsenv-prefix}{one}\esub\varfour\varthree\wstenv\kctx$ is pristine:} obvious because the property holds
    by \ih for $\wstenv\kctx$.

  \item \emph{$\env \esub\var \mol \envtwo $ is pristine}: by \ih{}, $\env \esub\var{\vartwo\,\varthree}$ is pristine. Note that $\wstenv\kctx(\vartwo)$ is a value, and therefore by the previous point, the body $\envthree$ of $\wstenv\kctx(\vartwo)$ is pristine. Therefore, $\esub{\varstar}\mol\envtwo$, which is an $\alpha$-renaming of $\envthree$, is pristine by \reflemmap{pristine-properties}{alpha}.
  In order to apply \reflemmap{pristine-properties}{concatenation} to conclude that $\env \esub\var \mol \envtwo $ is pristine, we need to show that $\allvars\env \Disj \bv\envtwo$ (which follows from the fact that $\esub{\varstar}\mol\envtwo$ is an $\alpha$-renaming of $\envthree$ where all bound variables are globally fresh) and $\bv\env \Disj \fv{\esub\varstar\mol\envtwo}$ ($\fv{\esub\varstar\mol\envtwo} = \fv{\envthree}$ by \reflemma{env-fv-alpha}, which are different than the bound variables in $\env$ by well-namedness).

\end{itemize}

  \medskip

\item ${\env \esub\var\vartwo}  \rlsep  \kctx
 \tomachsub
 {\env \isub\var\vartwo}  \rlsep  \kctx $
with $\var\neq\varstar$.

\begin{itemize}
  \item \emph{Every body in $\wstenv\kctx$ is pristine:} obvious because the property holds
  by \ih.

\item \emph{$\env\isub\var\vartwo $ is pristine:} by \ih, $\env \esub\var\vartwo$ is pristine and so $\env$ is pristine.
  Since $\state$ is well-named also $\env \esub\var\vartwo$ is well-named, and thus $\vartwo \not\in \bv\env$. Finally, $\env\isub\var\vartwo$ is pristine by \reflemmap{pristine-properties}{renaming}.
\end{itemize}
  \medskip

  \item ${\env \esub\var\mol}  \rlsep  \kctx
   \tomachcone 
  \env  \rlsep  \kctxp{ {\ctxhole\esub\var\mol}  }
  $
   when $\mol$ is an abstraction or when $\mol$ is $\vartwo$ or $\vartwo\varthree$ but $\vartwo$ is not defined in $\wstenv\kctx$ or $\wstenv\kctx(\vartwo)$ is not a value.

  \begin{itemize}
    \item \emph{Every body in $\wstenv{\kctxp{\ctxhole\esub\var\mol}} =_\reflemmaeqp{unsenv-prefix}{one}\esub\var\mol\wstenv\kctx$ is pristine:} by \ih, ${\env \esub\var\mol}  \rlsep  \kctx$ is \good and therefore $\env \esub\var\mol$ is pristine and thus $\mol$ is a pristine bite. By induction on the structure of $\mol$ one can show that every body in $\mol$ is pristine. By \ih, the property holds for every body of $\wstenv\kctx$ as well, concluding this case.
  \item 
   \emph{$\env$ is pristine} (if it is not empty): it follows by the fact that $\env \esub\var\mol $ is pristine.
  \end{itemize}

   \medskip

\item $\emptyenv  \rlsep  \kctx
 \tomachctwo 
\emptyenv  \lrsep  \kctx$.

\begin{itemize}
  \item \emph{Every body in $\wstenv\kctx$ is pristine:} obvious because the property holds by \ih.
  \item Obvious because $\env$ is pristine.
\end{itemize}

\medskip

\item $\env  \lrsep  \kctxp{ {\ctxhole\esub\var\mol  }}
 \tomachcthree 
{\env \esub\var\mol}  \lrsep  \kctx
$
where $\mol$ is a variable or an application.

\begin{itemize}
  \item \emph{Every body in $\wstenv\kctx$ is pristine:} obvious because the property holds
  by \ih for every body in $\wstenv{\kctxp{ {\ctxhole\esub\var\mol}}} =_\reflemmaeqp{unsenv-prefix}{one} \esub\var\mol\wstenv\kctx$.
  \item Nothing to prove because we are in the $\lrsep$ phase.
\end{itemize}

\medskip

\item $\env  \lrsep  \kctxp{ {\ctxhole\esub\var\val  }}
 \tomachgc 
\env  \lrsep  \kctx$
with $\var \notin \fv\env$.

\begin{itemize}
  \item \emph{Every body in $\wstenv\kctx$ is pristine:} obvious because the property holds
  by \ih for every body in $\wstenv{\kctxp{ {\ctxhole\esub\var\val  }}} =_\reflemmaeqp{unsenv-prefix}{one} \esub\var\val\wstenv\kctx$.
  \item Nothing to prove because we are in the $\lrsep$ phase.
\end{itemize}

\medskip

\item $\env  \lrsep  \kctxp{\envtwo {\esub\var{\la\vartwo\ctxhole }  }}
 \tomachcfour 
{\envtwo \esub\var{\la\vartwo\env}}  \lrsep   \kctx$. 

\begin{itemize}
  \item \emph{Every body in $\wstenv\kctx$ is pristine:} obvious because the property holds
  by \ih for every body in $\wstenv{\kctxp{\envtwo {\esub\var{\la\vartwo\ctxhole }  }}} =_\reflemmaeqp{unsenv-prefix}{two} \wstenv\kctx$.
  \item Nothing to prove because we are in the $\lrsep$ phase.
\end{itemize}

\item $\env  \lrsep  \kctxp{ {\ctxhole\esub\var{\la\vartwo\envtwo} } }
 \tomachcfive 
\envtwo  \rlsep  \kctxp{\env {\esub\var{\la\vartwo\ctxhole}}  }
$ with $\var \in \fv\env$.

\begin{itemize}
  
\item \emph{Every body in $\wstenv{\kctxp{\env {\esub\var{\la\vartwo\ctxhole}}  }} =_\reflemmaeqp{unsenv-prefix}{two} \wstenv\kctx$ is pristine:} obvious because the property holds
     by \ih for every body in $\wstenv{\kctxp{ {\ctxhole\esub\var{\la\vartwo\envtwo} } }} =_\reflemmaeqp{unsenv-prefix}{one}
     \esub\var{\la\vartwo\envtwo}\wstenv\kctx$.
  \item \emph{$\envtwo$ is pristine:} 
  $\envtwo$ is a body of $\wstenv{\kctxp{\ctxhole \esub\var{\la\vartwo\envtwo}}} =_\reflemmaeqp{unsenv-prefix}{one} \esub\var{\la\vartwo\envtwo}\wstenv\kctx$, and is therefore pristine by the previous point.
\qedhere
\end{itemize}

\end{itemize}
\end{proof}



\subsubsection{Well-crumbled invariant}

The Well-crumbling invariant is a technical invariant required to prove
the Garbage invariant. 


\begin{definition}[Well-crumbling]
An environment $\env$ is \emph{well-crumbled} iff
\begin{enumerate}
 \item for every decomposition
  $\env = \env_1 \esub\var\mol \env_2$ such that $\var \in \crnames$ and
  $\mol$ is not an abstraction, $\var \in \fv{\unf{\env_1}}$.
 \item the property holds recursively for the body of every abstraction
   that occurs in $\env$
\end{enumerate}
A state $\env \gensep \kctx$ is \emph{well-crumbled} if $\kctxp\env$ is. 
\end{definition}

\begin{lemma}\label{l:well-crumbled-renaming}
 ~
 \begin{enumerate}
  \item \label{p:well-crumbled-renaming-one} If $\env$ is well-crumbled then $\rename\env$ also is.
  \item \label{p:well-crumbled-renaming-two} If $\env \envtwo$ is well-crumbled
   and $\var \not\in \domain\envtwo$
   then $\env\isub\var\vartwo\envtwo$ is well-crumbled.
 \end{enumerate}
\end{lemma}
\begin{proof}
 We need to prove that:
 \begin{enumerate}
  \item If $\env$ is well-crumbled then $\rename\env$ also is.
   Variables in $\crnames$ can only be renamed into variables in $\crnames$, renaming turns non-abstractions
   into non-abstractions and, by \reflemma{env-fv-alpha} and \reflemmap{properties-unfolding}{alpha}, it does not change the set of free variables in the
   unfolding of an environment.
  \item If $\env \envtwo$ is well-crumbled and $\var \not\in \domain\envtwo$
   then $\env\isub\var\vartwo\envtwo$ is well-crumbled.
   The substitution $\isub\var\vartwo$ cannot turn a non-abstraction into an abstraction and, by \reflemmap{properties-unfolding}{renaming},
   $\fv{\unf{\env\isub\var\vartwo}} = \fv{\unf\env\isub\var\vartwo} \supseteq \fv{\unf\env} \setminus \set{x}$.
   Since $\var \not\in \domain\envtwo$, the property follows.
 \end{enumerate}
\end{proof}

\begin{lemma}\label{l:pristine-well-crumbled}
Every pristine environment is well-crumbled.
\end{lemma}
\begin{proof}
For every decomposition $\env_1\esub\var\mol\env_2$ of a pristine environment
it holds that $\unf{\env_1} = \openctxp\var$ for some open context $\openctx$.
Thus $\var \in \fv{\unf{\env_1}}$ as required to be well-crumbled.
The fact that the same holds recursively for each body of a pristine environment
is trivially proved by structural recursion over $\env$.
\end{proof}


\begin{theorem}[Well-crumbled invariant]
\label{thm:well-crumbled-invariant-app}
  Let $\state=\env \gensep\kctx$ be a state reachable from an initial state $\state_0$. Then $\state$ is well-crumbled.
\end{theorem}
\begin{proof}

By induction on the execution $\exec:\state_0 \tosmach^* \state$.
 If $\exec$ is empty then $\state = \state_0$ and, by definition of initial state, $\state_0 = \env_0 \rlsep \ctxhole$ for some well-named and pristine environment $\env_0$: $\state_0$ is well-crumbled by definition if $\env_0$ is well-crumbled, which holds by
 \reflemma{pristine-well-crumbled}.

 If $\exec$ is non-empty we look at the last transition $\statetwo  \tosmach \state$, knowing by \ih that the well-crumbled invariant holds
for $\statetwo$:

\begin{itemize}
  
\item 
$\env \esub\var{\vartwo\,\varthree} \rlsep \kctx
 \tomachbv 
\env (\esub\var \mol \envtwo \isub\varfour\varthree) \rlsep    \kctx$
with $\rename{(\wstenv\kctx(\vartwo))} = \la\varfour(\esub{\varstar}\mol\envtwo) $ and $\wstenv\kctx(\varthree)= \val$ for some $\val$.

     By \ih $\env \esub\var{\vartwo\,\varthree} \rlsep \kctx$ is well-crumbled.
     Since $\wstenv\kctx(\vartwo)$ is an abstraction that occurs in the well-crumbled state $\env \esub\var{\vartwo\,\varthree} \rlsep \kctx$,
     its body must be well-crumbled and therefore, by \reflemmap{well-crumbled-renaming}{one}, also $\esub{\varstar}\mol\envtwo$ is well-crumbled
     and, by \reflemmap{well-crumbled-renaming}{two}, also $\esub{\varstar}\mol\envtwo\isub\varfour\varthree$ is well-crumbled.
     The hypothesis over $\varthree$ used for applying \reflemmap{well-crumbled-renaming}{two} is trivially satisfied because
     $\varthree$ is a fresh variable that does not occur in $\kctx$ at all.
     Since $\var$ is bound to a non-abstraction, if $\var \in \crnames$ then $\var \in \fv{\unf\env}$. Therefore
     \sloppy $\env(\esub{\var}\mol\envtwo\isub\varfour\varthree)$ is well-crumbled.
     We conclude that $\env (\esub\var \mol \envtwo \isub\varfour\varthree) \rlsep \kctx$ is well-crumbled by noting that
     $\vartwo$ and $\varthree$ are bound to abstractions. 
  
     \medskip

\item 
$\env \esub\var{\vartwo\,\varthree}  \rlsep   \kctx
  \tomachbi 
\env \esub\var \mol \envtwo  \rlsep  \kctxp{ \ctxhole\esub\varfour\varthree   }$
with $\rename{(\wstenv\kctx(\vartwo))} = \la\varfour(\esub{\varstar}\mol\envtwo) $ and $\wstenv\kctx(\varthree)= \itm$ for some inert term $\itm$.

     By \ih $\env \esub\var{\vartwo\,\varthree} \rlsep \kctx$ is well-crumbled.
     Since $\wstenv\kctx(\vartwo)$ is an abstraction that occurs in the well-crumbled state $\env \esub\var{\vartwo\,\varthree} \rlsep \kctx$,
     its body must be well-crumbled and therefore, by \reflemmap{well-crumbled-renaming}{one}, also $\esub{\varstar}\mol\envtwo$ is well-crumbled
     and thus $\esub{\varstar}\mol\envtwo\esub\varfour\varthree$ is also well-crumbled because $\varfour \in \calcnames$.
     Since $\var$ is bound to a non-abstraction, if $\var \in \crnames$ then $\var \in \fv{\unf\env}$. Therefore
     $\env\esub{\var}\mol\envtwo\esub\varfour\varthree$ is well-crumbled.
     We conclude that $\env \esub\var \mol \envtwo \esub\varfour\varthree \rlsep \kctx$ is well-crumbled by noting that
     $\vartwo$ is bound to an abstraction and that $\varthree$ occurs in the unfolding of the environment because $\esub\varfour\varthree$ does too because $\varfour \in \calcnames$. 

     \medskip

\item ${\env \esub\var\vartwo}  \rlsep  \kctx
 \tomachsub
 {\env \isub\var\vartwo}  \rlsep  \kctx $
with $\var\neq\varstar$.

  Obvious because by \ih ${\env \esub\var\vartwo}  \rlsep  \kctx$ is well-crumbled and by \reflemmap{well-crumbled-renaming}{two}.

     \medskip

\item ${\env \esub\var\mol}  \rlsep  \kctx
     \tomachcone 
\env  \rlsep  \kctxp{ {\ctxhole\esub\var\mol}  }
$
     when $\mol$ is an abstraction or when $\mol$ is $\vartwo$ or $\vartwo\varthree$ but $\vartwo$ is not defined in $\wstenv\kctx$ or $\wstenv\kctx(\vartwo)$ is not a value.

     \sloppy $\env  \rlsep  \kctxp{ {\ctxhole\esub\var\mol}  }$ is obviously well-crumbled because ${\env \esub\var\mol}  \rlsep  \kctx$ is well-crumbled by \ih and, by definition of well-crumbled context, both the
hypothesis and the conclusion require $\kctxp{\env\esub\var\mol}$ to be well-crumbled.

     \medskip

\item $\emptyenv  \rlsep  \kctx
 \tomachctwo 
\emptyenv  \lrsep  \kctx$

$\emptyenv  \lrsep  \kctx$ is obviously well-crumbled because
     $\emptyenv \rlsep \kctx$ is well-crumbled by \ih.

     \medskip

\item $\env  \lrsep  \kctxp{ {\ctxhole\esub\var\mol  }}
 \tomachcthree 
{\env \esub\var\mol}  \lrsep  \kctx
$
where $\mol$ is a variable or an application.

\sloppy ${\env \esub\var\mol}  \lrsep  \kctx$ is obviously well-crumbled because $\env \lrsep \kctxp{\ctxhole\esub\var\mol}$ is well-crumbled by \ih and, by definition of well-crumbled context, both the
     hypothesis and the conclusion require $\kctxp{\env\esub\var\mol}$ to be well-crumbled.

     \medskip

\item $\env  \lrsep  \kctxp{ {\ctxhole\esub\var\val  }}
 \tomachgc 
\env  \lrsep  \kctx$
with $\var \notin \fv\env$.

$\env  \lrsep  \kctx$ is obviously well-crumbled because by \ih $\env \lrsep \kctxp{\ctxhole\esub\var\val}$ is well-crumbled.

     \medskip

\item $\env  \lrsep  \kctxp{\envtwo {\esub\var{\la\vartwo\ctxhole }  }}
 \tomachcfour 
{\envtwo \esub\var{\la\vartwo\env}}  \lrsep   \kctx$. 
\sloppy To prove that ${\envtwo \esub\var{\la\vartwo\env}}  \lrsep   \kctx$ is well-crumbled, simply note that $\env  \lrsep  \kctxp{\envtwo {\esub\var{\la\vartwo\ctxhole }  }}$ is well-crumbled by \ih and, by definition of well-crumbled context, both the hypothesis and the conclusion require $\kctxp{\envtwo {\esub\var{\la\vartwo\env}}}$ to be well-crumbled.

     \medskip

\item $\env  \lrsep  \kctxp{ {\ctxhole\esub\var{\la\vartwo\envtwo} } }
 \tomachcfive 
\envtwo  \rlsep  \kctxp{\env {\esub\var{\la\vartwo\ctxhole}}  }
$ with $\var \in \fv\env$.

$\envtwo  \rlsep  \kctxp{\env {\esub\var{\la\vartwo\ctxhole}}  }$ is obviously well-crumbled because $\env  \lrsep  \kctxp{ {\ctxhole\esub\var{\la\vartwo\envtwo} } }$ is well-crumbled by \ih and, by definition of well-crumbled context, both the hypothesis and the conclusion require $\kctxp{\env {\esub\var{\la\vartwo\envtwo}}  }$ to be well-crumbled.

\qedhere
\end{itemize}
\end{proof}

\subsubsection{Garbage invariant}

The Garbage invariant basically guarantees that the read-back $\unf\kctx$ of a machine context $\kctx$ is a proper multi-context. We formulate the invariant directly on the frame of $\kctx$, since being garbage-free is a property enforced in the already evaluated part of a state.


\begin{definition}[Garbage-free]~
  \begin{itemize}
  \item \emph{Environments:} an environment $\env$ is garbage-free if $\vartwo \in \fv\env$ implies $\vartwo \in \fv{\unf\env}$.

  \item \emph{Frames:} a frame $\frame$ is garbage-free if
 \begin{itemize}
  \item $\frame = \ctxhole$, or
  \item $\frame = \env\esub\var{\la\vartwo\frametwo}$ and
  \begin{itemize}
    \item $\env$ is garbage-free,
    \item $\var\in \fv\env$, and 
    \item $\frametwo$ is garbage-free.
  \end{itemize}  
 \end{itemize}
 \item \emph{States:} a state $\state=\env\gensep\kctx$ is garbage-free if $\framei\kctx$ is garbage-free and when ${\gensep} = {\lrsep}$ then $\env$ is garbage-free. 
\end{itemize}
\end{definition} 

Garbage-free frames are decomposable:
\begin{lemma}[Garbage-free decomposition]\label{l:gcfree-prop}
 ~
  \label{p:gcfree-prop-two} $\framep{\frametwo}$ is garbage-free iff
   $\frame$ and $\frametwo$ are.
\end{lemma}
\begin{proof}
By structural induction over $\frame$.
   \begin{itemize}
    \item Case $\ctxhole$. By definition.
    \item Case $\env\esub\var{\la\vartwo\framethree}$. We need to prove that
     $\env\esub\var{\la\vartwo\framethree}$ is garbage-free iff
     $\env\esub\var{\la\vartwo\framethreep\frametwo}$ is. The property follows from the \ih over $\frametwo$ and the definition
      of garbage-free context.
   \end{itemize}
\end{proof}


\begin{theorem}[Garbage-free invariant]
\label{thm:garbage-free-invariant-app}
  Let $\state=\env \gensep\kctx$ be a state reachable from an initial state $\state_0$. Then $\state$ is garbage-free.
\end{theorem}
\begin{proof}

By induction on the execution $\exec:\state_0 \tosmach^* \state$.
 If $\exec$ is empty then $\state = \state_0$ and, by definition of initial state, $\state_0 = \env_0 \rlsep \ctxhole$ for some well-named and pristine environment $\env_0$: then $\state_0$ is garbage-free by definition of garbage-free state.

 If $\exec$ is non-empty we look at the last transition $\statetwo  \tosmach \state$, knowing by \ih that the garbage-free invariant holds
for $\statetwo$:

\begin{itemize}
  
\item 
$\env \esub\var{\vartwo\,\varthree} \rlsep \kctx
 \tomachbv 
\env (\esub\var \mol \envtwo \isub\varfour\varthree) \rlsep    \kctx$
with $\rename{(\wstenv\kctx(\vartwo))} = \la\varfour(\esub{\varstar}\mol\envtwo) $ and $\wstenv\kctx(\varthree)= \val$ for some $\val$.

$\env (\esub\var \mol \envtwo \isub\varfour\varthree) \rlsep \kctx$ is garbage-free iff $\framei\kctx$ is garbage-free. The property holds because, by \ih, $\env \esub\var{\vartwo\,\varthree} \rlsep \kctx$ is garbage-free.

     \medskip

\item 
$\env \esub\var{\vartwo\,\varthree}  \rlsep   \kctx
  \tomachbi 
\env \esub\var \mol \envtwo  \rlsep  \kctxp{ \ctxhole\esub\varfour\varthree   }$
with $\rename{(\wstenv\kctx(\vartwo))} = \la\varfour(\esub{\varstar}\mol\envtwo) $ and $\wstenv\kctx(\varthree)= \itm$ for some inert term $\itm$.

$\env \esub\var \mol \envtwo  \rlsep  \kctxp{ \ctxhole\esub\varfour\varthree   }$ is garbage-free iff $\framei{\kctxp{ \ctxhole\esub\varfour\varthree   }}=_{\reflemmaeqp{wk-st-lam}{b}} \framei{\kctx}$ is garbage-free. The property holds because, by \ih, $\env \esub\var{\vartwo\,\varthree} \rlsep \framei\kctx$ is garbage-free.

     \medskip

\item ${\env \esub\var\vartwo}  \rlsep  \kctx
 \tomachsub
 {\env \isub\var\vartwo}  \rlsep  \kctx $
with $\var\neq\varstar$.

${\env\isub\var\vartwo}  \rlsep  \kctx$ is garbage-free iff $\framei\kctx$ is garbage-free. The property holds because, by \ih, ${\env \esub\var\vartwo}  \rlsep  \kctx$ is garbage-free.

     \medskip

     \item ${\env \esub\var\mol}  \rlsep  \kctx
      \tomachcone 
     \env  \rlsep  \kctxp{ {\ctxhole\esub\var\mol}  }
     $
      when $\mol$ is an abstraction or when $\mol$ is $\vartwo$ or $\vartwo\varthree$ but $\vartwo$ is not defined in $\wstenv\kctx$ or $\wstenv\kctx(\vartwo)$ is not a value.

      $\env  \rlsep  \kctxp{ {\ctxhole\esub\var\mol}  }$ is garbage-free iff $\framei{\kctxp{ \ctxhole\esub\var\mol   }}=_{\reflemmaeqp{wk-st-lam}{b}} \framei{\kctx}$. The property holds because $ \framei{\kctx}$ is garbage free because, by \ih, $\env \esub\var\mol \rlsep \kctx$ is garbage-free.

     \medskip

\item $\emptyenv  \rlsep  \kctx
 \tomachctwo 
\emptyenv  \lrsep  \kctx$

$\emptyenv  \lrsep  \kctx$ is garbage-free iff $\framei\kctx$ is garbage-free, which holds by \ih, and $\emptyenv$ is garbage-free, which is obvious by definition of
     garbage-free environment.
     \medskip

\item $\env  \lrsep  \kctxp{ {\ctxhole\esub\var\mol  }}
 \tomachcthree 
{\env \esub\var\mol}  \lrsep  \kctx
$
where $\mol$ is a variable or an application.

\sloppy ${\env \esub\var\mol}  \lrsep  \kctx$ is garbage-free iff $\framei\kctx$ and $\env\esub\var\mol$ are garbage-free. By \ih, $\env \lrsep \kctxp{\ctxhole\esub\var\mol}$ is garbage-free, i.e. $\framei{\kctxp{ \ctxhole\esub\var\mol   }}=_{\reflemmaeqp{wk-st-lam}{b}} \framei{\kctx}$ and $\env$ are garbage-free.
     By the well-crumbled invariant, $\env  \lrsep  \kctxp{ {\ctxhole\esub\var\mol  }}$ is well-crumbled and therefore, because $\var$ is bound to a non-abstraction,
     if $\var \in \crnames$ then $\var \in \fv{\unf\env}$. Therefore any variable that occurs free in $\mol$ also occurs free in
     $\unf{\env\esub\var\mol}$ that is either $\unf\env\isub\var\mol$, if $\var \in \crnames$, or $\unf\env\esub\var\mol$.

     \medskip

\item $\env  \lrsep  \kctxp{ {\ctxhole\esub\var\val  }}
 \tomachgc 
\env  \lrsep  \kctx$
with $\var \notin \fv\env$.

$\env  \lrsep  \kctx$ is garbage-free iff $\framei\kctx$ and $\env$ are garbage-free. The property holds because, by \ih, $\env  \lrsep  \kctxp{ {\ctxhole\esub\var\val  }}$ is garbage-free, i.e. $\env$ and $\framei{\kctxp{ \ctxhole\esub\var\val   }}=_{\reflemmaeqp{wk-st-lam}{b}} \framei{\kctx}$ are garbage-free.

     \medskip

\item $\env  \lrsep  \kctxp{\envtwo {\esub\var{\la\vartwo\ctxhole }  }}
 \tomachcfour 
{\envtwo \esub\var{\la\vartwo\env}}  \lrsep   \kctx$. 

${\envtwo \esub\var{\la\vartwo\env}}  \lrsep   \kctx$ is garbage-free iff $\kctx$ and $\envtwo \esub\var{\la\vartwo\env}$ are garbage-free. By \ih, $\env \lrsep \kctxp{\envtwo {\esub\var{\la\vartwo\ctxhole }  }}$ is garbage-free, i.e. $\framei{\kctxp{\envtwo {\esub\var{\la\vartwo\ctxhole }  }}}=_{\reflemmaeqp{wk-st-lam}{a}} \framei{\kctx}\ctxholep{\envtwo {\esub\var{\la\vartwo\ctxhole }  }}$ and $\env$ are garbage-free.
By \reflemma{gcfree-prop}, $\framei\kctx$ and $\framei{\envtwo {\esub\var{\la\vartwo\ctxhole }  }}$ are garbage-free and thus
$\envtwo$ is garbage-free and $\var \in \fv\envtwo$ by definition of garbage-free context. In order to conclude
that $\envtwo \esub\var{\la\vartwo\env}$ is garbage-free we need to show that every variable that occurs free in $\envtwo \esub\var{\la\vartwo\env}$ occurs free in $\unf{\envtwo \esub\var{\la\vartwo\env}} = \unf\envtwo \isub\var{\la\vartwo\unf\env}$. A variable that occurs free in $\envtwo \esub\var{\la\vartwo\env}$ occurs free either in $\envtwo$ and thus in $\unf{\envtwo}$ because $\envtwo$ is garbage-free, or in $\env$ and thus in $\unf\env$ because $\env$ is garbage-free. Therefore it occurs free in $\unf\envtwo \isub\var{\la\vartwo\unf\env}$ because $\var \in \unf\env$ because $\var \in \envtwo$ and $\envtwo$ is garbage-free.

     \medskip

\item $\env  \lrsep  \kctxp{ {\ctxhole\esub\var{\la\vartwo\envtwo} } }
 \tomachcfive 
\envtwo  \rlsep  \kctxp{\env {\esub\var{\la\vartwo\ctxhole}}  }
$ with $\var \in \fv\env$.

$\envtwo  \rlsep  \kctxp{\env {\esub\var{\la\vartwo\ctxhole}}  }$ is garbage-free iff $\framei{\kctxp{\env {\esub\var{\la\vartwo\ctxhole}}  }}$ is garbage-free. By \ih, $\env  \lrsep  \kctxp{ {\ctxhole\esub\var{\la\vartwo\envtwo} } }$ is garbage-free, i.e. $\framei{\kctxp{ {\ctxhole\esub\var{\la\vartwo\envtwo} } }} =_{\reflemmaeqp{wk-st-lam}{b}} \framei\kctx$ and $\env$ are garbage-free. By \reflemma{gcfree-prop}, it is sufficient to prove
    that $\env {\esub\var{\la\vartwo\ctxhole}}$ is garbage-free, i.e. that $\env$ is garbage-free, which we already proved, and that $\var \in \fv\env$, which holds by hypothesis.

\qedhere
\end{itemize}
\end{proof}

\subsubsection{Good invariant} \label{subsect:last-invariant}

The good invariant is definitely the most complex one.
The fundamental property, which is part of the requirement for a good state
$\env \gensep \kctx$, is that $\unf\kctx$ is a fine context.

However, in order to show that this holds as an invariant for all reachable states,
the notion of good state must be strengthened by imposing strict, technical requirements on
various fragments of the machine state.

One such requirement is called \emph{compatibility} and it is imposed on the environment that is being evaluated with respect to the substitution $\indsub{\wstenv\kctx}$ originated by the enclosing machine context $\kctx$.

\begin{definition}[Compatibility with a fireball substitution]
  Let $\sfire$ be a strong fireball. We say that \emph{$\sfire$ is compatible with a (fireball) substitution $\sigma$} if whenever a variable $\var$ such that $\sigma(\var) = \val$ occurs free in $\sfire$ then it does as the argument of an application. Compatibility for other syntactic categories, \eg external multi contexts, is defined similarly.
\end{definition}

Another fundamental requirement is called \emph{well-framing} and it is imposed on the frame $\framei\kctx$.
The frame $\framei\kctx$ is the part of the context that has already been strongly evaluated.
The induced property is that $\unf{\framei\kctx}$ is a normal fine
multi-context such that applying the substitution $\indsub{\wstenv\kctx}$ does not create
new redexes in the already computed part.

The well-framed requirement w.r.t. a substitution is a strengthening of that property that
requires it to hold hereditarily, since this is necessary in order to propagate it in the
proof that all reachable states are good.

\begin{definition}[Well-framed]
 A frame $\frame$ is \emph{well-framed} w.r.t. a substitution $\sigma$ if, for every decomposition
 $\frame = \frametwop\framethree$, $\unf\frametwo$ is a normal fine multi-context
 compatible with $\sigma$.
\end{definition}

The following lemma is a trivial technical property over well-framed frames.

\begin{lemma}\label{l:well-framed-match}
If $\framep\frametwo$ is well-framed w.r.t. $\sigma$, then $\frame$ is well-framed
w.r.t. $\sigma$.
\end{lemma}
\begin{proof}
Every decomposition $\frame = \frame_1\ctxholep{\frame_2}$ induces a decomposition
$\framep\frametwo = \frame_1\ctxholep{\frame_2\ctxholep\frametwo}$. The statement
follows by definition of well-framed frame.
\end{proof}

We are ready to define formally the good property:

\begin{definition}[Good stuff]
An environment $\envtwo$ is \emph{open \good} if 
\begin{itemize}
    \item $\indsub\envtwo$ is a fireball substitution;    
    \item $\indenv\envtwo$ is an inert context.
    \item $\envtwo$ has immediate values.    
  \end{itemize}  
A context $\kctx$ is \emph{good} when 
 \begin{itemize}
  \item $\wstenv\kctx$ is open good;
  \item $\framei\kctx$ is well-framed w.r.t. $\indsub{\wstenv\kctx}$;
  \item $\unf\kctx $ is a fine context.
 \end{itemize}  
A state $\env \gensep \kctx$ is \emph{good} if
\begin{itemize}
  \item $\kctx$ is good and 
  \item if ${\gensep} = {\lrsep}$ then $\unf\env$ is a strong fireball compatible with $\indsub{\wstenv\kctx}$.
  \end{itemize} 
\end{definition}

As already mentioned, the proof of the Good invariant is quite involved. Before proving that all reachable states are good (\refthm{goodness-invariants}) we need a good number of auxiliary results, which we prove in the following paragraphs. Since \SCAM{} transitions can add or remove ES from the machine state, we are going to show that goodness is stable under the addition and removal of ES, under suitable conditions. In order to do that, we need to prove multiple corresponding properties for multi contexts and compatible substitutions.

\paragraph{Basic properties of multi contexts}

In this paragraph we prove a couple of general properties of multi contexts that are required in the next paragraph.

The first lemma allows to see terms as (non-proper) multi contexts.

\begin{lemma}\hfill
\label{l:terms-to-multi-ctxs}
  \begin{enumerate}
  \item \label{p:terms-to-multi-ctxs-rigid-term}
  Every rigid term $\rtm$ is a rigid multi context with no holes.
  \item \label{p:terms-to-multi-ctxs-inert-ctx}
  Every inert context is a fine multi context.
  \end{enumerate} 
\end{lemma}
\begin{proof}
  By an easy inspection of the grammar of multi contexts, and by definition.
\end{proof}

The second lemma shows that the plugging of multi contexts amounts to syntactic substitution, in the case when no variable capture can occur:
\begin{lemma}[Substitution and plugging for multi contexts]
\label{l:strongmctx-plug-eq-sub}
  Let $\strongmctx$ and $\rmctx$ be a strong and a rigid multi contexts such that they do not capture variables in $\fv{\strongmctxtwo}$ and with no free occurrences of $\var$. Then 
  \begin{enumerate}
    \item \label{p:strongmctx-plug-eq-sub-a}
    $\rmctxp\var\isub\var{\strongmctxtwo} = \rmctxp{\strongmctxtwo}$.
    
    \item \label{p:strongmctx-plug-eq-sub-b}
    $\strongmctxp\var\isub\var{\strongmctxtwo} = \strongmctxp{\strongmctxtwo}$.
  \end{enumerate}
\end{lemma}

\begin{proof}
By mutual induction on $\rmctx$ and $\strongmctx$.
\begin{enumerate}
  \item \emph{Rigid}. Cases:
  \begin{itemize}
    \item \emph{Variable}, \ie $\rmctx = \vartwo$. Then $\rmctxp\var\isub\var{\strongmctxtwo} = \vartwo\isub\var{\strongmctxtwo} = \vartwo = \rmctxp{\strongmctxtwo}$.
    
    \item \emph{Application}, \ie $\rmctx = \rmctxtwo \strongmctx$. Then $$\begin{array}{l}\rmctxp\var\isub\var{\strongmctxtwo} = \rmctxtwop\var\isub\var{\strongmctxtwo} \strongmctxp\var\isub\var{\strongmctxtwo} \\ =_{\ih} \rmctxtwop\strongmctxtwo \strongmctxp\strongmctxtwo = \rmctxp{\strongmctxtwo} \end{array}$$
    
    \item \sloppy \emph{Explicit substitution}, \ie $\rmctx = \rmctxtwo \esub\vartwo\rmctxthree$.	Then $\rmctxp\var\isub\var{\strongmctxtwo} = (\rmctxtwop\var \esub\vartwo{\rmctxthreep\var})\isub\var{\strongmctxtwo}$. We have that $\vartwo\notin\fv\strongmctxtwo$ because $\rmctx$ does not capture variables in $\fv\strongmctxtwo$. Therefore, $(\rmctxtwop\var \esub\vartwo{\rmctxthreep\var})\isub\var{\strongmctxtwo} = \rmctxtwop\var\isub\var{\strongmctxtwo} \esub\vartwo{\rmctxthreep\var\isub\var{\strongmctxtwo}}$ without having to rename $\vartwo$ in $\rmctxtwop\var \esub\vartwo{\rmctxthreep\var}$. And then one can continue as expected: 
    $$\begin{array}{l}\rmctxtwop\var\isub\var{\strongmctxtwo} \esub\vartwo{\rmctxthreep\var\isub\var{\strongmctxtwo}}\\ =_{\ih} \rmctxtwop\strongmctxtwo \esub\vartwo{\rmctxthreep\strongmctxtwo} = \rmctxp{\strongmctxtwo}\end{array}$$
  \end{itemize}
  
  \item \emph{Strong}. Cases:	
  \begin{itemize}
    \item \emph{Empty}, \ie $\strongmctx = \ctxhole$. Then $\strongmctxp\var\isub\var\strongmctxtwo = \var\isub\var\strongmctxtwo = \strongmctxtwo = \strongmctxp\strongmctxtwo$.
    
    \item \emph{Term}, \ie $\strongmctx = \tm$. Remember that $\var\notin\fv\strongmctx=\fv\tm$. Then $\strongmctxp\var\isub\var{\strongmctxtwo} = \tm\isub\var{\strongmctxtwo} = \tm = \strongmctxp{\strongmctxtwo}$.
    
    \item \sloppy \emph{Abstraction}, \ie $\strongmctx = \la\vartwo\strongmctxthree$. Then $\strongmctxp\var\isub\var{\strongmctxtwo} = (\la\vartwo\strongmctxthreep\var)\isub\var{\strongmctxtwo}$. We have that $\vartwo\notin\fv\strongmctxtwo$ because $\strongmctx$ does not capture variables in $\fv\strongmctxtwo$. Therefore, $(\la\vartwo\strongmctxthreep\var)\isub\var{\strongmctxtwo} = \la\vartwo\strongmctxthreep\var\isub\var{\strongmctxtwo}$ without having to rename $\vartwo$ in $\la\vartwo\strongmctxthreep\var$. And then one can continue as expected: $\la\vartwo\strongmctxthreep\var\isub\var{\strongmctxtwo} =_{\ih} \la\vartwo\strongmctxthreep\strongmctxtwo = \strongmctxp\strongmctxtwo$.
      
    \item \emph{Rigid}, \ie $\strongmctx = \rmctx$. By \refpoint{strongmctx-plug-eq-sub-a}.
    
    \item \sloppy \emph{Explicit substitution}, \ie $\strongmctx = \strongmctxthree \esub\var\rmctx$. Then $\strongmctxp\var\isub\var{\strongmctxtwo} = (\strongmctxthreep\var \esub\vartwo{\rmctxp\var})\isub\var{\strongmctxtwo}$. We have that $\vartwo\notin\fv\strongmctxtwo$ because $\rmctx$ does not capture variables in $\fv\strongmctxtwo$. Therefore, $(\strongmctxthreep\var \esub\vartwo{\rmctxp\var})\isub\var{\strongmctxtwo} = \strongmctxthreep\var\isub\var{\strongmctxtwo} \esub\vartwo{\rmctxp\var\isub\var{\strongmctxtwo}}$ without having to rename $\vartwo$ in $\strongmctxthreep\var \esub\vartwo{\rmctxp\var}$. And then one can continue as expected: 
    $$\begin{array}{l}\strongmctxthreep\var\isub\var{\strongmctxtwo} \esub\vartwo{\rmctxp\var\isub\var{\strongmctxtwo}}\\ =_{\ih} \strongmctxthreep\strongmctxtwo \esub\vartwo{\rmctxp\strongmctxtwo} = \rmctxp{\strongmctxtwo}\end{array}$$	\end{itemize}	
\end{enumerate}
\end{proof}

\paragraph{Multi contexts and compatible substitutions}

The following lemma shows that compatibility of a multi context $\strongmctx$ with respect to a substitution $\sigma$ ensures that some nice properties of $\strongmctx$ are preserved in $\strongmctx\sigma$. 

\begin{lemma}[Multi contexts and compatible substitutions]
  \label{l:sctx-and-sub}
    Let $\rmctx$ be a rigid multi context and $\strongmctx$  and a fine multi context both compatible with a fireball substitution $\sigma$. Then 
    \begin{enumerate}
    \item \label{p:sctx-and-sub-a}
    $\rmctx\sigma$ is a rigid multi context. Moreover, if $\rmctx$ is proper then $\rmctx\sigma$ is proper.
    \item \label{p:sctx-and-sub-b}
    $\strongmctx\sigma$ is an external multi context. Moreover, if $\strongmctx$ is proper then $\strongmctx\sigma$ is proper.
    \end{enumerate}
  \end{lemma}
  \begin{proof}
  By mutual induction on $\rmctx$ and $\strongmctx$.
  \begin{enumerate}
    \item \emph{Rigid}. Cases:
    \begin{itemize}
      \item \emph{Variable}, \ie $\rmctx = \var$. Since $\var$ does not occur as an argument, by compatibility $\rmctx\sigma = \var\sigma = \sigma(\var)$ is an inert term. By \reflemmap{terms-to-multi-ctxs}{rigid-term} $\sigma(\var)$ can be seen as a rigid multi context.
      
      \item \emph{Application}, \ie $\rmctx = \rmctxtwo \strongmctx$. Note that $\rmctxtwo$ is compatible with $\sigma$ and that $\strongmctx$ is compatible only if $\strongmctx \neq \var$. By \ih, $\rmctxtwo\sigma$ is a rigid multi context. If $\strongmctx = \var$ then $\strongmctx\sigma = \var\sigma = \sigma(var)$ which is an inert term and thus a rigid multi context by \reflemmap{terms-to-multi-ctxs}{rigid-term}. If $\strongmctx \neq \var$ then by \ih $\strongmctx\sigma$ is an external multi context. Then $\rmctx\sigma = \rmctxtwo\sigma \strongmctx\sigma$ is a rigid multi context.
      
      If $\rmctx$ is proper then one among $\rmctxtwo$ and $\strongmctx$ is proper, and properness of $\rmctx\sigma$ follows from the \ih
      
      \item \emph{Explicit substitution}, \ie $\rmctx = \rmctxtwo \esub\var\rmctxthree$. Both $\rmctxtwo$ and $\rmctxthree$ are compatible with $\sigma$. By \ih, both $\rmctxtwo\sigma$ and $\rmctxthree\sigma$ are rigid multi contexts. Then $\rmctx\sigma = \rmctxtwo\sigma \rmctxthree\sigma$ is a rigid multi context.
      
          If $\rmctx$ is proper then one among $\rmctxtwo$ and $\rmctxthree$ is proper, and properness of $\rmctx\sigma$ follows from the \ih
    \end{itemize}
    
    \item \emph{Strong}. Cases:	
    \begin{itemize}
      \item \emph{Empty}, \ie $\strongmctx = \ctxhole$. Trivial, because $\ctxhole\sigma = \ctxhole$.
      
      \item \emph{Term}, \ie $\strongmctx = \tm$. Trivial because every term, and in particular $\tm\sigma$ is an external multi context.
      
      \item \emph{Abstraction}, \ie $\strongmctx = \la\vartwo\strongmctxtwo$. By \ih, $\strongmctxtwo\sigma$ is an external multi context, and so is $\strongmctx\sigma = \la\vartwo\strongmctxtwo\sigma$. The moreover par follows from the moreover part of the \ih
        
      \item \emph{Rigid}, \ie $\strongmctx = \rmctx$. It follows from \refpoint{sctx-and-sub-a}.
      
      \item \emph{Explicit substitution}, \ie $\strongmctx = \strongmctxtwo \esub\var\rmctx$. Both $\strongmctxtwo$ and $\rmctx$ are compatible with $\sigma$. By \ih, both $\strongmctxtwo\sigma$ and $\rmctx\sigma$ are external multi contexts. Then $\strongmctx\sigma = \strongmctxtwo\sigma \rmctx\sigma$ is an external multi context.
      
          If $\rmctx$ is proper then one among $\rmctxtwo$ and $\rmctxthree$ is proper, and properness of $\rmctx\sigma$ follows from the \ih
    \end{itemize}	
  \end{enumerate}
\end{proof}

The following two lemmas show that compatibility is preserved both by plugging and by composition of rigid and external multi contexts:

\begin{lemma}[Plugging preserves compatibility]
\label{l:smctx-fire-compatibility}
Let $\strongmctx$ and $\rmctx$ be a strong and a rigid multi contexts and $\sfire$ be a strong fireball such that they are all compatible with $\sigma$. Then 
\begin{enumerate}
  \item \label{p:smctx-fire-compatibility-a}
  $\rmctxp\sfire$ is compatible with $\sigma$.
  
  \item \label{p:smctx-fire-compatibility-b}
  $\strongmctxp\sfire$ is compatible with $\sigma$.
\end{enumerate}
\end{lemma}
\begin{proof}
By mutual induction on $\rmctx$ and $\strongmctx$.
\begin{enumerate}
  \item \emph{Rigid}. Cases:
  \begin{itemize}
    \item \emph{Variable}, \ie $\rmctx = \var$. Then $\rmctxp\sfire = \var=\rmctx$ is compatible with $\sigma$.
    
    \item \emph{Application}, \ie $\rmctx = \rmctxtwo \strongmctxthree$. It follows immediately from the \ih on $\rmctxtwo$ and $\strongmctxthree$, apart when $\strongmctxthree=\var$ and $\sigma(\var) = \val$. In such a case however the \ih gives compatibility of $\rmctxtwo$ which is enough to obtain compatibility of $\rmctx$.
    
    \item \emph{Explicit substitution}, \ie $\rmctx = \rmctxtwo \esub\var\rmctxthree$. It follows immediately from the \ih
  \end{itemize}
  
  \item \emph{Strong}. Cases:	
  \begin{itemize}
    \item \emph{Empty}, \ie $\strongmctx = \ctxhole$. Trivial, because $\strongmctxp\sfire = \sfire$ is compatible with $\sigma$ by hypothesis.
    
    \item \emph{Term}, \ie $\strongmctx = \tm$. Then $\strongmctxp\sfire = \tm$ which is compatible by hypothesis. 
    
    \item \emph{Abstraction}, \ie $\strongmctx = \la\var\sfire$. It follows immediately from the \ih
      
    \item \emph{Rigid}, \ie $\strongmctx = \rmctx$. By \refpoint{smctx-fire-compatibility-a}.
    
    \item \emph{Explicit substitution}, \ie $\strongmctx = \sfire \esub\var\rmctx$. It follows immediately from the \ih
  \end{itemize}	
\end{enumerate}
\end{proof}
  
\begin{lemma}[Composition of multi contexts]
  \label{l:strctx-compos}
  Let $\strongmctx$ and $\rmctx$ be a strong and a rigid multi contexts, and $\strongmctxtwo$ be a further external multi context. Then 
  \begin{enumerate}
    \item \label{p:strctx-compos-a}
    $\rmctxp\strongmctxtwo$ is a rigid multi context.
    
    \item \label{p:strctx-compos-b}
    $\strongmctxp\strongmctxtwo$ is an external multi context.
  \end{enumerate}
  Moreover, let $\mctx \in \set{\strongmctx, \rmctx}$ and 
   \begin{enumerate}
     \item \label{p:strctx-compos-fine}
    if both $\mctx$ and $\strongmctxtwo$ are proper (and thus fine), so does $
    \mctxp\strongmctxtwo$.
  
    \item if both $\mctx$ and $\strongmctxtwo$ are compatible with $\sigma$, so does $
    \mctxp\strongmctxtwo$.
    
    \item if both $\mctx$ and $\strongmctxtwo$ are normal, so does $
    \mctxp\strongmctxtwo$.
  \end{enumerate}
  \end{lemma}

\begin{proof}
  By mutual induction on $\rmctx$ and $\strongmctx$.
\begin{enumerate}
  \item \emph{Rigid}. Cases:
  \begin{itemize}
    \item \emph{Variable}, \ie $\rmctx = \var$. Then $\rmctxp\strongmctxtwo = \var$ which is a rigid multi context.
    
    \item \emph{Application}, \ie $\rmctx = \rmctxtwo \strongmctxthree$. It follows immediately from the \ih on $\rmctxtwo$ and $\strongmctxthree$, apart for compatibility with $\sigma$ when $\strongmctxthree=\var$ and $\sigma(\var) = \val$. In such a case however the \ih gives compatibility of $\rmctxtwo$ which is enough to obtain compatibility of $\rmctx$.
    
    \item \emph{Explicit substitution}, \ie $\rmctx = \rmctxtwo \esub\var\rmctxthree$. It follows immediately from the \ih
  \end{itemize}
  
  \item \emph{Strong}. Cases:	
  \begin{itemize}
    \item \emph{Empty}, \ie $\strongmctx = \ctxhole$. Trivial.
    
    \item \emph{Term}, \ie $\strongmctx = \tm$. Then $\strongmctxp\strongmctxtwo = \tm$ which is an external multi context. 
    
    \item \emph{Abstraction}, \ie $\strongmctx = \la\var\strongmctxtwo$. It follows immediately from the \ih
      
    \item \emph{Rigid}, \ie $\strongmctx = \rmctx$. By \refpoint{strctx-compos-a}.
    
    \item \emph{Explicit substitution}, \ie $\strongmctx = \strongmctxtwo \esub\var\rmctx$. It follows immediately from the \ih
  \qedhere
  \end{itemize}	
\end{enumerate}
\end{proof}

The following auxiliary, technical lemma allows to extract from a term all the occurrences of a free variable, decomposing the term into a multi context that plugs that variable:

\begin{lemma}[Context extraction from fireballs]
  \label{l:str-fir-ctx-dec}
    Let $\tm$ be a well-named term such that $\var\in\fv{\tm}$ and $\tm$ is compatible with $\isub\var\val$.
    \begin{enumerate}
    \item if $\tm$ is a strong inert term then there exists a normal and proper rigid multi context $\rmctx$ such that $\tm = \rmctxp\var$ and $\var\notin\fv\rmctx$.
    
    \item if $\tm$ is a strong fireball then there exists a fine and normal multi context $\strongmctx$ such that $\tm = \strongmctxp\var$ and $\var\notin\fv\strongmctx$.
    \end{enumerate}
  \end{lemma}
\begin{proof}
By induction on $\tm$. Cases:
\begin{itemize}
  \item \emph{Variable}, \ie $\tm = \var$: 
  \begin{enumerate}
    \item trivially true, as the compatibility hypothesis is not verified ($\var$ occurs free but not as an argument);
    \item Simply take $\strongmctx \defeq \ctxhole$.
  \end{enumerate}
  
  \item \emph{Application}, \ie $\tm = \sitm \sfire$:
  \begin{enumerate}
    \item Suppose that $\var$ occurs in both $\sitm$ and $\sfire$. Note that $\sitm$ is compatible and that $\sfire$ is compatible only if $\sfire\neq\var$. By \ih there is a normal and proper rigid normal multi context $\rmctxtwo$ such that $\sitm = \rmctxtwop\var$ and $\var\notin\fv\rmctxtwo$. If $\sfire = \var$ then take $\strongmctx \defeq \ctxhole$, otherwise by \ih there exists a fine and normal multi context $\strongmctx$ such that $\sfire = \strongmctxp\var$ and $\var\notin\fv\strongmctx$. Then $\rmctx \defeq \rmctxtwo \strongmctx$ satisfies the statement.
    
    If $\var$ does not occur in $\sitm$ then one uses \reflemmap{terms-to-multi-ctxs}{rigid-term} to see $\sitm$ as a normal rigid multi context and reason as before. If $\var$ does not occur in $\sfire$ then one sees $\sfire$ as a normal and fine multi context with no holes, as all terms are external multi contexts. Note that $\var$ has to occur in $\itm$ or $\sfire$, because it occurs in $\tm$, and so the context $\rmctx$ is always proper.
    
    \item Simply take $\strongmctx$ as the rigid multi context obtained in the previous point.
  \end{enumerate}
  
  \item \emph{Abstraction}, \ie $\tm = \la\vartwo\sfire$:
  \begin{enumerate}
    \item trivially true, as the hypothesis is not verified;
    \item Simply take $\strongmctx \defeq \la\vartwo\strongmctxtwo$, where $\strongmctxtwo$ is the fine and normal multi context given by the \ih on $\sfire$.
  \end{enumerate}
  
  \item \emph{Explicit substitution}, \ie $\tm = \tmtwo\esub\vartwo\sitm$:
  \begin{enumerate}
    \item if $\tm$ is a strong inert term then so is $\tmtwo$. Suppose that $\var$ occurs in both $\tmtwo$ and $\itm$. Note that both $\tmtwo$ and $\sitm$ are compatible with $\isub\var\val$. Therefore, we can apply the \ih on both terms, obtaining two normal and proper rigid multi context $\rmctxtwo$ and $\rmctxthree$ such that $\rmctxtwop\var = \tmtwo$, $\rmctxthreep\var = \sitm$, $\var\notin\fv\rmctxtwo$, and $\var\notin\fv\rmctxthree$. Then $\rmctx \defeq \rmctxtwo \esub\vartwo\rmctxthree$ verifies the statement. If $\var$ does not occur in one among $\tmtwo$ and $\itm$ then one uses \reflemmap{terms-to-multi-ctxs}{rigid-term} to see it as a normal rigid multi context and reason as before. Note that $\var$ has to occur in $\tmtwo$ or $\itm$, because it occurs in $\tm$, and so the context $\rmctx$ is always proper.
    
    \item Along the lines of the previous point, spelled out in the following. If $\tm$ is a strong fireball then so is $\tmtwo$. Suppose that $\var$ occurs in both $\tmtwo$ and $\sitm$. We can apply the \ih obtaining a fine and normal multi context $\strongmctxtwo$ such that $\sfire = \strongmctxp\var$ and $\var\notin\fv\strongmctx$. By the compatibility hypothesis $\sitm$ does not have shape $\sctxp\var$, and so we can apply the \ih, obtaining a normal and proper rigid multi context $\rmctxtwo$ such that $\rmctxtwop\var = \sitm$ and $\var\notin\fv\rmctxtwo$. Then $\strongmctx \defeq \strongmctxtwo \esub\vartwo\rmctxthree$ verifies the statement. If $\var$ does not occur in $\sitm$ then one uses \reflemmap{terms-to-multi-ctxs}{rigid-term} to see $\sitm$ as a normal rigid multi context and reason as before. If $\var$ does not occur in $\tmtwo$ then one sees $\tmtwo$ as a normal and fine multi context with no holes, as all terms are external multi contexts. Note that $\var$ has to occur in $\tmtwo$ or $\itm$, because it occurs in $\tm$, and so the context $\strongmctx$ is always proper.
  \end{enumerate}
\end{itemize}
\end{proof}

\paragraph{Stability of goodness by addition/removal of ES} The invariance of goodness for the $\rlsep$-transitions requires (weak) goodness to be stable by addition appropriate ES next to the hole of $\kctx$. Similarly, $\lrsep$-transitions require stability of (weak) goodness by removal of the innermost ES next to the hole in $\kctx$. We first focus on adding ES, in the next two lemmas.

\begin{lemma}[Open goodness addition]
\label{l:hereditarily-weakly-good}
 Let $\env$ be open good and $\itm$ be an inert term. Then:
 \begin{enumerate}
	\item if $\env(\var)$ undefined or $\env(\var) = \itm$ then $\esub\vartwo\var\env$ and $\esub{\vartwo}{\var\varthree}\env$ are open good. 
	\item $\esub{\vartwo}{\la\var\envtwo}\env$ is open good.
  \end{enumerate}
\end{lemma}

\begin{proof}
\hfill
 \begin{enumerate}
 \item 
  We have to prove three facts:
 \begin{enumerate}
  \item \emph{$\indsub{\esub\vartwo\var\env}$ is a fireball substitution}. Three sub-cases:
  \begin{itemize}
    \item $\vartwo\in\crnames$ and $\env(\var)$ undefined: then  
       $\indsub{\esub\vartwo\var\env} =_\reflemmaeqp{sigma-wk-cons}{a} \isub{\vartwo}{\var}\cup\indsub{\env}$. 
       By hypothesis, $\indsub{\env}$ is a fireball substitution, and thus so is $\indsub{\esub\vartwo\var\env}$.

    \item $\vartwo\in\crnames$ and $\env(\var) = \itm$: 
       $\indsub{\esub\vartwo\var\env} =_\reflemmaeqp{sigma-wk-cons}{a} \isub{\vartwo}{\indsub{\env}(\var)}\cup\indsub{\env}$. 
       By hypothesis, $\indsub{\env}$ is a fireball substitution, thus $ \indsub{\env}(\var)$ is a fireball, and we conclude.

    \item $\vartwo\in\calcnames$:
    $\indsub{\esub\vartwo\var\env} =_\reflemmaeqp{sigma-wk-cons}{a} \indsub{\env}$, which by hypothesis is a fireball substitution.
  \end{itemize}
  
  \item \emph{$\indenv{\esub\vartwo\var\env}$ is a inert context}. Three sub-cases:
  \begin{itemize}     
    \item $\vartwo\in\crnames$:
    $\indenv{\esub\vartwo\var\env} =_\reflemmaeqp{sigma-wk-cons}{c} \indenv{\env}$, which by hypothesis is an inert context.

    \item $\vartwo\in\calcnames$ and $\env(\var)$ undefined:
    $\indenv{\esub\vartwo\var\env} =_\reflemmaeqp{sigma-wk-cons}{c} \esub{\vartwo}{\var}\indenv{\env}$.
    By hypothesis, $\indenv{\env}$ is an inert context, and thus so is $\indenv{\esub\vartwo\var\env}$.

    \item $\vartwo\in\calcnames$ and $\env(\var) = \itm$:
    $\indenv{\esub\vartwo\var\env} =_\reflemmaeqp{sigma-wk-cons}{c} \esub{\vartwo}{\indsub{\env}(\var)}\indenv{\env}$.
    By hypothesis, $\indenv{\env}$ is an inert context, and we need to prove that $\indsub{\env}(\var)$ is an inert term. Since $\env$ is open good, 
    $\indsub\env(\var)$ is a fireball, and if it is a value then $\env(\var)$ is also a value. By the side condition of the rule, $\env(\var)$ is an inert, and thus $\indsub\env(\var)$ is not a value, that is, it is a inert term.
    \end{itemize}
    
    \item \emph{$\esub\vartwo\var\env$ has immediate values}. Assume that $\indsub{\esub\vartwo\var\env} (\varthree)$ is a value for $\varthree \neq \varstar$, we have to prove that $(\esub\vartwo\var\env) (\varthree)$ is a value. Three sub-cases:    
    \begin{itemize}
    \item $\vartwo\in\crnames$ and $\env(\var)$ undefined: 
    $\indsub{\esub\vartwo\var\env}(\varthree) =_\reflemmaeqp{sigma-wk-cons}{a} (\isub{\vartwo}{\var}\cup\indsub{\env})(\varthree)$.
      If $\varthree\neq\vartwo$ then it follows by the fact that $\env$ has immediate values (by hypothesis). Otherwise, $(\isub{\vartwo}{\var}\cup\indsub{\env})(\vartwo) = \var$ which is not a value, and so the statement trivially holds.

    \item $\vartwo\in\crnames$ and $\env(\var) = \itm$: 
    $\indsub{\esub\vartwo\var\env}(\varthree) =_\reflemmaeqp{sigma-wk-cons}{a} (\isub{\vartwo}{\indsub{\env}(\var)}\cup\indsub{\env})(\varthree)$.
      If $\varthree\neq\vartwo$ then it follows by the fact that $\env$ has immediate values (by hypothesis). Otherwise, $(\isub{\vartwo}{\indsub{\env}(\var)}\cup\indsub{\env})(\vartwo) = \indsub{\env}(\var)$. Since $\env$ has immediate values and $\var \neq \varstar$ (because it occurs in $\esub\vartwo\var$), if $\indsub{\env}(\var)$ is a value then $\env(\var)$ is a value, against the hypothesis that it is an inert term---then this case is not possible.
      
    \item $\vartwo\in\calcnames$: 
      $\indsub{\esub\vartwo\var\env}(\varthree) = \indsub{\env}(\varthree)$ and the property follows from the fact that $\env$ has immediate value (by hypothesis).
  \end{itemize} 
 \end{enumerate}
 
 \item We have to prove three facts:
 \begin{enumerate}
  \item \emph{$\indsub{\esub\vartwo{\la\var\envtwo}\env}$ is a fireball substitution}. Then 
       $\indsub{\esub\vartwo{\la\var\envtwo}\env} =_\reflemmaeqp{sigma-wk-cons}{a} \isub{\vartwo}{\unf{\la\var\envtwo}\indsub{\env}}\cup\indsub{\env}$. 
       By hypothesis, $\indsub{\env}$ is a fireball substitution, thus so is $\indsub{\esub\vartwo{\la\var\envtwo}\env}$.
  
  \item \emph{$\indenv{\esub\vartwo\var\env}$ is a inert context}. Then 
    $\indenv{\esub\vartwo{\la\var\envtwo}\env} =_\reflemmaeqp{sigma-wk-cons}{c} \indenv{\env}$, which by hypothesis is an inert context.
    
    \item \sloppy \emph{$\esub\vartwo{\la\var\envtwo}\env$ has immediate values}. Assume that $\indsub{\esub\vartwo{\la\var\envtwo}\env} (\varthree)$ is a value for $\varthree \neq \varstar$, we have to prove that $(\esub\vartwo{\la\var\envtwo}\env) (\varthree)$ is a value. Note that $\indsub{\esub\vartwo{\la\var\envtwo}\env}(\varthree) =_\reflemmaeqp{sigma-wk-cons}{a} (\isub{\vartwo}{\unf{\la\var\envtwo}\indsub{\env}}\cup\indsub{\env})(\varthree)$. If $\varthree\neq\vartwo$ then it follows by the fact that $\env$ has immediate values (by hypothesis). Otherwise, $(\isub{\vartwo}{\unf{\la\var\envtwo}\indsub{\env}}\cup\indsub{\env})(\vartwo) = \unf{\la\var\envtwo}\indsub{\env}$ and $(\esub\vartwo{\la\var\envtwo}\env)(\vartwo) = \la\var\envtwo$ which, as required, is a value.
 \end{enumerate}
  \end{enumerate}
\end{proof}

\begin{lemma}[Goodness addition]
\label{l:hereditarily-good}
 Let $\kctx$ be good and $\itm$ be an inert term.
 \begin{enumerate}
	\item If $\wstenv\kctx(\var)$ undefined or $\wstenv\kctx(\var) = \itm$ then $\kctxp{\ctxhole\esub\vartwo\var}$ and $\kctxp{\ctxhole\esub\vartwo{\var\varthree}}$ are good. 
	\item If $\vartwo\notin \fv{\framei\kctx}$ then $\kctxp{\ctxhole\esub\vartwo{\la\var\envtwo}}$ is good.
  \end{enumerate}
\end{lemma}

\begin{proof}
\sloppy Let $\kctxtwo$ be either $\kctxp{\ctxhole\esub\vartwo\var}$, $\kctxp{\ctxhole\esub\vartwo{\var\varthree}}$ or $\kctxp{\ctxhole\esub\vartwo{\la\var\envtwo}}$ depending on which case we are proving.
For both points, by \reflemma{hereditarily-weakly-good} $\wstenv\kctxtwo$ is open good and $\framei\kctx$ is well-framed, so that we only have to show that the unfolding of the frame $\unf{\framei\kctxtwo}$ of $\kctxtwo$ is a fine and normal multi context compatible with $\wstenv\kctxtwo$ and that
the unfolding $\unf\kctxtwo$ of $\kctxtwo$ is fine.
\begin{itemize}
 \item $\unf{\framei\kctxtwo}$ is a fine and normal multi context compatible with $\wstenv\kctxtwo$:
 \begin{enumerate}
  \item We treat the case of $\kctxp{\ctxhole\esub\vartwo\var}$, for $\kctxp{\ctxhole\esub\varfour{\var\varthree}}$ the reasoning is identical. Note that $\framei\kctx = \framei{\kctxp{\ctxhole\esub\vartwo\var}}$. So by \ih we know that $\unf{\framei\kctx}$ is a fine  and normal multi context. We only need to show that it is compatible with $\indsub{\esub\vartwo\var\wstenv\kctx}$, knowing what we refer to as the \emph{compatibility hypothesis}, that is, that it is compatible with $\indsub{\esub\vartwo\var\wstenv\kctx}$. Three sub-cases:
  \begin{enumerate}
    \item $\vartwo\in\crnames$ and $\wstenv\kctx(\var)$ undefined: then by \reflemmap{sigma-wk-cons}{a} we have
    $\indsub{\esub\vartwo\var\wstenv\kctx} =\isub{\vartwo}{\var}\cup\indsub{\wstenv\kctx}$.
    Since $\var$ is inert, compatibility then follows from the compatibility hypothesis. 

    \item $\vartwo\in\crnames$ and $\wstenv\kctx(\var) = \itm$: then by \reflemmap{sigma-wk-cons}{a} we have
    $\indsub{\esub\vartwo\var\wstenv\kctx} =\isub{\vartwo}{\indsub{\wstenv\kctx}(\var)}\cup\indsub{\wstenv\kctx}$.
    By hypothesis, $\wstenv\kctx(\var)$ is an inert term. Since $\wstenv\kctx$ is open good, it has immediate values, and so $\indsub{\wstenv\kctx}(\var)$ is an inert term as well because $\var \neq \varstar$ because $\var$ occurs in $\esub\vartwo\var$. Compatibility then follows from the compatibility hypothesis.

    \item $\vartwo\in\calcnames$: then by \reflemmap{sigma-wk-cons}{c} we have
    $\indsub{\esub\vartwo\var\wstenv\kctx} = \indsub{\wstenv\kctx}$ and the property follows by compatibility hypothesis.
  \end{enumerate}
  
  \item \sloppy Note that also in this case we have $\framei\kctx =_{\reflemmaeqp{wk-st-lam}{b}} \framei{\kctxp{\ctxhole\esub\vartwo{\la\var\envtwo}}}$, and the \ih gives that it is a fine and normal multi context follows. We have to show it compatible with $\indsub{\esub\vartwo\var\wstenv\kctx}$, knowing that it is compatible with $\indsub{\wstenv\kctx}$. This is immediate, because by hypothesis $\vartwo \notin\fv{\framei\kctx}$ and thus $\vartwo \notin\fv{\unf{\framei\kctx}}$.
 \end{enumerate}
 \item $\unf\kctxtwo$ is fine:
 \begin{enumerate}
  \item \sloppy We treat the case of $\kctxp{\ctxhole\esub\vartwo\var}$, for $\kctxp{\ctxhole\esub\varfour{\var\varthree}}$ the reasoning is identical. We have $\unf{\kctxp{\ctxhole\esub\vartwo\var}} = \unf\kctx \ctxholep{\unf{\ctxhole\esub\vartwo\var}\indsub{\wstenv\kctx}}$ by \reflemmap{read-back-decomposition}{c}. Note that $\unf{\ctxhole\esub\vartwo\var}$ is either the inert context $\ctxhole$ or the inert context $\ctxhole\esub\vartwo\var$. Now we apply various lemmas about multi contexts:
\begin{itemize}
\item $\unf{\ctxhole\esub\vartwo\var}$ is a fine multi context by \reflemmap{terms-to-multi-ctxs}{inert-ctx},
\item  it is also is compatible with the fireball substitution $\indsub{\wstenv\kctx}$ because by hypothesis $\var$ is not bound to an abstraction in $\wstenv\kctx$, then 
\item $\unf{\ctxhole\esub\vartwo\var}\indsub{\wstenv\kctx}$ is a fine context by \reflemma{sctx-and-sub}, and finally
\item $\unf\kctx \ctxholep{\unf{\ctxhole\esub\vartwo\var}\indsub{\wstenv\kctx}}$ is a fine context by \reflemmap{strctx-compos}{fine}.
\end{itemize}

  \item Trivial because by \reflemmap{read-back-decomposition}{c} $\unf{\kctxp{\ctxhole\esub\vartwo{\la\var\envtwo}}}
   = \unf\kctx \ctxholep{\unf{\ctxhole\esub\vartwo{\la\var\envtwo}}\indsub{\wstenv\kctx}}
   = \unf\kctx \ctxholep{\ctxhole\indsub{\wstenv\kctx}} = \unf\kctx$ that is fine by hypothesis because $\kctx$ is good.
 \end{enumerate}
 
 \end{itemize}
\end{proof}

Goodness is also stable under removal of ES. In order to show that, we first prove in the following lemma a corresponding property for multi contexts: rigidity, strength, and properness are stable under removal.

\begin{lemma}
  \label{l:strong-proper-context-es-removal}
  Let $\mctx$ be a multi context.
  \begin{enumerate}
    \item \label{p:strong-proper-context-es-removal-a}
    If $\mctxp{\ctxhole\esub\var\tm}$ is a rigid multi context then so is $\mctx$.
    \item \label{p:strong-proper-context-es-removal-b}
    If $\mctxp{\ctxhole\esub\var\tm}$ is an external multi context then so is $\mctx$. 
  \end{enumerate}
  Moreover, if $\mctxp{\ctxhole\esub\var\tm}$ is proper then $\mctx$ is proper.
  \end{lemma}
  
  \begin{proof}
  By induction on $\mctx$.
  \begin{itemize}
    \item \emph{Empty}, \ie $\mctx = \ctxhole$. 
    \begin{enumerate}
      \item Then $\mctxp{\ctxhole\esub\var\tm} = \ctxhole\esub\var\tm$ but no rigid multi context can have this shape, so this case is impossible.
      
      \item Then $\ctxhole$ is an external multi context.
    \end{enumerate}
      
    \item \emph{Variable}, \ie $\rmctx = \var$. 
    \begin{enumerate}
      \item $\var$ is a rigid multi context.
      \item $\var$ is an external multi context.
    \end{enumerate}
      
    \item \emph{Abstraction}, \ie $\mctx = \la\var\mctxtwo$.
    \begin{enumerate}
      \item Then $\mctxp{\ctxhole\esub\var\tm} = \la\var\mctxtwop{\ctxhole\esub\var\tm}$ but no rigid multi context can have this shape, so this case is impossible.		
      \item By \ih $\mctxtwo$ is an external multi context, and then so is $\mctx$.
    \end{enumerate}
        
    \item \emph{Application}, \ie $\mctx = \mctxtwo \mctxthree$. 
    \begin{enumerate}
      \item If $\mctxp{\ctxhole\esub\var\tm} = \mctxtwop{\ctxhole\esub\var\tm} \mctxthreep{\ctxhole\esub\var\tm}$ is a rigid multi context then $\mctxtwop{\ctxhole\esub\var\tm}$ is a rigid multi context and $\mctxthreep{\ctxhole\esub\var\tm}$ is an external multi context. By \ih, $\mctxtwo$ is rigid and $\mctxthree$ is strong. Then $\mctx$ is rigid.
      \item $\mctxp{\ctxhole\esub\var\tm}$ is an external multi context only if it is rigid. By \refpoint{strong-proper-context-es-removal-a}, $\mctx$ is rigid, and thus strong.
    \end{enumerate}
      
    \item \emph{Explicit substitution}, \ie $\mctx = \mctxtwo \esub\vartwo\mctxthree$.
    \begin{enumerate}
      \item If $\mctxp{\ctxhole\esub\var\tm} = \mctxtwop{\ctxhole\esub\var\tm} \esub\vartwo{\mctxthreep{\ctxhole\esub\var\tm}}$ is a rigid multi context then $\mctxtwop{\ctxhole\esub\var\tm}$ is a rigid multi context and $\mctxthreep{\ctxhole\esub\var\tm}$ is a rigid multi context. By \ih, both $\mctxtwo$ and $\mctxthree$ are rigid, and then so is $\mctx$.
      \item If $\mctxp{\ctxhole\esub\var\tm} = \mctxtwop{\ctxhole\esub\var\tm} \esub\vartwo{\mctxthreep{\ctxhole\esub\var\tm}}$ is an external multi context then $\mctxtwop{\ctxhole\esub\var\tm}$ is an external multi context and $\mctxthreep{\ctxhole\esub\var\tm}$ is a rigid multi context. By \ih, $\mctxtwo$ is an external multi context and $\mctxthree$ is a rigid multi context. Then $\mctx$ is an external multi context.
    \end{enumerate}
  \end{itemize}	
  The moreover part follows evidently holds in the base cases, and it follows immediately from the \ih in the inductive cases.
  \end{proof}

The invariance of goodness by removals of the innermost ES next to the hole in $\kctx$ is simpler than stability for addition, and it is given by the next lemma.

\begin{lemma}[Goodness removal]
\label{l:goodness-removal}~
\begin{enumerate}
  \item If $ \esub\var\mol\env$ is open good
  then $\env$ is open good.
  
  \item If $\kctxp{\ctxhole\esub\var\mol}$ is good, then $\kctx$ is good.  
\end{enumerate}
\end{lemma}

\begin{proof}
\hfill
\begin{enumerate}
 \item We have to prove three facts:
 \begin{enumerate}
  \item \emph{$\indsub\env$ is a fireball substitution}. By \reflemmap{sigma-wk-cons}{a}, $\indsub{\esub\var\mol\env}$ is $\indsub\env$ plus possibly a substitution on $\var$. Therefore, if $\indsub{\esub\var\mol\env}$ is a fireball substitution then so is $\indsub\env$.
  
  \item \emph{$\indenv\env$ is a inert context}. By \reflemmap{sigma-wk-cons}{c},
  $\indenv{\esub\var\mol\env}$ is $\indenv\env$ plus possibly an ES on $\var$. Therefore, if $\indenv{\esub\var\mol\env}$ is a inert substitution context then so is $\indenv\env$.
    
    \item \emph{$\env$ has immediate values}. By \reflemmap{sigma-wk-cons}{a}, $\indsub{\esub\var\mol\env}$ is $\indsub\env$ plus possibly a substitution on $\var$. Therefore, if $\indsub{\esub\var\mol\env}$ has immediate values then so does $\indsub\env$.
 \end{enumerate}
 
 \item By hypothesis $\wstenv{\kctxp{\ctxhole\esub\var\mol}}=_{\reflemmaeqp{unsenv-prefix}{one}}\esub\var\mol \wstenv\kctx$ is open good, and so $\wstenv\kctx$ is open good by the previous point.
\begin{enumerate}
 \item We prove that $\framei\kctx$ is well-framed w.r.t. $\indsub{\wstenv\kctx}$. By hypothesis, ${\framei{\kctxp{\ctxhole\esub\var\mol}}} =_\reflemmaeqp{wk-st-lam}{b} \framei\kctx$ is well-framed w.r.t. $\indsub{\wstenv{\kctxp{\ctxhole\esub\var\mol}}} = \indsub{\esub\var\mol\wstenv\kctx} $.
 Now, by \reflemmap{sigma-wk-cons}{a}, we have that $\indsub{\esub\var\mol\wstenv{\kctx}}$ is $\indsub{\wstenv\kctx}$ plus possibly a substitution on $\var$. Therefore, compatibility with respect to $\indsub{\esub\var\mol\wstenv{\kctx}}$ implies compatibility with respect to $\indsub{\wstenv\kctx}$ and thus $\framei\kctx$ is also well-framed w.r.t. $\indsub{\wstenv\kctx}$.
 \item We prove $\unf\kctx$ is a fine multi context. By hypothesis we know that $\kctxp{\ctxhole\esub\var\mol}$ is good and
  therefore that $\unf{\kctxp{\ctxhole\esub\var\mol}} =_\reflemmaeqp{read-back-decomposition}{c} \unf\kctx \ctxholep{\unf{\ctxhole\esub\var\mol}\indsub{\wstenv\kctx}}$ is a fine multi context. Two cases:
  \begin{itemize}
  	\item If $\mol = \val$ or $\var\in \crnames$ then $\unf\kctx \ctxholep{\unf{\ctxhole\esub\var\mol}\indsub{\wstenv\kctx}} = \unf\kctx \ctxholep{\ctxhole\isub\var{\unf\mol}\indsub{\wstenv\kctx}} = \unf\kctx$, which is then a fine multi context. 
	\item Otherwise, $\unf\kctx \ctxholep{\unf{\ctxhole\esub\var\mol}\indsub{\wstenv\kctx}}= \unf\kctx \ctxholep{\ctxhole\esub\var{\mol\indsub{\wstenv\kctx}}}$. By \reflemmap{strong-proper-context-es-removal}{b}, $\unf\kctx $ is a fine multi context.
	\qedhere
 	\end{itemize}
 \end{enumerate}
\end{enumerate} 
\end{proof}

\bigskip

We now have all the ingredients to conclude that also Goodness is propagated.


\begin{theorem}[Goodness invariant]
     \label{thm:goodness-invariants}
  Let $\state=\env \gensep\kctx$ be a state reachable from an initial state $\state_0$. Then $\state$ is good.
\end{theorem}
\begin{proof}

By induction on the execution $\exec:\state_0 \tosmach^* \state$.
 If $\exec$ is empty then $\state = \state_0$ and, by definition of initial state, $\state_0 = \env_0 \rlsep \ctxhole$ for some well-named and pristine environment $\env_0$: $\env_0 \rlsep \ctxhole$ is \good by the definition of good, since $\env_0$ is pristine.

 If $\exec$ is non-empty we look at the last transition $\statetwo  \tosmach \state$, knowing by \ih that the goodness invariant holds
for $\statetwo$:

\begin{itemize}
  
\item 
$\env \esub\var{\vartwo\,\varthree} \rlsep \kctx
 \tomachbv 
\env (\esub\var \mol \envtwo \isub\varfour\varthree) \rlsep    \kctx$
with $\rename{(\wstenv\kctx(\vartwo))} = \la\varfour(\esub{\varstar}\mol\envtwo) $ and $\wstenv\kctx(\varthree)= \val$ for some $\val$.

Goodness follows from the \ih

     \medskip

\item 
$\env \esub\var{\vartwo\,\varthree}  \rlsep   \kctx
  \tomachbi 
\env \esub\var \mol \envtwo  \rlsep  \kctxp{ \ctxhole\esub\varfour\varthree   }$
with $\rename{(\wstenv\kctx(\vartwo))} = \la\varfour(\esub{\varstar}\mol\envtwo) $ and $\wstenv\kctx(\varthree)= \itm$ for some inert term $\itm$.

By \ih{}, $\kctx$ is good; then $\kctxp{\ctxhole\esub\varfour\varthree}$ is good by \reflemma{hereditarily-good}.

     \medskip

\item ${\env \esub\var\vartwo}  \rlsep  \kctx
 \tomachsub
 {\env \isub\var\vartwo}  \rlsep  \kctx $
with $\var\neq\varstar$.

Goodness follows from the \ih

\medskip

\item ${\env \esub\var\mol}  \rlsep  \kctx
     \tomachcone 
\env  \rlsep  \kctxp{ {\ctxhole\esub\var\mol}  }
$
     when $\mol$ is an abstraction or when $\mol$ is $\vartwo$ or $\vartwo\varthree$ but $\vartwo$ is not defined in $\wstenv\kctx$ or $\wstenv\kctx(\vartwo)$ is not a value.

     To show that $\kctxp{\ctxhole\esub\var\mol}$ is good, note that by \ih{}, $\kctx$ is good. For all cases but when $\mol$ is an abstraction we can immediately apply \reflemma{hereditarily-good} and obtain that $\kctxp{\ctxhole\esub\var\mol}$ is good. When $\mol$ is an abstraction, by \ih we obtain that $\statetwo$ is well-named, and so $\var\notin\fv{\framei\kctx}$. Then we can apply \reflemma{hereditarily-good} and obtain that $\kctxp{\ctxhole\esub\var\mol}$ is good. 

\medskip

\item $\emptyenv  \rlsep  \kctx
 \tomachctwo 
\emptyenv  \lrsep  \kctx$

$\emptyenv  \lrsep  \kctx$ is obviously good because the property holds by \ih

     \medskip

\item $\env  \lrsep  \kctxp{ {\ctxhole\esub\var\mol  }}
 \tomachcthree 
{\env \esub\var\mol}  \lrsep  \kctx
$
where $\mol$ is a variable or an application.

To show that ${\env \esub\var\mol}  \lrsep  \kctx$ is \good:
\begin{itemize}
 \item \emph{$\kctx$ is \good}: by \ih{}, $\kctxp{\ctxhole\esub\var\mol}$ is \good. \sloppy By \reflemma{goodness-removal}, $\kctx$ is good.
 \item \emph{$\unf{\env \esub\var\mol}$ is a strong fireball compatible with $\indsub{\wstenv\kctx}$}: by \ih{}, $\unf{\env}$ is a strong fireball compatible with $\indsub{\wstenv{\kctxp{\ctxhole\esub\var\mol}}} =_{\reflemmaeqp{unsenv-prefix}{one}} \indsub{\esub\var\mol\wstenv\kctx}$. Two cases:
  \begin{itemize}
    \item $\var\in\crnames$:
    $\unf{\env\esub\var\mol} = \unf\env\isub\var{\unf\mol}$
     and $\indsub{\esub\var\mol\wstenv\kctx} =_\reflemmaeqp{sigma-wk-cons}{a} \isub{\var}{\indsub{\wstenv\kctx}(\mol)}\cup\indsub{\wstenv\kctx}$.
     Because $\unf\env$ is a strong fireball (by \ih) and $\unf\mol=\mol$ is a strong inert (by hypothesis of the transition), then $\unf\env\isub\var{\unf\mol}$ is a strong fireball because strong fireballs are stable under the substitution of strong inerts. Let $\vartwo$ be such that $\indsub{\wstenv\kctx}$ is a value. 
     By compatibility of $\unf\env$ with $\isub{\var}{\indsub{\wstenv\kctx}(\mol)}\cup\indsub{\wstenv\kctx}$, $\vartwo$ occurs in $\unf\env$ only as an argument, if ever. Two cases:
     \begin{itemize}
     \item $\mol$ is a variable $\varthree$ and $\varthree \neq \varstar$ since it occurs in $\esub\var\mol$. Note that $\varthree\neq \vartwo$, otherwise $\wstenv\kctx$ would not have immediate values, against goodness of $\kctx$, and that for the same reason one also has $\var\neq\vartwo$. Then $\vartwo$ does not occur as an argument in $\unf{\env \esub\var\mol} = \unf\env\isub\var\varthree$. 
     
     \item $\mol$ is an application $\varthree\varfour$. 
     Note that if $\varthree = \vartwo$ then $\indsub{\wstenv\kctx}(\mol)$ is  not a fireball, against the hypothesis that $\isub{\var}{\indsub{\wstenv\kctx}(\mol)}\cup\indsub{\wstenv\kctx}$ is a fireball substitution. Then $\vartwo$ occurs only as an argument in $\unf{\env \esub\var\mol}$, and compatibility holds.
    \end{itemize}
     
    \item $\var\in\calcnames$: then $\unf{\env \esub\var\mol} = \unf\env \esub\var{\mol}$ and $\indsub{\esub\var\mol\wstenv\kctx} =_\reflemmaeqp{sigma-wk-cons}{a} \indsub{\wstenv\kctx}$. Since $\unf\env$ is a strong fireball (by \ih) and $\mol$ is a strong inert term (by hypothesis of the transition), then $\unf{\env \esub\var\mol}$ is a strong fireball. Compatibility for $\unf\env$ follows from the \ih, we only have to analyze $\esub\var\mol$. Let $\vartwo$ be such that $\indsub{\wstenv\kctx}$ is a value. Two cases:
     \begin{itemize}
     \item $\mol$ is a variable $\varthree$ and $\varthree \neq \varstar$ since it occurs in $\esub\var\mol$. Note that $\varthree\neq \vartwo$, otherwise $\wstenv\kctx$ would not have immediate values, against goodness of $\kctx$, and that for the same reason one also has $\var\neq\vartwo$. Then compatibility holds.
     
     \item $\mol$ is an application $\varthree\varfour$. 
     Note that if $\varthree = \vartwo$ then $\indsub{\wstenv\kctx}(\mol)$ is not a fireball, against the hypothesis that $\isub{\var}{\indsub{\wstenv\kctx}(\mol)}\cup\indsub{\wstenv\kctx}$ is a fireball substitution. Then $\varthree \neq \vartwo$ and compatibility holds.
    \end{itemize}    
  \end{itemize}
\end{itemize}
     \medskip

\item $\env  \lrsep  \kctxp{ {\ctxhole\esub\var\val  }}
 \tomachgc 
\env  \lrsep  \kctx$
with $\var \notin \fv\env$.

By \ih, $\kctxp{\ctxhole\esub\var\val}$ is \good. By \reflemma{goodness-removal}, $\kctx$ is good.

     \medskip

\item $\env  \lrsep  \kctxp{\envtwo {\esub\var{\la\vartwo\ctxhole }  }}
 \tomachcfour 
{\envtwo \esub\var{\la\vartwo\env}}  \lrsep   \kctx$. 

To show that ${\envtwo \esub\var{\la\vartwo\env}}  \lrsep   \kctx$ is \good:
\begin{itemize}
 \item \emph{$\kctx$ is \good}: by \ih{} $\kctxp{\envtwo \esub\var{\la\vartwo\ctxhole}}$ is \good. We have to prove that:
 \begin{itemize}
  \item \emph{$\wstenv\kctx$ is open good}: note that $\wstenv\kctx =_\reflemmaeqp{unsenv-prefix}{two} \wstenv{\kctxp{\envtwo \esub\var{\la\vartwo\ctxhole}}}$ and $\wstenv{\kctxp{\envtwo \esub\var{\la\vartwo\ctxhole}}}$ is open good because $\kctxp{\envtwo \esub\var{\la\vartwo\ctxhole}}$ is good.
  
  \item \emph{$\framei\kctx$ is well-framed w.r.t. $\indsub{\wstenv\kctx}$}.
   By \ih $\framei{\kctxp{\envtwo \esub\var{\la\vartwo\ctxhole}}} =_\reflemmaeqp{wk-st-lam}{a} \framei\kctx\ctxholep{\envtwo\esub\var{\la\vartwo\ctxhole}} $ is well-framed w.r.t.
   $\indsub{\wstenv{\kctxp{\envtwo \esub\var{\la\vartwo\ctxhole}}}} =_\reflemmaeqp{unsenv-prefix}{two}
    \indsub{\wstenv\kctx}$. Therefore also $\framei\kctx$ is well-framed w.r.t. $\indsub{\wstenv\kctx}$ by
   \reflemma{well-framed-match}.
 \item \emph{$\unf\kctx$ is fine}:
  by hypothesis we know that $\kctxp{\ctxhole\esub\var{\la\vartwo\env}}$ is good and
  therefore that $\unf{\kctxp{\ctxhole\esub\var{\la\vartwo\env}}} =_\reflemmaeqp{read-back-decomposition}{c} \unf\kctx\ctxholep{\unf{\ctxhole\esub\var{\la\vartwo\env}}\indsub{\wstenv\kctx}} = \unf\kctx\ctxholep{\ctxhole\isub\var{\la\vartwo\unf\env}\indsub{\wstenv\kctx}} = \unf\kctx$ is a fine multi context. 

 \end{itemize}

 \item \emph{$\unf{\envtwo \esub\var{\la\vartwo\env}}$ is a strong fireball compatible with $\indsub{\wstenv\kctx}$}: since $\env \lrsep \kctxp{\envtwo \esub\var{\la\vartwo\ctxhole}}$ is \good, we have that
 \begin{itemize}
 \item $\unf{\env}$ is a strong fireball compatible with $\indsub{\wstenv{\kctxp{\envtwo \esub\var{\la\vartwo\ctxhole}}}} =_\reflemmaeqp{unsenv-prefix}{two} \indsub{\wstenv\kctx}$.

 \item $\unf{\framei{\kctxp{\envtwo \esub\var{\la\vartwo\ctxhole}}}} =_\reflemmaeqp{wk-st-lam}{a} \unf{\framei\kctx\ctxholep{\envtwo\esub\var{\la\vartwo\ctxhole}}}$ is a fine context compatible with $\indsub{\wstenv{\kctxp{\envtwo \esub\var{\la\vartwo\ctxhole}}}} =_\reflemmaeqp{unsenv-prefix}{two} \indsub{\wstenv\kctx}$. 
 \end{itemize}
\sloppy Then $\unf{\envtwo \esub\var{\la\vartwo\env}} = \unf{\framei{\kctxp{\envtwo \esub\var{\la\vartwo\ctxhole}}}\ctxholep\env} =_{\reflemmaeqp{frame-unf-factorization}{b}} \unf{\framei{\kctxp{\envtwo \esub\var{\la\vartwo\ctxhole}}}} \ctxholep{\unf\env}$ is a strong fireball because $\framei{\kctxp{\envtwo \esub\var{\la\vartwo\ctxhole}}}$ is normal by \ih

By \reflemma{smctx-fire-compatibility}, compatibility of $\unf{\framei{\kctxp{\envtwo \esub\var{\la\vartwo\ctxhole}}}}$ and $\unf\env$ with $\indsub{\wstenv\kctx}$ implies compatibility of $\unf{\framei{\kctxp{\envtwo \esub\var{\la\vartwo\ctxhole}}}} \ctxholep{\unf\env}$ with $\indsub{\wstenv\kctx}$.
  \end{itemize}

     \medskip

\item $\env  \lrsep  \kctxp{ {\ctxhole\esub\var{\la\vartwo\envtwo} } }
 \tomachcfive 
\envtwo  \rlsep  \kctxp{\env {\esub\var{\la\vartwo\ctxhole}}  }
$ with $\var \in \fv\env$.

We need to prove that $\envtwo  \rlsep  \kctxp{\env {\esub\var{\la\vartwo\ctxhole}}  }$ is \good.
By \ih{}, $\env \lrsep \kctxp{\ctxhole \esub\var{\la\vartwo\envtwo}}$ is good, which means that:
\begin{itemize}
 \item ${\wstenv{\kctxp{\ctxhole \esub\var{\la\vartwo\envtwo}}}} =_\reflemmaeqp{unsenv-prefix}{one} \esub\var{\la\vartwo\envtwo}\wstenv\kctx$ is open good,
 
 \item $\framei{\kctxp{ {\ctxhole\esub\var{\la\vartwo\envtwo} } }}$ is well-framed w.r.t.
  $\indsub{\wstenv{\kctxp{\env {\esub\var{\la\vartwo\ctxhole}}  }}} =_\reflemmaeqp{unsenv-prefix}{two} \indsub{\wstenv\kctx}$.
  Therefore, by \reflemma{well-framed-match}, $\framei\kctx$ is also well-framed w.r.t. $\indsub{\wstenv\kctx}$.

 \item 
  $\unf{\kctxp{\ctxhole\esub\var{\la\vartwo\envtwo}}} =_\reflemmaeqp{read-back-decomposition}{c} \unf\kctx\ctxholep{\unf{\ctxhole\esub\var{\la\vartwo\envtwo}}\indsub{\wstenv\kctx}} = \unf\kctx\ctxholep{\unf{\ctxhole\isub\var{\la\vartwo\unf\envtwo}}\indsub{\wstenv\kctx}} = \unf\kctx$ is a fine multi context. 
 
  \item $\unf\env$ is a strong fireball compatible with $\indsub{\esub\var{\la\vartwo\envtwo}\wstenv\kctx} =_\reflemmaeqp{sigma-wk-cons}{a} \isub\var{\la\vartwo\unf\envtwo\indsub{\wstenv\kctx}} \indsub{\wstenv\kctx}$,
\end{itemize}

We have to prove that:
\begin{itemize}
  \item \emph{$\kctxp{\env \esub\var{\la\vartwo\ctxhole}}$ is \good}: that is, we have to prove that
  \begin{itemize}
    \item \sloppy \emph{$\wstenv{\kctxp{\env \esub\var{\la\vartwo\ctxhole}}}$ is open good}: note that $\wstenv{\kctxp{\env \esub\var{\la\vartwo\ctxhole}}} =_\reflemmaeqp{unsenv-prefix}{two} \wstenv\kctx$. By \ih (first item above), we know that $\esub\var{\la\vartwo\envtwo}\wstenv\kctx$ is open good. Then, $\wstenv\kctx$ is open good by the open goodness removal lemma (\reflemma{goodness-removal}).
    
    \item
\emph{$\framei {\kctxp{\env \esub\var{\la\vartwo\ctxhole}}}$ is well-framed w.r.t. $\indsub{\wstenv{\kctxp{\env \esub\var{\la\vartwo\ctxhole}}}} =_\reflemmaeqp{unsenv-prefix}{two} \indsub{\wstenv\kctx}$}: note that $\framei {\kctxp{\env \esub\var{\la\vartwo\ctxhole}}} =_\reflemmaeqp{wk-st-lam}{a} \framei\kctx\ctxholep{\env \esub\var{\la\vartwo\ctxhole}}$ and that we already proved that
$\framei\kctx$ is well-framed w.r.t. $\indsub{\wstenv\kctx}$. \sloppy Therefore we just need to show $\unf{\framei\kctx\ctxholep{\env \esub\var{\la\vartwo\ctxhole}}} =_{\reflemmaeqp{frame-unf-factorization}{a}} \unf{\framei\kctx}\ctxholep{\unf\env \isub\var{\la\vartwo\ctxhole}}$
to be a normal fine multi-context compatible with $\indsub{\wstenv\kctx}$.
    
Since $\framei\kctx$ is well-framed w.r.t. $\indsub{\wstenv\kctx}$, $\unf{\framei\kctx}$ is a fine and normal multi context compatible with $\indsub{\wstenv\kctx}$. By the fourth item of the \ih above, $\unf\env$ is a strong fireball compatible with $ \isub\var{\la\vartwo\unf\envtwo\indsub{\wstenv\kctx}} \indsub{\wstenv\kctx}$. By the hypothesis on the transition, $\var\in\fv\env$, and by the garbage invariant $\var\in\fv{\unf\env}$. Note that the compatibility hypothesis on $\unf\env$ implies in particular that  $\unf\env$ is compatible with $\isub\var\val$. 

Then by \reflemma{str-fir-ctx-dec} there exists a fine and normal context $\strongmctx$ such that $\var\notin\fv\strongmctx$ and $\unf\env = \strongmctxp\var$, and (by the hypothesis on $\unf\env$) $\strongmctx$ is compatible with $\indsub{\wstenv\kctx}$. \sloppy By \reflemma{strctx-compos}, $\unf{\env\isub\var{\la\vartwo\ctxhole}} = \strongmctxp\var\isub\var{\la\vartwo\ctxhole} =_{\reflemmaeq{strongmctx-plug-eq-sub}} \strongmctxp{\la\vartwo\ctxhole}$ is a fine and normal context, compatible with $\indsub{\wstenv\kctx}$ because $\strongmctx$ is. Note that $\unf{\framei {\kctxp{\env \esub\var{\la\vartwo\ctxhole}}}} = \unf{\framei\kctx}\ctxholep{\unf\env \isub\var{\la\vartwo\ctxhole}} = \unf{\framei\kctx}\ctxholep{\strongmctxp{ \la\vartwo\ctxhole}}$, which is finally proved to be a fine and normal context compatible with  $\indsub{\wstenv\kctx}$ by applying \reflemma{strctx-compos} once more.

 \item \emph{$\unf{\kctxp{\env \esub\var{\la\vartwo\ctxhole}}}$ is fine:} 
  $\unf{\kctxp{\env \esub\var{\la\vartwo\ctxhole}}} =_\reflemmaeqp{read-back-decomposition}{c}
   \unf{\kctx}\ctxholep{\unf{\env \esub\var{\la\vartwo\ctxhole}}\indsub{\wstenv\kctx}}
   =\unf{\kctx}\ctxholep{\strongmctxp{\la\vartwo\ctxhole}\indsub{\wstenv\kctx}}$

 \sloppy where $\strongmctxp{\la\vartwo\ctxhole}$ is the fine (and normal) context compatible with $\indsub{\wstenv\kctx}$ built in the previous
 proof item. Thus, by \reflemma{sctx-and-sub}, $\strongmctxp{\la\vartwo\ctxhole}\indsub{\wstenv\kctx}$ is a fine multi context and by \reflemmap{strctx-compos}{fine} so is
 $\unf{\kctx}\ctxholep{\strongmctxp{\la\vartwo\ctxhole}\indsub{\wstenv\kctx}}$.
 
 \qedhere
  \end{itemize}
\end{itemize}
\end{itemize}
\end{proof}

We conclude this section with the following theorem, proving the property of machine contexts that motivated the introduction of the Good invariant in the first place:

\begin{theorem}[Contextual read-back]
   \label{thmappendix:ctx-read-back}
      \NoteState{thm:ctx-read-back}
Let $\state = \env\gensep\kctx$ be a reachable state. Then $\unf\kctx$ is a proper external multi context.
\end{theorem}

\begin{proof}
By \refthm{goodness-invariants} $\state$ is good. Then $\unf\kctx$ is a proper external multi context.
\end{proof}

\subsection{Proof of the \SCAM{} Implementation Theorem }

Here we are supposed to prove the following theorem.
\begin{theorem}[\SCAM{} implementation]
\label{thmappendix:machine-final}
\NoteState{th:machine-final}
\begin{enumerate}
  \item \emph{Relaxed $\beta$-projection}: let $\state$ be 
a reachable state. If $\state \tomachbv \statetwo$ then $\unf\state (\toms\toes)^+\unf\statetwo$, and if 
$\state \tomachbi \statetwo$ then $\unf\state \toms^+\eqstruct \unf\statetwo$.

  \item \emph{Strong implementation}: the \SCAM is a relaxed implementation of the strategy $(\tovsubs,\eqstruct)$.
\end{enumerate}
\end{theorem}

We prove the two points separately.

\subsection{Proof of Relaxed Projection}

In this section we provide the proof of \refthboth{machine-final}: the read back projects $\beta$ transitions to steps in the calculus. More precisely, each $\tomachbv$ transition projects to one or more $\toms$ steps followed by as many $\toes$ steps --- \refthpboth{machine-final}{projection-ms} ---, and each $\tomachbi$ transition projects to one or more $\toms$ steps \emph{up to structural equivalence}, \ie{} $\toms^+\eqstruct$ (hence the term ``relaxed'') --- \refthpboth{machine-final}{impl-scam}.

For this reason, we first need an auxiliary lemma that shows how substitutions can be permutated in terms. The following lemma proves (under suitable requirements on variables) that substitution contexts commute with open contexts, up to the structural equality of VSC.
\begin{lemma}[Open and substitution contexts commute up to $\eqstruct$]
  \label{l:open-contextual}
    Let $\sctx$ be a substitution context and $\openctx$ be an open context such that:
    \begin{itemize}
      \item $\domain\sctx \Disj \fv\openctx$,
      \item $\fv\sctx \Disj \domain\openctx$,
      \item $\domain\sctx \Disj \domain\openctx$.
    \end{itemize} 
    Then $\sctxp{\openctxp\tm} \equiv \openctxp{\sctxp\tm}$.
  \end{lemma}
  \begin{proof}
    By induction on the structure of $\openctx$:
    \begin{itemize}
      \item If $\openctx=\ctxhole$, the statement follows trivially.
      \item If $\openctx = \openctxtwo\tmtwo$, then $\sctxp{\openctxp\tm} = \sctxp{\openctxtwop\tm\tmtwo}$.
      We prove that $\sctxp{\openctxtwop\tm\tmtwo} \tostructapl^* \sctxp{\openctxtwop\tm}\tmtwo$ by induction on the structure of $\sctx$:
      \begin{itemize}
        \item The case $\sctx=\ctxhole$ is trivial.
        \item When $\sctx = \sctxtwo\esub\vartwo\tmthree$, $\sctxp{\openctxtwop\tm\tmtwo} = \sctxtwop{\openctxtwop\tm\tmtwo}\esub\vartwo\tmthree$. By \ih{} $\sctxtwop{\openctxtwop\tm\tmtwo}\esub\vartwo\tmthree \tostructapl^* \sctxtwop{\openctxtwop\tm}\tmtwo\esub\vartwo\tmthree$ (note that the additional requirement that $\vartwo\not\in\domain\openctxtwo$ follows from the hypothesis that $\domain\sctx \Disj \domain\openctx$). In order to apply one last time $\tostructapl$ and conclude, we need to show that $\vartwo\not\in\fv\tmtwo$, which follows from the hypothesis that $\domain\sctx \Disj \fv\openctx$.
      \end{itemize}
      Finally, we conclude by using the \ih{}
  
      \item The case $\openctx = \tmtwo\openctxtwo$ is similar to the case above.
      \item If $\openctx = \openctxtwo\esub\var\tmtwo$ then $\sctxp{\openctxp\tm} = \sctxp{\openctxtwo\esub\var\tmtwo}$.
      We prove that $\sctxp{\openctxtwo\esub\var\tmtwo} \tostructcom^* \sctxp\openctxtwo\esub\var\tmtwo$ by induction on the structure of $\sctx$:
      \begin{itemize}
        \item The case $\sctx=\ctxhole$ is trivial.
        \item When $\sctx = \sctxtwo\esub\vartwo\tmthree$, $\sctxp{\openctxtwo\esub\var\tmtwo} = \sctxtwop{\openctxtwo\esub\var\tmtwo}\esub\vartwo\tmthree$. By \ih{} $\sctxtwop{\openctxtwo\esub\var\tmtwo}\esub\vartwo\tmthree \tostructcom^* \sctxtwop\openctxtwo\esub\var\tmtwo\esub\vartwo\tmthree$ (again, the additional requirement that $\vartwo\not\in\domain\openctxtwo$ follows from the hypothesis that $\domain\sctx \Disj \domain\openctx$). In order to apply one last time $\tostructcom$ and conclude, we need to show that $\vartwo\not\in\fv\tmtwo$ and $\var\not\in\fv\tmthree$: the first follows from the hypothesis that $\domain\sctx \Disj \fv\openctx$, the second from $\domain\openctx \Disj \fv\sctx$.
      \end{itemize}
      Finally, we conclude by using the \ih{}
      \item If $\openctx = \tmtwo\esub\var\openctxtwo$ then $\sctxp{\openctxp\tm} = \sctxp{\tmtwo\esub\var{\openctxtwop\tm}} $.
      We prove that $\sctxp{\tmtwo\esub\var{\openctxtwop\tm}} \tostructes^* \tmtwo\esub\var{\sctxp{\openctxtwop\tm}}$ by induction on the structure of $\sctx$:
      \begin{itemize}
        \item The case $\sctx=\ctxhole$ is trivial.
        \item When $\sctx = \sctxtwo\esub\vartwo\tmthree$, $\sctxp{\tmtwo\esub\var{\openctxtwop\tm}} = \sctxtwop{\tmtwo\esub\var{\openctxtwop\tm}}\esub\vartwo\tmthree$. By \ih{} $\sctxtwop{\tmtwo\esub\var{\openctxtwop\tm}}\esub\vartwo\tmthree \tostructes^* \tmtwo\esub\var{\sctxtwop{\openctxtwop\tm}}\esub\vartwo\tmthree$ (again, the additional requirement that $\vartwo\not\in\domain\openctxtwo$ follows from the hypothesis that $\domain\sctx \Disj \domain\openctx$). In order to apply one last time $\tostructes$ and conclude, we need to show that $\vartwo\not\in\fv\tmtwo$, which follows from the hypothesis that $\domain\sctx \Disj \fv\openctx$ and $\domain\sctx \Disj \domain\openctx$.
      \end{itemize}
      Finally, we conclude by using the \ih{}
    \qedhere
    \end{itemize}
  \end{proof}

\begin{theorem}[Relaxed projection]
\label{thmappendix:projection-ms}
\NoteStateP{machine-final}{projection-ms}
Let $\state$ be a reachable state. 
\begin{enumerate} 
 \item If $\state \tomachbv \statetwo$ then $\unf\state \mathrel{(\toms\toes)}^+\unf\statetwo$.
 \item If $\state \tomachbi \statetwo$ then $\unf\state \toms^+\eqstruct \unf\statetwo$.
\end{enumerate}
\end{theorem}
\begin{proof}
  \sloppy The state $\state$ must be $\env \esub\var{\vartwo\,\varthree} \rlsep \kctx$ with $\rename{(\wstenv\kctx(\vartwo))} = \la\varfour(\esub{\varstar}\mol\envthree)$.
There are two cases, depending on whether $\wstenv\kctx(\varthree)$ is an abstraction or an inert.
  \begin{itemize}
    \item \emph{Abstraction}, \ie{} $\wstenv\kctx(\varthree) = \val$.
    Since $\state$ is reachable it is well-named by \refthm{named-invariant} and thus by \reflemmap{wstenv-well-named}{e} $\kctx$ is well-named and so by \reflemmap{wstenv-well-named}{c} $\wstenv\kctx$ is well named and thus
    by \reflemma{lookup-indsub} $\varthree\indsub{\wstenv\kctx} = \indsub{\wstenv\kctx}(\varthree)$
    is a value. The machine transition is
    \[
    \state = \env \esub\var{\vartwo\,\varthree} \rlsep \kctx
      \tomachbv
    \env (\esub \var \mol \envthree \isub\varfour\varthree) \rlsep \kctx = \statetwo
    \]
    and it projects as follows:
    \[\begin{array}{rlllcccc}
      \unf\state & = & \unf{(\env \esub\var{\vartwo\,\varthree} \rlsep \kctx)}
      \\
      & = &
      \unf\kctx \ctxholep{\unf{(\env \esub\var{\vartwo\,\varthree})}{\indsub{\wstenv\kctx}}}
      \\
      & =_{\kctx \textit{ good}} &
      \strongmctxp{\unf{(\env \esub\var{\vartwo\,\varthree})}{\indsub{\wstenv\kctx}}}
      \\
      & =_{\reflemmaeq{weak-unfolding-pristine-new}} &
      \strongmctxp{\openctxp{\unf{\esub\varstar{\vartwo\varthree}}{\indsub{\wstenv\kctx}}}}
      \\
      & = &
      \strongmctxp{\openctxp{\vartwo{\indsub{\wstenv\kctx}}~\varthree{\indsub{\wstenv\kctx}}}}
      \\
      & = &
      \strongmctxp{\openctxp{(\la\varfour\unf{(\esub{\varstar}\mol\envthree)}{\indsub{\wstenv\kctx}})~(\varthree{\indsub{\wstenv\kctx}})}}
      \\
      & {(\toms\toes)^+}_{\!\!\!\!\scriptsize \begin{array}{l}\reflemmaeq{multi-step},\\ \varthree{\indsub{\wstenv\kctx}}\textit{ is a value}\end{array}} &
      \strongmctxp{\openctxp{\unf{(\esub{\varstar}\mol\envthree)}{\indsub{\wstenv\kctx}}\isub\varfour{\varthree{\indsub{\wstenv\kctx}}}}}
      \\
      & = &
      \strongmctxp{\openctxp{\unf{(\esub{\varstar}\mol\envthree)}\isub\varfour\varthree{\indsub{\wstenv\kctx}}}}
      \\
      & =_{\reflemmaeqp{properties-unfolding}{renaming}} &
      \strongmctxp{\openctxp{\unf{(\esub{\varstar}\mol\envthree\isub\varfour\varthree)}{\indsub{\wstenv\kctx}}}}
      \\
      & = &
      \strongmctxp{\openctxp{\unf{(\esub\varstar{\mol\isub\varfour\varthree}(\envthree\isub\varfour\varthree))}{\indsub{\wstenv\kctx}}}}
      \\
      & =_{\reflemmaeq{weak-unfolding-pristine-new}} &
      \strongmctxp{\unf{(\env \esub\var{\mol\isub\varfour\varthree} (\envthree \isub\varfour\varthree)))}{\indsub{\wstenv\kctx}}} 
      \\
      & = &
      \strongmctxp{\unf{(\env (\esub\var \mol \envthree \isub\varfour\varthree))}{\indsub{\wstenv\kctx}}} 
      \\
      & =_{\kctx \textit{ good}} & \unf\kctx \ctxholep{\unf{(\env (\esub\var \mol \envthree \isub\varfour\varthree))}{\indsub{\wstenv\kctx}}} 
      \\
      & = & \unf{(\env (\esub\var \mol \envthree \isub\varfour\varthree) \rlsep \kctx)} & = & \unf\statetwo
      \\ & = & \unf\statetwo
      \end{array}\]


    \item \emph{Inert}, \ie{} $\wstenv\kctx(\varthree) = \itm$.
    Since $\state$ is reachable it is good by \refthm{goodness-invariants} and so 
    $\indsub{\wstenv\kctx}$ has immediate values.
    Thus, since $\varthree \neq \varstar$ because it occurs in $\esub\var{\vartwo\varthree}$ and since $\wstenv\kctx(\varthree)$ is an inert,
    $\varthree\indsub{\wstenv\kctx} = \indsub{\wstenv\kctx}(\varthree)$ must be an inert by definition of the immediate values property.
    The machine transition is
    \[
    \state = \env \esub\var{\vartwo\,\varthree} \rlsep \kctx
      \tomachbi
    \env \esub\var \mol \envthree \rlsep \kctxp{\ctxhole\esub\varfour\varthree} = \statetwo
    \]
    and it projects as follows:
    \[\begin{array}{rlllcccc}
      \unf\state & = & \unf{(\env \esub\var{\vartwo\,\varthree} \rlsep \kctx)}
      \\
      & = &
      \unf\kctx \ctxholep{\unf{(\env \esub\var{\vartwo\,\varthree})}{\indsub{\wstenv\kctx}}}
      \\
      & =_{\kctx \textit{ good}} &
      \strongmctxp{\unf{(\env \esub\var{\vartwo\,\varthree})}{\indsub{\wstenv\kctx}}}
      \\
      & =_{\reflemmaeq{weak-unfolding-pristine-new}} &
      \strongmctxp{\openctxp{\unf{\esub\varstar{\vartwo\varthree}}{\indsub{\wstenv\kctx}}}}
      \\
      & = &
      \strongmctxp{\openctxp{\vartwo{\indsub{\wstenv\kctx}} ~\varthree{\indsub{\wstenv\kctx}}}}
      \\
      & = &
      \strongmctxp{\openctxp{(\la\varfour\unf{(\esub{\varstar}\mol\envthree)}{\indsub{\wstenv\kctx}}) ~(\varthree{\indsub{\wstenv\kctx}})}}
      \\
      & {\toms^+}_{\!\!\!\scriptsize \begin{array}{l}\reflemmaeq{multi-step},\\ \varthree{\indsub{\wstenv\kctx}}\textit{ is an inert}\end{array}} &
      \strongmctxp{\openctxp{\unf{(\esub{\varstar}\mol\envthree)}{\indsub{\wstenv\kctx}} \esub\varfour{\varthree{\indsub{\wstenv\kctx}}}}}
      \\
      & \equiv_{\reflemmaeq{open-contextual}} &
      \strongmctxp{\openctxp{\unf{(\esub{\varstar}\mol\envthree)}{\indsub{\wstenv\kctx}}} \esub\varfour{\varthree{\indsub{\wstenv\kctx}}}}
      \\
      & =_{\reflemmaeq{weak-unfolding-pristine-new}} &
      \strongmctxp{\unf{\env \esub\var \mol \envthree} {\indsub{\wstenv\kctx}} \esub\varfour{\varthree{\indsub{\wstenv\kctx}}}}
      \\
      & = &
      \strongmctxp{\unf{\env \esub\var \mol \envthree} \esub\varfour{\varthree}{\indsub{\wstenv\kctx}} }
      \\
      & = & \strongmctxp{\unf{(\env \esub\var \mol \envthree \esub\varfour\varthree )}{\indsub{\wstenv\kctx}}} 
      \\
      & =_{\kctx \textit{ good}} & \unf\kctx \ctxholep{\unf{(\env \esub\var \mol \envthree \esub\varfour\varthree )}{\indsub{\wstenv\kctx}}} 
      \\
      & = & \unf{(\env \esub\var \mol \envthree \rlsep \kctxp{\ctxhole\esub\varfour\varthree})}
      \\ & = & \unf\statetwo
      \end{array}\]

  \qedhere
  \end{itemize}
\end{proof}

\subsection{Proof of overhead transparency}

In this section we provide the proof of \reflemma{overhead-transp}, \ie{} that the read back projects overhead transitions to equality in the VSC calculus.

\begin{lemma}[Overhead transparency]\label{l:overhead-transp}
Let $\state$ be reachable. If $\state \tomach \statetwo$ with a non-$\beta$ transition, then $\unf\state \equiv \unf\statetwo$.
\end{lemma}
\begin{proof}
By inspection of the non-multiplicative machine transitions.
The proof for all the search steps is trivial since a search step has the form
$\env_1 \gensep \kctx_1 \to \env_2 \gensep \kctx_2$ with
$\kctx_1\ctxholep{\env_1} = \kctx_2\ctxholep{\env_2}$ and thus
$\unf\state = \unf{\kctx_1\ctxholep{\env_1}} = \unf{\kctx_2\ctxholep{\env_2}} = \unf\statetwo$.
We analyze the remaining steps.
\begin{itemize}
\item Case ${\env \esub\var\vartwo}  \rlsep  \kctx
\tomachsub
{\env \isub\var\vartwo}  \rlsep  \kctx $
with $\var\neq\varstar$.
Since ${\env \esub\var\vartwo}  \rlsep  \kctx$ is reachable, it is well-named and therefore,
by \reflemmap{wstenv-well-named}{d}, $\env \esub\var\vartwo$ is well-named and thus, by \reflemma{decomp-well-named-two},
$\env$ is well-named and $\vartwo \not\in \bv\env$.

$$\begin{array}{lll}
& \unf{{\env \esub\var\vartwo}  \rlsep  \kctx}\\
= &\unf{\kctxp{\env \esub\var\vartwo}}\\
=_\reflemmaeqp{read-back-decomposition}{c} & \unf\kctx\ctxholep{\unf{\env\esub\var\vartwo}\indsub{\wstenv\kctx}}\\
=_{\var \in \crnames} & \unf\kctx\ctxholep{\unf\env\isub\var\vartwo\indsub{\wstenv\kctx}}\\
=_\reflemmaeqp{properties-unfolding}{renaming} & \unf\kctx\ctxholep{\unf{\env\isub\var\vartwo}\indsub{\wstenv\kctx}}\\
=_\reflemmaeqp{read-back-decomposition}{c} &\unf{\kctxp{\env \isub\var\vartwo}}\\
= & \unf{{\env \isub\var\vartwo}  \rlsep  \kctx}\\
\end{array}$$


\item Case
  $\env  \lrsep  \kctxp{ {\ctxhole\esub\var\val  }}
  \tomachgc 
  \env  \lrsep  \kctx$
  with $\var \notin \fv\env$.

 $$\begin{array}{lll}
 & \unf{\env  \lrsep  \kctxp{ {\ctxhole\esub\var\val  }}}\\
 = & \unf{\kctxp{ {\env\esub\var\val  }}}\\
 =_\reflemmaeqp{read-back-decomposition}{c} & \unf\kctx\ctxholep{\unf{\env\esub\var\val}\indsub{\wstenv\kctx}}\\
 = &  \unf\kctx\ctxholep{\unf\env\isub\var{\unf\val}\indsub{\wstenv\kctx}}\\
 =_{\reflemmaeqp{properties-unfolding}{fv},\,\var \not\in \fv\env} & \unf\kctx\ctxholep{\unf\env\indsub{\wstenv\kctx}}\\
 =_\reflemmaeqp{read-back-decomposition}{c} & \unf{\kctxp\env}\\
 = & \unf{\env  \lrsep  \kctx}
 \end{array}$$
\end{itemize}
\end{proof}

We now have all the ingredients to prove \refthpboth{machine-final}{impl-scam}, \ie{} the implementation theorem for the \SCAM machine:
\begin{theorem}[Strong CbV Implementation]
    \label{thmappendix:impl-scam}
    \NoteStateP{machine-final}{impl-scam}
    The \SCAM is a relaxed implementation of the strong fireball strategy $(\tovsubs,\eqstruct)$.
\end{theorem}
\begin{proof}
The strong fireball strategy $(\tovsubs,\eqstruct)$ is a structural strategy by \Cref{propappendix:strong-bisimulation}.
We show that $(\tomachscam, \tovsubs, \eqstruct, \unf\cdot)$
form a relaxed implementation system, and obtain the statement by \refthm{abs-impl}. First of all, the 
initialization constraint for the \SCAM{} is given by \reflemma{transl-properties}. About the 
conditions for a relaxed implementation system:
    \begin{enumerate}        
        
        \item \emph{Relaxed $\beta$-projection:} by \refthpboth{machine-final}{projection-ms}.
        
        \item \emph{Overhead transparency:} by \reflemma{overhead-transp}.
        
        \item \emph{Overhead transitions terminates:} it follows by \refcorollary{bound-com}.
        
    \item \emph{Halt:} by inspection of the \SCAM{} transitions, the only normal states of the machine have the form 
$\env\lrsep$. Since this state is good, $\unf\env$ is a strong fireball. By \reflemma{harmony}%
    , $\unf\env$ is a $\tovsubs$-normal form.
    
    \item \emph{Relaxed determinism}:
	  	\begin{itemize}
          \item \emph{$\tovsubs$ is diamond}: by \refpropp{vsc-diamond}{diamond}.
			\item \emph{$\tomachscam$ is deterministic}: by a simple inspection of the transitions and the well-naming 
property. Well-naming grants uniqueness of lookup in the environment during transitions $\tomachbv$ and 
$\tomachbi$.
		\end{itemize}

    \end{enumerate}
\end{proof}

\section{Proofs of Section~\ref*{SECT:COMPLEXITY} (Complexity)}
\label{app:complexity}

In this section we prove that the \SCAM can be implemented within a
bilinear time overhead. The fundamental invariant is the \emph{size} invariant, proved in \refthm{size-invariant}: it basically shows that the size of the abstractions present in the unevaluated parts of a reachable state is bound by the size of the initial state. 

First of all, we show that the measure of the initial state is linearly related to the size of the initial \lat{}:

\begin{lemma}[Linear compilation]
   \label{lappendix:bound-measure-size}
   \NoteState{l:bound-measure-size}
   Let $\tm$ be a \lat{}. Then 
   $\size{\mytr\tm} \leq 2\size\tm$.
\end{lemma}
\begin{proof}
   We prove this statement mutually with the corresponding statement for the auxiliary translation, \ie{} that $\size{\env} \leq 2\size\tm$ when $(\_,\env) \defeq \auxtr\tm$.
   We proceed by induction on the structure of $\tm$:
   \begin{itemize}
      \item If $\tm=\var$, then $\mytr\tm = \esub\varstar\var$, and $\size{\mytr\tm} = 2 = 2\size\tm$.
      \item If $\tm=\la\var\tmtwo$, then $\mytr\tm = \esub\varstar{\la\var{\mytr\tmtwo}}$, and $\size{\mytr\tm} = 2 + \size{\mytr\tmtwo} \leq_{\ih{}} 2 + 2\size\tmtwo = 2(\size\tmtwo + 1) = 2\size\tm$.
      \item If $\tm=\tmtwo\tmthree$, then $\mytr\tm = \esub\varstar{\var\vartwo}\env\envtwo$ where $(\var,\env) \defeq \mytr\tmtwo$ and $(\vartwo,\envtwo) \defeq \mytr\tmthree$. By \ih{} $\size\env \leq 2\size\tmtwo$ and $\size\envtwo \leq 2\size\tmthree$. Hence $\size{\mytr\tm} = \size{\esub\varstar{\var\vartwo}\env\envtwo} = 2 + \size\env + \size\envtwo \leq_{\ih{}} 2 + 2\size\tmtwo + 2\size\tmthree = 2(\size\tmtwo + \size\tmthree + 1) = 2\size\tm$.
   \end{itemize}
   Concerning the auxiliary translation:
   \begin{itemize}
      \item If $\tm=\var$, then $\auxtr\tm = (\var, \emptyenv)$, and $\size{\emptyenv} = 0 < 2 = 2\size\tm$.
      \item If $\tm=\la\var\tmtwo$, then $\auxtr\tm = (\varthree,\esub\varthree{\la\var{\mytr\tmtwo}})$, and $\size{\env} = 2 + \size{\mytr\tmtwo} \leq_{\ih{}} 2 + 2\size\tmtwo = 2\size\tm$.
      \item If $\tm=\tmtwo\tmthree$, then $\auxtr\tm = (\varthree,\esub\varthree{\var\vartwo}\env\envtwo)$ where $(\var,\env) \defeq \mytr\tmtwo$ and $(\vartwo,\envtwo) \defeq \mytr\tmthree$. By \ih{} $\size\env \leq 2\size\tmtwo$ and $\size\envtwo \leq 2\size\tmthree$. Hence $\size{\esub\varthree{\var\vartwo}\env\envtwo} = 2 + \size\env + \size\envtwo \leq_{\ih{}} 2 + 2\size\tmtwo + 2\size\tmthree = 2\size\tm$.
   \qedhere
   \end{itemize}
\end{proof}

In the proof of the size invariant we shall use repeatedly the following trivial properties of size:
\begin{lemma}[Properties of $\size\cdot$]
   \label{l:easy-size-ps}
   For all environments $\env,\envtwo$:
   \begin{enumerate}
      \item \label{p:easy-size-ps-a}
         $\size{\env\envtwo} = \size\env + \size\envtwo$
      \item \label{p:easy-size-ps-d}
         $\size{\env\isub\var\vartwo} = \size\env$
      \item \label{p:easy-size-ps-e}
         if $\env \AlphaEq \envtwo$, then $ \size\env = \size\envtwo $
      \item \label{p:easy-size-ps-g}
         $\size{\esub\var \mol \envtwo}= \size{\esub\varstar \mol \envtwo}$
   \end{enumerate}
\end{lemma}
\begin{proof}~
   \begin{enumerate}
      \item
         By induction on the structure of $\envtwo$:
         \begin{itemize}
            \item Case $\emptyenv$: $\size{\env\emptyenv} = \size{\env} = \size{\env} + 0 =\size{\env} + \size{\emptyenv}$.
            \item Case $\envthree\esub\var\mol$: $\size{\env\envtwo} = \size{\env\envthree\esub\var\mol} = 1 + \size{\env\envthree} + \size\mol =_\ih
             1 + \size\env + \size\envthree + \size\mol
             = \size\env + \size{\envthree\esub\var\mol} = \size\env + \size\envtwo$.
         \end{itemize}
      \item Obvious because the definition of $\size\cdot$ does not care about names.
      \item Obvious because the definition of $\size\cdot$ does not care about names.
      \item Obvious because the definition of $\size\cdot$ does not care about names.
   \qedhere
   \end{enumerate}
\end{proof}

We turn to the proof of the size invariant:
   \begin{theorem}[Size invariant]\label{thm:size-invariant}
      Let $\state=\env \gensep\kctx$ be a state reachable starting from an initial state $\state_0 = \env_0 \rlsep \ctxhole$.
      Then $\size\val \leq \size{\env_0}$ holds for every abstraction $\val$ either in $\env$
      (when the state is $\env \mathrel\rlsep \kctx$) or in $\wstenv\kctx$.
   \end{theorem}
   \begin{proof}
By induction on the execution $\exec:\state_0 \tosmach^* \state$.

If $\exec$ is empty then $\state = \state_0$; in this case, the size invariant is trivial by the definition of size.

If $\exec$ is non-empty we look at the last transition $\statetwo  \tosmach \state$, knowing by \ih that the size invariant holds 
for $\statetwo$:
\begin{itemize}

 \item 
$\env \esub\var{\vartwo\,\varthree} \rlsep \kctx
\tomachbv 
\env (\esub\var \mol \envtwo \isub\varfour\varthree) \rlsep    \kctx$
with $\rename{(\wstenv\kctx(\vartwo))} = \la\varfour(\esub{\varstar}\mol\envtwo) $ and $\wstenv\kctx(\varthree)= \val$ for some $\val$. Every value $\val$ in $\wstenv\kctx$ is such that $\size\val \leq \size{\env_0}$ because
    the property holds by \ih.
    Moreover, values in $\env (\esub\var \mol \envtwo \isub\varfour\varthree)$ are either
    renamings of $\alpha$-renamings of values in $\wstenv\kctx$ or
    values in $\env$ and thus in
    $\env \esub\var{\vartwo\,\varthree}$. Therefore the property holds by \ih, \reflemmap{easy-size-ps}{d} and \reflemmap{easy-size-ps}{e}.

\item 
$\env \esub\var{\vartwo\,\varthree}  \rlsep   \kctx
 \tomachbi 
\env \esub\var \mol \envtwo  \rlsep  \kctxp{ \ctxhole\esub\varfour\varthree   }$
with $\rename{(\wstenv\kctx(\vartwo))} = \la\varfour(\esub{\varstar}\mol\envtwo) $ and $\wstenv\kctx(\varthree)= \itm$ for some inert term $\itm$. Every value $\val$ in $\wstenv{\kctxp{\ctxhole\esub\varfour\varthree}} =_\reflemmaeqp{unsenv-prefix}{one}\esub\varfour\varthree\wstenv\kctx$ is such that $\size\val \leq \size{\env_0}$ because
    the property holds by \ih for $\wstenv\kctx$.
    Moreover, values in $\env \esub\var \mol \envtwo$ are either
    $\alpha$-renamings of values in $\wstenv\kctx$ or
    values in $\env$ and thus in
    $\env \esub\var{\vartwo\,\varthree}$. Therefore the property holds by \ih and \reflemmap{easy-size-ps}{e}.

\item ${\env \esub\var\vartwo}  \rlsep  \kctx
\tomachsub
{\env \isub\var\vartwo}  \rlsep  \kctx $
with $\var\neq\varstar$. Every abstraction in $\wstenv\kctx$ or in $\env\isub\var\vartwo$ is an abstraction in $\wstenv\kctx$ or a
    renaming of an abstraction in $\env\esub\var\vartwo$. The property follows from \ih by \reflemmap{easy-size-ps}{d}.

\item ${\env \esub\var\mol}  \rlsep  \kctx
\tomachcone 
\env  \rlsep  \kctxp{ {\ctxhole\esub\var\mol}  }
$
when $\mol$ is an abstraction or when $\mol$ is $\vartwo$ or $\vartwo\varthree$ but $\vartwo$ is not defined in $\wstenv\kctx$ or $\wstenv\kctx(\vartwo)$ is not a value. Every abstraction in $\env$ or in $\wstenv{\kctxp{{\ctxhole\esub\var\mol}}} =_\reflemmaeqp{unsenv-prefix}{one}\esub\var\mol\wstenv\kctx$ is also in $\env \esub\var\mol$ or in $\wstenv\kctx$. Therefore the property holds by \ih.

\item $\emptyenv  \rlsep  \kctx
\tomachctwo 
\emptyenv  \lrsep  \kctx$. Immediate because the context does not change and the environment is empty both before and after the transition.

\item $\env  \lrsep  \kctxp{ {\ctxhole\esub\var\mol  }}
\tomachcthree 
{\env \esub\var\mol}  \lrsep  \kctx
$
where $\mol$ is a variable or an application. Every abstraction in $\wstenv\kctx$ is also in $\wstenv{\kctxp{{\ctxhole\esub\var\mol}}} =_\reflemmaeqp{unsenv-prefix}{one}\esub\var\mol\wstenv\kctx$. Therefore the property holds by \ih.

\item $\env  \lrsep  \kctxp{ {\ctxhole\esub\var\val  }}
\tomachgc 
\env  \lrsep  \kctx$
with $\var \notin \fv\env$. Every abstraction in $\wstenv\kctx$ is also in $\wstenv{\kctxp{{\ctxhole\esub\var\val}}} =_\reflemmaeqp{unsenv-prefix}{one}\esub\var\val\wstenv\kctx$. Therefore the property holds by \ih.

\item $\env  \lrsep  \kctxp{\envtwo {\esub\var{\la\vartwo\ctxhole }  }}
\tomachcfour 
{\envtwo \esub\var{\la\vartwo\env}}  \lrsep   \kctx$. 
Trivial because $\wstenv{\kctxp{\envtwo {\esub\var{\la\vartwo\ctxhole }  }}} =_\reflemmaeqp{unsenv-prefix}{two} \wstenv\kctx$.

\item $\env  \lrsep  \kctxp{ {\ctxhole\esub\var{\la\vartwo\envtwo} } }
\tomachcfive 
\envtwo  \rlsep  \kctxp{\env {\esub\var{\la\vartwo\ctxhole}}  }
$ with $\var \in \fv\env$. Every abstraction in $\envtwo$ or in $\wstenv{\kctxp{\env {\esub\var{\la\vartwo\ctxhole}}  }} =_\reflemmaeqp{unsenv-prefix}{two} \wstenv\kctx$ is also in
    $\wstenv{\kctxp{ {\ctxhole\esub\var{\la\vartwo\envtwo} } }} =_\reflemmaeqp{unsenv-prefix}{one}
    \esub\var{\la\vartwo\envtwo}\wstenv\kctx$. Therefore the property holds by \ih.
    \qedhere
\end{itemize}
\end{proof}

\subsection{Number of overhead transitions}
The aim of this sub-section is to provide a bound on the number of machine steps as a function of the
number of $\beta$-steps, \ie{} \refcorollaryboth{bound-com}. The result is obtained
as a corollary of \reflemmaboth{strong-bound-open} and \reflemmaboth{meas-during-exec}.

\bigskip

We estimate the number of overhead transitions in a modular way: first bounding those in all strong phases, then those in the open phases.

The number of transitions of strong phases is bound by the number of transitions of open phases: we provide a global analysis in the following lemma.

\begin{lemma}[Open phases bound strong phases]
   \label{lappendix:strong-bound-open}
   \NoteState{l:strong-bound-open}
   Let $\exec\colon\state_0 \tomach{} \state$ an execution of the \SCAM{}. Then:
   \[\size\exec_{\admsym_2} + \size\exec_{\admsym_3} + \size\exec_{\gc} + \size\exec_{\admsym_4} + \size\exec_{\admsym_5} \leq \size\exec_{\betain} + 4\size\exec_{\admsym_1} + 1.\]
\end{lemma}
\begin{proof} The statement is a consequence of the following inequalities:
   \begin{itemize}
   \item $\size\exec_{\admsym_2} \leq \size\exec_{\admsym_5} + 1$ because only ${\admsym_5}$ switches the phase to $\rlsep$ (plus 1 because the initial state is in the $\rlsep$ phase).
   \item $\size\exec_{\admsym_3} + \size\exec_{\gc}\leq \size\exec_{\betain} + \size\exec_{\admsym_1}$ because ${\admsym_3}$ and $\gc$ pop entries of the form $\esub\var\mol$ from the machine environment, which are pushed only by $\betain$ and $\admsym_1$.
   \item $\size\exec_{\admsym_4} \leq \size\exec_{\admsym_5}$ because ${\admsym_4}$ pops entries of the form $\esub\var{\la\vartwo\ctxhole}$ from the machine environment, which are pushed only by ${\admsym_5}$.
   \item $\size\exec_{\admsym_5} \leq \size\exec_{\admsym_1}$ because ${\admsym_5}$ pops entries of the form $\esub\var{\la\vartwo\envtwo}$ from the machine environment, which are pushed only by ${\admsym_1}$.
\end{itemize}
\end{proof}

To bound the number of overhead transitions of open phases, we introduce a new measure $\tmpMeasure{\cdot}$ over machine states:

\begin{center}
	{\small
   $\begin{array}{lllcl}
        \textsc{Environments}
  &
  &\size\epsilon \defeq 0 
  &&
  \size{\env\esub\var\mol} \defeq 1 + \size\env + \size\mol
\end{array}$

 $\begin{array}{lllll}
  \textsc{Contexts}
  &
  \tmpMeasure{\ctxhole} \defeq 0
  &
  \tmpMeasure{\env\esub\var{\la\vartwo\kctx}} \defeq \tmpMeasure{\kctx}
  &
  \tmpMeasure{\kctx\esub\var\mol} \defeq \tmpMeasure{\kctx} + \size{\mol} 
\end{array}$

	$\begin{array}{llllll}
  \textsc{States}
  &&
  \tmpMeasure{\env ~\rlsep~ \kctx} \defeq \size\env + \tmpMeasure\kctx
  &&
  \tmpMeasure{\env ~\lrsep~ \kctx} \defeq \tmpMeasure\kctx
\end{array}$
}
\end{center}

Note: the measure $\tmpMeasure{\cdot}$ completely ignores the environment $\env$ on the left of the cursor during the strong phase (case $\tmpMeasure{\env ~\lrsep~ \kctx}$). In fact, that environment has already been fully evaluated, and thus it should not contribute in any way. The environment $\env$ in $\tmpMeasure{\env ~\rlsep~ \kctx}$, instead, contributes with its size: in this way, the measure strictly decreases after an open overhead transition, since $\tomachsub$ and $\csym_1$ pop an entry from the environment on the left of the cursor.

\medskip

First, we prove a couple of trivial properties of the measure, that will be used
repeatedly in the proof of \reflemmaappendix{meas-during-exec}.
\begin{lemma}[Properties of $\tmpMeasure\cdot$]
   \label{l:easy-measure-ps}
   For all contexts $\kctx,\kctxtwo$:
   \begin{enumerate}
      \item \label{p:easy-measure-ps-c}
         $ \tmpMeasure{\kctxp{\kctxtwo}} = \tmpMeasure\kctx + \tmpMeasure\kctxtwo $
      \item \label{p:easy-measure-ps-f}
         $\tmpMeasure{\kctxp{\ctxhole\esub\var\mol}} = \tmpMeasure\kctx + \size\mol$
   \end{enumerate}
\end{lemma}
\begin{proof}~
   \begin{enumerate}
      \item By induction on the structure of $\kctx$:
       \begin{itemize}
         \item Case $\ctxhole$: trivial.
        \item Case $\kctx\esub\var\mol$:
         $\tmpMeasure{\kctxthree\esub\var\mol\ctxholep\kctxtwo} = \tmpMeasure{\kctxthreep\kctxtwo\esub\var\mol} = \tmpMeasure{\kctxthreep\kctxtwo} + \size\mol =_\ih
\tmpMeasure{\kctxthree} + \tmpMeasure{\kctxtwo} + \size\mol = \tmpMeasure{\kctxthree\esub\var\mol} + \tmpMeasure\kctxtwo$.
        \item Case $\env\esub\var{\la\vartwo\kctxthree}$:
         $\tmpMeasure{\env\esub\var{\la\vartwo\kctxthree}\ctxholep\kctxtwo} = \tmpMeasure{\env\esub\var{\la\vartwo\kctxthreep\kctxtwo}} =
          \tmpMeasure{\kctxthreep\kctxtwo} =_\ih
          \tmpMeasure\kctxthree + \tmpMeasure\kctxtwo =
          \tmpMeasure{\env\esub\var{\la\vartwo\kctxthree}} + \tmpMeasure\kctxtwo
         $.
       \end{itemize}

      \item
       $\tmpMeasure{\kctxp{\ctxhole\esub\var\mol}}
       =_{\refpointeq{easy-measure-ps-c}} \tmpMeasure\kctx + \tmpMeasure{\ctxhole\esub\var\mol}
       = \tmpMeasure\kctx + \size\mol$.
      
   \qedhere
   \end{enumerate}
\end{proof}

\begin{lemma}[Measure during execution]
   \label{lappendix:meas-during-exec}
   \NoteState{l:meas-during-exec}
   Let $\state$ be a state reachable from $\state_0$, and $\state \tomachhole a \statetwo$.
   \begin{itemize}
      \item Beta transitions increase the measure: if $a \in\set{\betaabs, \betain}$ then $\tmpMeasure\statetwo \leq \tmpMeasure\state + \tmpMeasure{\state_0}$.
      \item Open overhead decreases the measure: if $a \in\set{\rensym, \admsym_1}$ then $\tmpMeasure\statetwo < \tmpMeasure\state$.
      \item Strong phase does not increase the measure: if $a \in\set{\admsym_2, \admsym_3, \gc, \admsym_4, \admsym_5 }$ then $\tmpMeasure\statetwo \leq \tmpMeasure\state$.
   \end{itemize}
\end{lemma}
\begin{proof}
   We proceed by cases on the transition $a$:
   \begin{itemize}
      \item Case $\betaabs$:
         $ \env \esub\var{\vartwo\,\varthree} \rlsep \kctx
          \tomachbv
         \env (\esub\var \mol \envtwo \isub\varfour\varthree) \rlsep \kctx$
         where $\wstenv\kctx(\vartwo) = \val$ and $\rename\val = \la\varfour{\esub\varstar\mol\envtwo}$.
         First of all, note that by the size invariant:
          $$\begin{array}{ll}
            & \size{\esub\varstar\mol\envtwo}\\
          < & \size{\la\varfour{\esub\varstar\mol\envtwo}} \\
          = & \size{\rename\val} \\
          =_{\reflemmaeqp{easy-size-ps}{e}} & \size{\val} \\
          \leq_{\refthmeq{size-invariant}} & \tmpMeasure{\state_0}
          \end{array}$$

         Then:
         $$\begin{array}{ll}
         & \tmpMeasure{\env (\esub\var \mol \envtwo \isub\varfour\varthree) \rlsep \kctx}\\
         = & \size{\env (\esub\var \mol \envtwo \isub\varfour\varthree)} + \tmpMeasure\kctx\\
         =_{\reflemmaeqp{easy-size-ps}{a}} & \size\env + \size{\esub\var \mol \envtwo \isub\varfour\varthree} + \tmpMeasure\kctx\\
         =_{\reflemmaeqp{easy-size-ps}{d}} & \size\env + \size{\esub\var \mol \envtwo} + \tmpMeasure\kctx\\
         =_{\reflemmaeqp{easy-size-ps}{g}} & \size\env + \size{\esub\varstar \mol \envtwo} + \tmpMeasure\kctx\\
         \leq & \size\env + \tmpMeasure{\state_0} + \tmpMeasure\kctx\\
         < & \size{\env \esub\var{\vartwo\,\varthree}} + \tmpMeasure\kctx + \tmpMeasure{\state_0}\\
         = & \tmpMeasure{\env \esub\var{\vartwo\,\varthree} \rlsep \kctx} + \tmpMeasure{\state_0}
         \end{array}$$


      \item Case $\betain$:
      $ \env \esub\var{\vartwo\,\varthree} ~\rlsep~ \kctx
       \tomachbi 
      \env \esub\var \mol \envtwo ~\rlsep~ \kctxp{\ctxhole\esub\varfour\varthree} $
      where $\wstenv\kctx(\vartwo) = \val$ and $\rename\val = \la\varfour{\esub\varstar\mol\envtwo}$.

         Then:
         $$\begin{array}{ll}
         & \tmpMeasure{\env \esub\var \mol \envtwo \rlsep \kctxp{\ctxhole \esub\varfour\varthree}}\\
         = & \size{\env \esub\var \mol \envtwo} + \tmpMeasure{\kctxp{\ctxhole \esub\varfour\varthree}}\\
         =_{\reflemmaeqp{easy-measure-ps}{f}} & \size{\env \esub\var \mol \envtwo} + \tmpMeasure\kctx \\
         =_{\reflemmaeqp{easy-size-ps}{a}} & \size\env + \size{\esub\var \mol \envtwo} + \tmpMeasure\kctx \\
         =_{\reflemmaeqp{easy-size-ps}{g}} & \size\env + \size{\esub\varstar \mol \envtwo} + \tmpMeasure\kctx \\
         \leq & \size\env + \tmpMeasure{\state_0} + \tmpMeasure\kctx \\
         < & \size{\env \esub\var{\vartwo\,\varthree}} +\tmpMeasure\kctx + \tmpMeasure{\state_0}\\
         = & \tmpMeasure{\env \esub\var{\vartwo\,\varthree} \rlsep \kctx} + \tmpMeasure{\state_0}
         \end{array}$$


      \item Case $\tomachsub$:
      $ \env \esub\var\vartwo \rlsep \kctx
      \tomachsub
      \env \isub\var\vartwo \rlsep \kctx$ where $\var \neq \varstar$.
      Then $$\begin{array}{ll}
         & \size{\env \isub\var\vartwo \rlsep \kctx}\\
       = & \size{\env \isub\var\vartwo} + \tmpMeasure\kctx\\
       =_{\reflemmaeqp{easy-size-ps}{d}} & \size\env + \tmpMeasure\kctx\\
       < & \size{\env \esub\var\vartwo} + \tmpMeasure\kctx\\
       = & \tmpMeasure{\env \esub\var\vartwo \rlsep \kctx}
       \end{array}$$


      \item Case $\csym_1$:
      $ \env \esub\var\mol \rlsep \kctx
      \tomachcone
      \env \rlsep \kctxp{\ctxhole \esub\var\mol  }$.
      Then $$\begin{array}{ll}
        & \tmpMeasure{\env \rlsep \kctxp{\ctxhole \esub\var\mol}}\\
      = & \size\env + \tmpMeasure{\kctxp{\ctxhole\esub\var\mol}}\\
      =_{\reflemmaeqp{easy-measure-ps}{f}} & \size\env + \tmpMeasure\kctx + \size{\mol}\\
      < & 1 + \size{\env} + \size\mol + \tmpMeasure \kctx\\
      = & \size{\env \esub\var\mol} + \tmpMeasure\kctx \\
      = & \tmpMeasure{\env \esub\var\mol \rlsep \kctx}
      \end{array}$$


      \item Case $\csym_2$:
      $ \rlsep \kctx
       \tomachctwo 
       \lrsep \kctx $.
      Then $$\begin{array}{ll}
          & \tmpMeasure{\lrsep \kctx}\\
        = & \tmpMeasure\kctx\\
        = & \tmpMeasure{\rlsep \kctx}\\
       \end{array}$$


      \item Case $\csym_3$:
      $ \env \lrsep \kctxp{\ctxhole \esub\var\mol}
       \tomachcthree 
      \env \esub\var\mol \lrsep \kctx$.
      Then $$\begin{array}{ll}
        & \tmpMeasure{\env \esub\var\mol \lrsep \kctx}\\
      = & \tmpMeasure\kctx\\
      < & \tmpMeasure\kctx + \size\mol \\
      =_{\reflemmaeqp{easy-measure-ps}{f}} & \tmpMeasure{\kctxp{\ctxhole\esub\var\mol}}\\
      = & \tmpMeasure{\env \lrsep \kctxp{\ctxhole \esub\var\mol}}\\
      \end{array}$$


      \item Case $\gc$:
      $ \env \lrsep \kctxp{\ctxhole \esub\var\val  }
       \tomachgc 
      \env \lrsep \kctx $.
      Then
       $$\begin{array}{ll}
         & \tmpMeasure{\env \lrsep \kctx}\\
       = & \tmpMeasure\kctx\\
       < & \tmpMeasure\kctx + \size\val\\
       =_{\reflemmaeqp{easy-measure-ps}{f}} & \tmpMeasure{\kctxp{\ctxhole\esub\var\val}}\\
       = & \tmpMeasure{\env \lrsep \kctxp{\ctxhole \esub\var\val}}\\
       \end{array}$$


      \item Case $\csym_4$:
       $\env \lrsep \kctxp{\envtwo \esub\var{\la\vartwo\ctxhole }}
        \tomachcfour
        \envtwo \esub\var{\la\vartwo\env} \lrsep  \kctx$.
      Then
      $$\begin{array}{ll}
        & \tmpMeasure{\envtwo \esub\var{\la\vartwo\env} \lrsep \kctx}\\
      = & \tmpMeasure\kctx\\
      = & \tmpMeasure\kctx + \tmpMeasure{\envtwo \esub\var{\la\vartwo\ctxhole}}\\
      =_{\reflemmaeqp{easy-measure-ps}{c}} & \tmpMeasure{\kctxp{\envtwo \esub\var{\la\vartwo\ctxhole }}}\\
      = & \tmpMeasure{\env \lrsep \kctxp{\envtwo \esub\var{\la\vartwo\ctxhole }}}
      \end{array}$$


      \item Case $\csym_5$:
      $\env \lrsep \kctxp{ \ctxhole\esub\var{\la\vartwo\envtwo}}
       \tomachcfive
       \envtwo \rlsep \kctxp{\env \esub\var{\la\vartwo\ctxhole}}$.
      Then
      $$\begin{array}{ll}
        & \tmpMeasure{\envtwo \rlsep \kctxp{\env \esub\var{\la\vartwo\ctxhole}}}\\
      = & \size\envtwo + \tmpMeasure{\kctxp{\env \esub\var{\la\vartwo\ctxhole}}}\\
      =_{\reflemmaeqp{easy-measure-ps}{c}} & \size\envtwo + \tmpMeasure\kctx + \tmpMeasure{\env \esub\var{\la\vartwo\ctxhole}}\\
      = & \size\envtwo + \tmpMeasure\kctx \\
      < & \tmpMeasure\kctx + \size{\la\vartwo\envtwo}\\
      =_{\reflemmaeqp{easy-measure-ps}{f}} & \tmpMeasure{\kctxp{ \ctxhole\esub\var{\la\vartwo\envtwo}}}\\
      = & \tmpMeasure{\env \lrsep \kctxp{ \ctxhole\esub\var{\la\vartwo\envtwo}}}
      \end{array}$$
   \end{itemize}
\end{proof}

As a consequence of the two previous lemmas, we obtain the following combined bound on the number of machine transitions:

 \begin{corollary}[Bi-linear number of overhead transitions]
    \label{cappendix:bound-com}
    \NoteState{c:bound-com}
    Let $\tm$ be a \lat{} and $\exec\colon \compil\tm \tomachscam^* \state$ a  \SCAM{} execution. Then $\size\exec \in \bigo((1 + \sizebeta\exec) \cdot \size\tm)$.
 \end{corollary}
\begin{proof}
   Let $\state_0 \defeq \mytr\tm \rlsep \ctxhole$. Then $\tmpMeasure{\state_0} = \size{\mytr\tm}$, and by \reflemma{bound-measure-size} $\tmpMeasure{\state_0} \in \bigo(\size\tm)$. To prove the required statement, it suffices to show that overhead transitions are bilinear, \ie{} that:
   \[\size\exec_{\rensym} + \size\exec_{\admsym_1} +\size\exec_{\admsym_2} + \size\exec_{\admsym_3} + \size\exec_{\gc} + \size\exec_{\admsym_4} + \size\exec_{\admsym_5} \in \bigo((1 + \size\exec_\beta) \cdot \size\tm).\]
   Note that \reflemmaappendix{meas-during-exec} implies that $\tmpMeasure{\state} \leq \tmpMeasure{\state_0} + \size\exec_\beta  \cdot \size{\mytr\tm} - \size\exec_{\rensym} - \size\exec_{\admsym_1}$. Since $\tmpMeasure{\state_0} = \size{\mytr\tm}$ and $\tmpMeasure\state \geq 0$, we obtain $\size\exec_{\rensym} + \size\exec_{\admsym_1} \leq (1 + \size\exec_\beta) \cdot \size{\mytr\tm}$, \ie{} the number of $\rensym$ and $\admsym_1$ transitions is bilinear. To bound the number of the remaining overhead transitions, just use \reflemmaappendix{strong-bound-open}.
\end{proof}

\subsection{Bi-linearity of the \SCAM{}}
In this subsection, we formalize the arguments in the paper about the cost of implementing \SCAM{} execution on Random Access Machines. We proceed in the following way:
\begin{enumerate}
 \item we prove in \reflemma{bilinear-scam} that all machine steps but garbage collection run in bi-linear time,
 \item we derive in \cref{c:bound-cost-gc} that the space consumption of the machine is also bi-linear,
 \item we conclude noting that garbage collection, which runs in time proportional to the amount of space to
  be freed, is also bi-linear and therefore the \SCAM{} runs in bi-linear time (\cref{c:bilinear-scam} which proves
  \refthm{bilinear-scam}).
\end{enumerate}

\begin{lemma}[The \SCAM{} without garbage-collection is bilinear]
   \label{l:bilinear-scam}
   For any \lat{} $\tm$ and any \SCAM{} execution $\exec\colon \compil\tm \rlsep \tomachscam^* \env \lrsep\ctxhole$, the cost of implementing $\exec$ on a RAM, excluding the cost of $\tomachgc$ steps, is $\bigo((1 + \sizebeta\exec)\size\tm)$.
\end{lemma}
\begin{proof}
 Consider that:
 \begin{itemize}
  \item each $\tofs$ step costs $\bigo(\size\tm)$ because, by \refthm{size-invariant}, the actual representation that encodes the value to be copied and renamed has size $\bigo(\size\tm)$. There are $\sizebeta\exec$ such steps.
  \item each overhead step except $\tomachgc$ costs $\bigo(1)$ and there are $\bigo((1 + \sizebeta\exec) \size\tm)$ such steps by
    \refcorollary{bound-com}.
 \end{itemize}
 Adding all the costs together yields the expected bound.
\end{proof}
\begin{corollary}[Cost of $\tomachgc$ steps]\label{c:bound-cost-gc}
 A RAM in $\bigo((1 + \sizebeta\exec)\size\tm)$ cannot create more than $\bigo((1 + \sizebeta\exec)\size\tm)$ new constructors.
 Since the number of initial node is also $\bigo(\size\tm)$, garbage-collection cannot recover more than
 $\bigo((1 + \sizebeta\exec)\size\tm)$ constructors and thus the overall cost of garbage-collection is also
 $\bigo((1 + \sizebeta\exec)\size\tm)$.
\end{corollary}
\begin{corollary}[The \SCAM{} is bilinear]
   \label{c:bilinear-scam}
   \NoteState{thm:bilinear-scam}
   Let $\tm$ be a $\l$-term and  $\exec\colon \compil\tm \tomachscam^* \state$ a \SCAM{} execution. Then $\exec$ can be implemented on a RAM in $\bigo((1 + \sizebeta\exec)\size\tm)$ time.
\end{corollary}
\begin{proof}
 Immediate by \reflemma{bilinear-scam} and \refcorollary{bound-cost-gc}.
\end{proof}

\section{Implementation in OCaml}
\label{sect:ocaml}
\lstset{
 language=caml,
 mathescape=true,
 columns=[c]fixed,
 basicstyle=\small\ttfamily,
 keywordstyle=\bfseries,
 upquote=true,
 commentstyle=\color{gray},
 breaklines=true,
 showstringspaces=false,
 emph={None, Some}, emphstyle=\color{blue},
 emph={[2]content, copying, prev, rc, occurs, Body, Container, Var, Lam, App, Shared}, emphstyle={[2]\color{blue}},
 emph={[3]int, option, bool, exp_subst, bite, environment, revenvironment, body_or_container, zipper, var}, emphstyle={[3]\color{olive}},
}

We describe now an implementation of the \SCAM{} in OCaml, that can be
downloaded at \url{https://tinyurl.com/y5remxo8}.
The code is meant to demonstrate concretely that every
machine transition can be implemented in $\bigo(1)$. Moreover we tried very hard to minimize the
memory footprint, avoiding doubly-linked data structures almost everywhere.

The implementation is also useful to experiment with ideas and
variations and to trace execution on interesting terms. The code is not meant
instead to be tested via benchmarks, to compare the implementation with other ones:
it heavily employs mutation, which is costly in OCaml because of the write
barrier of the garbage collector. Moreover OCaml is garbage collected,
but our machine already takes care of garbage collecting. Therefore, for the
sake of comparing with other implementations, the code should be rewritten in
a low-level language like C.

The code is attached to the submission. It is self-contained and it implements
a parser for pure terms, the crumbling translation, and the \SCAM{}
itself. It is also possible to compile the code to JavaScript and run it in a browser. In this case the user can write down a term to be reduced in the browser and look at every intermediate machine states.

The code snippets that we show in this section have been obtained by removing
from the submitted code the pretty-printing statements used only to show how
the machine runs.

\subsection{Preliminaries: machine moves and zippers}

The \SCAM{} works on a graph of memory cells that encodes machine states.
It crawls the graph looking for the next redex, which is reduced performing
local graph modifications only. To reach the next redex, the machine moves in two different directions:
\begin{enumerate}
\item \emph{bi-directional horizontal visits of environments}: the machine alternates left to right (in the open phase) and right to left (in the strong phase) walks on environments. Note that, because of sharing, environments are not simply lists but DAG structures.
`\item \emph{bi-directional vertical visits of abstractions}: switching from an open phase to a strong one is done when entering into the body of an abstraction, itself an environment, and the opposite switch requires to exit the abstraction.
\end{enumerate}
As we recall next, the space-conscious way of visiting bidirectionally a list is using a \emph{zipper}, rather than doubly linking the list---the technique also smoothly scales up to trees. Our implementation uses the same principle, except that we have to deal with a rich graph structure (an enriched DAG), not just lists or trees.

\subsubsection{Zippers} In functional programming, imperative data structures equipped with a cursor, such that the only modifications can happen locally where the cursor is placed, are effectively replaced by zipper-like data structures \cite{DBLP:journals/jfp/Huet97}.

The zipper --- the first zipper-like data structure ever discovered --- represents a list and a cursor into it with two lists --- actually two stacks: the first one is the tail of the list, i.e. the suffix of the list that starts at the cursor; the second one is the prefix of the list already traversed, with all pointers reversed to obtain again a list. Moving the cursor one position to the right or to the left, for example, just corresponds to popping the head of one list and pushing it on top of the other. Inserting an element where the cursor is corresponds to pushing an element on top of the suffix list, etc.

The zipper is interesting also in an imperative setting because it provides in $\bigo(1)$ most operations normally implemented using bi-directional lists, but it uses the same amount of space of a normal list (plus one additional pointer).

Zipper-like data structures can be obtained for all kind of algebraic data types via formal derivation \cite{Mcbride01thederivative}. In particular, it is easy to obtain a zipper-like data structure for trees. To obtain the zipper the idea is simply: every time a pointer is traversed, it is also reversed and the entry-point into the data structure becomes a tuples of forward pointers --- to proceed in the visit --- and backward pointers --- to go back.

\begin{figure*}[!t]
\includegraphics[width=15cm]{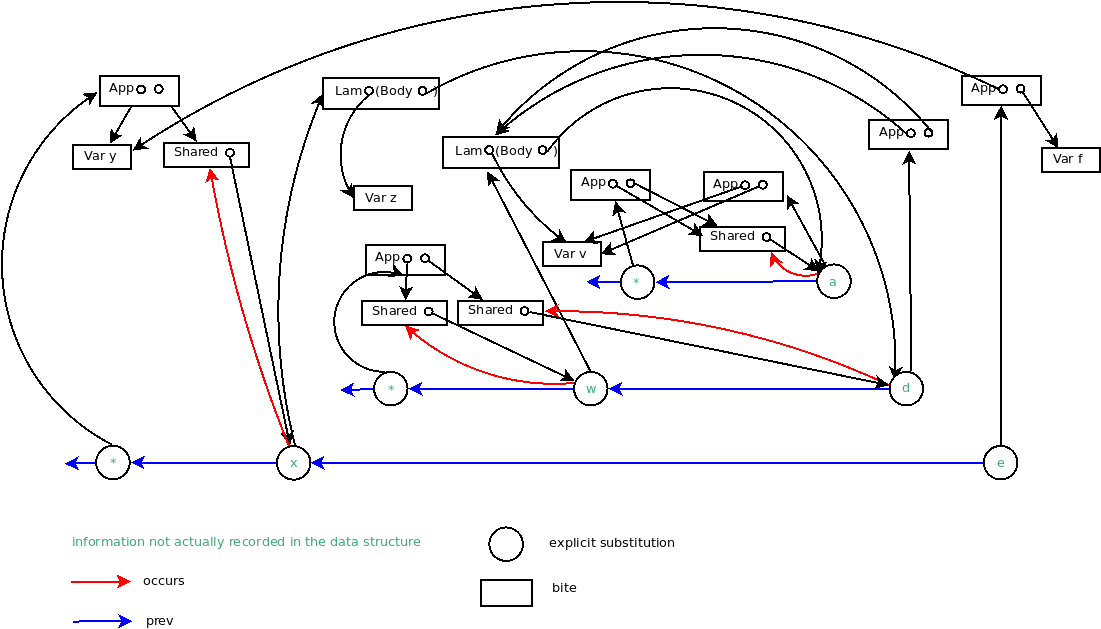}
\caption{\label{f:exampleterm}Example of a crumbled environment as a graph: $\esub{\varstar}{yx}
 \esub{x}{\la z
   \esub{\varstar}{wd}
   \esub{w}{\la v \esub{\varstar}{aa}\esub{a}{vv}}
   \esub{d}{zz}
 }
 \esub{e}{yf}
$ }
\end{figure*}

\subsection{Data structures}

We discuss now all the data structures used for the implementation, shown in
Table~\ref{t:data}.

\begin{table}[h]
\begin{lstlisting}
type var = { name : int }

type exp_subst =
 { mutable content : bite
 ; mutable copying : bool
 ; mutable prev : exp_subst option
 ; mutable rc : int
 ; mutable occurs : bite option }
and environment = exp_subst
and revenvironment = exp_subst
and body_or_container =
 | Body of environment
 | Container of zipper option
and bite =
 | Var of var
 | Lam of {v: var ; mutable b: body_or_container}
 | App of bite * bite
 | Shared of {mutable c: exp_subst}
and zipper = environment option * revenvironment option
\end{lstlisting}
\caption{Data structures\label{t:data}}
\end{table}

\subsubsection{Bites and crumbled environments}

An example of the representation in memory of a crumbled environment is given in
Figure~\ref{f:exampleterm}.

The type of bites is \verb+bite+. A bite is either an occurrence of a variable, an abstraction or an application.

Free variables and variables bound by abstractions are encoded by \verb+Var+ cells. A \verb+Var+ holds a name (an integer), used only for pretty-printing purposes. Variables bound in the context are encoded by \verb+Shared+ cells, that hold (mutable) pointers to elements of the \verb+exp_subst+ datatype, which encodes explicit substitutions. We shall come back soon to the details of \verb+exp_subst+.

Application cells \verb+App+ just hold two pointers to the two arguments.

Abstraction cells \verb+Lam+ hold a pointer to the bound variables and a (mutable) value of type \verb+body_or_container+, that encodes the two possible forms of an abstraction in our zipper-like data structure: when the body of the abstraction is not being visited, the abstraction holds a pointer to its body, which is an environment encoded as a pointer to its rightmost explicit substitution. When the body is being visited, instead, the abstraction points back to its innermost enclosing abstraction, if it exists. To identify the enclosing abstraction, however, a single pointer to the bite is not sufficient to resume the visit later; instead we identify the enclosing abstraction --- which is always the definiens of an explicit substitution --- by pointing to its explicit substitution via a zipper over the environment the explicit substitution belongs to. We use the \verb+Body+ constructor to label the pointer to the body and the \verb+Container+ constructor to label the optional label to the zipper pointing to the enclosing abstraction.

The datatype \verb+zipper+ represents standard zippers over environments, which are lists. It is defined as a pair made of an environment, of type \verb+environment+, and an environment where all pointers are reversed, of type \verb+revenvironment+. Both environments are identified by the first explicit substitution of the list.

\begin{figure*}[t]
\includegraphics[width=15cm]{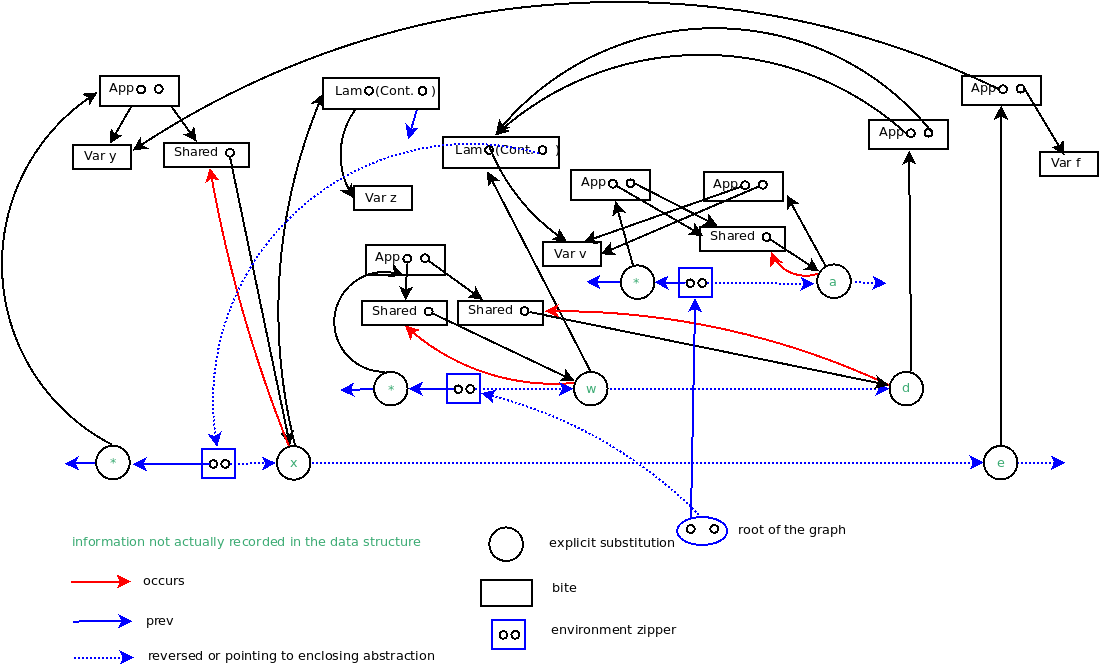}
\caption{\label{f:examplestate}Example of a machine state as a graph:
$\esub{\varstar}{aa}
 \lrsep
\esub{\varstar}{yx}
 \esub{x}{\la z
   \esub{\varstar}{wd}
   \esub{w}{\la v \ctxhole\esub{a}{vv}}
   \esub{d}{zz}
 }
 \esub{e}{yf}
$ }
\end{figure*}

Explicit substitutions are represented by records of type \verb+exp_subst+ made of several mutable fields:
\begin{itemize}
 \item \verb+content+, a bite, it is the bite $\mol$ of the ES $\esub\var\mol$. The variable $\var$, instead, is unnamed and identified with the memory location  of the cell.\footnote{We write (in green) the name $\var$ in the example graphs of this section for legibility purposes only, even if the name is lost in the implementation.}
 \item \verb+copying+ is a boolean that is used to implement the linear graph
   copying algorithm described in \cite{DBLP:conf/ppdp/AccattoliB17}, used to implement the $\alpha$-renaming in the $\beta$-transitions of the \SCAM{}. The algorithm is the same used
   in mark\&sweep garbage collectors to copy a graph of linked cells to a new
   memory region without loosing sharing. Normally, the boolean is set to
   false. It temporarily becomes \verb+true+ during the copy
   when the cell has already been copied and therefore pointers to the cell
   are to be forwarded to the copy.
 \item \verb+prev+, another optional explicit substitution, is the next substitution in the environment containing the explicit substitution. Therefore an environment is represented like in
   C lists by cells pointing to the next cell or to \verb+None+
   if the cell is the last one of the list.
 \item \verb+rc+, for \emph{reference counter}, an integer, holding the number of occurrences in the graph of
   the variable bound by the
    explicit substitution. When it becomes zero the explicit substitution can be garbage
   collected.

 \item \verb+occurs+, an optional bite: roughly, it indicates whether the variable bound by the ES has been 
introduced by the crumbling transformation. More precisely, variables generated during crumbling
   occurs exactly once; later, when rule $\tomachsub$ is fired, a machine
   invariant (see \reflemma{aux-aux-aux-most-once})
   grants the variable to be still occurring at most
   once. The \verb+occurs+ field, if set and if the explicit substitution is
   part of a pristine environment, points to the unique occurrence and it
   is used to implement rule $\tomachsub$ in $\bigo(1)$. This is the only pointer that
   violates acyclicity of the memory cells graph.
\end{itemize}

\subsubsection{Contexts and states}

We recall here the definition of machine contexts and machine states:

\begin{center}$\begin{array}{r\colspace rcl \colspace\colspace\colspace r\colspace rcl }
\textsc{Machine contexts} & \kctx & \grameq & \ctxhole\envtwo \mid \env\esub\var{\la\vartwo\kctx}\envtwo
\\
\textsc{State} & \state & \grameq & \env~\rlsep~\kctx \mid \env~\lrsep~\kctx
\end{array}$\end{center}

A machine state is of the form $\env \gensep \kctx$ where $\env$ is a crumbled environment, $\kctx$ is a machine context and together they identify a precise position in the crumbled environment $\kctxp\env$. Observe that
$\kctxp\env$ can be uniquely rewritten as $\kctxp{\env \envtwo}$ where
$\kctx$ is either $\ctxhole$ or it has the form $\kctxtwop{\env_1\esub\var{\la\vartwo\ctxhole}\env_2}$.
Therefore, to identify the position, it is sufficient to use a zipper for
environments to encode $\env\envtwo$ and a second optional zipper for environments to identify $\esub\var{\la\vartwo\ctxhole}$ inside ${\env_1\esub\var{\la\vartwo\ctxhole}\env_2}$. We already took care of identifying ${\env_1\esub\var{\la\vartwo\ctxhole}\env_2}$ inside $\kctxtwop{\env_1\esub\var{\la\vartwo\ctxhole}\env_2}$ by making the traversed abstraction ${\la\vartwo\ctxhole}$ point back, via another optional zipper, to its innermost enclosing abstraction.

Figure~\ref{f:examplestate} shows the OCaml representation of a machine state
$\env \lrsep \kctx$ such that $\kctxp\env$ is identical to the crumbled environment of Figure~\ref{f:exampleterm}. The root of the graph is the pair of zippers just discussed. Dotted blue lines represents \verb+prev+ arcs that have been reversed or arcs used to make abstractions point to the innermost enclosing abstraction.

To summarize so far, the number of memory cells used to represent a machine state is linear in the number of cells used to represent a crumbled environment: for each traversed abstraction we added a new cell holding the two pointers of a zipper.

\subsection{Machine moves and garbage collection}
\begin{table}[t]
\begin{small}
\begin{lstlisting}
let rec gc {prev=p;content=c;_} =
 Option.iter gc p;
 gc_bite c
and gc_bite =
 function
    Var _ -> ()
  | Shared {c=v} -> v.rc <- v.rc - 1
  | App(t1,t2) -> gc_bite t1; gc_bite t2
  | Lam{b=Body e;_} -> gc e
  | Lam{b=Container _;_} -> assert false
\end{lstlisting}
\end{small}
\caption{\label{l:gc}Garbage collection}
\end{table}

This section ends with the code of the
right-to-left and left-to-right evaluation phases. The code also depends on
two additional functions: \verb+gc+, shown in Table~\ref{l:gc}, that
is meant to garbage collect the term and \verb+copy_env+ that
creates in linear time an $\alpha$-renamed copy of an environment using a
modified mark\&sweep algorithm described in \cite{DBLP:conf/ppdp/AccattoliB17}.

Both evaluation functions take in input the entry point of the graph, i.e.
the pairs made of the zipper and the optional zipper previously discussed.
The call to start evaluation on a crumbled environment \verb+e+ is
\verb+eval_RL (Some e,None) None+ where $\verb+(Some e, None)+$ is the
initial zipper on the unvisited environment \verb+e+ and the second \verb+None+
signifies that we never crossed an abstraction so far.

The two \verb+eval+ functions trivially implement the reduction rules of the
machine. The only interesting observations are:
\begin{enumerate}
 \item \verb+eval_LR+ and \verb+eval_LR+ are mutual functions as expected.
 \item all recursive calls to one of the two eval functions are in tail position
   and therefore there is no inner cost in space to evaluate them.
 \item the calls to the two auxiliary functions
   \verb+gc+ and \verb+copy_env+ are not in tail position. These functions,
   moreover, have complexity linear in the size of their input both in space
   and in time.
 \item all other operations except the calls to \verb+gc+ and \verb+copy_env+
   have constant complexity as expected: they just increment/decrement numbers, 
   read or assign pointers, pattern-match and build algebraic data types, and
   test pointer equalities.
\end{enumerate}

The \verb+gc+ function written in OCaml is a bit weird: since OCaml has automatic garbage collection, \verb+gc+ just needs to traverse the data structure to decrement all the reference counters. Then the implementation of the $\tomachgc$ rule will not use in the recursive call the explicit substitution to be garbage collected, triggering the garbage collector of OCaml. An implementation in C would just free all cells during the recursion.

\begin{small}
\begin{lstlisting}
let rec eval_RL ((n,z) as zip : zipper) (k : zipper option) =
 match n with
   None ->
    (* $\tomachctwo$ *)
    eval_LR (None,z) k
 | Some n ->
   match n.content with
    | (App(_ ,App _) | App(App _,_)
      |App(_, Lam _) | App(Lam _,_)
      |App(Shared {c={content=Lam{b=Container _;_};_}},_)) -> assert false
    | App(Shared {c={content=Lam{v=y;b=Body e;_};_} as r}, (Shared {c={content=Lam _;_} as y'})) ->
       (* $\tomachbv$ *)
       r.rc <- r.rc - 1 ;
       y'.rc <- y'.rc - 1 ;
       let e' = copy_env y y' n e in
       eval_RL (Some e',z) k
    | App(Shared {c={content=Lam{v=y;b=Body e};_} as r}, t) ->
       (* $\tomachbi$ *)
       r.rc <- r.rc - 1 ;
       let y' = mk_exp_subst t in
       let e' = copy_env y y' n e in
       y'.prev <- z;
       eval_RL (Some e',Some y') k
    | Shared {c={content=(Shared _ | Var _);_} as c} when n.prev <> None ->
        (* $\tomachsub$ *)
        (match n.occurs with
            Some (Shared r as o) ->
             c.occurs <- Some o;
             r.c <- c ;
             eval_RL (n.prev,z) k
          | _ -> assert false) ;
    | (Lam _ | Var _ | App(Var _, _)
       |Shared _ | App(Shared _, _)) ->
         (* $\tomachcone$ *)
         let p = n.prev in
       eval_RL (Some e',z) k
    | App(Shared {c={content=Lam{v=y;b=Body e};_} as r}, t) ->
       (* $\tomachbi$ *)
       r.rc <- r.rc - 1 ;
       let y' = mk_exp_subst t in
       let e' = copy_env y y' n e in
       y'.prev <- z;
       eval_RL (Some e',Some y') k
    | Shared {c={content=(Shared _ | Var _);_} as c} when n.prev <> None ->
        (* $\tomachsub$ *)
        (match n.occurs with
            Some (Shared r as o) ->
             c.occurs <- Some o;
             r.c <- c ;
             eval_RL (n.prev,z) k
          | _ -> assert false) ;
    | (Lam _ | Var _ | App(Var _, _)
       |Shared _ | App(Shared _, _)) ->
         (* $\tomachcone$ *)
         let p = n.prev in
         n.prev <- z;
         eval_RL (p,Some n) k
and eval_LR ((n,z) as zip : zipper) (k : zipper option) =
   match z,k with
      None,None ->
       (* normal form reached! *)
       n
   | None,Some (n',Some ({prev=z'';content=Lam({b=Container k';_} as r);_} as z')) ->
      (* $\tomachcfour$ *)
      r.b <- Body (match n with None -> assert false | Some n -> n);
      z'.prev <- n';
      eval_LR (Some z',z'') k'
   | None,Some _ -> assert false
   | Some ({prev=p;rc;_} as zz),_ ->
      match zz.content with
       | Lam {b=Body e;_} when rc = 0 ->
           (* $\tomachgc$ *)
           gc e;
           eval_LR (n,p) k
       | Lam ({b=Body e;_} as r) ->
           (* $\tomachcfive$ *)
           r.b <- Container k;
           eval_RL (Some e,None) (Some (n,z))
       | Lam {b=Container _;_} ->
           assert false
       | (Var _ | Shared _ | App _) ->
           (* $\tomachcthree$ *)
           zz.prev <- n;
           eval_LR (Some zz,p) k
\end{lstlisting}
\end{small}

\section{Proofs of Section~\ref*{SECT:COMPLEMENTS} (Implosiveness at Work)}
\label{app:complements}

\begin{proposition}[Implosive family]
\label{propappendix:implosive-family}
   \NoteState{prop:implosive-family}
Let $\tm_n$ and $\tmtwo_n$ as in \Cref{SECT:COMPLEMENTS}.
\begin{enumerate}
	\item \emph{External strategy (exponentially many steps)}: $\tm_n (\toms\toes)^{2^n-1} \tmtwo_n$ and $\tmtwo_n$ is a strong fireball.
\item \emph{\SCAM{} Implosion (linearly many steps)}: $\exec_n: \compil{\tm_n} \tomachscam^* \state_n$ with $\unf{\state_n} = \tmtwo_n$ and $\sizebeta{\exec_n} = n$.
\end{enumerate}
\end{proposition}

\begin{proof}
\hfill
\begin{enumerate}
	\item By induction on $n$. 
	\begin{itemize}
	\item Case $n=1$: 
	$$\tm_1 = \pi I = (\la\var\la\vartwo ((\vartwo \var)\var)) I \toms \la\vartwo ((\vartwo \var)\var)\esub\var I \toes \la\vartwo((\vartwo I) I) = \tmtwo_1$$
	and $\tmtwo_1$ is evidently a strong fireball.
	\item Case $n+1$: 
	$$\tm_{n+1} = \pi (\la\varthree \tm_n) = (\la\var\la\vartwo ((\vartwo \var)\var)) (\la\varthree \tm_n) \toms (\la\vartwo ((\vartwo \var)\var)) \esub\var{\la\varthree \tm_n} \toes \la\vartwo((\vartwo(\la\varthree \tm_n)) (\la\varthree \tm_n)).$$ 
	By \ih, $\tm_n (\toms\toes)^{2^n-1} \tmtwo_n$ and the two copies of $\tm_n$ in $ \la\vartwo((\vartwo(\la\varthree \tm_n)) (\la\varthree \tm_n))$ are inside an external evaluation context. Then $ \la\vartwo((\vartwo(\la\varthree \tm_n)) (\la\varthree \tm_n)) (\toms\toes)^{2^n-1}  \la\vartwo((\vartwo(\la\varthree \tmtwo_n)) (\la\varthree \tm_n)) (\toms\toes)^{2^n-1}  \la\vartwo((\vartwo(\la\varthree \tmtwo_n)) (\la\varthree \tmtwo_n)) = \tmtwo_{n+1}$. The number of $\toms\toes$ double steps is $2^n-1 + 2^n-1 +1 = 2^{n+1} -1$. By \ih, $\tmtwo_n$ is a strong fireball and so $\tmtwo_{n+1}$ is a strong fireball.
	\end{itemize}
	
	\item We sketch the idea. The crumbling of $\tm_{n+1}$ is:
	$$\esub\varstar{\varfour\varthree'}\esub\varfour{\la\var\mytr{\la\vartwo ((\vartwo \var)\var)}}\esub{\varthree'}{\la\varthree\mytr{\tm_n}} $$
	Thus the machine does
	\[\begin{array}{llllll}
	\esub\varstar{\varfour\varthree'}\esub\varfour{\la\var\mytr{\la\vartwo ((\vartwo \var)\var)}}\esub{\varthree'}{\la\varthree\mytr{\tm_n}} \rlsep \ctxhole
	& \tomachcone
	\\
	\esub\varstar{\varfour\varthree'}\esub\varfour{\la\var\mytr{\la\vartwo ((\vartwo \var)\var)}} \rlsep \ctxhole\esub{\varthree'}{\la\varthree\mytr{\tm_n}}
	& \tomachcone
	\\
	\esub\varstar{\varfour\varthree'} \rlsep \ctxhole\esub\varfour{\la\var\mytr{\la\vartwo ((\vartwo \var)\var)}}\esub{\varthree'}{\la\varthree\mytr{\tm_n}}
	& \tomachbv
	\\
	\esub\varstar{\la{\vartwo'}\mytr{((\vartwo' \varthree')\varthree')}} \rlsep \ctxhole\esub\varfour{\la\var\mytr{\la\vartwo ((\vartwo \var)\var)}}\esub{\varthree'}{\la\varthree\mytr{\tm_n}}
		& \tomachcone
	\\
	\rlsep\ctxhole\esub\varstar{\la{\vartwo'}\mytr{((\vartwo' \varthree')\varthree')}}  \esub\varfour{\la\var\mytr{\la\vartwo ((\vartwo \var)\var)}}\esub{\varthree'}{\la\varthree\mytr{\tm_n}}
	\end{array}\]
	At this point the machine changes phase and enters $\l{\vartwo'}$. Since the body is normal and the occurrences of $\varthree'$ appear as arguments, the machine shall go through it without performing any $\beta$-transition, and thus without copying $\la\varthree\mytr{\tm_n}$. Thus the machine gets to the state
	\[\begin{array}{llllll}
		\esub\varstar{\la{\vartwo'}\mytr{((\vartwo' \varthree')\varthree')}}  \lrsep \ctxhole\esub\varfour{\la\var\mytr{\la\vartwo ((\vartwo \var)\var)}}\esub{\varthree'}{\la\varthree\mytr{\tm_n}}
	\end{array}\]
	Next, the machine garbage collects the ES on $\varfour$ and enters $\l\varthree$.
		\[\begin{array}{llllll}
		 \tomachgc &\esub\varstar{\la{\vartwo'}\mytr{((\vartwo' \varthree')\varthree')}}  \lrsep \ctxhole\esub{\varthree'}{\la\varthree\mytr{\tm_n}}
		 \\
		 		 \tomachcfive &\mytr{\tm_n}  \rlsep \esub\varstar{\la{\vartwo'}\mytr{((\vartwo' \varthree')\varthree')}}\esub{\varthree'}{\la\varthree\ctxhole}

	\end{array}\]
Note that $\tm_n$ is closed, so that its execution does not depend on the context\\ $\esub\varstar{\la{\vartwo'}\mytr{((\vartwo' \varthree')\varthree')}}\esub{\varthree'}{\la\varthree\ctxhole}$. Therefore the machine repeats the same sequence of transitions---that contains only one $\beta$-transition---for $\tm_n$. Therefore, the machine executes $\tm_{n+1}$ doing only $n+1$ $\beta$-transitions. By bilinearity, the whole execution has length $\bigo(n+1)$.
\qedhere
\end{enumerate}
\end{proof}

}{}

\end{document}